\documentclass[11pt]{article}
\usepackage[a4paper,margin=25mm]{geometry}
\usepackage[utf8]{inputenc}
\usepackage[T1]{fontenc}
\usepackage[english]{babel}
\usepackage{
    amsmath,
    amssymb,
    mathtools,
    amsthm,
    bbm,
    amsbsy,
    stackrel,
    mathrsfs
}
\usepackage{dsfont}
\usepackage{subcaption}
\usepackage{enumerate}
\usepackage{hyperref}
\usepackage{color,xcolor}
\usepackage{csquotes}
\usepackage[labelfont=bf]{caption}
\usepackage{algpseudocode,algorithmicx,algorithm}
\usepackage{multirow,tabularx,booktabs}

\newtheorem{theorem}{Theorem}
\newtheorem{proposition}[theorem]{Proposition}

\newtheorem{lemma}[theorem]{Lemma}

\theoremstyle{definition}

\newtheorem{assumption}[theorem]{Assumption}

\theoremstyle{remark}
\newtheorem{remark}[theorem]{Remark}

\numberwithin{equation}{section}
\numberwithin{theorem}{section}
\DeclareMathOperator{\Pas}{\mathds{P}\text{-as}}
\DeclareMathOperator*{\argmin}{argmin}
\DeclareMathOperator*{\Argmin}{Argmin}
\DeclareMathOperator*{\supp}{supp}

\DeclareMathOperator{\oo}{o}
\DeclareMathOperator{\OO}{O}
\DeclareMathOperator{\Cost}{Cost}

\DeclareMathOperator{\Var}{\mathbb{V}ar}
\DeclareMathOperator{\Refine}{Refine}

\newcommand{\e}{\mathrm{e}}
\newcommand{\intl}{[\![}
\newcommand{\intr}{]\!]}

\hypersetup{
    colorlinks=true,
    urlcolor=black,
    linkcolor=black,
    citecolor=black
}

\title{Adaptive Multilevel Stochastic Approximation of the Value-at-Risk}
\author{%
St{\'e}phane~Cr{\'e}pey\footnote{Universit\'e Paris Cit\'e, CNRS, Laboratoire de Probabilit\'es, Statistique et Mod\'elisation (LPSM, UMR 8001), 75013 Paris.
\texttt{stephane.crepey@lpsm.paris}}
\and
Noufel~Frikha\footnote{Universit\'e Paris 1 Panth\'eon-Sorbonne, Centre d'Economie de la Sorbonne (CES), 106 Boulevard de l’H\^opital, 75642 Paris Cedex 13. \texttt{noufel.frikha@univ-paris1.fr}}
\and
Azar~Louzi\footnote{Universit\'e Paris Cit\'e, CNRS, Laboratoire de Probabilit\'es, Statistique et Mod\'elisation (LPSM, UMR 8001), 75013 Paris. \texttt{azar.louzi@lpsm.paris}}
\and
Jonathan~Spence\footnote{ Maxwell Institute for Mathematical Sciences, School of Mathematics, University of Edinburgh, Edinburgh, UK, EH9 3FD. \texttt{jspence5@ed.ac.uk}}
}
\date{\today}

\begin{document}

\maketitle

\begin{center}
\begin{minipage}{.8\linewidth}
\small
{\bf Abstract.}
Cr\'epey, Frikha, and Louzi (2025) introduced a multilevel stochastic approximation scheme to compute the value-at-risk of a financial loss that is only simulatable by Monte Carlo.
The best complexity of the scheme is in $\OO(\varepsilon^{-\frac52})$, $\varepsilon>0$ being a prescribed accuracy, which is suboptimal compared to the canonical multilevel Monte Carlo performance.
This suboptimality stems from the discontinuity of the Heaviside function involved in the biased stochastic gradient that is recursively evaluated to derive the value-at-risk.
To mitigate this issue, this paper proposes and analyzes a multilevel stochastic approximation algorithm that adaptively selects the number of inner samples at each level, and proves that its best complexity is in $\OO(\varepsilon^{-2}|\ln{\varepsilon}|^\frac52)$. Our theoretical analysis is exemplified through numerical experiments.
\medskip

\noindent
{\bf Keywords.}
stochastic approximation~;
value-at-risk~;
nested Monte Carlo~;
multilevel Monte Carlo~;
adaptive Monte Carlo

\noindent
{\bf MSC.}
Primary
65C05,
62L20;
secondary
62G32,
91G60
%\medskip

%\noindent
%{\bf DOI.}
%XXX
%\quad
\noindent
{\bf arXiv.}
2408.06531
\quad
\noindent
{\bf HAL.}
04670735

\noindent
Accepted for publication in Annals of Applied Probability

\end{minipage}
\end{center}

\section{Introduction}

The value-at-risk (VaR) is the predominantly used risk metric in finance \cite{BISFRTB}.
It stands for the quantile, at some confidence level, of the loss of a portfolio.
The probability that such loss exceeds the VaR expresses the chances of occurrence of large losses on the portfolio.
Evaluating a portfolio's VaR is thus paramount to assessing its risk.
Note that a large class of financial portfolios can only be valued by Monte Carlo (MC), which introduces a bias when estimating the loss, hence increasing the VaR computation's complexity. Determining a portfolio's VaR is thus considered a challenging task.

Following the celebrated works on VaR convexification (Rockafellar and Uryasev \cite{RU00}, Ben-Tal and Teboulle \cite{BT07} and Bardou et al.~\cite{BFP09}), a nascent line of research has adopted a stochastic approximation (SA) viewpoint to estimate the VaR (Bardou et al.~\cite{BFP09,BFP09:2,BFP16}, Costa and Gadat \cite{CG21} and Bercu et al.~\cite{BCG21}).
Recent developments (Barrera et al.~\cite{Bar+19} and Cr\'epey et al.~\cite{CFL23,amlsa}) have focused on the nested nature of the problem.
Notably, Barrera et al.~\cite{Bar+19} propose a nested SA (NSA) method that couples an outer SA scheme with an inner nested MC scheme to compute the VaR. It is shown by Cr\'epey et al.~\cite{CFL23} that such a scheme runs in $\OO(\varepsilon^{-3})$ time, $\varepsilon>0$ being a prescribed accuracy.
Meanwhile, subsequently to the wide success of multilevel MC (MLMC) methods (Heinrich \cite{Hei98,Hei01}, Giles \cite{Gil08} and Dereich \cite{Der11}), multilevel SA (MLSA) algorithms (Frikha \cite{Fri16} and Dereich \cite{Der19}) have emerged as a means to accelerate biased SA schemes such as NSA.
While the latter injects biased nested simulations directly into an iterative SA scheme, MLSA telescopes a sequence of correlated estimate pairs of lower and lower biases, hence reducing the complexity by an order of magnitude.
Cr\'epey et al.~\cite{CFL23} leverage the multilevel acceleration proposed by Frikha \cite{Fri16} to reduce the complexity of the VaR NSA scheme of Barrera et al.~\cite{Bar+19} to $\OO(\varepsilon^{-2-\delta})$, $\frac12\leq\delta<1$ being a small number governed by the integrability degree of the loss.
They also identify a parametrization for their MLSA algorithm that achieves a complexity for the estimation of the expected shortfall (ES), a risk measure closely linked to the VaR, in $\OO(\varepsilon^{-2}|\ln{\varepsilon}|^2)$.
Cr\'epey et al.~\cite{amlsa} further obtain central limit theorems (CLTs) for VaR and ES NSA and MLSA schemes, as well as their Polyak-Ruppert variants that enjoy stronger numerical stability properties.
\\

Albeit significant, the performance gain achieved by MLSA over NSA in estimating the VaR \cite{CFL23} is to be nuanced.
According to Giles~\cite{Gil08} and Frikha~\cite{Fri16}, the canonical optimal complexity for multilevel techniques is of order $\varepsilon^{-2}$.
The suboptimality of the VaR MLSA scheme arises from its inner recursions relying on the evaluation of a discontinuous gradient akin to a Heaviside function ($x\mapsto \mathds{1}_{x>0}$). When a biased simulation of the loss falls on the opposite side of the discontinuity relative to the exact target loss, the multilevel recursion incurs an $\OO(1)$ update error. The accumulation of these errors propagates to the final estimator, leading to the suboptimality.

Motivated by valuing a digital option with Heaviside payoff in a local volatility model, Giles \cite{Gil24} surveys multiple approaches to the above discontinuity issue in a multilevel Monte Carlo (MLMC) context.
Most techniques described therein emanate from the derivative sensitivity computation literature and their common goal is to lower the variance of the paired estimators at each level of the multilevel simulation.
We retain three ideas that may apply to our nested SA setting.
Firstly, a natural idea is to smoothen the Heaviside payoff, but this only attains a complexity in $\OO(\varepsilon^{-2-\frac14})$ due to a large smoothing Lipschitz constant that is involved in.
Secondly, Malliavin calculus could be used in conjunction with a cubic spline interpolation, which however propagates regression noises to the final estimator.
Finally, levelwise adaptive refinement of inner Monte Carlo samplings seems to be a versatile technique that achieves the desired performance gain.
It advocates to dynamically refine the biased simulations in order to deal with the discontinuity issue (Giles and Haji-Ali \cite{GH19}, Haji-Ali et al.~\cite{HST22} and Giles et al.~\cite{GHS23}).
Let us provide a brief overview of this technique.

An early example of adaptive nested Monte Carlo simulation can be traced back to Broadie et al.~\cite{BDM11}, who endeavor to efficiently evaluate the probability of occurrence of large losses within the nested Monte Carlo setting of Gordy and Juneja \cite{GJ10}.
The latter had proposed an algorithm for estimating the VaR, which computes an order statistic out of multiple loss scenarios, each one being obtained by averaging out multiple inner MC samples.
The authors in \cite{BDM11} are interested in estimating the probability that such nested loss, simulated with the help of inner MC samples, exceeds the VaR. This estimation involves evaluating a Heaviside function that is centered on the VaR itself.
The adaptive strategy developed in \cite{BDM11} consists in refining simulations of risk scenarios that fall too close to the VaR to reduce their chances of ending up on the wrong side thereof, by increasing the amount of inner simulations involved.
In so doing, evaluating the Heaviside function at the nested loss becomes almost equivalent to evaluating it at the true loss. Giles and Haji-Ali \cite{GH19} and Haji-Ali et al.~\cite{HST22} revisit the adaptive refinement concept of \cite{BDM11} in the context of nested MLMC, whereby the adaptive refinement is used to increase the number of inner simulations at each level according to how close their functional falls with respect to the Heaviside discontinuity.

We note that the work of Elfverson et al.~\cite{doi:10.1137/140984294} has led to similar adaptive ideas in a partial differential equation Monte Carlo approximation setting, where the Monte Carlo estimation error is assumed to be almost surely bounded.
The terminology \emph{``adaptive MLMC''} has also been employed in other meanings, such as adaptive path simulation (Hoel et al.~\cite{HoelvonSchwerinSzepessyTempone+2014+1+41}) and importance sampling (Ben Alaya et al.~\cite{doi:10.1080/17442508.2022.2084338}).
\\

In this article, we develop an adaptive refinement strategy to efficiently estimate the VaR using the MLSA algorithm derived in \cite{CFL23}.
As in MLMC, our adaptive strategy prioritizes aligning the samples with their targets relative to the discontinuities of the update functions, which are centred at the SA iterates.
Note that the strategy that was developed for MLMC in \cite{GH19,HST22} does not directly apply to our SA setting, since refining the samples that drive an SA scheme shifts their distributions relatively to the SA iterates, who themselves evolve from one iteration to the next.
This is problematic since the shift in sample distribution affects the original SA problem formulation, rendering it possibly non-convex and possessing multiple stationary points.
To address this issue, we carefully incorporate a saturation factor that varies along the iterations in order to mitigate the undesirable effects of the change in sample distribution.

We apply our strategy to both NSA and MLSA paradigms and provide sharp $L^2(\mathbb{P})$-controls.
For a prescribed accuracy $\varepsilon>0$, their optimal complexities are shown to be in $\OO(\varepsilon^{-\frac52}|\ln{\varepsilon}|^\frac12)$ and in $\OO(\varepsilon^{-2}|\ln{\varepsilon}|^\frac52)$ respectively, resulting in an order of magnitude $\varepsilon^\frac12$ speed-up on their non-adaptive counterparts, up to a logarithmic factor.
The adaptive MLSA algorithm largely attains a complexity order that is comparable to the standard unbiased Robbins-Monro algorithm, thus demonstrating the narrowing of the performance gap between nested MLSA schemes and Robbins-Monro schemes in the context of Heaviside-type update functions.
Our numerical analyses, conducted on a toy model as well as a more concrete financial setup, strongly support our theoretical findings.
\\

The paper is structured as follows.
Section~\ref{sec:mlsa} recalls the problem and main results in \cite{CFL23} regarding the estimation of the VaR using NSA and MLSA.
Section~\ref{sec:amlsa} develops an adaptive refinement strategy to enhance the efficiency of the Monte Carlo sampling that is nested within the MLSA approach. Sections~\ref{sec:adansa} and~\ref{sec:adamlsa} exploit this strategy to reduce the complexities of NSA and MLSA and provide subsequent $L^2(\mathbb{P})$-controls and complexity rates.
Section~\ref{sec:num} presents numerical studies to demonstrate the performance improvement resulting from the adaptive strategy.

\section{Stochastic Approximation Approach}
\label{sec:mlsa}

This section recalls the setting and main results in \cite{CFL23} on the stochastic approximation of the VaR. They will serve as a baseline for the adaptive strategy that we develop in Section~\ref{sec:amlsa}.

\subsection{Unbiased Stochastic Approximation Algorithm}
\label{sec:problem}

Let $(\Omega,\mathcal{F},\mathbb{P})$ be a probability space accommodating all of the subsequent random variables.
Let $X^0\in L^1(\mathbb{P})$ be an $\mathbb{R}$-valued random loss of a portfolio, defined at some time horizon $\tau>0$.
Following Acerbi and Tasche \cite{AT02} and F\"ollmer and Schied \cite{FS10}, the VaR of $X^0$ at some confidence level $\alpha\in(0,1)$, denoted $\xi^0_\star$, is defined as
\begin{equation}
\xi^0_\star
:=\inf{\{\xi\in\mathbb{R}:\mathbb{P}(X^0\leq\xi)\geq\alpha\}}.
\label{eq:VaR}
\end{equation}
As stated in \cite{RU00,BT07,BFP09}, if the cdf $F_{X^0}$ of $X^0$ is continuous, then $\xi^0_\star$ is the left-end solution to
\begin{equation}
\label{eq:sa:opt}
\min_{\xi\in\mathbb{R}}{V_0(\xi)},
\quad\text{where}\quad
V_0(\xi):=\xi+\frac1{1-\alpha}\mathbb{E}[(X^0-\xi)^+],
\quad\xi\in\mathbb{R}.
\end{equation}
Moreover, $V_0$ is convex and continuously differentiable on $\mathbb{R}$, with
\begin{equation*}
V_0'(\xi)
=\frac{F_{X^0}(\xi)-\alpha}{1-\alpha}
=\mathbb{E}[H(\xi,X^0)],
\quad \xi\in\mathbb{R},
\end{equation*}
where
\begin{equation}
H(\xi,x)=1-\frac1{1-\alpha}\mathds1_{x \geq\xi},
\quad\xi,x\in\mathbb{R}.
\label{eq:H1}
\end{equation}
If $F_{X^0}$ is additionally increasing, then $V_0$ is strictly convex and $\xi^0_\star$ is the unique minimizer of $V_0$:
\begin{equation}
\label{eq:sa:opt:sol}
\xi^0_\star=\argmin{V_0}.
\end{equation}
Besides, if $X^0$ admits a continuous pdf $f_{X^0}$, then $V_0$ is twice continuously differentiable on $\mathbb{R}$, with
\begin{equation*}
V_0''(\xi)=\frac1{1-\alpha}f_{X^0}(\xi),
\quad \xi\in\mathbb{R}.
\end{equation*}

When i.i.d.~samples of $X^0$ are available, \cite{BFP09} proposes to estimate $\xi^0_\star$ using the unbiased SA scheme of dynamics
\begin{equation}
\xi^0_{n+1}=\xi^0_n-\gamma_{n+1}H(\xi^0_n,X^0_{n+1}),
\quad n\in\mathbb{N},
\label{eq:sa:alg:xi}
\end{equation}
where $(X^0_n)_{n\geq1}\stackrel{\text{\tiny\rm i.i.d.}}{\sim}X^0$, $\xi^0_0$ is a real-valued random initialization independent of the samples $(X^0_n)_{n\geq1}$, and $(\gamma_n)_{n\geq1}$ is a positive non-increasing sequence such that $\sum_{n=1}^\infty\gamma_n=\infty$ and $\lim_{n\to\infty}\gamma_n=0$.

We are interested in the setting where $X^0$ can be written as a conditional expectation:
\begin{equation}
\label{def:X0:cond:expect}
X^0=\mathbb{E}[\varphi(Y,Z)|Y].
\end{equation}
Here, $Y$ and $Z$ are two independent random variables taking values in $\mathbb{R}^d$ and $\mathbb{R}^q$ respectively, and $\varphi:\mathbb{R}^d\times\mathbb{R}^q\to\mathbb{R}$ is a measurable function such that $\varphi(Y,Z)\in L^1(\mathbb{P})$.
From a financial standpoint, $Y$ typically models the dynamics of the portfolio's risk factors up to the time horizon $\tau$ of the loss simulation, $Z$ their dynamics beyond $\tau$ and $\varphi(Y,Z)$ the subsequent future cash flows.

\begin{remark}
Generally, the risk factors $Y$ and $Z$ can be sampled from a model at a fixed cost $c_0>0$, and the cash flow structure $\varphi$ is known.
If $Y$ and $Z$ are simulated by discretizing a stochastic differential equation (SDE), one must take into account additional terms in our error controls, as well as extra factors in our complexity formulas.
We refer for instance to Rhee and Glynn \cite{RG12} for a multilevel method to simulate the functional $\varphi(Y,Z)$ based on an SDE.
\end{remark}

Hereafter, we are interested in the case where $X^0$ cannot be sampled exactly. One must then leverage the nested formulation of $X^0$ in \eqref{def:X0:cond:expect} to produce suitable samples for the SA approach.

\subsection{Nested Stochastic Approximation Algorithm}
\label{ssec:nested:sa}

Let $\mathbb{N}^*:=\mathbb{N}\setminus\{0\}$. A natural idea is to approximate $X^0$ by the nested Monte Carlo estimator
\begin{equation}
\label{eq:Xh}
X^h:=\frac1K\sum_{k=1}^K\varphi(Y,Z^k),
\quad\text{where}\quad
h:=\frac1K\in\mathcal{H}=\bigg\{\frac1{K'}: K'\in\mathbb{N}^*\bigg\}
\end{equation}
and $Y$ and $(Z^k)_{1\leq k\leq K}\stackrel{\text{\tiny\rm i.i.d.}}{\sim}Z$ are independent.
$h$ is termed \emph{bias parameter}, since it helps control the bias associated with the VaR estimation.
For some $h\in\mathcal{H}$, we swap $X^0$ by $X^h$ in the original problem \eqref{eq:sa:opt} and obtain the approximate problem
\begin{equation}
\label{eq:sa:nested:opt}
\min_{\xi\in\mathbb{R}}{V_h(\xi)},
\quad\text{where}\quad
V_h(\xi):=\xi+\frac1{1-\alpha}\mathbb{E}[(X^h-\xi)^+],
\quad\xi\in\mathbb{R}.
\end{equation}
Again, assuming that $\varphi(Y,Z)\in L^1(\mathbb{P})$, if the cdf $F_{X^h}$ of $X^h$ is continuous, then $V_h$ is convex and continuously differentiable on $\mathbb{R}$, with
\begin{equation*}
V_h'(\xi)
=\frac{F_{X^h}(\xi)-\alpha}{1-\alpha}
=\mathbb{E}[H(\xi,X^h)],
\quad \xi\in\mathbb{R},
\end{equation*}
recalling the definition \eqref{eq:H1}.
If $F_{X^h}$ is additionally increasing, then $V_h$ is strictly convex and
\begin{equation}
\label{eq:sa:nested:sol}
\xi^h_\star=\argmin{V_h}
\end{equation}
is well defined and constitutes a biased estimator of $\xi^0_\star$.
Finally, if $X^h$ admits a continuous pdf $f_{X^h}$, then $V_h$ is twice continuously differentiable on $\mathbb{R}$, with
\begin{equation*}
V_h''(\xi)=\frac1{1-\alpha}f_{X^h}(\xi),
\quad \xi\in\mathbb{R}.
\end{equation*}

Assuming $F_{X^h}$ continuous and increasing, $\xi^0_\star$ can be estimated with bias $h\in\mathcal{H}$ using the NSA scheme of dynamics
\begin{equation}
\xi^h_{n+1}=\xi^h_n-\gamma_{n+1}H(\xi^h_n,X^h_{n+1}),\quad n\in\mathbb{N},
\label{eq:sa:nested:alg:xi}
\end{equation}
where $(X^h_n)_{n\geq1}\stackrel{\text{\tiny\rm i.i.d.}}{\sim}X^h$, $\xi^h_0$ is a real-valued random initialization independent of the samples $(X^h_n)_{n\geq1}$, and $(\gamma_n)_{n\geq1}$ is a positive non-increasing sequence such that $\sum_{n=1}^\infty\gamma_n=\infty$ and $\lim_{n\to\infty}\gamma_n=0$.

\subsubsection{Convergence Analysis}
Let $\overline{\mathcal{H}}:=\mathcal{H}\cup\{0\}$.
The global error of NSA can be written as a sum of statistical and bias errors as such:
\begin{equation}
\label{eq:global:error:VaR:NSA}
\xi^h_n-\xi^0_\star
=(\xi^h_n-\xi^h_\star)
+(\xi^h_\star-\xi^0_\star).
\end{equation}

\begin{remark}
In the MC realm, a statistical error has zero expectation, which is not the case for the quantity $\xi^h_n-\xi^h_\star$. Our terminology here is borrowed from the SA literature (e.g.~Duflo \cite{Duf96} or Frikha \cite{Fri16}).
Nevertheless, as recalled later in Lemma~\ref{lmm:error}(\ref{lmm:error:statistical}), we do have $\xi^h_n-\xi^h_\star\to0$ in $L^2(\mathbb{P})$ as $n\uparrow\infty$, exactly as in the MC framework.
\end{remark}

\begin{assumption}[{\cite[Assumptions~3.2 \&~3.5]{CFL23}}]\label{asp:misc}\
\begin{enumerate}[(i)]
    \item\label{asp:misc:i}
For all $h\in \mathcal{H}$, $F_{X^h}$ admits the first order Taylor expansion
\begin{equation*}
F_{X^h}(\xi)-F_{X^0}(\xi)=v(\xi)h+w(\xi,h)h,
\quad\xi\in\mathbb{R},
\end{equation*}
for some functions $v,w(\cdot,h):\mathbb{R}\to\mathbb{R}$ satisfying, for all $\xi^0_\star\in\Argmin{V_0}$,
\begin{equation}\label{eq:w}
\int^\infty_{\xi^0_\star}v(\xi)\mathrm{d}\xi<\infty
,\qquad
\lim_{\mathcal{H}\ni h\downarrow0}w(\xi^0_\star,h)=\lim_{\mathcal{H}\ni h\downarrow0}\int^\infty_{\xi^0_\star}w(\xi,h)\mathrm{d}\xi=0.
\end{equation}

    \item\label{asp:misc:ii}
For all $h\in\overline{\mathcal{H}}$, the law of $X^h$ admits a continuous pdf $f_{X^h}$ with respect to the Lebesgue measure. Moreover, the pdf $(f_{X^h})_{h\in\mathcal{H}}$ converge locally uniformly to $f_{X^0}$.

    \item\label{asp:misc:iii}
For all $R>0$,
\begin{equation*}
\inf_{h\in\overline{\mathcal{H}},\xi\in B(\xi^0_\star,R)}{f_{X^h}(\xi)}>0.
\end{equation*}

    \item\label{asp:misc:iv}
The pdf $(f_{X^h})_{h\in\overline{\mathcal{H}}}$ are uniformly bounded and uniformly Lipschitz, namely,
\begin{equation*}
\sup_{h\in\overline{\mathcal{H}}}{(\|f_{X^h}\|_\infty+[f_{X^h}]_{\text{\rm Lip}})}<\infty,
\end{equation*}
where $[f_{X^h}]_\text{\rm Lip}$ denotes the Lipschitz constant of $f_{X^h}$, $h\in\overline{\mathcal{H}}$.
\end{enumerate}
\end{assumption}

\begin{remark}[{\cite[Remark~1.2]{amlsa}}]\label{rmk:nested:misc}
Giorgi et al.'s \cite[Propositions~5.1(a,b)]{Giorgi2020} provide guarantees on the postulates  of Assumptions~\ref{asp:misc}(\ref{asp:misc:i})-(\ref{asp:misc:ii}).
The second part of \eqref{eq:w} also reads $w(\xi^0_\star,h)=\int_{\xi^0_\star}^\infty w(\xi,h)\mathrm{d}\xi=\oo(1)$ as $\mathcal{H}\ni h\downarrow0$, $\xi^0_\star\in\Argmin{V_0}$.
Assumptions~\ref{asp:misc}(\ref{asp:misc:iii})-(\ref{asp:misc:iv}) are natural in view of Assumption~\ref{asp:misc}(\ref{asp:misc:ii}) and the increasing nature of $F_{X^0}$.
By Assumptions~\ref{asp:misc}(\ref{asp:misc:ii})-(\ref{asp:misc:iii}), for all $h\in\overline{\mathcal{H}}$ and all $\xi^h_\star\in\Argmin{V_h}$, $V_h''(\xi^h_\star)=(1-\alpha)^{-1}f_{X^h}(\xi^h_\star)>0$, subsequently reducing $\Argmin{V_h}$ to a singleton $\{\xi^h_\star\}$.
Assumptions~\ref{asp:misc}(\ref{asp:misc:iii}) may be easily relaxed to supposing $f_{X^0}(\xi^0_\star)>0$, but the proposed formulation facilitates the derivation of the Lyapunov properties reprised in Lemma~\ref{lmm:lyapunov}, which are fundamental to the statistical error controls of Lemma~\ref{lmm:error}(\ref{lmm:error:statistical}) below.
\end{remark}

We refer to Lemma~\ref{lmm:lyapunov} for important consequences of Assumption~\ref{asp:misc}.
For $q\in\mathbb{N}^*$, we let 
\begin{equation}\label{eq:lambda}
\bar\lambda_q:=\inf_{h\in\overline{\mathcal{H}}}{\bar\lambda_{h,q}}>0,
\end{equation}
where the $(\bar\lambda_{h,q})_{h\in\overline{\mathcal{H}},q\geq1}$ are defined in Lemma~\ref{lmm:lyapunov}(\ref{lmm:lyapunov-ii}). Note that $(\bar\lambda_q)_{q\geq1}$ are non-decreasing with respect to $q\in\mathbb{N}^*$.

\begin{lemma}[{\cite[Proposition~3.4 \& Theorem~3.7(i)]{CFL23}}]\label{lmm:error}

\begin{enumerate}[(i)]

    \item\label{lmm:error:weak}
Suppose that $\varphi(Y,Z)\in L^1(\mathbb{P})$, that Assumptions \ref{asp:misc}(\ref{asp:misc:i})--(\ref{asp:misc:iii}) hold, and that $f_{X^0}>0$.
Then, as $\mathcal{H}\ni h\downarrow0$,
\begin{equation*}
\xi^h_\star-\xi^0_\star
=-\frac{v(\xi^0_\star)}{f_{X^0}(\xi^0_\star)}h+\oo(h).
\end{equation*}
Thus, for all $h\in\mathcal{H}$,
\begin{equation}
\label{eq:oO}
|\xi^h_\star-\xi^0_\star|
\leq Ch,
\quad\text{where}\quad
C
=\bigg|\frac{v(\xi^0_\star)}{f_{X_0}(\xi^0_\star)}\bigg|
+\sup_{h\in\mathcal{H}}\bigg\{\bigg|\frac{\xi^h_\star-\xi^0_\star}h+\frac{v(\xi^0_\star)}{f_{X_0}(\xi^0_\star)}\bigg|\bigg\}.
\end{equation}

    \item\label{lmm:error:statistical}
Suppose that $\varphi(Y,Z)\in L^2(\mathbb{P})$, that Assumption~\ref{asp:misc} holds, and that
\begin{equation}
\label{integrability:condition:xi0h}
\sup_{h\in\overline{\mathcal{H}}}
\mathbb{E}\bigg[|\xi^h_0|^2
\exp\bigg(\frac{2k_\alpha}{1-\alpha}
\sup_{h'\in\mathcal{H}}\|f_{X^{h'}}\|_\infty
|\xi_0^h|\bigg)\bigg]
<\infty,
\end{equation}
where $k_\alpha=1\vee\frac\alpha{1-\alpha}$.
If $\gamma_n=\gamma_1n^{-\beta}$, $n\in\mathbb{N}^*$, with $\gamma_1>0$ and $\beta\in(0,1]$, and if $\bar\lambda_1\gamma_1>1$ when $\beta=1$, then, for all $h\in\mathcal{H}$ and all $n\in\mathbb{N}^*$,
\begin{equation*}
\mathbb{E}[(\xi^h_n-\xi^h_\star)^2]\leq C\gamma_n,
\end{equation*}
where
\begin{equation}
\label{eq:C:recall}
C=\sup_{h\in\mathcal{H}}{\kappa_{h,1}(C_{h,1,1}+C_{h,2})},
\end{equation}
and, for $h\in\mathcal{H}$ and $q\in\mathbb{N}^*$, $C_{h,1,q}=$
\begin{equation*}
\begin{cases}
\gamma_1^{-1}(\mathbb{E}[\mathcal{L}^{\mu_{h,q}}_{h,1}(\xi^h_0)]+\frac{\e+1}{\e\mu_{h,q}+1})\\
\quad\times\sup_{m\geq1}{\{m^\beta\exp{(2^{2\beta+1}(1+\e\mu_{h,q})\zeta_{h,1}^{\mu_{h,q}}\gamma_1^2\Phi_{1-2\beta}(m+1)-\frac{\lambda^{\mu_{h,q}}_{h,1}}2\gamma_1\Phi_{1-\beta}(m+1))}\}}\\
+(1+\e)\zeta^{\mu_{h,q}}_{h,1}(2^{2\beta}\gamma_1\sup_{m\geq1}{\{m^\beta\exp(-2^{-(\beta+2)}\lambda^{\mu_{h,q}}_{h,1}\gamma_1m^{1-\beta})\Phi_{1-2\beta}(m+1)\}}+\frac{2^{\beta+1}}{\lambda^{\mu_{h,q}}_{h,1}})\\
\qquad\qquad\qquad
\qquad\qquad\qquad
\qquad\qquad\qquad
\qquad\qquad\qquad
\qquad\qquad\qquad
\text{if $\beta\in(0,1)$,}\\
\gamma_1^{-1}\exp((1+\e\mu_{h,q})\zeta^{\mu_{h,q}}_{h,1}\frac{\pi^2}6\gamma_1^2+\frac{\gamma_1\lambda^{\mu_{h,q}}_{h,1}}2)
(\mathbb{E}[\mathcal{L}^{\mu_{h,q}}_{h,1}(\xi^h_0)]+\gamma_1^2(1+\e)\zeta^{\mu_{h,q}}_{h,1}2^{(\gamma_1\lambda^{\mu_{h,q}}_{h,1})\vee2})\\
\qquad\qquad\qquad
\qquad\qquad\qquad
\qquad\qquad\qquad
\qquad\qquad\qquad
\qquad\qquad\qquad
\text{if $\beta=1$,}
\end{cases}
\end{equation*}
and $C_{h,2}=$
\begin{equation*}
\begin{cases}
(\gamma_1^{-1}\mathbb{E}[\bar{\mathcal{L}}_{h,2}(\xi^h_0)]+\frac{1+2^\beta C_{h,1,2}}{1+\e\mu_{h,2}^2})\\
\quad\times\sup_{m\geq1}{\{m^\beta\exp{(2^{2\beta+1}\bar\zeta_{h,2}(1+\e\mu_{h,2}^2)\gamma_1^2\Phi_{1-2\beta}(m+1)-\frac{\bar\lambda_{h,2}}2\gamma_1\Phi_{1-\beta}(m+1))}\}}\\
+(2^\beta C_{h,1,2}+1)\\
\quad\times\bar\zeta_{h,2}(\gamma_1^22^{2\beta}\sup_{m\geq1}{\{m^\beta\exp(-2^{-(\beta+2)}\bar\lambda_{h,2}\gamma_1m^{1-\beta})\Phi_{1-2\beta}(m+1)\}}+\gamma_1\frac{2^{2\beta+1}}{\bar\lambda_{h,2}}),\\
\qquad\qquad\qquad
\qquad\qquad\qquad
\qquad\qquad\qquad
\qquad\qquad\qquad
\qquad\qquad\qquad
\text{if $\beta\in(0,1)$,}\\
\gamma_1^{-1}\exp((1+\e\mu_{h,2}^2)\bar\zeta_{h,2}\frac{\pi^2}6\gamma_1^2+\frac{\gamma_1\bar\lambda_{h,2}}2)\\
\quad\times(\mathbb{E}[\bar{\mathcal{L}}_{h,2}(\xi^h_0)]+\gamma_1^3(2C_{h,1,2}+1)\bar\zeta_{h,2}2^{(\gamma_1\bar\lambda_{h,2})\vee3}\sup_{m\geq1}{\{\frac{\Phi_{\bar\lambda_{h,2}\gamma_1-2}(m+1)}{(m+1)^{\bar\lambda_{h,2}\gamma_1-1}}\}})\\
\qquad\qquad\qquad
\qquad\qquad\qquad
\qquad\qquad\qquad
\qquad\qquad\qquad
\qquad\qquad\qquad
\text{if $\beta=1$,}
\end{cases}
\end{equation*}
and also, for $\mu\geq0$,
\begin{equation*}
\begin{aligned}
\bar\zeta_{h,q}&=\zeta^{\mu_{h,q}}_{h,q},
\quad
\zeta^\mu_{h,q}=\frac12\chi^\mu_{h,q}k_\alpha^2\big((\gamma_1\sigma^\mu_q+\sigma^\mu_{q-1})\vee\sigma^\mu_q\big),\\
\sigma^\mu_{h,q}&=2^{q-1}\e^{\mu k_\alpha^2\gamma_1}\big((1+\e\mu^qk_\alpha^{2q}\gamma_1^q)\vee\e k_\alpha^{2q}\big),
\end{aligned}
\end{equation*}
with the constants $\mu_{h,q}$, $\kappa_{h,q}$, $\lambda^\mu_{h,q}$, $\bar\lambda_{h,q}$ and $\chi^\mu_{h,q}$, and the functions $\mathcal{L}$ and $\Phi$, being defined in Lemmas~\ref{lmm:lyapunov} and \ref{lmm:gamma}.
\end{enumerate}
\end{lemma}

\begin{remark}
\begin{enumerate}[(i)]
\item
The value of the constant $C$ in Lemma~\ref{lmm:error}(\ref{lmm:error:statistical}) is retrieved from the proof of \cite[Theorem~3.7(i)]{CFL23}. Its value is based on the results of Lemmas~\ref{lmm:lyapunov} and \ref{lmm:gamma}.

\item
The proof of Lemma~\ref{lmm:error}(\ref{lmm:error:statistical}) in \cite{CFL23} can be replicated line by line by taking $h=0$, which corresponds to the unbiased case of Section~\ref{sec:problem}. We can thereby obtain the following $L^2(\mathbb{P})$-control over the statistical error of the SA algorithm \eqref{eq:sa:alg:xi}:
\begin{equation}\label{eq:unbiased:statistical:error}
\mathbb{E}[(\xi^0_n-\xi^0_\star)^2]\leq C\gamma_n,
\quad n\in\mathbb{N},
\end{equation}
where
\begin{equation*}
C=\kappa_{0,1}(C_{0,1}+C_{0,2}),
\end{equation*}
with $C_{0,1}$ and $C_{0,2}$ being defined respectively like $C_{h,1}$ and $C_{h,2}$ in Lemma~\ref{lmm:error}(\ref{lmm:error:statistical}), once $h$ is substituted with $0$, and $\kappa_{0,1}$ being defined in Lemma~\ref{lmm:lyapunov}(\ref{lmm:lyapunov-iv}).
\end{enumerate}
\end{remark}

\subsubsection{Complexity Analysis}
The controls of Lemma~\ref{lmm:error} on the bias and statistical errors in \eqref{eq:global:error:VaR:NSA} result in the following complexities.

\begin{proposition}[{\cite[Theorem~3.8]{CFL23}}]
\label{prp:cost:nsa}
Let $\varepsilon>0$ be a prescribed accuracy. Within the framework of Lemma~\ref{lmm:error}(\ref{lmm:error:statistical}), setting
\begin{equation*}
h=\frac1{\lceil C\varepsilon^{-1}\rceil}
%\sim\varepsilon
\quad\text{and}\quad
n=\lceil\bar{C}^\frac1\beta\gamma_1^\frac1\beta\varepsilon^{-\frac2\beta}\rceil,
\end{equation*}
where $C,\bar{C}>0$ are the constants described in \eqref{eq:oO} and \eqref{eq:C:recall} respectively,
yields a global $L^2(\mathbb{P})$ error for NSA of order $\varepsilon$ as $\varepsilon\downarrow0$. The corresponding computational cost verifies
\begin{equation*}
\Cost^\beta_\text{\tiny\rm NSA}\leq C'nh^{-1}\sim C''\varepsilon^{-\frac2\beta-1}
\quad\text{as}\quad
\varepsilon\downarrow0,
\end{equation*}
for some constants $C',C''>0$ independent of $\varepsilon$.
The minimal corresponding computational cost satisfies
\begin{equation*}
\Cost^1_\text{\tiny\rm NSA}\leq C''\varepsilon^{-3}
\quad\text{as}\quad
\varepsilon\downarrow0,
\end{equation*}
and is attained for $\gamma_n=\gamma_1n^{-1}$, i.e.~for $\beta=1$ under the constraint $\bar\lambda_1\gamma_1>1$.
\end{proposition}

\subsection{Multilevel Stochastic Approximation Algorithm}
\label{ssec:mlsa}

With a complexity ceiling of $\OO(\varepsilon^{-3})$ for nested VaR estimation, as achieved by NSA, a multilevel approach is proposed in \cite{CFL23} to accelerate the latter scheme.
We recall below the MLSA scheme and the associated $L^2(\mathbb{P})$-control and complexity.

Take $h_0:=\frac1K\in\mathcal{H}$, $M\geq2$ an integer and $L\in\mathbb{N}^*$ a number of levels, and let
\begin{equation}
\label{eq:hl}
h_\ell:=\frac{h_0}{M^\ell}=\frac1{KM^\ell}\in\mathcal{H},
\quad\ell\in[\![0,L]\!].
\end{equation}
To each bias parameter $h_\ell$, $\ell\in[\![0,L]\!]$, corresponds an approximate problem $\min_{\xi\in\mathbb{R}}{V_{h_\ell}(\xi)}$ to \eqref{eq:sa:opt} with unique solution $\xi^{h_\ell}_\star$. These solutions can be telescoped into
\begin{equation}
\label{eq:xi:chi:hL}
\xi^{h_L}_\star
=\xi^{h_0}_\star+\sum_{\ell=1}^L\xi^{h_\ell}_\star-\xi^{h_{\ell-1}}_\star.
\end{equation}
By denoting $\mathbf{N}:=(N_\ell)_{0\leq\ell\leq L}$ a sequence of iteration amounts, the MLSA scheme \cite{CFL23,Fri16} consists in approximating $\xi^0_\star \approx \xi^{h_L}_\star$ with
\begin{equation}\label{alg:mlsa}
\xi^\text{\tiny\rm ML}_\mathbf{N}=\xi^{h_0}_{N_0}+\sum_{\ell=1}^L\xi^{h_\ell}_{N_\ell}-\xi^{h_{\ell-1}}_{N_\ell},
\end{equation}
where each level $\ell\in[\![0,L]\!]$ is simulated independently.
More precisely, once $\xi^{h_0}_{N_0}$ is simulated using $N_0$ iterations of \eqref{eq:sa:nested:alg:xi} with bias $h_0$, at each level $\ell\in[\![1,L]\!]$, for $j\in\{(\ell-1),\ell\}$, $\xi^{h_j}_{N_\ell}$ is obtained by iterating
\begin{equation}
\xi^{h_j}_{n+1}=\xi^{h_j}_n-\gamma_{n+1}H(\xi^{h_j}_n,X^{h_j}_{n+1}),\quad0\leq n\leq N_\ell-1,
\label{eq:sa:ml:alg:xi}
\end{equation}
where $(X^{h_{\ell-1}}_n,X^{h_\ell}_n)_{1\leq n\leq N_\ell}\stackrel{\text{\tiny\rm i.i.d.}}{\sim}(X^{h_{\ell-1}},X^{h_\ell})$ and $(\xi^{h_{\ell-1}}_0,\xi^{h_\ell}_0)$ is an $\mathbb{R}^2$-valued random initialization that is independent of the samples $(X^{h_{\ell-1}}_n,X^{h_\ell}_n)_{1\leq n\leq N_\ell}$.
Crucially, at each level $\ell\in[\![1,L]\!]$, $X^{h_\ell}$ and $X^{h_{\ell-1}}$ are strongly correlated, in the sense that
\begin{equation}\label{perfect:correlation}
X^{h_{\ell-1}}=\frac1{KM^{\ell-1}}\sum_{k=1}^{KM^{\ell-1}}\varphi(Y,Z^k),
\quad
X^{h_\ell}=\frac1MX^{h_{\ell-1}}+\frac1{KM^\ell}\sum_{k=KM^{\ell-1}+1}^{KM^\ell}\varphi(Y,Z^k),
\end{equation}
where $(Z^k)_{1\leq k\leq KM^\ell}\stackrel{\text{\rm\tiny i.i.d.}}{\sim}Z$ are independent of $Y$.

At each level $\ell\in[\![1,L]\!]$, $\xi^{h_\ell}_{N_\ell}$ and $\xi^{h_{\ell-1}}_{N_\ell}$ are referred to respectively as the fine and coarse approximations.
$\xi^{h_0}_{N_0}$ is referred to as the initial (or level $0$) approximation.
We can infer from \eqref{eq:sa:ml:alg:xi} that MLSA correlates multiple NSA pairs, each pair being set with consecutive bias parameters on the geometric scale $\mathcal{H}_0:=\{h_\ell,\ell\in\mathbb{N}\}$.
The produced effect is an incremental bias correction of the level $0$ approximation with bias parameter $h_0$.
Conversely, NSA can be viewed as an MLSA instance with $0$ higher levels.

The correlation between $X^{h_\ell}$ and $X^{h_{\ell-1}}$ is primarily induced by their shared $Y$ simulation.
A quick look at the proof of \cite[Theorem~4.6]{CFL23} (recalled below in Theorem~\ref{thm:ml-variance-cv}) reveals that their overlap in simulations of $Z$ has little impact on the error controls on the MLSA algorithm.
Indeed, in MLMC, antithetic sampling \cite{GS14} which seeks to exploit the correlation between $X^{h_\ell}$ and $X^{h_{\ell-1}}$ in terms of $Z$ simulations, only achieves an improvement on standard MLMC in the order of $|\ln{\varepsilon}|$, for a prescribed accuracy $\varepsilon>0$.

\subsubsection{Convergence Analysis}
The following lemma delimits three frameworks under which we derive $L^2(\mathbb{P})$-controls and complexities for the MLSA scheme.

\begin{lemma}[{\cite[Proposition~4.2]{CFL23}}]\label{lmm:local:strong:error:indicator:func}\
\begin{enumerate}[(i)]
\item\label{lmm:local:strong:error:indicator:func-i}
Assume that the real-valued random variables $X^h$ admit pdf $f_{X^h}$ that are bounded uniformly in $h\in\overline{\mathcal{H}}$.
\begin{enumerate}[\rm a.]
\item\label{lmm:local:strong:error:indicator:func-ia}
If there exists $p_\star>1$ such that
\begin{equation}\label{assumption:finite:Lp:moment}
\mathbb{E}\big[\big|\varphi(Y,Z)-\mathbb{E}[\varphi(Y,Z)|Y]\big|^{p_\star}\big]<\infty,
%\quad\text{for some $p_\star>1$},
\end{equation}
then, for all $h,h'\in\overline{\mathcal{H}}$ such that $0\leq h\leq h'$, and all $\xi\in\mathbb{R}$,
\begin{equation*}
\mathbb{E}[|\mathds1_{\{X^h>\xi\}}-\mathds1_{\{X^{h'}>\xi\}}|]\leq C(h'-h)^\frac{p_\star}{2(p_\star+1)},
\end{equation*}
where
\begin{equation*}
C=B_{p_\star}\mathbb{E}\big[\big|\varphi(Y, Z)-\mathbb{E}[\varphi(Y, Z)|Y]\big|^{p_\star}\big]^\frac1{p_\star+1}\Big(\sup_{h\in\overline{\mathcal{H}}}\|f_{X^h}\|_\infty\Big)^\frac{p_\star}{p_\star+1},
\end{equation*}
with $B_{p_\star}$ being a positive constant that depends only on $p_\star$.

\item\label{lmm:local:strong:error:indicator:func-ib}
Assume that there exists a constant $\mathfrak{g}\geq0$ such that, for all $u\in\mathbb{R}$,
\begin{equation}
\label{assumption:conditional:gaussian:concentration}
\mathbb{E}\Big[\exp\Big(u\big(\varphi(Y,Z)-\mathbb{E}[\varphi(Y,Z)|Y]\big)\Big)\Big|Y\Big]\leq\e^{\mathfrak{g}u^2}\quad\Pas.
\end{equation}
Then, for all $h,h'\in\overline{\mathcal{H}}$ such that $0\leq h<h'$ and all $\xi\in\mathbb{R}$,
\begin{equation*}
\begin{aligned}
\mathbb{E}[|\mathds1_{\{X^h>\xi\}}&-\mathds1_{\{X^{h'}>\xi\}}|]\\
&\leq2\sqrt{\mathfrak{g}(h'-h)}\bigg(1+\sqrt2\Big(\sup_{h''\in\overline{\mathcal{H}}}\|f_{X^{h''}}\|_\infty\Big)\big|\ln{\big(\mathfrak{g}(h'-h)\big)}\big|^\frac12\bigg).
\end{aligned}
\end{equation*}
\end{enumerate}

\item\label{lmm:local:strong:error:indicator:func-ii}
For $\ell\in\mathbb{N}^*$, let
\begin{equation*}
G^\ell:= \frac{X^{h_\ell}-X^{h_{\ell-1}}}{\sqrt{h_\ell}},
\end{equation*}
and
\begin{equation*}
F_{X^{h_{\ell-1}}|G^\ell=g}: x\mapsto\mathbb{P}(X^{h_{\ell-1}}\leq x|G^\ell=g),
\quad g\in\supp(\mathbb{P}_{G^\ell}).
\end{equation*}
Consider the sequence of random variables $(K^\ell)_{\ell\geq1}$ given by $K^\ell:= K^\ell(G^\ell)$, where
\begin{equation*}
K^\ell(g):=\sup_{x\neq y}\frac{|F_{X^{h_{\ell-1}}|G^\ell=g}(x)-F_{X^{h_{\ell-1}}|G^\ell=g}(y)|}{|x-y|},
\quad g\in\supp(\mathbb{P}_{G^\ell}),
\quad\ell\in\mathbb{N}^*.
\end{equation*}
Assume that the $(K^\ell)_{\ell\geq1}$ satisfy
\begin{equation}\label{assump:unif:lipschitz:integrability:conditional:cdf}
\sup_{\ell\geq1}\mathbb{E}[|G^\ell|K^\ell]<\infty.
\end{equation}
Then, for all $\xi\in\mathbb{R}$,
\begin{equation*}
\mathbb{E}[|\mathds1_{\{X^{h_\ell}>\xi\}}-\mathds1_{\{X^{h_{\ell-1}}>\xi\}}|]
\leq Ch_\ell^\frac12,
\end{equation*}
where
\begin{equation*}
C=\sup_{\ell\geq1}\mathbb{E}[|G^\ell|K^\ell].
\end{equation*}
\end{enumerate}
\end{lemma}

\begin{remark}
\label{rmk:framework}
Lemma~\ref{lmm:local:strong:error:indicator:func}'s frameworks are ordered from weakest to strongest.
\begin{enumerate}[(i)]
\item\label{rmk:framework:p}
The constant $B_{p_\star}$ in the postulate of Lemma~\ref{lmm:local:strong:error:indicator:func}(\ref{lmm:local:strong:error:indicator:func-i})\hyperref[lmm:local:strong:error:indicator:func-ia]{a} emanates from a Burkholder-Davis-Gundy inequality.
According to Revuz and Yor~\cite[Chapter~IV, Section~4]{RY99},
\begin{equation*}
B_{p_\star}
=\begin{cases}
\frac{4-p_\star}{2-p_\star}2^{p_\star}
&\text{if $p_\star<2$,}\\
(\frac{p_\star}{p_\star-1})^{\frac{p_\star^2}2}
&\text{if $p_\star\geq2$.}
\end{cases}
\end{equation*}

\item\label{rmk:framework:sub:gaussian}
According to Fathi and Frikha \cite{fathi:frikha} and Frikha and Menozzi \cite{FM12}, if
\begin{equation*}
\mathbb{E}\Big[\exp\Big(u_0\big(\varphi(Y,Z)-\mathbb{E}[\varphi(Y,Z)|Y]\big)^2\Big)\Big]<\infty
\end{equation*}
for some $u_0>0$, then there exists $\mathfrak{g}>0$ such that, for all $u\geq0$,
\begin{equation*}
\mathbb{E}\Big[\exp\Big(u\big(\varphi(Y,Z)-\mathbb{E}[\varphi(Y,Z)|Y]\big)\Big)\Big]\leq\e^{\mathfrak{g}u^2}.
\end{equation*}
This is analogous to our Gaussian concentration framework \eqref{assumption:conditional:gaussian:concentration}.

Note that the framework \eqref{assumption:conditional:gaussian:concentration} entails the framework \eqref{assumption:finite:Lp:moment} for all $p_\star>1$, which can be shown via an exponential series expansion.
\item (\cite[Remark~4.3]{CFL23})
The framework \eqref{assump:unif:lipschitz:integrability:conditional:cdf} relaxes Gordy and Juneja's \cite[Assumption~1]{GJ10}, and is  the most advantageous computationally (c.f.~Proposition~\ref{prp:mlsa:complexity}(\ref{prp:mlsa:complexity:iii})).
It holds for instance if $(g,\ell)\mapsto F_{X^{h_{\ell-1}}|G^\ell=g}$ is uniformly Lipschitz in $g$ and $\ell$.
\end{enumerate}
\end{remark}

\begin{assumption}[{\cite[Assumption~4.4]{CFL23}}]\label{asp:fh-f0}
There exist some constants $c,\delta_0>0$ such that, for all $h\in\mathcal{H}$ and for any compact set $\mathcal{K}\subset\mathbb{R}$,
\begin{equation*}
\sup_{\xi\in\mathcal{K}}{|f_{X^h}(\xi)-f_{X^0}(\xi)|}\leq ch^{\frac14+\delta_0}.
\end{equation*}
\end{assumption}

\begin{remark}[{\cite[Remark~4.5]{CFL23}}]
Assumption~\ref{asp:fh-f0} holds within the framework of \cite[Proposition~5.1(a)]{Giorgi2020}.
\end{remark}

\begin{theorem}[{\cite[Theorem~4.6]{CFL23}}]\label{thm:ml-variance-cv}
Suppose that $\varphi(Y,Z)\in L^2(\mathbb{P})$, that Assumptions~\ref{asp:misc} and~\ref{asp:fh-f0} hold, and that
\begin{equation*}
\sup_{h\in\overline{\mathcal{H}}}
\mathbb{E}\bigg[|\xi^h_0|^2
\exp\bigg(\frac{4k_\alpha}{1-\alpha}
\sup_{h'\in\mathcal{H}}\|f_{X^{h'}}\|_\infty
|\xi_0^h|\bigg)\bigg]
<\infty,
\end{equation*}
where $k_\alpha=1\vee\frac\alpha{1-\alpha}$.
If $\gamma_n=\gamma_1n^{-\beta}$, $n\in\mathbb{N}^*$, with $\gamma_1>0$ and $\beta\in(0,1]$, and if $\bar\lambda_2\gamma_1>2$ when $\beta=1$, then,
there exists some constant $C>0$ such that, for all $L\in\mathbb{N}^*$ and all $\mathbf{N}=(N_0,\dots,N_L)\in(\mathbb{N}^*)^{L+1}$,
\begin{equation}\label{L2:norm:ML:VaR}
\mathbb{E}[(\xi^\text{\tiny\rm ML}_\mathbf{N}-\xi_\star^{h_L})^2]
\leq C\bigg(\gamma_{N_0}+\bigg(\sum_{\ell=1}^L\gamma_{N_\ell}\bigg)^2+\sum_{\ell=1}^L\gamma_{N_\ell}^\frac32+\sum_{\ell=1}^L\gamma_{N_\ell}\epsilon(h_\ell)\bigg),
\end{equation}
where
\begin{equation}
\label{eq:eps(hl)}
\epsilon(h)
:=\begin{cases}
h^\frac{p_\star}{2(1+p_\star)}
&\text{if \eqref{assumption:finite:Lp:moment} holds,}\\
h^\frac12|\ln{h}|^\frac12
&\text{if \eqref{assumption:conditional:gaussian:concentration} holds,}\\
h^\frac12
&\text{if \eqref{assump:unif:lipschitz:integrability:conditional:cdf} holds,}
\end{cases}
\quad h\in\mathcal{H}.
\end{equation}
\end{theorem}

\subsubsection{Complexity Analysis}
The global error of MLSA can be decomposed into statistical and bias errors:
\begin{equation}\label{mlsa:error:global}
\xi^\text{\tiny ML}_\mathbf{N}-\xi^0_\star=(\xi^\text{\tiny ML}_\mathbf{N}-\xi^{h_L}_\star)+(\xi^{h_L}_\star-\xi^0_\star).
\end{equation}

\begin{proposition}[{\cite[Proposition~4.7, Lemma~4.8 \& Theorem~4.9]{CFL23}}]\label{prp:mlsa:complexity}\
\begin{enumerate}[(i)]
    \item \label{prp:mlsa:complexity:levels}
Suppose that Assumption~\ref{asp:misc}(\ref{asp:misc:i}) is satisfied.
Let $\varepsilon>0$ be a prescribed accuracy and $C>0$ the constant defined in \eqref{eq:oO}.
If $h_0>C^{-1}\varepsilon$, then, setting
\begin{equation}
\label{selection:number:level:MLSA}
L=\bigg\lceil\frac{\ln{(Ch_0\varepsilon^{-1})}}{\ln{M}}\bigg\rceil\geq1
\end{equation}
achieves a bias error for MLSA of order $\varepsilon$.
    \item
The computational cost of MLSA satisfies
\begin{equation*}
\Cost_\text{\tiny\rm MLSA}^\beta
\leq\bar{C}\sum_{\ell=0}^L\frac{N_\ell}{h_\ell},
\end{equation*}
for some constant $\bar{C}>0$.

    \item\label{prp:mlsa:complexity:iii}
Let $\varepsilon>0$ be a prescribed accuracy.
Within the framework of Theorem \ref{thm:ml-variance-cv}, there exists a constant $C'>0$ such that, setting
\begin{equation*}
N_\ell=\bigg\lceil(C'\gamma_1)^\frac1\beta\varepsilon^{-\frac2\beta}\bigg(\sum_{\ell'=0}^Lh_{\ell'}^{-\frac\beta{1+\beta}}\epsilon(h_{\ell'})^\frac1{1+\beta}\bigg)^\frac1\beta h_\ell^\frac1{1+\beta}\epsilon(h_\ell)^\frac1{1+\beta}\bigg\rceil,
\quad\ell\in[\![0,L]\!],
\end{equation*}
i.e.
\begin{equation*}
    N_\ell=
\begin{cases}
   \lceil(C'\gamma_1)^\frac1\beta\varepsilon^{-\frac2\beta}h_\ell^{\frac1{1+\beta}(1+\frac{p_\star}{2(1+p_\star)})}
   (\sum_{\ell'=0}^Lh_{\ell'}^{-{\frac1{1+\beta}(\beta-\frac{p_\star}{2(1+p_\star)})}})^\frac1\beta\rceil\\
      \qquad\qquad\qquad\qquad
      \qquad\qquad\qquad\qquad
      \qquad\qquad\qquad\qquad
      \quad
      \text{if \eqref{assumption:finite:Lp:moment} holds,}\\
   \lceil(C'\gamma_1)^\frac1\beta \varepsilon^{-\frac2\beta}h_\ell^\frac3{2(1+\beta)}|\ln{h_\ell}|^\frac1{2(1+\beta)}
   (\sum_{\ell'=0}^Lh_{\ell'}^{-\frac{2\beta-1}{2(1+\beta)}}|\ln{h_{\ell'}}|^\frac1{2(1+\beta)})^\frac1\beta\rceil\\
      \qquad\qquad\qquad\qquad
      \qquad\qquad\qquad\qquad
      \qquad\qquad\qquad\qquad
      \quad
      \text{if \eqref{assumption:conditional:gaussian:concentration} holds,}\\
   \lceil(C'\gamma_1)^\frac1\beta\varepsilon^{-\frac2\beta}h_\ell^\frac3{2(1+\beta)}
   (\sum_{\ell'=0}^Lh_{\ell'}^{-\frac{2\beta-1}{2(1+\beta)}})^\frac1\beta\rceil\\
      \qquad\qquad\qquad\qquad
      \qquad\qquad\qquad\qquad
      \qquad\qquad\qquad\qquad
      \quad
      \text{if \eqref{assump:unif:lipschitz:integrability:conditional:cdf} holds,}
\end{cases}
\end{equation*}
achieves a statistical error on the estimation of $\xi^0_\star$ of order $\varepsilon$.
The optimal computational cost is attained when $\beta=1$ under the constraint $\bar\lambda_2\gamma_1>2$, in which case
\begin{equation*}
\Cost_\text{\tiny\rm MLSA}^1
\leq C''
\begin{cases}
   \varepsilon^{-3 +\frac{p_\star}{2(1+p_\star)}}
      &\text{if \eqref{assumption:finite:Lp:moment} holds,}\\
   \varepsilon^{-\frac52}|\ln{\varepsilon}|^\frac12
      &\text{if \eqref{assumption:conditional:gaussian:concentration} holds,}\\
   \varepsilon^{-\frac52}
      &\text{if \eqref{assump:unif:lipschitz:integrability:conditional:cdf} holds,}
\end{cases}
\end{equation*}
for some constant $C''>0$.
\end{enumerate}
\end{proposition}

\begin{remark}
The advantageous framework described by \eqref{assump:unif:lipschitz:integrability:conditional:cdf} results in its best computational cost in $\OO(\varepsilon^{-\frac52})$, which remains suboptimal compared to the canonical optimum of $\OO(\varepsilon^{-2})$ that is achievable by multilevel techniques \cite{Gil08,Fri16}.
\end{remark}

The next section devises a strategy to adaptively refine the inner Monte Carlo samples of NSA and MLSA.

\section{Adaptive Refinement Strategy}
\label{sec:amlsa}

The suboptimality of MLSA is linked to the discontinuity of the gradient $H$ \eqref{eq:H1} intervening in the VaR recursion \eqref{eq:sa:ml:alg:xi}.
Indeed, for $n\in\mathbb{N}^*$, if the simulated loss 
$$
X^h_n=\frac1K\sum_{k=1}^K\varphi(Y_n,Z^k_n)
$$ 
is too close to the estimate $\xi^h_{n-1}$ but falls on its opposite side with respect to its sampling target 
\begin{equation}\label{def:X0:n}
X^0_n:=\mathbb{E}[\varphi(y,Z)]_{|y=Y_n},
\end{equation}
an $\OO(1)$ error is registered due to the discontinuity of $H(\xi^h_{n-1},\cdot)$ at $\xi^h_{n-1}$, thus lowering the overall performance of the multilevel algorithm.
To address this issue, we investigate the incorporation of an adaptive refinement layer into MLSA.
\\

The following steps elucidate the intuition behind our adaptive refinement strategy.
For simplicity, we consider a nested simulation $X^{h_\ell}$, targeting an unbiased simulation $X^0$, that we want to refine adaptively given a current iterate $\xi$ at the iteration number $n$.
Assuming that $Y$ and $Z^1,\dots,Z^{KM^\ell}\stackrel{\text{\tiny\rm i.i.d.}}{\sim}Z$ were used to simulate $X^{h_\ell}$ according to \eqref{eq:Xh}, refining the latter to $X^{h_{\ell+1}}$ consists in simulating additional samples $Z^{KM^\ell+1},\dots,Z^{KM^{\ell+1}}\stackrel{\text{\tiny\rm i.i.d.}}{\sim}Z$ independently from $Y$ and $Z^1,\dots,Z^{KM^\ell}$ and setting
\begin{equation}
\label{refinement}
X^{h_{\ell+1}}=\frac1MX^{h_\ell}+\frac1{KM^{\ell+1}}\sum_{k=KM^{\ell}+1}^{KM^{\ell+1}}\varphi(Y,Z^k).
\end{equation}

\noindent
\emph{Step~1. A confidence based heuristic.}
\newline
We loosely adapt the reasonings employed in \cite[Section~2.3.1]{GH19} and \cite[Section~3]{HST22} to derive a preliminary refinement strategy. Roughly speaking, considering a refinement amount $\eta\in\mathbb{N}$, we want to ensure that $H(\xi,X^{h_{\ell+\eta}})=H(\xi,X^0)$ with high probability by aligning $X^{h_{\ell+\eta}}$ with $X^0$ on the same side of the discontinuity of $H(\xi,\cdot)$ at $\xi$.
Via a conditional CLT,
\begin{equation}
\label{eq:CLT}
X^{h_{\ell+\eta}}\approx\mathcal{N}\Big(X^0,h_{\ell+\eta}\Var\big(\varphi(Y,Z)\big|Y\big)\Big).
\end{equation}
To achieve the desired alignment, we consider a critical value $C_\mathrm{ad}$ corresponding to a confidence level $p\in(0,1)$, and choose $\eta$ minimal such that $|X^{h_{\ell+\eta}}-X^0|\leq|X^0-\xi|$ with confidence $p$, which can be expressed as
\begin{equation}\label{strategy:basic}
C_\mathrm{ad}\sqrt{h_{\ell+\eta}\Var\big(\varphi(Y,Z)\big|Y\big)}
\leq|X^0-\xi|,
\quad\text{or}\quad
|X^0-\xi|
\geq C_\mathrm{ad}h_{\ell+\eta}^\frac12,
\end{equation}
up to a modification of $C_\mathrm{ad}$ conditionally on $Y$.
However, this approach is impractical as $X^0$ is inaccessible.

We thus change our perspective and require instead that $|X^{h_{\ell+\eta}} - X^0| \leq |X^{h_{\ell+\eta}} - \xi|$ with confidence $p$, to ensure that $X^0$ and $X^{h_{\ell+\eta}}$ are on the same side relative to $\xi$. We then select $\eta$ minimal such that
\begin{equation}\label{simple:strategy}
C_\mathrm{ad}\sqrt{h_{\ell+\eta}\Var\big(\varphi(Y,Z)\big|Y\big)}
\leq|X^{h_{\ell+\eta}}-\xi|,
\quad\text{i.e.}\quad
|X^{h_{\ell+\eta}}-\xi|
\geq C_\mathrm{ad}h_{\ell+\eta}^\frac12,
\end{equation}
allowing for a modification of $C_\mathrm{ad}$ conditional on $Y$. 
In effect, we augment the number of inner simulations of $X^{h_\ell}$ by $\eta$ refinements until the threshold $C_\mathrm{ad}h_{\ell+\eta}^\frac12$ around $\xi$ is crossed.
\cite{GH19} views this process as estimating $X^0$ by $X^{h_{\ell+\eta}}$ in the original strategy \eqref{strategy:basic}, while \cite{GHS23} interprets the evaluation of the criterion \eqref{simple:strategy} as performing a Student t-test on the null hypothesis ``$X^{h_{\ell+\eta}}=\xi$''.
\\

For more flexibility, we introduce two parameters $r>1$ and $0<\theta\leq1$ that respectively control the strictness and budgeting of the refinement strategy.
\\

\noindent
\emph{Step~2. Refinement strictness.}
\newline
We redefine the refinement strategy as choosing $\eta$ minimal such that
\begin{equation}\label{eq:naive:strategy}
|X^{h_{\ell+\eta}}-\xi|
\geq C_\mathrm{ad}h_{\ell+\eta}^\frac1r.
\end{equation}
The parameter $r$ allows us to adjust the strategy's propensity to refine. Larger $r$ values impose more strictness. Setting $r=2$ retrieves the previous baseline strategy.
\\

\noindent
\emph{Step~3. Refinement budgeting.}
\newline
The strategy outlined above may be risky, as the refinement amount $\eta$ required to satisfy \eqref{eq:naive:strategy} could be excessively large, hence increasing the subsequent complexity. To address this, we impose an upper limit on $\eta$ at $\lceil \theta \ell \rceil$:
\begin{equation}\label{eq:good:strategy}
\eta(\xi)=\lceil\theta\ell\rceil\wedge\min\{k:|X^{h_{\ell+k}}-\xi|\geq C_\mathrm{ad}h_{\ell+k}^\frac1r\}.
\end{equation}
Note importantly the dependency of $\eta$ on $\xi$. For this strategy to be computationally efficient, the entailed number of inner simulations should, on average, be comparable to the number of inner simulations absent the strategy. To ensure this, we relax the constant $C_\mathrm{ad}$ to a normalization factor $c(h_\ell)$ that is calibrated such that $\mathbb{E}[h_{\ell+\eta(\xi)}^{-1}] = \OO(h_{\ell}^{-1})$.
We have
\begin{equation*}
\mathbb{E}[h_{\ell+\eta(\xi)}^{-1}]
=\sum_{k=0}^{\lceil\theta\ell\rceil}h_{\ell+k}^{-1}\mathbb{P}\big(\eta(\xi)=k\big)
\leq h_\ell^{-1}+\sum_{k=1}^{\lceil\theta\ell\rceil}h_{\ell+k}^{-1}\mathbb{P}\big(\eta(\xi)=k\big).
\end{equation*}
From \eqref{eq:good:strategy} and Assumption~\ref{asp:misc}(\ref{asp:misc:iv}),
\begin{equation*}
\mathbb{P}\big(\eta(\xi)=k\big)
\leq\mathbb{P}\big(|X^{h_{\ell+k-1}}-\xi|<c(h_\ell)h_{\ell+k-1}^\frac1r\big)
\leq2M^\frac1r\sup_{h\in\bar{H}}\|f_{X^h}\|_\infty c(h_\ell)h_{\ell+k}^\frac1r.
\end{equation*}
Extending the definition of $s \mapsto h_s$ to $\mathbb{R}$ with $h_s := \frac{h_0}{M^s}$, we obtain, for some constant $C>0$,
\begin{equation*}
\mathbb{E}[h_{\ell+\eta(\xi)}^{-1}]
\leq C\bigg(h_\ell^{-1}+c(h_\ell)\sum_{k=1}^{\lceil\theta\ell\rceil}h_{\ell+k}^\frac1r\bigg)
\leq C\big(h_\ell^{-1}+c(h_\ell)h_{(1+\theta)\ell(\frac1r-1)}\big).
\end{equation*}
To balance the terms on the right-hand side, we normalize with 
\begin{equation*}
c(h_\ell)=h_{\theta\ell(r-1)-\ell}^\frac1r,
\end{equation*}
yielding the refinement strategy
\begin{equation*}
\eta(\xi)=\lceil\theta\ell\rceil\wedge\min\{k:|X^{h_{\ell+k}}-\xi|\geq C_\mathrm{ad}h_{\theta\ell(r-1)+k}^\frac1r\}.
\end{equation*}
This approach is essentially the same as the one delineated in \cite{HST22}. The amount of inner simulations is increased until the varying threshold $C_\mathrm{ad} h_{\theta \ell (r - 1) + k}^{\frac{1}{r}}$ around $\xi$ is crossed by the sample $X^{h_{\ell+k}}$.
\\

\noindent
\emph{Step~4. Refinement saturation.}
\newline
To illustrate the strategy thus far, let's test it along a single dynamic of \eqref{eq:sa:nested:alg:xi} with bias parameter $h_\ell$:
\begin{equation*}
\xi^{h_\ell}_{n+1}=\xi^{h_\ell}_n-\gamma_{n+1}H(\xi^{h_\ell}_n,X^{h_{\ell+\eta(\xi^{h_\ell}_n)}}_{n+1}).
\end{equation*}
Unlike the MC setting \cite{HST22} where the iterate $\xi$ remains constant across the recursions, the iterates $(\xi^{h_\ell}_n)_{n\geq1}$ above change from one step to the next. This causes the samples $(X^{h_\ell}_n)_{n\geq1}$ to refine to $(X^{h_{\ell+\eta(\xi^{h_\ell}_{n-1})}}_n)_{n\geq1}$, de facto solving the root finding program $\xi:\mathbb{E}[H(\xi,X^{h_{\ell+\eta(\xi)}})]=0$, or equivalently, $\xi:\mathbb{E}[\mathds{1}_{\{X^{h_{\ell+\eta(\xi)}}<\xi\}}]=\alpha$.
Since $\eta(\xi)$ evolves with $\xi$, $\xi\mapsto\mathbb{E}[\mathds{1}_{\{X^{h_{\ell+\eta(\xi)}}<\xi\}}]$ is no longer given by a cdf, as it may not be monotone and could have several roots.

To address this issue, we saturate the refinement amount $\eta$ for large iteration numbers $n$ (corresponding to terminal SA phases) to the cap $\lceil\theta\ell\rceil$, turning it into the convex program $\min_\xi V_{h_{\ell+\lceil\theta\ell\rceil}}(\xi)$ for $n$ large. This is achieved by incorporating an increasing dependency on $n$ into the strategy's threshold.
Such dependency was calibrated from the ensuing error controls provided in Proposition~\ref{prp:error:statistical:bis} and Theorem~\ref{thm:amlsa:L2}.
We refer to Figure~\ref{fig:adaptivity} for a visual representation of the adaptive refinement procedure, to \eqref{eq:eta} for our definitive strategy and to Remark~\ref{rmk:n} for additional related comments.

\begin{figure}[H]
\includegraphics[width=\textwidth]{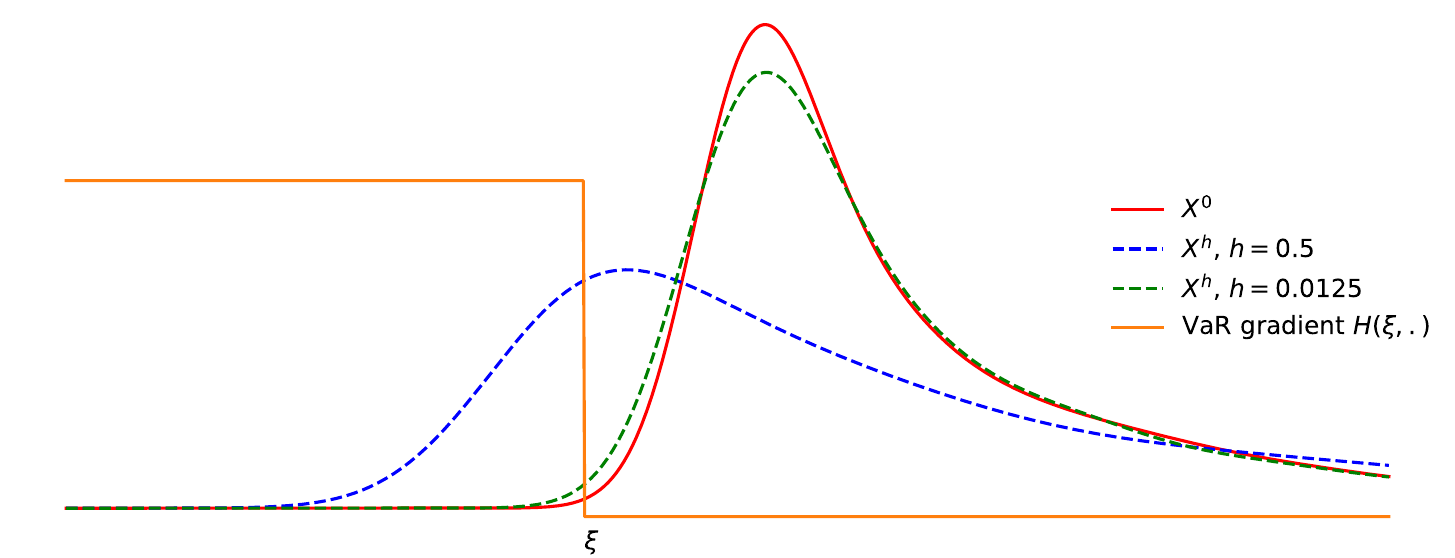}
\caption{An illustration of the adaptive refinement procedure.
The sample $X^h$ is refined by lowering its bias parameter $h$ from $\frac12$ to $\frac18$, to ensure that $H(\xi,X^h)$ is equal to $H(\xi,X^0)$ with great probability.
The plotted pdf of $X^0$ and $X^h$, $h\in\{\frac12,\frac18\}$, are based on the case study of Section~\ref{ssec:euption}.
The gradient $x\mapsto H(\xi,\cdot)$ is not represented at scale.}
\label{fig:adaptivity}
\end{figure}

The following lemma revisits the framework \eqref{assump:unif:lipschitz:integrability:conditional:cdf} to grant us a flexible basis for our adaptive refinement strategy.
Its proof is postponed to Appendix~\ref{proofs}.

\begin{lemma}\label{lmm:local:strong:error:indicator:func-ii:bis}
For $0\leq h<h'\in\overline{\mathcal{H}}$, define
\begin{equation}\label{eq:G}
G_h^{h'}:=\frac{X^{h'}-X^h}{\sqrt{h'}},
\end{equation}
and
\begin{equation*}
F_{X^h|G_h^{h'}=g}: x\mapsto\mathbb{P}(X^h\leq x|G_h^{h'}=g),
\quad g\in\supp(\mathbb{P}_{G_h^{h'}}).
\end{equation*}
Assume that the random variables $(K_h^{h'})_{0\leq h<h'\in\overline{\mathcal{H}}}$, defined by $K_h^{h'}:= K_h^{h'}(G_h^{h'})$, where
\begin{equation*}
K_h^{h'}(g):=\sup_{x\neq y}\frac{|F_{X^h|G_h^{h'}=g}(x)-F_{X^h,G_h^{h'}=g}(y)|}{|x-y|},
\quad g\in\supp(\mathbb{P}_{G_h^{h'}}),
\quad 0\leq h<h'\in\overline{\mathcal{H}},
\end{equation*}
satisfy
\begin{equation}\label{assump:unif:lipschitz:integrability:conditional:cdf:bis}
\sup_{0\leq h<h'\in\overline{\mathcal{H}}}\mathbb{E}[K_h^{h'}|G_h^{h'}|]<\infty.
\end{equation}
Then, for all $0\leq h<h'\in\overline{\mathcal{H}}$,
\begin{equation*}
\mathbb{E}[|\mathds1_{\{X^{h'}>\xi\}}-\mathds1_{\{X^h>\xi\}}|]\leq C(h')^\frac12,
\end{equation*}
where
\begin{equation*}
C=\sup_{0\leq h_1<h_2\in\overline{\mathcal{H}}}{\mathbb{E}[K_{h_1}^{h_2}|G_{h_1}^{h_2}|]}.
\end{equation*}
\end{lemma}

\begin{remark}
The framework \eqref{assump:unif:lipschitz:integrability:conditional:cdf} can be viewed as a special case of \eqref{assump:unif:lipschitz:integrability:conditional:cdf:bis} for consecutive bias pairs on the geometric scale $\mathcal{H}_0=\{h_\ell,\ell\in\mathbb{N}\}\subset\mathcal{H}$.
In standard MLSA, the fine and coarse approximations at each level are controlled by consecutive bias parameters in $\mathcal{H}_0$. If one is to refine either approximation separately, the refined fine and coarse approximations are no longer controlled by consecutive bias parameters in $\mathcal{H}_0$, hence the need for the generalized framework above.
\end{remark}

Following the roadmap of \cite{GH19,HST22,GHS23}, we define
\begin{equation}
\label{eq:eta}
\eta^\ell_n(\xi):=\lceil\theta\ell\rceil\wedge\min{\{k\in\intl0,\lceil\theta\ell\rceil\intr:|X^{h_{\ell+k}}-\xi|\geq C_\mathrm{ad}\psi^{\ell,k}_n\}},
\quad\ell,  n\in\mathbb{N}^*,
\end{equation}
with the convention $\min\varnothing=+\infty$,
where $C_\mathrm{ad}>0$, $r>1$, $0<\theta\leq1$,
\begin{equation}
\psi^{\ell,k}_n
:=
\begin{cases}
    u_n^{-\frac1{p_\star}}
    h_{\theta\ell(r-1)+k}^\frac1r
    &\text{if \eqref{assumption:finite:Lp:moment} holds,}\\
    \ln^\frac12{((\gamma_n/\gamma_1)^{-\frac12}h^{-\frac{1+\theta}2}_{\ell+k})}
    h_{\theta\ell(r-1)+k}^\frac1r
    &\text{if \eqref{assumption:conditional:gaussian:concentration} or \eqref{assump:unif:lipschitz:integrability:conditional:cdf:bis} holds,}
\end{cases}
\label{eq:psi}
\end{equation}
and
\begin{equation}\label{eq:kappa}
u_n=\gamma_1n^{-\delta},
\quad\text{with}\quad\delta\in(0,1].
\end{equation}
Note the property 
\begin{equation}
\label{eq:sum:psi}
\begin{aligned}
\sum_{k=0}^{\lceil\theta\ell\rceil}\psi^{\ell,k}_n
&\leq
\begin{cases}
(1-M^{-\frac1r})^{-1}\psi^{\ell,0}_n
&\text{if \eqref{assumption:finite:Lp:moment} holds,}\\
(1-M^{-\frac1r})^{-1}\sqrt{2+\theta}\psi^{\ell,0}_n
&\text{if \eqref{assumption:conditional:gaussian:concentration} or \eqref{assump:unif:lipschitz:integrability:conditional:cdf:bis} holds,}
\end{cases}\\
&\leq\frac{\sqrt3}{1-M^{-\frac12}}\psi^{\ell,0}_n,
\end{aligned}
\end{equation}
where we used in the last inequality the constraints $\theta\in(0,1]$ and $r\in(1,2]$, which stem from the frameworks of Lemmas~\ref{lmm:bias} and~\ref{lmm:local:strong:error:indicator:func:variance}.

We refer to $C_\mathrm{ad}$ as the confidence constant \cite{GH19,HST22}, $r$ as the refinement strictness parameter \cite{HST22,GHS23}, $\theta$ as the budgeting parameter, the $n$-dependent factor in \eqref{eq:psi} as the saturation factor, and the quantity $\frac{|X^{h_{\ell+k}}-\xi|}{\psi^{\ell,k}_n}$, that our adaptive strategy seeks to greedily make as large as possible, as the error margin \cite{BDM11}.

\begin{algorithm}[H]
\caption{Adaptive refinement strategy}
\label{alg:refine}
\begin{algorithmic}[1]
\Require {A sample $X^{h_\ell}$, an iterate $\xi$, refinement parameters $C_\mathrm{ad}>0$, $r>1$ and $0<\theta\leq1$, a level $\ell\in\mathbb{N}^*$, an iteration number $n \in\mathbb{N}^*$}
\State {$\eta\gets0$}
\While {$\eta<\lceil\theta\ell\rceil$ \textbf{and} $|X^{h_{\ell+\eta}}-\xi|<C_\mathrm{ad}\psi^{\ell,\eta}_n$}
    \State {Refine $X^{h_{\ell+\eta}}$ to $X^{h_{\ell+\eta+1}}$}
    \State {$\eta\gets\eta+1$}
\EndWhile
\State \Return {$X^{h_{\ell+\eta}}$}
\end{algorithmic}
\end{algorithm}

We refer to \cite[Algorithm~3.1]{HST22} for an analogous refinement strategy in an MLMC setting, that is independent of the iteration number $n$.

\begin{remark}\label{rmk:refine}
\begin{enumerate}[(i)]
\item
At each level $\ell$ and iteration number $n$, the adaptive strategy refines the simulation $X^{h_\ell}_n$ used in \eqref{eq:sa:ml:alg:xi} to $X^{h_{\ell+\eta}}_n$, with the aim of escaping the region $[\xi\pm C_\mathrm{ad}\psi^{\ell,\eta}_n]$ around the discontinuity of $x\mapsto H(\xi,x)$ at $x=\xi$.

\item
\label{rmk:level:0}
No refinement is applied at the level $\ell=0$, as $\eta^0_n=0$ for all $n\in\mathbb{N}^*$.

\item
\label{rmk:variation}
The thresholds $C_\mathrm{ad}\psi^{\ell,k}_n$ depend on the hyperparameters $r>1$ and $0<\theta\leq1$.
It decreases for $r$ or $\theta$ large.
It also depends on the level $\ell$ and the iteration number $n$.
It decreases for $\ell$ large and increases for $n$ large.

\item
Recalling the definition \eqref{def:X0:n}, by virtue of the almost sure convergence of $X^{h_\ell}_n$ to $X^0_n$ as $\ell\uparrow\infty$ (by the conditional law of large numbers, under suitable assumptions), the refined simulation $X^{h_{\ell+\eta^\ell_n(\xi)}}_n$ is in theory closer to the actual simulation target $X^0_n$ than $X^{h_\ell}_n$.
It is therefore expected that $H(\xi,X^{h_{\ell+\eta^\ell_n(\xi)}}_n)=H(\xi,X^0_n)$ with high probability.
See Remark~\ref{rmk:n} for a complementary comment.

\item\label{rmk:C}

In our derivation of the adaptive refinement strategy, the confidence constant $C_\mathrm{ad}$ took in fact the form of $C_p\sqrt{\Var(\varphi(Y,Z)|Y)}$, where $C_p$ is a critical value of the law $\mathcal{N}(0,1)$ at some confidence level $p$ and $\sqrt{\Var(\varphi(Y,Z)|Y)}$ is the sampling standard deviation.
For a $99\%$-confidence level, a natural choice for $C_\mathrm{ad}$ is three standard deviations, i.e.~$C_p=3$. The standard deviation can be estimated empirically.
However, the ensuing convergence and complexity analyses show that our choice of $C_\mathrm{ad}$ constant for the adaptive strategy \eqref{eq:eta} accomplishes the desired performance gain. In practice, $C_\mathrm{ad}$ should be fine-tuned to avoid over- or under-adapting.

\end{enumerate}
\end{remark}

\section{Adaptive Nested Stochastic Approximation Algorithm}
\label{sec:adansa}

Before delving into the adaptive refinement of MLSA, let us first investigate the influence of our refinement strategy on NSA, which we recall can be considered an instance of $0$ level MLSA.

Consider a bias parameter $h_\ell$, $\ell\in\mathbb{N}^*$, on the geometric scale $\mathcal{H}_0=\{h_{\ell'},\ell'\in\mathbb{N}\}$, which allows resorting to single-level refinements \eqref{refinement}.
The adaptively refined nested SA algorithm for estimating the VaR can be written as
\begin{equation}
\label{adNSA}
\widetilde\xi^{h_\ell}_{n+1} = \widetilde\xi^{h_\ell}_{n} - \gamma_{n+1} H(\widetilde\xi_n^{h_\ell}, \widetilde{X}^{h_{\ell}}_{n+1}),
\quad n\in\mathbb{N},
\end{equation}
where
\begin{equation}\label{eq:eta:l:n}
\widetilde{X}^{h_{\ell}}_{n+1} := X^{h_{\ell+\widetilde\eta^\ell_{n+1}}}_{n+1},
\quad
\widetilde\eta^\ell_{n+1}:=\eta^\ell_{n+1}(\widetilde\xi^{h_\ell}_n),
\quad n\in\mathbb{N},
\end{equation}
and $\widetilde\xi^{h_\ell}_0$ is a random real-valued initialization that is independent of the samples $(\widetilde{X}^{h_{\ell}}_n)_{n\geq1}$.

\begin{algorithm}[H]
\caption{Adaptive Nested SA for estimating the VaR}
\label{alg:nested-sa}
\begin{algorithmic}[1]
\Require {$K,M,\ell,N\in\mathbb{N}^*$ with $M\geq2$, a positive non-increasing sequence
$(\gamma_n)_{n\geq1}$ such that $\sum_{n=1}^\infty\gamma_n=\infty$ and $\lim_{n\to\infty}\gamma_n=0$, refinement parameters $C_\mathrm{ad}>0$, $r>1$ and $0<\theta\leq1$}
\State {Sample $\widetilde\xi^{h_\ell}_0$ randomly}
\For {$n=0\mathrel{.\,.}N-1$}
   \State {Simulate $Y_{n+1}\sim Y$ and $Z^1_{n+1},\dots,Z^{KM^\ell}_{n+1}\stackrel[]{\text{\rm\tiny i.i.d.}}{\sim}Z$ independently of $Y_{n+1}$}
   \State {$X^{h_\ell}_{n+1}\gets\frac1{KM^\ell}\sum_{k=1}^{KM^\ell}\varphi(Y_{n+1},Z^k_{n+1})$}
   \State{$\widetilde{X}^{h_\ell}_{n+1}\gets\Refine_{C_\mathrm{ad},r,\theta,\ell,n}(X^{h_\ell}_{n+1},\widetilde\xi^{h_\ell}_n)$}
   \State {$\widetilde\xi^{h_\ell}_{n+1}\gets\widetilde\xi^{h_\ell}_n-\gamma_{n+1}H(\widetilde\xi^{h_\ell}_n,\widetilde{X}^{h_\ell}_{n+1})$}
\EndFor
\State \Return {$\widetilde\xi^{h_\ell}_N$}
\end{algorithmic}
\end{algorithm}

\begin{remark}
\label{rmk:n}
Unlike in \cite{GH19,HST22,GHS23}, the adaptive refinement \eqref{eq:eta} depends on the iteration number $n$.
As previously discussed, using an adaptive refinement independent of $n$ can be viewed as seeking a root of $\xi\mapsto\mathbb{E}[H(\xi,X^{h_{\ell+\eta^\ell(\xi)}})]$, which is not guaranteed to retrieve a problem like \eqref{eq:sa:nested:opt}, as it could have multiple roots. The introduced dependency on $n$ saturates the refinement amount to $\lceil\theta\ell\rceil$ for large $n$, thus aligning $X^{h_{\ell+\eta^\ell_n(\xi)}}_n$ with $X^{h_{\ell+\lceil\theta\ell\rceil}}_n$ with high probability and practically solving the strictly convex nested SA problem $\min_\xi V_{h_{\ell+\lceil\theta\ell\rceil}}(\xi)$ with bias parameter $h_{\ell+\lceil\theta\ell\rceil}$.

This behavior can be verified by calculation. Assuming that $\varphi(Y,Z)\in L^1(\mathbb{P})$, by Markov's inequality,
\begin{align}
\mathbb{P}\big(\eta^\ell_n(\xi)<\lceil\theta\ell\rceil\big)
&=\mathbb{P}(\exists k\in[\![0,\lceil\theta\ell\rceil-1]\!],
|X^{h_{\ell+k}}-\xi|\geq C_\mathrm{ad}\psi^{\ell,k}_n)\notag\\
&\leq\sum_{k=0}^{\lceil\theta\ell\rceil-1}\mathbb{P}(|X^{h_{\ell+k}}-\xi|\geq C_\mathrm{ad}\psi^{\ell,k}_n)\notag\\
&\leq\frac1{C_\mathrm{ad}\psi^{\ell,k}_n}\sum_{k=0}^{\lceil\theta\ell\rceil-1}\mathbb{E}[|X^{h_{\ell+k}}-\xi|]
\to0\quad\text{as}\quad n\uparrow\infty.
\label{eq:P(eta<theta*l)}
\end{align}
Observe that
\begin{equation*}
\mathbb{P}\big(\exists n_0\in\mathbb{N},\forall n\geq n_0,\eta^\ell_n(\xi)=\lceil\theta\ell\rceil\big)
=1-\mathbb{P}\big(\limsup_{n\uparrow\infty}{\{\eta^\ell_n(\xi)<\lceil\theta\ell\rceil\}}\big).
\end{equation*}
Given the increasing monotonicity of $n\mapsto\psi^{\ell,k}_n$ following Remark~\ref{rmk:refine}(\ref{rmk:variation}), the events $(\{\eta^\ell_n(\xi)<\lceil\theta\ell\rceil\})_{n\geq0}$ are decreasing, so that by Fatou's lemma and \eqref{eq:P(eta<theta*l)},
\begin{equation*}
\mathbb{P}\big(\limsup_{n\uparrow\infty}{\{\eta^\ell_n(\xi)<\lceil\theta\ell\rceil\}}\big)
\leq\liminf_{n\uparrow\infty}\mathbb{P}\big(\eta^\ell_n(\xi)<\lceil\theta\ell\rceil\big)
=0.
\end{equation*}
Finally,
\begin{equation*}
\mathbb{P}\big(
    \exists n_0\in\mathbb{N},
    \forall n\geq n_0,
    \eta^\ell_n(\xi)
    =\lceil\theta\ell\rceil
\big)=1.
\end{equation*}
In words, since $(\eta^\ell_n(\xi))_{n\geq1}$ are discrete random variables taking values in $[\![0,\lceil\theta\ell\rceil]\!]$, considering \eqref{eq:P(eta<theta*l)}, there exists a $\Pas$-finite random time $n_0$ that depends on $\ell$ and $\xi$, such that for $n\in\mathbb{N}$, the equality $\eta^\ell_n(\xi)=\lceil\theta\ell\rceil$ holds on the event $\{n\geq n_0\}$.
Hence the desired saturation effect for $n$ large.

Of course, the iterate $\xi$ depends on the iteration number $n$. The boundedness of the iterates in some $L^p(\mathbb{P})$ space, $0<p\leq2$, is hence essential to retrieve the asymptotic behavior described in \eqref{eq:P(eta<theta*l)}. Such boundedness is established in Proposition~\ref{prp:error:statistical:bis}.
\end{remark}

\subsection{Convergence Analysis}
The proofs of the ensuing results are postponed to Appendix~\ref{proofs}.

\begin{lemma}\label{lmm:bias}
Let $\gamma_n=\gamma_1n^{-\beta}$, $n\in\mathbb{N}^*$, with $\gamma_1>0$ and $\beta\in(0,1]$.
\begin{enumerate}[(i)]
\item\label{lmm:bias:i}
Assume that the real-valued random variables $X^h$ admit pdf $f_{X^h}$ that are bounded uniformly in $h\in\overline{\mathcal{H}}$.
\begin{enumerate}[\rm a.]
\item\label{lmm:bias:i:a}
If \eqref{assumption:finite:Lp:moment} is satisfied, with
\begin{equation}
\label{assumption:finite:Lp:moment:bias}
p_\star>2,
\quad
1<r<2,
\quad\text{and}\quad
0<\theta\leq\frac{p_\star-2}{p_\star+2},
\end{equation}
then, for all $\ell,n\in\mathbb{N}^*$,
\begin{equation*}
|\mathbb{E}[\mathds1_{\{X^{h_{\ell+\eta^\ell_n(\xi)}}>\xi\}}
-\mathds1_{\{X^{h_{\ell+\lceil\theta\ell\rceil}}>\xi\}}]|
\leq C_1h_\ell^{1+\theta}u_n,
\end{equation*}
where
\begin{equation*}
C_1=\frac{h_0^{-\frac{2p_\star}{p_\star+2}}h_1^{-p_\star(\frac1r-\frac12)}}{M^{p_\star(\frac1r-\frac12)}-1}\sqrt2C_\mathrm{ad}^{-p_\star}B_{p_\star}\mathbb{E}\big[\big|\varphi(Y,Z)-\mathbb{E}[\varphi(Y,Z)|Y]\big|^{p_\star}\big]^\frac1{p_\star},
\end{equation*}
with $B_{p_\star}$ being a positive constant that depends only on $p_\star$.

\item\label{lmm:bias:i:b}
If \eqref{assumption:conditional:gaussian:concentration} is satisfied for some $\mathfrak{g}>0$, with
\begin{equation}
\label{assumption:conditional:gaussian:concentration:bias}
h_0\geq\bigg(\frac{8\mathfrak{g}}{C_\mathrm{ad}^2}\bigg)^\frac{r}{2-r},\quad
1<r\leq2,
\quad\text{and}\quad
0<\theta\leq1,
\end{equation}
then, for all $\ell,n\in\mathbb{N}^*$,
\begin{equation*}
|\mathbb{E}[\mathds1_{\{X^{h_{\ell+\eta^\ell_n(\xi)}}>\xi\}}
-\mathds1_{\{X^{h_{\ell+\lceil\theta\ell\rceil}}>\xi\}}]|
\leq C_2h_\ell^{1+\theta}\gamma_n,
\end{equation*}
where
\begin{equation*}
C_2=\frac2{\gamma_1(1-M^{-1})}.
\end{equation*}
\end{enumerate}

\item\label{lmm:bias:ii}
If \eqref{assump:unif:lipschitz:integrability:conditional:cdf:bis} is satisfied with, for some $\upsilon_0>0$,
\begin{equation}
\label{assump:unif:lipschitz:integrability:conditional:cdf:bias}
\begin{gathered}
\sup_{0\leq h<h'\in\overline{\mathcal{H}}}\mathbb{E}[\exp(\upsilon_0|G_h^{h'}|^2)]<\infty,
\quad
h_0\geq\bigg(\frac2{\upsilon_0C_\mathrm{ad}^2}\bigg)^\frac{r}{2-r},\\
1<r\leq2,
\quad\text{and}\quad
0<\theta\leq1,
\end{gathered}
\end{equation}
then, for all $\ell,n\in\mathbb{N}^*$,
\begin{equation*}
|\mathbb{E}[\mathds1_{\{X^{h_{\ell+\eta^\ell_n(\xi)}}>\xi\}}
-\mathds1_{\{X^{h_{\ell+\lceil\theta\ell\rceil}}>\xi\}}]|
\leq C_3h_\ell^{1+\theta}\gamma_n,
\end{equation*}
where
\begin{equation*}
C_3=\frac{\sup_{0\leq h<h'\in\overline{\mathcal{H}}}\mathbb{E}[\exp(\upsilon_0|G_h^{h'}|^2)]}{\gamma_1(1-M^{-1})}.
\end{equation*}
\end{enumerate}
\end{lemma}

\begin{remark}
\begin{enumerate}[(i)]
\item
We refer to Remark~\ref{rmk:framework}(\ref{rmk:framework:p}) for an elaboration on the value of $B_{p_\star}$ in Lemma~\ref{lmm:bias}(\ref{lmm:bias:i})\hyperref[lmm:bias:i:a]{a}.

\item
Unlike \cite[Lemma~3.9]{HST22}, the bias controls above display a dependency on $n$ that decays roughly in the order of the step size $\gamma_n$. This property will prove useful to control the numerical error induced by the refinement strategy \eqref{eq:eta} in the adaptive nested scheme \eqref{adNSA}.

\item
Using the special values of $h=0$ and $h'=1$ within the first postulate of the framework \eqref{assump:unif:lipschitz:integrability:conditional:cdf:bias} yields that the quantity
\begin{equation*}
\mathbb{E}\Big[\exp\Big(\upsilon_0\big(\varphi(Y,Z)-\mathbb{E}[\varphi(Y,Z)|Y]\big)^2\Big)\Big]
\end{equation*}
is bounded.
As discussed in Remark~\ref{rmk:framework}(\ref{rmk:framework:sub:gaussian}), this entails a Gaussian concentration of the random variable $\varphi(Y,Z)-\mathbb{E}[\varphi(Y,Z)|Y]$.
\end{enumerate}
\end{remark}

Define
\begin{equation}
\label{def:zeta}
\widetilde\gamma_n^\ell=
\begin{cases}
(u_nh_\ell^{1+\theta})\vee\gamma_n
&\text{if \eqref{assumption:finite:Lp:moment:bias} holds,}\\
\gamma_n
&\text{if \eqref{assumption:conditional:gaussian:concentration:bias} or \eqref{assump:unif:lipschitz:integrability:conditional:cdf:bias} holds,}
\end{cases}
\quad\ell,n\in\mathbb{N}^*.
\end{equation}
Recalling the definition of $(\lambda_q)_{q\geq1}$ in \eqref{eq:lambda}, the next lemma provides an adaptive counterpart to Lemma~\ref{lmm:error}(\ref{lmm:error:statistical}).

\begin{proposition}
\label{prp:error:statistical:bis}
Let $\gamma_n=\gamma_1n^{-\beta}$, $n\in\mathbb{N}^*$, with $\gamma_1>0$ and $\beta\in(0,1]$.
Suppose that Assumption~\ref{asp:misc} holds.
\begin{enumerate}[(i)]
    \item\label{prp:error:statistical:bis:i}
If $\bar\lambda_1\gamma_1>1$ when $\beta=1$, and
\begin{equation*}
\sup_{\ell\geq0}\mathbb{E}\bigg[|\widetilde\xi^{h_\ell}_0|^2
\exp\bigg(\frac{2k_\alpha}{1-\alpha}
\sup_{\ell'\geq1}\|f_{X^{h_{\ell'}}}\|_\infty
|\widetilde\xi_0^{h_\ell}|\bigg)\bigg]<\infty,
\end{equation*}
where $k_\alpha=1\vee\frac\alpha{1-\alpha}$,
then, for all $\ell,n\in\mathbb{N}^*$,
\begin{equation}
\label{eq:L2:stat:error}
\mathbb{E}[(\widetilde\xi^{h_\ell}_n-\xi^{h_{\ell+\lceil\theta\ell\rceil}}_\star)^2]
\leq C^\beta\widetilde\gamma_n^\ell,
\quad
\text{where}
\quad
C^\beta=\sup_{\ell\geq1}C^{\ell,\beta},
\end{equation}
with the $(C^{\ell,\beta})_{\ell\geq1}$ being the positive constants given in \eqref{eq:C^l,b}.
    \item\label{prp:error:statistical:bis:ii}
If $\bar\lambda_2\gamma_1>2$ when $\beta=1$, and
\begin{equation*}
\sup_{\ell\geq0}\mathbb{E}\bigg[|\widetilde\xi^{h_\ell}_0|^4
\exp\bigg(\frac{4k_\alpha}{1-\alpha}
\sup_{\ell'\geq1}\|f_{X^{h_{\ell'}}}\|_\infty
|\widetilde\xi_0^{h_\ell}|\bigg)\bigg]<\infty,
\end{equation*}
where $k_\alpha=1\vee\frac\alpha{1-\alpha}$,
then, for all $\ell,n\in\mathbb{N}^*$,
\begin{equation}
\label{eq:L4:control}
\mathbb{E}[(\widetilde\xi^{h_\ell}_n-\xi^{h_{\ell+\lceil\theta\ell\rceil}}_\star)^4]
\leq\bar{C}^\beta(\widetilde\gamma_n^\ell)^2,
\quad
\text{where}
\quad
\bar{C}^\beta=\sup_{\ell\geq1}\bar{C}^{\ell,\beta},
\end{equation}
with the $(\bar{C}^{\ell,\beta})_{\ell\geq1}$ being the positive constants given in \eqref{eq:C:bar:l,b}.
\end{enumerate}
\end{proposition}

\begin{remark}
The above results are comparable to the non-adaptive case that is dealt with in Lemma~\ref{lmm:error}(\ref{lmm:error:statistical}).
The $L^4(\mathbb{P})$-control is necessary to our study of the adaptive MLSA scheme in Section \ref{sec:adamlsa}.
\end{remark}

\subsection{Complexity Analysis}
Proposition~\ref{prp:error:statistical:bis}(\ref{prp:error:statistical:bis:i}) provides an insight into the behavior of adNSA.
Fixing $\ell\in\mathbb{N}^*$ and running \eqref{adNSA} $n\in\mathbb{N}^*$ times results in a global error
\begin{equation*}
\widetilde\xi^{h_\ell}_n-\xi^0_\star
=(\widetilde\xi^{h_\ell}_n-\xi^{h_{\ell+\lceil\theta\ell\rceil}}_\star)+(\xi^{h_{\ell+\lceil\theta\ell\rceil}}_\star-\xi^0_\star),
\end{equation*}
where the first term represents the statistical error and the second the bias error.
The proofs of the following results are available in Appendix~\ref{complexity}.

\begin{proposition}\label{prp:selection:number:level:adNSA}
Suppose that Assumption~\ref{asp:misc}(\ref{asp:misc:i}) holds.
Let $\varepsilon>0$ be a fixed prescribed accuracy, and let $C>0$ denote the constant described in \eqref{eq:oO}.
If $h_0>C^{-1}\varepsilon$, then, setting
\begin{equation}
\label{selection:number:level:adNSA}
\ell=\bigg\lceil\frac{\ln{(Ch_0\varepsilon^{-1})}}{(1+\theta)\ln{M}}\bigg\rceil
\geq1
\end{equation}
achieves a bias error for adNSA of order $\varepsilon$.
\end{proposition}

\begin{remark}
According to the decomposition \eqref{eq:global:error:VaR:NSA}, when using NSA with a bias $h_\ell$ on the geometric scale $\mathcal{H}_0=\{h_{\ell'},\ell'\in\mathbb{N}\}$, one must set
\begin{equation*}
\ell=\bigg\lceil\frac{\ln{(Ch_0\varepsilon^{-1})}}{\ln{M}}\bigg\rceil
\geq1
\end{equation*}
to achieve a bias error of order $\varepsilon$. adNSA requires a smaller $\ell$, which in turn translates into a less expensive, albeit biased, simulation $X^{h_\ell}$ that is then refined. This adjustment contributes to a reduction in the complexity of adNSA.
\end{remark}

The next lemma quantifies the average simulation amounts performed per iteration under the adaptive strategy.

\begin{lemma}\label{lmm:E:h:-1}
Under Assumptions~\ref{asp:misc}(\ref{asp:misc:ii})-(\ref{asp:misc:iv}), by setting $r>1$, for all $\ell,n\in\mathbb{N}^*$,
\begin{equation*}
\mathbb{E}[h_{\ell+\widetilde\eta^\ell_n}^{-1}]
\leq C\begin{cases}
    h_\ell^{-1}n^\frac\delta{p_\star}
    &\text{if \eqref{assumption:finite:Lp:moment} holds,}\\
    h_\ell^{-1}(\ell+\ln{n})^\frac12
    &\text{if \eqref{assumption:conditional:gaussian:concentration} or \eqref{assump:unif:lipschitz:integrability:conditional:cdf:bis} holds,}
\end{cases}
\end{equation*}
where $C>0$ is described in \eqref{eq:C=C1:max:C2}.
\end{lemma}

\begin{remark}
Note the analogous \cite[Proposition~3.1]{HST22} for adaptive Monte Carlo, independent nonetheless of the iteration number $n$. The average amount of inner simulations therein is in the order of $h_\ell^{-1}$.
The saturation of our adaptive strategy for $n$ large is responsible for expanding these amounts by a small order depending on $n$.
We refer to Remark~\ref{rmk:n} for further comments on this dependence on $n$.
\end{remark}

\begin{proposition}
\label{prp:ansa:complexity}
Suppose that Assumptions~\ref{asp:misc}(\ref{asp:misc:ii})-(\ref{asp:misc:iv}) hold.
By setting $r>1$,
the average complexity of adNSA satisfies
\begin{equation*}
\Cost_{\text{\rm adNSA}}
\leq C
\begin{cases}
n^{1+\frac\delta{p_\star}}h_\ell^{-1}
&\text{if \eqref{assumption:finite:Lp:moment} holds,}\\
n(\ell^\frac12+\ln^\frac12{n})h_\ell^{-1}
&\text{if \eqref{assumption:conditional:gaussian:concentration} or \eqref{assump:unif:lipschitz:integrability:conditional:cdf:bis} holds,}
\end{cases}
\end{equation*}
where $C>0$ is the constant described in \eqref{eq:C:meta}.
\end{proposition}

\begin{remark}
Absent the saturation factor, the complexity would simply be in $\OO(nh_\ell^{-1})$, similar to the complexity order of NSA, as per Proposition~\ref{prp:cost:nsa}.
\end{remark}

\begin{theorem}
\label{thm:cost:adansa}
Let $\varepsilon>0$ be a prescribed accuracy.
Assume that $h_0>(C')^{-1}\varepsilon$, where $C'>0$ is the constant described in \eqref{eq:oO}, and that $\bar\lambda_1\gamma_1>1$ if $\beta=1$, and set $\ell$ as in \eqref{selection:number:level:adNSA}.
Then, denoting $C>0$ the constant defined in \eqref{eq:cst}, by taking
\begin{equation*}
n=\begin{cases}
\lceil C^\frac2\delta\varepsilon^{-\frac1\delta}\rceil
&\text{if \eqref{assumption:finite:Lp:moment:bias} holds with $\delta\leq\frac\beta2$,}\\
\lceil C^\frac2\beta\varepsilon^{-\frac2\beta}\rceil
&\text{if \eqref{assumption:finite:Lp:moment:bias} holds with $\delta\geq\frac\beta2$, or if \eqref{assumption:conditional:gaussian:concentration:bias} or \eqref{assump:unif:lipschitz:integrability:conditional:cdf:bias} holds,}
\end{cases}
\end{equation*}
the statistical $L^2(\mathbb{P})$-error for adNSA is of order $\varepsilon$ as $\varepsilon\downarrow0$.

Under this setting, by taking $r>1$, the computational cost of adNSA satisfies
\begin{equation*}
\Cost_{\text{\rm adNSA}}
\leq\bar{C}\begin{cases}
\varepsilon^{-\frac{p_\star+\delta}{p_\star\delta}-\frac1{1+\theta}}
&\text{if \eqref{assumption:finite:Lp:moment:bias} holds with $\delta\leq\frac\beta2$,}\\
\varepsilon^{-\frac{2(\delta+p_\star)}{\beta p_\star}-\frac1{1+\theta}}
&\text{if \eqref{assumption:finite:Lp:moment:bias} holds with $\delta\geq\frac\beta2$,}\\
\varepsilon^{-\frac2\beta-\frac1{1+\theta}}|\ln{\varepsilon}|^\frac12
&\text{if \eqref{assumption:conditional:gaussian:concentration:bias} or \eqref{assump:unif:lipschitz:integrability:conditional:cdf:bias} holds,}
\end{cases}
\quad\text{as}\quad
\varepsilon\downarrow0,
\end{equation*}
where $\bar{C}>0$ is the constant defined in \eqref{eq:cst:bar}.
The minimal computational cost, which is attained under the constraint $\bar\lambda_1\gamma_1>1$, satisfies
\begin{equation*}
\Cost_{\text{\rm adNSA}}
\leq\bar{C}\begin{cases}
\varepsilon^{-\frac{5p_\star+4}{2p_\star}}
&\text{if \eqref{assumption:finite:Lp:moment:bias} holds, $\beta=1$, $\delta=\frac\beta2=\frac12$, and $\theta=\frac{p_\star-2}{p_\star+2}$,}\\
\varepsilon^{-\frac52}|\ln{\varepsilon}|^\frac12
&\text{if \eqref{assumption:conditional:gaussian:concentration:bias} or \eqref{assump:unif:lipschitz:integrability:conditional:cdf:bias} holds, and $\beta=\theta=1$,}
\end{cases}
\quad\text{as}\quad
\varepsilon\downarrow0.
\end{equation*}
\end{theorem}

\begin{remark}
\begin{enumerate}[(i)]
\item
Within the framework of Assumption~\eqref{assumption:finite:Lp:moment:bias}, adNSA achieves a speed-up over NSA when $p_\star>4$.
By contrast, heavier-tailed distributions, corresponding to an integrability exponent $p_\star\leq4$, do not appear to benefit from adaptive refinement. This may be explained by the fact that adaptivity is primarily sensitive to higher-order moments.

\item
By comparison with Proposition~\ref{prp:cost:nsa}, the complexities outlined in Theorem~\ref{thm:cost:adansa} are an order of magnitude of $\varepsilon^\frac12$ lower than those of NSA for large $p_\star$ under the framework \eqref{assumption:finite:Lp:moment:bias} and up to a logarithmic factor under the frameworks \eqref{assumption:conditional:gaussian:concentration:bias} and \eqref{assump:unif:lipschitz:integrability:conditional:cdf:bias}.
As we will see in the next section, adNSA serves as a fundamental component of the adaptive MLSA scheme. Consequently, the performance improvements noted on the former scheme should result in an acceleration of the latter.

\end{enumerate}
\end{remark}

The table below recaps the differences between NSA and adNSA.

\begin{table}[H]
\centerline{
    \begin{tabular}{!{\vrule width 1pt}c!{\vrule width 1pt}c|c!{\vrule width 1pt}}
         \specialrule{0.1em}{0em}{0em}
         Algorithm
         & NSA & adNSA \\
         \specialrule{0.1em}{0em}{0em}
         Bias parameter
         & $h\sim\varepsilon$
         & $h_\ell\sim\varepsilon^\frac1{1+\theta}\Leftrightarrow\ell\sim\dfrac{|\ln{\varepsilon}|}{(1+\theta)\ln{M}}$\\
         \hline
         Iteration amount
         & $n\sim\varepsilon^{-\frac2\beta}$
         & $n\sim\begin{cases}
\varepsilon^{-\frac1\delta}
&\text{if \eqref{assumption:finite:Lp:moment:bias} holds with $\delta\leq\frac\beta2$,}\\
\varepsilon^{-\frac2\beta}
&\text{if \eqref{assumption:finite:Lp:moment:bias} holds with $\delta\geq\frac\beta2$,}\\
&\text{\;or if \eqref{assumption:conditional:gaussian:concentration:bias} or \eqref{assump:unif:lipschitz:integrability:conditional:cdf:bias} holds,}
\end{cases}$\\
         \hline
         \multirow{3}{*}[-0.5em]{Best complexity}
         & \multirow{3}{*}[-0.5em]{$\Cost_\text{\tiny\rm NSA}=\OO(\varepsilon^{-3})$}
         & $\Cost_\text{\tiny\rm adNSA}=$\\
         & 
         & $\begin{cases}
\OO(\varepsilon^{-\frac{5p_\star+4}{2p_\star}})
&\text{if \eqref{assumption:finite:Lp:moment:bias} holds,}\\
\OO(\varepsilon^{-\frac52}|\ln{\varepsilon}|^\frac12)
&\text{if \eqref{assumption:conditional:gaussian:concentration:bias} or \eqref{assump:unif:lipschitz:integrability:conditional:cdf:bias} holds.}
\end{cases}$ \\
         \specialrule{0.1em}{0em}{0em}
    \end{tabular}
}
    \caption{Comparison of NSA and adNSA. $\varepsilon\>0$ designates the prescribed accuracy. The statements refer to Propositions~\ref{prp:cost:nsa} and~\ref{thm:cost:adansa}.}
    \label{tbl:cpx:nested}
\end{table}

\section{Adaptive Multilevel Stochastic Approximation Algorithm}
\label{sec:adamlsa}

Recalling that MLSA telescopes multiple paired NSA schemes, the previous development on adNSA provides a framework for the extension of the adaptive refinement strategy to the multilevel paradigm. We define the adaptive multilevel SA (adMLSA) estimator for the VaR as
\begin{equation}
\label{alg:amlsa:var}
\widetilde\xi^\text{\rm\tiny ML}_\mathbf{N}=\xi^{h_0}_{N_0}+\sum_{\ell=1}^L\widetilde\xi^{h_\ell}_{N_\ell}-\widetilde\xi^{h_{\ell-1}}_{N_\ell},
\end{equation}
where $\mathbf{N}:=(N_0,\dots,N_L)\in(\mathbb{N}^*)^{L+1}$ represents the number of iterations at each level.
Each level $\ell\in[\![0,L]\!]$ is simulated independently.
As detailed in Remark~\ref{rmk:refine}(\ref{rmk:level:0}), the level $0$ estimator is not refined, resulting in $N_0$ iterations of NSA.
Each of the remaining levels $\ell\in[\![1,L]\!]$ is obtained as follows: after initializing $(\widetilde\xi^{h_{\ell-1}}_0,\widetilde\xi^{h_\ell}_0)$, for each $n=0,\dots,N_\ell-1$, once the components of $(X^{h_{\ell-1}}_{n+1},X^{h_\ell}_{n+1})$ have been simulated according to \eqref{perfect:correlation}, $X^{h_{\ell-1}}_{n+1}$ and $X^{h_\ell}_{n+1}$ are separately refined as in \eqref{eq:eta:l:n}, relative to the fine and coarse iterates $\widetilde\xi^{h_{\ell-1}}_n$ and $\widetilde\xi^{h_\ell}_n$, into $\widetilde{X}^{h_{\ell-1}}_{n+1}$ and $\widetilde{X}^{h_\ell}_{n+1}$ respectively.
These refined samples are then injected into separate single updates of \eqref{adNSA} to obtain the following iterates $(\widetilde\xi^{h_{\ell-1}}_{n+1},\widetilde\xi^{h_\ell}_{n+1})$.

\begin{algorithm}[H]
\caption{Adaptive Multilevel SA for estimating the VaR}
\label{alg:adamlsa}
\begin{algorithmic}[1]
\Require {A number of levels $L\geq1$, a bias parameter $h_0=\frac1K\in\mathcal{H}$, an integer $M\geq2$, $N_0,\dots,N_L\in\mathbb{N}^*$, a positive non-increasing sequence 
$(\gamma_n)_{n\geq1}$ such that $\sum_{n=1}^\infty\gamma_n=\infty$ and $\lim_{n\to\infty}\gamma_n=0$, refinement parameters $C_\mathrm{ad}>0$, $r>1$ and $0<\theta\leq1$}
\State {Sample $\xi^{h_0}_0$ randomly}
\State {$\xi^{h_0}_{N_0}\gets\mathrm{NSA}(h_0,N_0)$}
\For {$\ell=1\mathrel{.\,.}L$}
   \State {Set $h_\ell\gets\frac{h_0}{M^\ell}$}
   \State {Sample $(\widetilde\xi^{h_{\ell-1}}_0,\widetilde\xi^{h_\ell}_0)$ randomly}
   \label{alg:adamlsa:sample:0}
   \For {$n=0\mathrel{.\,.}N_\ell-1$}
      \State {Simulate $Y_{n+1}\sim Y$ and $Z^1_{n+1},\dots,Z^{KM^\ell}_{n+1}\stackrel[]{\text{\rm\tiny i.i.d.}}{\sim}Z$ independently of $Y_{n+1}$}
      \State {$X^{h_{\ell-1}}_{n+1}\gets\frac1{KM^{\ell-1}}\sum_{k=1}^{KM^{\ell-1}}\varphi(Y_{n+1},Z^k_{n+1})$}
      \State {$X^{h_\ell}_{n+1}\gets\frac1MX^{h_{\ell-1}}_{n+1}+\frac1{KM^\ell}\sum_{k=KM^{\ell-1}+1}^{KM^\ell}\varphi(Y_{n+1},Z^k_{n+1})$}
%       \algstore{adaMLSA}
% \end{algorithmic}
% \end{algorithm}

% \begin{algorithm}[H]
% \begin{algorithmic}[1]
%       \algrestore{adaMLSA}
      \For {$j=\ell-1\mathrel{.\,.}\ell$}
         \State{$\widetilde{X}^{h_j}_{n+1}\gets\Refine_{C_\mathrm{ad},r,\theta,j,n}(X^{h_j}_{n+1},\widetilde\xi^{h_j}_n)$}
         \State {$\widetilde\xi^{h_j}_{n+1}\gets\widetilde\xi^{h_j}_n-\gamma_{n+1}H(\widetilde\xi^{h_j}_n,\widetilde{X}^{h_j}_{n+1})$}
      \EndFor
   \EndFor
\EndFor
\State {$\widetilde\xi^\text{\tiny\rm ML}_\mathbf{N}\gets\xi^{h_0}_{N_0}+\sum_{\ell=1}^L \widetilde\xi^{h_\ell}_{N_\ell}-\widetilde\xi^{h_{\ell-1}}_{N_\ell}$}
\State \Return {$\widetilde\xi^\text{\tiny\rm ML}_\mathbf{N}$}
\end{algorithmic}
\end{algorithm}

\begin{remark}
In line~\ref{alg:adamlsa:sample:0} of Algorithm~\ref{alg:adamlsa}, the intializations $\widetilde\xi^{h_{\ell-1}}_0$ and $\widetilde\xi^{h_\ell}_0$ may be set equal, or even be set deterministically.
\end{remark}

\subsection{Convergence Analysis}
The proofs of the following results are postponed to Appendix~\ref{proofs}.

The next result guarantees stronger error controls for adMLSA comparatively with its non-adaptive counterpart, Lemma~\ref{lmm:local:strong:error:indicator:func}, for MLSA.

\begin{lemma}\label{lmm:local:strong:error:indicator:func:variance}
Let $\gamma_n=\gamma_1n^{-\beta}$, $n\in\mathbb{N}^*$, with $\gamma_1>0$ and $\beta\in(0,1]$.
\begin{enumerate}[(i)]
\item\label{lmm:local:strong:error:indicator:func-i:variance}
Assume that the $X^h$ admit pdf $f_{X^h}$ that are bounded uniformly in $h\in\overline{\mathcal{H}}$.
\begin{enumerate}[\rm a.]
\item\label{lmm:local:strong:error:indicator:func-ia:variance}
If \eqref{assumption:finite:Lp:moment} is satisfied, with
\begin{equation}
\label{assumption:finite:Lp:moment:variance}
1<r<2,
\quad\text{and}\quad
0<\theta\leq\frac{p_\star}{p_\star+2},
\end{equation}
then, for all $\ell,n\in\mathbb{N}^*$,
\begin{equation*}
\mathbb{E}[|\mathds1_{\{X^{h_{\ell+\eta^\ell_n(\xi)}}>\xi\}}-\mathds1_{\{X^0>\xi\}}|]
\leq C_1h_\ell^{(1+\theta)\frac{p_\star}{2(p_\star+1)}},
\end{equation*}
where
\begin{equation*}
\begin{aligned}
C_1
&=h_0^{-\frac{p_\star}{2(p_\star+1)}}B_{p_\star}\mathbb{E}\big[\big|\varphi(Y, Z)-\mathbb{E}[\varphi(Y, Z)|Y]\big|^{p_\star}\big]^\frac1{p_\star+1}\Big(\sup_{h\in\overline{\mathcal{H}}}\|f_{X^h}\|_\infty\Big)^\frac{p_\star}{p_\star+1}\\
&\quad+\frac{\gamma_1h_0^{-\frac{2(p_\star+1)}{p_\star+2}}h_1^{-p_\star(\frac1r-\frac12)}}{M^{p_\star(\frac1r-\frac12)}-1}\sqrt2C_\mathrm{ad}^{-p_\star}B_{p_\star}\mathbb{E}\big[\big|\varphi(Y, Z)-\mathbb{E}[\varphi(Y, Z)|Y]\big|^{p_\star}\big]^\frac1{p_\star},
\end{aligned}
\end{equation*}
with $B_{p_\star}$ being a positive constant that depends only on $p_\star$.

\item\label{lmm:local:strong:error:indicator:func-ib:variance}
If \eqref{assumption:conditional:gaussian:concentration} is satisfied for $\mathfrak{g}>0$, with
\begin{equation}
\label{assumption:conditional:gaussian:concentration:variance}
h_0\geq\bigg(\frac{4\mathfrak{g}}{C_\mathrm{ad}^2}\bigg)^\frac{r}{2-r},
\quad
1<r\leq2,
\quad\text{and}\quad
0<\theta\leq1,
\end{equation}
then, for all $\ell,n\in\mathbb{N}^*$,
\begin{equation*}
\mathbb{E}[|\mathds1_{\{X^{h_{\ell+\eta^\ell_n(\xi)}}>\xi\}}
-\mathds1_{\{X^0>\xi\}}|]
\leq C_2h_\ell^\frac{1+\theta}2|\ln{h_\ell}|^\frac12,
\end{equation*}
where
\begin{equation*}
C_2=4\mathfrak{g}^\frac12h_0^{-\frac12}\bigg(1+\Big(\sup_{h\in\overline{\mathcal{H}}}{\|f_{X^h}\|_\infty}\Big)\Big(2\big(2\vee|\ln{(\mathfrak{g}h_0^{-1})}|\big)\Big)^\frac12\bigg)+\frac2{1-M^{-\frac12}}.
\end{equation*}
\end{enumerate}

\item\label{lmm:local:strong:error:indicator:func-ii:variance}
If \eqref{assump:unif:lipschitz:integrability:conditional:cdf:bis} is satisfied, with, for some $\upsilon_0>0$,
\begin{equation}
\label{assump:unif:lipschitz:integrability:conditional:cdf:variance}
\begin{gathered}
\sup_{h\in\mathcal{H}}\mathbb{E}[\exp(\upsilon_0|G_0^h|^2)]<\infty,
\quad
h_0\geq\bigg(\frac1{\upsilon_0C_\mathrm{ad}^2}\bigg)^\frac{r}{2-r},\\
1<r\leq2,
\quad\text{and}\quad
0<\theta\leq1,
\end{gathered}
\end{equation}
then, for all $\ell,n\in\mathbb{N}^*$,
\begin{equation*}
\mathbb{E}[|\mathds1_{\{X^{h_{\ell+\eta^\ell_n(\xi)}}>\xi\}}
-\mathds1_{\{X^0>\xi\}}|]
\leq C_3h_\ell^\frac{1+\theta}2,
\end{equation*}
where
\begin{equation*}
C_3=h_0^{-\frac12}\sup_{0\leq h_1<h_2\in\overline{\mathcal{H}}}{\mathbb{E}[K_{h_1}^{h_2}|G_{h_1}^{h_2}|]}+\frac{\sup_{h\in\mathcal{H}}{\mathbb{E}[\exp{(\upsilon_0|G_0^h|^2)}]}}{1-M^{-\frac12}}.
\end{equation*}
\end{enumerate}
\end{lemma}

\begin{remark}
\begin{enumerate}[(i)]
    \item
Although the frameworks of Lemmas~\ref{lmm:bias} and~\ref{lmm:local:strong:error:indicator:func:variance} are similar, those of Lemma~\ref{lmm:bias} are stronger.

    \item
The value of $B_{p_\star}$ in Lemma~\ref{lmm:local:strong:error:indicator:func:variance}(\ref{lmm:local:strong:error:indicator:func-i:variance})\hyperref[lmm:local:strong:error:indicator:func-ia:variance]{a} is discussed in Remark~\ref{rmk:framework}(\ref{rmk:framework:p}).

    \item
In comparison with the non-adaptive frameworks of Lemma~\ref{lmm:local:strong:error:indicator:func}, the variance controls for adMLSA display an extra exponent factor of $1+\theta$, which should bring down the algorithm's statistical error faster to $0$.

\end{enumerate}
\end{remark}

The main convergence result follows. Its proof is deferred to Appendix~\ref{apx:prf:main}.

\begin{theorem}
\label{thm:amlsa:L2}
Suppose that $\varphi(Y,Z)\in L^2(\mathbb{P})$, that Assumptions~\ref{asp:misc} and~\ref{asp:fh-f0} hold, and that
\begin{equation*}
\sup_{\ell\geq0}\mathbb{E}\bigg[|\widetilde\xi^{h_\ell}_0|^4
\exp\bigg(\frac{4k_\alpha}{1-\alpha}
\sup_{\ell'\geq1}\|f_{X^{h_{\ell'}}}\|_\infty
|\widetilde\xi_0^{h_\ell}|\bigg)\bigg]<\infty,
\end{equation*}
where $k_\alpha=1\vee\frac\alpha{1-\alpha}$.
If $\gamma_n=\gamma_1n^{-\beta}$, $n\in\mathbb{N}^*$, with $\gamma_1>0$ and $\beta\in(0,1]$, and if $\bar\lambda_2\gamma_1>2$ when $\beta=1$,
then, for all $L\in\mathbb{N}^*$ and all $\mathbf{N}=(N_0,\dots,N_L)\in(\mathbb{N}^*)^{L+1}$,
\begin{equation}
\label{L2:norm:adaML:VaR}
\mathbb{E}[(\widetilde\xi^\text{\tiny\rm ML}_\mathbf{N}-\xi^{h_{L+\lceil\theta L\rceil}}_\star)^2]
\leq C\bigg(\gamma_{N_0}+\bigg(\sum_{\ell=1}^L\widetilde\gamma_{N_\ell}^\ell\bigg)^2+\sum_{\ell=1}^L\gamma_{N_\ell}(\widetilde{\gamma}_{N_\ell}^\ell)^\frac12+\sum_{\ell=1}^L\gamma_{N_\ell}\epsilon(h_\ell)^{1+\theta}\bigg),
\end{equation}
where $C>0$ is the constant described in \eqref{eq:master:cst}, $(\widetilde\gamma_n^\ell)_{\ell,n\geq1}$ are defined in \eqref{def:zeta}, and $\epsilon$ is defined as in \eqref{eq:eps(hl)}:
\begin{equation*}
\epsilon(h)
:=\begin{cases}
h^\frac{p_\star}{2(1+p_\star)}
&\text{if \eqref{assumption:finite:Lp:moment:bias} holds,}\\
h^\frac12|\ln{h}|^\frac12
&\text{if \eqref{assumption:conditional:gaussian:concentration:bias} holds,}\\
h^\frac12
&\text{if \eqref{assump:unif:lipschitz:integrability:conditional:cdf:bias} holds,}
\end{cases}
\quad h\in\mathcal{H}.
\end{equation*}
\end{theorem}

\begin{remark}
Similarly to Theorem~\ref{thm:ml-variance-cv}, the upper $L^2(\mathbb{P})$ control for the statistical error of adMLSA is comprised four terms: the first term governs the level $0$ simulation and the remaining three govern the drifts and martingales arising from the multilevel linearization \eqref{eq:xiML-xi*(hL)}.
The first difference with MLSA lies in the inclusion of $\widetilde\gamma_n^\ell$, which emanates from the controls of Proposition~\ref{prp:error:statistical:bis}.
Additionally, the final term exhibits an extra exponent factor of $1+\theta$, resulting in an accelerated convergence rate compared to MLSA.
\end{remark}

The convergence rate speed-up, shown in the previous theorem, translates a significant performance gain that should reflect in the complexity of adMLSA, as we clarify next.

\subsection{Complexity Analysis}

Approximating $\xi^0_\star$ by $\widetilde\xi^\text{\tiny\rm ML}_\mathbf{N}$ results in a global error that decomposes into statistical and bias errors as such:
\begin{equation}
\label{amlsa:global:error:var}
\widetilde\xi^\text{\tiny\rm ML}_\mathbf{N}-\xi^0_\star
=(\widetilde\xi^\text{\tiny\rm ML}_\mathbf{N}-\xi^{h_{L+\lceil\theta L\rceil}}_\star)+(\xi^{h_{L+\lceil\theta L\rceil}}_\star-\xi^0_\star).
\end{equation}
The proofs related to this section are deferred to Appendix~\ref{complexity}.

\begin{proposition}
\label{prp:selection:number:level:adMLSA}
Suppose that Assumption~\ref{asp:misc}(\ref{asp:misc:i}) holds.
Let $\varepsilon>0$ be a prescribed accuracy and $C>0$ the constant described in \eqref{eq:oO}.
If $h_0>C^{-1}\varepsilon$, then, setting the number of levels to
\begin{equation}
\label{selection:number:level:adMLSA}
L=\bigg\lceil\frac{\ln{(Ch_0\varepsilon^{-1})}}{(1+\theta)\ln{M}}\bigg\rceil\geq1
\end{equation}
achieves a bias error for adMLSA of order $\varepsilon$.
\end{proposition}

\begin{remark}
In view of Propositions~\ref{prp:selection:number:level:adMLSA} and~\ref{prp:mlsa:complexity}(\ref{prp:mlsa:complexity:levels}), to achieve a comparable bias error order, adMLSA requires significantly less levels than MLSA.
In fact, for $\theta=1$, adMLSA requires half as many levels as MLSA to achieve a bias error of order $\varepsilon$.
This property alone makes adMLSA much faster than MLSA.
It is linked to the fact that, for a given number of levels $L$, adMLSA achieves a bias error of order $h_L^{1+\theta}$, an order of magnitude $h_L^\theta$ lower than MLSA, which achieves a bias error of order $h_L$.
\end{remark}

\begin{proposition}
\label{prp:amlsa:complexity}
Suppose that Assumptions~\ref{asp:misc}(\ref{asp:misc:ii})-(\ref{asp:misc:iv}) hold.
By setting $r>1$,
the average complexity of adMLSA satisfies
\begin{equation}
\label{eq:adMLSA:cost}
\Cost_\text{adMLSA}\leq C
\begin{cases}
        \sum_{\ell=0}^LN_\ell^{1+\frac\delta{p_\star}}h_\ell^{-1}
    &\text{if \eqref{assumption:finite:Lp:moment} holds,}\\
        \sum_{\ell=0}^LN_\ell(L^\frac12+\ln^\frac12{N_\ell})h_\ell^{-1}
    &\text{if \eqref{assumption:conditional:gaussian:concentration} or \eqref{assump:unif:lipschitz:integrability:conditional:cdf:bis} holds,}
\end{cases}
\end{equation}
where $C>0$ is the constant given in \eqref{eq:new:C}.
\end{proposition}

\begin{remark}
Without a saturation factor in the refinement strategy, the computational cost would be lowered to $\OO(\sum_{\ell=0}^NN_\ell h_\ell^{-1})$, similar to both MLSA and the adaptive MLMC algorithm \cite{HST22}.
However, omitting this factor would adversely affect the convergence of the statistical error \eqref{L2:norm:adaML:VaR}. See Remark~\ref{rmk:n} for further comments on the necessity to saturate the refinements.
\end{remark}

\begin{theorem}
\label{thm:amlsa:complexity}
Let $C_\mathrm{ad}>0$, $r>1$ and $0<\theta\leq1$, and suppose that Assumptions~\ref{asp:misc} and~\ref{asp:fh-f0} hold.
Let $\varepsilon>0$ be a prescribed accuracy.
Denote $C,\bar{C}>0$ the constants on the right-hand sides of \eqref{L2:norm:adaML:VaR} and \eqref{eq:adMLSA:cost} respectively.
\begin{enumerate}[(i)]
\item\label{thm:amlsa:complexity:indicator:func-ia:N}
If \eqref{assumption:finite:Lp:moment:bias} holds,
then, by setting 
\begin{equation}
\label{eq:Nl:i}
N_\ell=\bigg\lceil C_1^{\beta,\delta}\varepsilon^{-\frac2\beta}\bigg(\sum_{\ell'=0}^Lh_{\ell'}^{-\frac{(2\beta-(1+\theta))p_\star+(2\beta-(1+\theta)\delta)}{2(p_\star+1)((1+\beta)p_\star+\delta)}p_\star}\bigg)^\frac1\beta h_\ell^{\frac{(3+\theta)p_\star+2}{2(p_\star+1)((1+\beta)p_\star+\delta)}p_\star}\bigg\rceil,
\quad\ell\in[\![0,L]\!],
\end{equation}
where
\begin{equation*}
\begin{aligned}
C_1^{\beta,\delta}&=C^\frac{p_\star}{(1+\beta)p_\star+\delta}(\bar{C}\gamma_1)^\frac{(p_\star+\delta)p_\star}{(1+\beta)p_\star+\delta}\\
&\hphantom{=}\times
\bigg(\bigg(\frac{\beta p_\star}{p_\star+\delta}\bigg)^\frac{p_\star+\delta}{(1+\beta)p_\star+\delta}+\bigg(\frac{\beta p_\star}{p_\star+\delta}\bigg)^{-\frac{\beta p_\star}{(1+\beta)p_\star+\delta}}\bigg)^\frac1\beta
\bigg(\frac{p_\star+\delta}{(1+\beta)p_\star+\delta}\bigg)^\frac1\beta,
\end{aligned}
\end{equation*}
adMLSA scores a global $L^2(\mathbb{P})$-error of order $\varepsilon$.
By taking $L$ as in \eqref{selection:number:level:adMLSA}, the corresponding average complexity satisfies
\begin{equation*}
\Cost_\text{adMLSA}
\leq\widehat{C}_1
\begin{cases}
\varepsilon^{-\frac{2(p_\star+\delta)}{\beta p_\star}}
&\text{if $\theta>\frac{(2\beta-1)p_\star+2\beta-\delta}{p_\star+\delta}$,}\\
\varepsilon^{-\frac{2(p_\star+\delta)}{\beta p_\star}}
|\ln{\varepsilon}|^\frac{(1+\beta)p_\star+\delta}{\beta p_\star}
&\text{if $\theta=\frac{(2\beta-1)p_\star+2\beta-\delta}{p_\star+\delta}$,}\\
\varepsilon^{-\frac{2(p_\star+\delta)}{\beta p_\star}-\frac{(2\beta-(1+\theta))p_\star+2\beta-(1+\theta)\delta}{(1+\theta)\beta(p_\star+1)}}
&\text{if $\theta<\frac{(2\beta-1)p_\star+2\beta-\delta}{p_\star+\delta}$,}
\end{cases}
\;\;\text{as}\;\;
\varepsilon\downarrow0,
\end{equation*}
where $\widehat{C}_1>0$ is the constant on the right hand side of \eqref{eq:1:C}.
The best complexity is attained if $\beta=1$, $\delta\downarrow0$, and $\theta=\frac{p_\star-2}{p_\star+2}$:
\begin{equation*}
\inf_{\beta,\delta,\theta}
\Cost_\text{adMLSA}
\leq\widehat{C}_1
\varepsilon^{-2
-\frac{3p_\star+2}{2p_\star(p_\star+1)}-\iota},
\quad\text{for any $\iota>0$,}
\quad\text{as $\varepsilon\downarrow0$.}
\end{equation*}

\item\label{thm:amlsa:complexity:indicator:func-ib:N}
If \eqref{assumption:conditional:gaussian:concentration:bias} holds,
then, taking $L$ as in \eqref{selection:number:level:adMLSA} and
\begin{equation*}
N_\ell
=\bigg\lceil(\bar{C}\gamma_1)^\frac1\beta
\varepsilon^{-\frac2\beta}
\bigg(\sum_{\ell'=0}^Lh_{\ell'}^{-\frac{2\beta-(1+\theta)}{2(1+\beta)}}|\ln{h_{\ell'}}|^\frac{1+\theta}{2(1+\beta)}\bigg)^\frac1\beta
h_\ell^\frac{3+\theta}{2(1+\beta)}
|\ln{h_\ell}|^\frac{1+\theta}{2(1+\beta)}\bigg\rceil,
\;\ell\in[\![0,L]\!],
\end{equation*}
leads to a global $L^2(\mathbb{P})$ error for adMLSA of order $\varepsilon$ as $\varepsilon\downarrow0$.
The corresponding average complexity satisfies
\begin{equation*}
\Cost_\text{adMLSA}
\leq\widehat{C}_2
\begin{cases}
\varepsilon^{-\frac2\beta}
|\ln{\varepsilon}|^\frac{1+\theta+\beta}{2\beta}
&\text{if $\theta>2\beta-1$,}\\
\varepsilon^{-\frac2\beta}
|\ln{\varepsilon}|^\frac{3(1+\beta)+\theta}{2\beta}
&\text{if $\theta=2\beta-1$,}\\
\varepsilon^{-\frac{3(1+\theta)+2\beta}{2(1+\theta)\beta}}
|\ln{\varepsilon}|^\frac{1+\theta+\beta}{2\beta}
&\text{if $\theta<2\beta-1$,}
\end{cases}
\quad\text{as}\quad
\varepsilon\downarrow0,
\end{equation*}
where $\widehat{C}_2>0$ is the constant on the right hand side of \eqref{eq:cost=O(eps):bis}.
This complexity is minimized for $\beta=\theta=1$, in which case
\begin{equation*}
\inf_{\beta,\theta}
\Cost_\text{adMLSA}
\leq \widehat{C}_2\varepsilon^{-2}
|\ln{\varepsilon}|^\frac72,
\quad\text{as}\quad
\varepsilon\downarrow0.
\end{equation*}

\item\label{thm:amlsa:complexity:indicator:func-ii:N}
If \eqref{assump:unif:lipschitz:integrability:conditional:cdf:bias} holds,
then, setting $L$ as in \eqref{selection:number:level:adMLSA} and
\begin{equation*}
N_\ell
=\bigg\lceil(\bar{C}\gamma_1)^\frac1\beta
\varepsilon^{-\frac2\beta}
\bigg(\sum_{\ell'=0}^Lh_{\ell'}^{-\frac{2\beta-(1+\theta)}{2(1+\beta)}}\bigg)^\frac1\beta
h_\ell^\frac{3+\theta}{2(1+\beta)}\bigg\rceil,
\quad\ell\in[\![0,L]\!],
\end{equation*}
gives a global $L^2(\mathbb{P})$ error for adMLSA of order $\varepsilon$ as $\varepsilon\downarrow0$.
The corresponding average complexity of
\begin{equation*}
\Cost_\text{adMLSA}
\leq\widehat{C}_3
\begin{cases}
\varepsilon^{-\frac2\beta}
|\ln{\varepsilon}|^\frac{1+\theta+\beta}{2\beta}
&\text{if $\theta>2\beta-1$,}\\
\varepsilon^{-\frac2\beta}
|\ln{\varepsilon}|^\frac{3(1+\beta)+\theta}{2\beta}
&\text{if $\theta=2\beta-1$,}\\
\varepsilon^{-\frac{3(1+\theta)+2\beta}{2(1+\theta)\beta}}
|\ln{\varepsilon}|^\frac{1+\theta+\beta}{2\beta}
&\text{if $\theta<2\beta-1$,}
\end{cases}
\quad\text{as}\quad
\varepsilon\downarrow0,
\end{equation*}
where $\widehat{C}_3>0$ is the constant on the right hand side of \eqref{eq:cost=O(eps):ter}.
This complexity is minimized for $\beta=\theta=1$, whereby
\begin{equation*}
\inf_{\beta,\theta}
\Cost_{\text{adMLSA}}
\leq\widehat{C}_3
\varepsilon^{-2}
|\ln{\varepsilon}|^\frac52,
\quad\text{as}\quad
\varepsilon\downarrow0.
\end{equation*}
\end{enumerate}
\end{theorem}

\begin{remark}
\begin{enumerate}[(i)]
\item
Within the frameworks of Theorems~\ref{thm:amlsa:complexity}(\ref{thm:amlsa:complexity:indicator:func-ib:N})-(\ref{thm:amlsa:complexity:indicator:func-ii:N}), adMLSA scores a significant performance gain over MLSA and retrieves the canonical multilevel performance \cite{Gil08,Fri16} of order $\varepsilon^{-2}$, up to a logarithmic factor.

\item
As for the framework of Theorem~\ref{thm:amlsa:complexity}(\ref{thm:amlsa:complexity:indicator:func-ia:N}), it is faster than its MLSA counterpart for any $p_\star>2$.
Better performance is expected by taking $p_\star$ sufficiently large, recalling that a large class of portfolio payoffs are $L^{p_\star}(\mathbb{P})$-integrable for all $p_\star>2$.
Note that, asymptotically as $p_\star\uparrow\infty$, $\Cost_{\text{adMLSA}}=\OO(\varepsilon^{-2-\iota})$, for all $\iota>0$, while $\Cost_{\text{MLSA}}=\OO(\varepsilon^{-\frac52})$, retrieving the same overperformance reported above.
\end{enumerate}
\end{remark}

The table below compares MLSA and adMLSA.

\begin{table}[H]
\centerline{
    \begin{tabular}{!{\vrule width 1pt}c!{\vrule width 1pt}c|c!{\vrule width 1pt}}
         \specialrule{0.1em}{0em}{0em}
         Algorithm
         & MLSA & adMLSA \\
         \specialrule{0.1em}{0em}{0em}
         \multirow{3}{*}{$L^2(\mathbb{P})$-control}
         & $\OO(h_L^2+\sum_{\ell=0}^L\gamma_{N_\ell}\epsilon(h_\ell))$
         & $\OO(h_L^{2(1+\theta)}+\sum_{\ell=0}^L\widetilde\gamma_{N_\ell}^\ell\epsilon(h_\ell)^{1+\theta})$ \\
         & (Decomposition \eqref{mlsa:error:global},
         & (Decomposition \eqref{amlsa:global:error:var},\\
         & Lemma~\ref{lmm:error}(\ref{lmm:error:weak}) and Theorem~\ref{thm:ml-variance-cv})
         & Lemma~\ref{lmm:error}(\ref{lmm:error:weak}) and Theorem~\ref{thm:amlsa:L2})\\
         \hline
         \multirow{2}{*}[-0.25em]{Number of levels}
         & $L\sim\dfrac{|\ln{\varepsilon}|}{\ln{M}}$
         & $L\sim\dfrac{|\ln{\varepsilon}|}{(1+\theta)\ln{M}}$\\
         & (Proposition~\ref{prp:mlsa:complexity}(\ref{prp:mlsa:complexity:levels}))
         & (Proposition~\ref{prp:selection:number:level:adMLSA})\\
         \hline
         \multirow{3}{*}[-2em]{Best complexity}
         & $\Cost_\text{\tiny\rm MLSA}=$
         & $\Cost_\text{\tiny\rm adMLSA}=$\\
         & $\begin{cases}
   \OO(\varepsilon^{-3+\frac{p_\star}{2(1+p_\star)}})
      &\text{if \eqref{assumption:finite:Lp:moment} holds,}\\
   \OO(\varepsilon^{-\frac52}|\ln{\varepsilon}|^\frac12)
      &\text{if \eqref{assumption:conditional:gaussian:concentration} holds,}\\
   \OO(\varepsilon^{-\frac52})
      &\text{if \eqref{assump:unif:lipschitz:integrability:conditional:cdf} holds.}
\end{cases}$
         & $\begin{cases}
   \OO(\varepsilon^{-2-\frac{3p_\star+2}{2p_\star(p_\star+1)}-\iota}),
   \forall\iota>0,
      &\text{if \eqref{assumption:finite:Lp:moment:bias} holds,}\\
   \OO(\varepsilon^{-2}|\ln{\varepsilon}|^\frac72)
      &\text{if \eqref{assumption:conditional:gaussian:concentration:bias} holds,}\\
    \OO(\varepsilon^{-2}|\ln{\varepsilon}|^\frac52)
      &\text{if \eqref{assump:unif:lipschitz:integrability:conditional:cdf:bias} holds.}
\end{cases}$ \\
         & (Proposition~\ref{prp:mlsa:complexity}(\ref{prp:mlsa:complexity:iii}))
         & (Theorem~\ref{thm:amlsa:complexity}) \\
         \specialrule{0.1em}{0em}{0em}
    \end{tabular}
}
    \caption{Comparison of MLSA and adMLSA.
    $\varepsilon>0$ designates a prescribed accuracy.}
    \label{tbl:cpx}
\end{table}

\section{Heuristics}
\label{heuristics}

\subsection{Refinement Parameters}
In view of Theorem~\ref{thm:cost:adansa} and Theorem~\ref{thm:amlsa:complexity}, to compute the VaR efficiently, we need to set $\theta=\frac{p_\star-2}{p_\star+2}$ within the framework \eqref{assumption:finite:Lp:moment:bias} and $\theta=1$ within the frameworks \eqref{assumption:conditional:gaussian:concentration:bias} and \eqref{assump:unif:lipschitz:integrability:conditional:cdf:bias}. Note that $\theta\approx1$ for $p_\star$ large enough within the framework \eqref{assumption:finite:Lp:moment:bias}.

The choice $r=1+\frac1\theta$ together with the assumption $\theta\approx1$ lead to thresholds in~(\ref{eq:psi}) of the form
\begin{equation*}
\psi^{\ell,k}_n=h_{\ell+k}^\frac\theta{1+\theta}w_n\approx h_{\ell+k}^\frac12w_n,
\end{equation*}
where $w_n$ is the saturation factor.
Complementing \cite{GHS23}, under this choice of $r$ and $\theta$, our strategy can be seen as conducting consecutive Student t-tests on the null hypotheses ``$X^{h_{\ell+k}}=\xi$'', $0\leq k\leq\lceil\theta\ell\rceil$, until the first rejection.

\subsection{Refinement Strategy}
Let us fix some level $\ell$ and iteration number $n$, and denote $\xi^\ell$ and $\xi^{\ell-1}$ the corresponding fine and coarse iterates. For simplicity, we omit $n$ from our notation.
To increase efficiency, we propose to only refine the fine sample $X^{h_\ell}$ to $\widetilde{X}^{h_\ell}=X^{h_{\ell+\eta^\ell(\xi^\ell)}}$, then discard the $Z$ simulations indexed from $KM^{\ell-1+\eta^\ell(\xi^\ell)}+1$ onward to recover the coarse sample $\widetilde{X}^{h_{\ell-1}}$:
\begin{align*}
\widetilde{X}^{h_{\ell-1}}
&=\frac1{KM^{\ell-1+\eta^\ell(\xi^\ell)}}\sum_{k=1}^{KM^{\ell-1}+\eta^\ell(\xi^\ell)}\varphi(Y,Z^k),\\
\widetilde{X}^{h_\ell}
&=\frac1M\widetilde{X}^{h_{\ell-1}}+\frac1{KM^{\ell+\eta^\ell(\xi^\ell)}}\sum_{k=KM^{\ell-1+\eta^\ell(\xi^\ell)+1}}^{KM^{\ell+\eta^\ell(\xi^\ell)}}\varphi(Y,Z^k),
\end{align*}
where $(Z^k)\stackrel{\text{\tiny\rm i.i.d.}}\sim Z$ are independent of $Y$.
This amounts to setting $\eta^{\ell-1}(\xi^{\ell-1})=(\eta^\ell(\xi^\ell)-2)^+$.
While this strategy for $\eta^{\ell-1}$ is not centered on $\xi^{\ell-1}$, numerical trials have shown it to produce better results.

\subsection{Confidence Constant}
In the following, we tune the confidence constant $C_\mathrm{ad}$ appearing in \eqref{eq:psi} on a grid.
Remark~\ref{rmk:refine}(\ref{rmk:C}) however suggests a different treatment for this constant.

Denote $\sigma=\sqrt{\Var(\varphi(Y,Z)|Y)}$ the sampling standard deviation, and $\sigma^h$, $h=\frac1K\in\mathcal{H}$, its empirical approximation given by
\begin{equation}
\sigma^h
=\bigg(\frac1K\sum_{k=1}^K\varphi(Y,Z^k)^2-(X^h)^2\bigg)^\frac12.
\end{equation}
An alternate refinement strategy consists in incorporating an estimation of $C_\mathrm{ad}$ in \eqref{eq:psi} with the quantity $C_p\sigma^{h_{\ell+k}}$ at the refinement step $k$, where $C_p$ is set to $3$ to retrieve a $99\%$-confidence on the closeness between $X^{h_{\ell+k}}$ and $X^0$.

We term $\sigma$-adNSA and $\sigma$-adMLSA the adNSA and adMLSA versions where $C_\mathrm{ad}$ is estimated by $(C_p\sigma^{h_{\ell+k}})_{0\leq k\leq\lceil\theta\ell\rceil}$. For comparative purposes, these versions are run as well and their performances are reported and discussed in subsequent analyses.

\subsection{Unsaturated Refinement}
To further reduce complexity, we propose to look into a desaturated variant of the refinement strategy \eqref{eq:eta}, whereby the saturation factor is set to $1$:
\begin{equation}\label{eq:unsat}
\eta^\ell(\xi)
=\lceil\theta\ell\rceil\wedge\min{
    \{k\in\intl0,\lceil\theta\ell\rceil\intr:
    |X^{h_{\ell+k}}-\xi|
    \geq C_\mathrm{ad}\psi^{\ell,k}\}
},
\quad\ell\in\mathbb{N}^*,
\quad\xi\in\mathbb{R},
\end{equation}
with the convention $\min\varnothing=+\infty$, where $C_\mathrm{ad}>0$, $r>1$, $0<\theta\leq1$, and
\begin{equation*}
\psi^{\ell,k}=h_{\theta\ell(r-1)+k}^\frac1r,
\quad k\in\intl0,\lceil\theta\ell\rceil\intr,
\quad\ell\in\mathbb{N}^*.
\end{equation*}

To be able to incorporate this refinement strategy in an adaptive multilevel scheme, we need the corresponding optimal iteration amounts.
These unfortunately cannot be easily obtained for this heuristic.
Given that the strategy \eqref{eq:unsat} resembles that of \cite[Eq (3.2)]{HST22} for adaptive MLMC, one could simply consider the iteration amounts derived in \cite{HST22}.
The complexity theorems \cite[Theorem~3.10 and 3.11]{HST22} content with pointing to the best complexities, without providing the corresponding iteration amounts.
These however may be found in the proof of the earlier work of \cite[Theorem~3.1]{Gil08}.
But adopting these iteration amounts remains problematic for a couple of reasons.
First of all, as pointed out in Section~\ref{sec:amlsa}, the saturation factor incorporated in the refinement strategy \eqref{eq:eta} has been worked out from the bias and variance upper estimates derived in Lemma~\ref{lmm:bias} and Proposition~\ref{prp:error:statistical:bis}.
Substituting this factor with one, which is tantamount to unsaturation, leads to upper estimates that do not asymptotically vanish with large iteration amounts.
Besides, the frameworks used by \cite{Gil08} or \cite{HST22} are different from ours in \eqref{assumption:finite:Lp:moment:bias}, \eqref{assumption:conditional:gaussian:concentration:bias} and \eqref{assump:unif:lipschitz:integrability:conditional:cdf:bias}.
Considering all of these concerns, we opt for the prudent choice of using our own optimal iteration amounts calculated in Theorem~\ref{thm:amlsa:complexity} as obtained by factoring in a saturation effect.
In other words, the unsaturated refinement strategy \eqref{eq:unsat} is simply treated as a heuristic in the pursuit of faster execution.

The adaptive NSA and MLSA algorithms corresponding to the heuristic \eqref{eq:unsat}, that we dub u-adNSA and u-adMLSA, will therefore be run with the same iteration amounts $n$ and $(N_\ell)_{0\leq\ell\leq L}$ delineated in Theorems~\ref{thm:cost:adansa} and~\ref{thm:amlsa:complexity}.

\section{Financial Case Studies}
\label{sec:num}

In the ensuing numerical studies, we illustrate the closure of the performance gap between MLSA and SA, which is made possible by using our adaptive refinement strategy.
To that end, we revisit the VaR use cases already handled by SA, NSA and MLSA in \cite[Sections~5 \&~6]{CFL23}.
The implementations for the case studies can be found at \href{https://github.com/azarlouzi/ada_mlsa}{\tt github.com/azarlouzi/ada\_mlsa}.

\subsection{European Option}
\label{ssec:euption}

We succinctly recall here the setting of \cite[Section~5]{CFL23}.
We refer to the developments therein for a rigorous derivation of the ensuing statements.

Consider a European option of maturity $T=1$ and payoff $\mathbb{R}\ni x\mapsto-x^2$, on an underlying asset following a standard Brownian motion dynamic $(W_t)_{0\leq t\leq1}$.
The risk-free rate is null and the pricing is performed under $\mathbb{P}$.
The option's value at time $t\in[0,1]$ is
\begin{equation*}
V_t=\mathbb{E}[-W_1^2|W_t],
\end{equation*}
and its associated loss at a horizon $\tau\in(0,1)$ is
\begin{equation*}
X^0=V_0-V_\tau.
\end{equation*}
We are interested in retrieving the VaR $\xi^0_\star$ of this loss at some confidence level $\alpha\in(0,1)$.

\subsubsection{Analytical and Simulation Formulas}
Let $\varphi:\mathbb{R}^2\to\mathbb{R}$,
\begin{equation*}
\varphi(y,z)
:=-(\sqrt{\tau}y+\sqrt{1-\tau}z)^2,
\quad
y,z\in\mathbb{R}.
\end{equation*}

On the one hand,
\begin{equation*}
X^0
\stackrel{\mathcal{L}}{=}
-1-\mathbb{E}[\varphi(Y,Z)|Y]
=\tau(Y^2-1),
\end{equation*}
where $Y$ and $Z$ are independent and of law $\mathcal{N}(0,1)$.
$X^0$ can thus be simulated exactly, hence the benchmarking unbiased SA scheme \cite{BFP09} is applicable to estimate the VaR.

On the other hand, for a bias parameter $h=\frac1K\in\mathcal{H}$, $X^0$ can be approximated by
\begin{equation*}
X^h=-1-\frac1K\sum_{k=1}^K\varphi(Y,Z^k),
\end{equation*}
where $Y,Z^1,\dots,Z^K\stackrel[]{\text{\rm\tiny i.i.d.}}{\sim}\mathcal{N}(0,1)$.
We can then apply NSA, MLSA, adNSA, $\sigma$-adNSA, u-adNSA, adMLSA, $\sigma$-adMLSA and u-adMLSA on this basis to approximate the VaR.

Finally, the VaR $\xi^0_\star$ at level $\alpha$ has an analytical form:
\begin{equation}\label{trueVaR}
\xi^0_\star=\tau\bigg\{\bigg(F^{-1}\bigg(\frac{1-\alpha}2\bigg)\bigg)^2-1\bigg\},
\end{equation}
where $F$ is the standard Gaussian cdf.
Its evaluation will help assess the estimation errors of the aforementioned SA schemes.

\subsubsection{Numerical Results}
We conduct below a performance comparison of the different SA schemes discussed in this paper.
We set the confidence level to $\alpha=97.5\%$ and the time horizon to $\tau=0.5$, which yields $\xi^0_\star\approx2.012$ via \eqref{trueVaR}.

To minimize complexity, all algorithms are run with $\beta=1$ and their respective theoretical optimal iteration amounts.
SA, NSA, adNSA and $\sigma$-adNSA are run with $\gamma_n=\frac1{100+n}$.
MLSA, adMLSA, $\sigma$-adMLSA and u-adMLSA are run with $M=2$ under the framework \eqref{assumption:finite:Lp:moment:bias}, with $\delta=0.05$ and $p_\star=11$. The parameter $p_\star$ has been fine-tuned on a grid of values to minimize RMSE and execution time.
We set $\theta=\frac{p_\star-2}{p_\star+2}$ and $r=1+\frac1\theta$ for the adaptive schemes, as recommended in Section~\ref{heuristics}.
Identical $h_0$ parametrization was applied to adNSA, $\sigma$-adNSA, MLSA, adMLSA, $\sigma$-adMLSA and u-adMLSA, and identical $\gamma_n$ parametrization was applied to MLSA, adMLSA, $\sigma$-adMLSA and u-adMLSA.
Further parametrizations of the adaptive and multilevel schemes, obtained via a grid search, are described in Table~\ref{tbl:option}.
The subsequent level amounts computed via \eqref{selection:number:level:MLSA}, \eqref{selection:number:level:adNSA} and \eqref{selection:number:level:adMLSA} are reported in columns $L$, $\ell_\mathrm{ad}$ and $L_\mathrm{ad}$.
For adNSA and adMLSA, the iteration amounts scaling factor $C$ (e.g.~$C$ represents $C^{1,\delta}_1$ for adMLSA as per \eqref{eq:Nl:i}) and the confidence constant $C_\mathrm{ad}$ were jointly optimized by grid search, which led to the choices $C=2$ and $C_\mathrm{ad}=0.5$ for adNSA and $C=2$ and $C_\mathrm{ad}=12$ for adMLSA and u-adMLSA.
As for $\sigma$-adNSA and $\sigma$-adMLSA, we use the critical value $C_p=3$ with scaling factor $C=2$.

\begin{table}[H]
\centerline{
\begin{tabular}{|c|c|c|c|c|c|c|}
\hline
$\varepsilon$
& $h_0$ & $\ell_\mathrm{ad}$ \\
\specialrule{0.1em}{0em}{0em}
$\frac{1}{32}$ & $\frac{1}{16}$ & $1$ \\
\hline
$\frac{1}{64}$ & $\frac{1}{32}$ & $1$ \\
\hline
$\frac{1}{128}$ & $\frac{1}{32}$ & $2$ \\
\hline
$\frac{1}{256}$ & $\frac{1}{32}$ & $2$ \\
\hline
$\frac{1}{512}$ & $\frac{1}{32}$ & $3$ \\
\hline
\end{tabular}
\hfill
\begin{tabular}{|c|c|c|c|c|c|c|}
\hline
$\varepsilon$ & $h_0$ & $L$ & $\gamma_n$ \\
\specialrule{0.1em}{0em}{0em}
$\frac{1}{32}$ & $\frac{1}{16}$ & $1$ & $\frac{2}{2.5\times10^3+n}$ \\
\hline
$\frac{1}{64}$ & $\frac{1}{32}$ & $1$ & $\frac{2}{4\times10^3+n}$ \\
\hline
$\frac{1}{128}$
& $\frac{1}{32}$ & $2$ & $\frac{0.75}{9\times10^3+n}$ \\
\hline
$\frac{1}{256}$ & $\frac{1}{32}$ & $3$ & $\frac{0.25}{10^4+n}$ \\
\hline
$\frac{1}{512}$ & $\frac{1}{32}$ & $4$ & $\frac{0.09}{10^4+n}$ \\
\hline
\end{tabular}
\hfill
\begin{tabular}{|c|c|c|c|c|c|c|}
\hline
$\varepsilon$ & $h_0$ & $L_\mathrm{ad}$ & $\gamma_n$ \\
\specialrule{0.1em}{0em}{0em}
$\frac{1}{32}$ & $\frac{1}{16}$ & $1$ & $\frac{2}{2.5\times10^3+n}$ \\
\hline
$\frac{1}{64}$ & $\frac{1}{32}$ & $1$ & $\frac{2}{4\times10^3+n}$ \\
\hline
$\frac{1}{128}$ & $\frac{1}{32}$ & $2$ & $\frac{0.75}{9\times10^3+n}$ \\
\hline
$\frac{1}{256}$ & $\frac{1}{32}$ & $2$ & $\frac{0.25}{10^4+n}$ \\
\hline
$\frac{1}{512}$ & $\frac{1}{32}$ & $3$ & $\frac{0.09}{10^4+n}$ \\
\hline
\end{tabular}
}
\caption{Parametrizations of adNSA and $\sigma$-adNSA (left), MLSA (center) and adMLSA, $\sigma$-adMLSA and u-adMLSA (right), by prescribed accuracy.}
\label{tbl:option}
\end{table}

Root-mean-square errors (RMSEs) relative to the true $\xi^0_\star$, as well as corresponding average execution times over $200$ runs, for a prescribed accuracy ranging in $\{\frac1{7},\frac1{13},\frac1{26},\frac1{52},\frac1{103}\}$, are graphed for each algorithm on a logarithmic scale in Figure~\ref{fig:euption}.
Recalling that $\xi^0_\star\approx2.012$, these prescribed accuracies correspond to accuracy rates ranging from $\frac{1/7}{2.012}\approx 7.10\%$ down to $\frac{1/103}{2.012}\approx0.48\%$. Before feeding them into the different SA algorithms, we adjust each accuracy $\varepsilon$ by a scaling factor of $C'=\frac15$ into $C'\varepsilon$ ($C'$ represents the constant on the right hand side of \eqref{L2:norm:adaML:VaR}).

Figure~\ref{fig:euption:eps} plots the average execution times against the prescribed accuracies themselves to illustrate the achieved complexities. 
The slopes fitted on these curves, shown in dashed lines in Figures~\ref{fig:euption} and~\ref{fig:euption:eps}, are reported in Table~\ref{tbl:option:slopes}.

\begin{figure}[H]
\includegraphics[width=\textwidth]{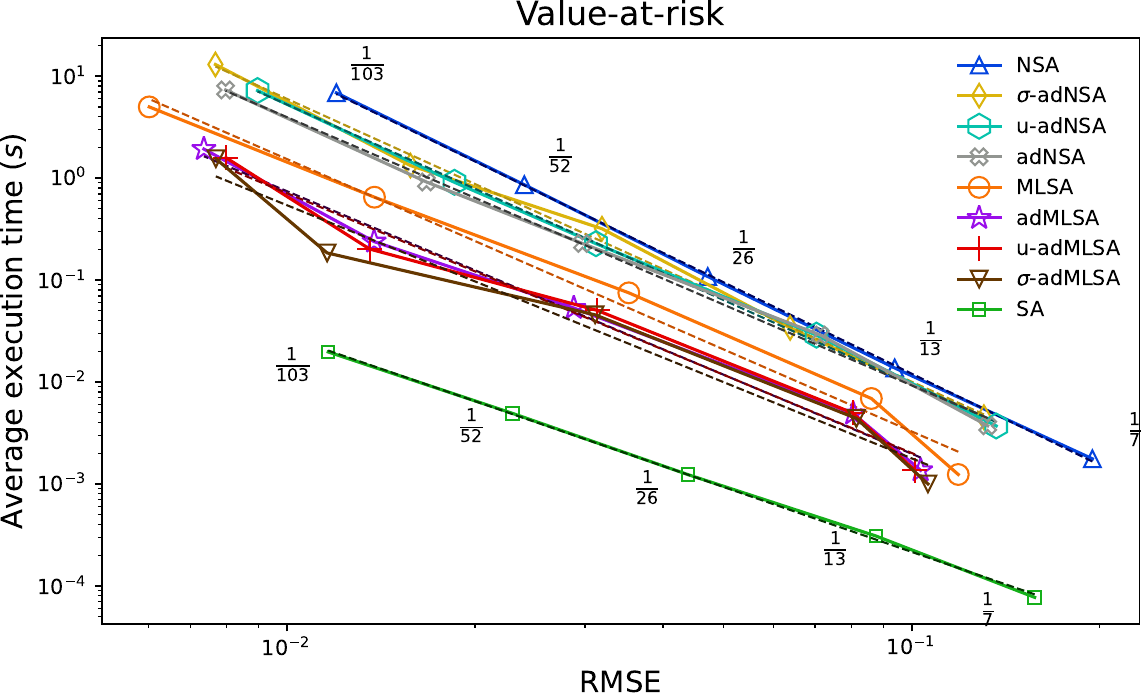}
\caption{Performance comparison of the different SA schemes.}
\label{fig:euption}
\end{figure}

\begin{figure}[H]
\includegraphics[width=\textwidth]{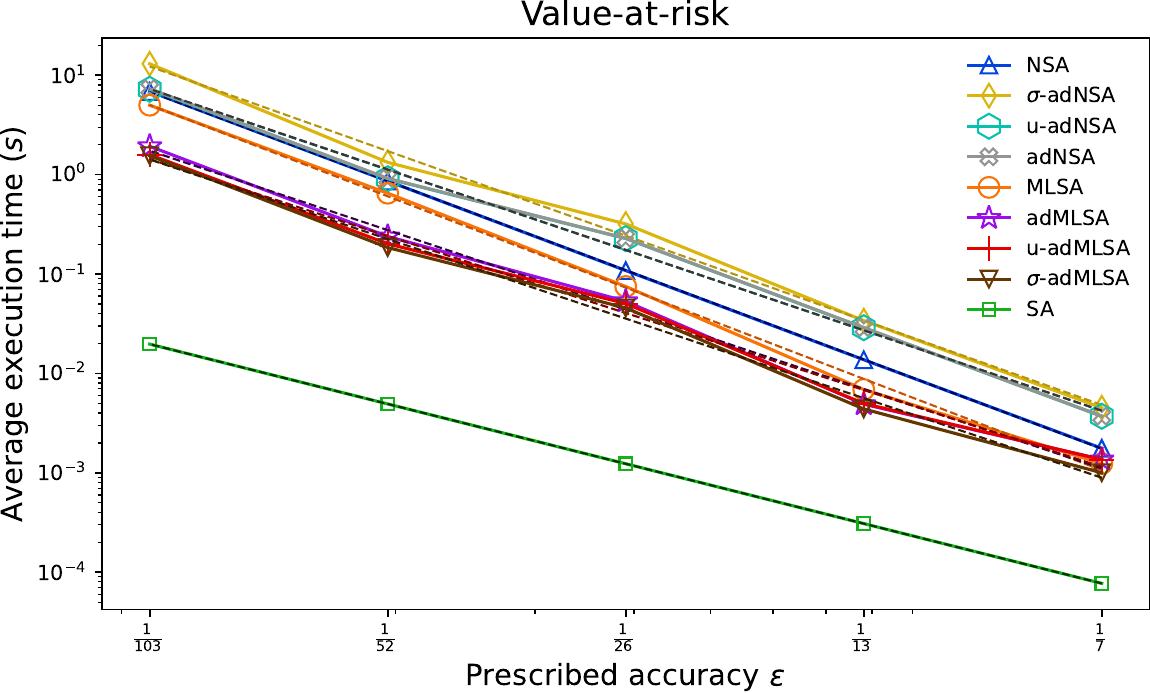}
\caption{Complexity comparison of the different SA schemes.}
\label{fig:euption:eps}
\end{figure}

\begin{table}[H]
\centerline{
\begin{tabular}{|c!{\vrule width 1pt}c|c|c|c|c|c|c|c|c|c|}
\hline
SA scheme
& NSA
& $\sigma$-adNSA
& u-adNSA
& adNSA
& MLSA 
& adMLSA
& u-adMLSA
& $\sigma$-adMLSA
& SA \\
\specialrule{0.1em}{0em}{0em}
RMSE
& $-2.98$
& $-2.78$
& $-2.74$
& $-2.63$
& $-2.67$
& $-2.58$
& $-2.55$
& $-2.48$
& $-2.12$
\\\hline
$\varepsilon$
& $-2.98$
& $-2.83$
& $-2.68$
& $-2.69$
& $-3.05$
& $-2.65$
& $-2.57$
& $-2.66$
& $-2.00$
\\\hline
\end{tabular}
}
\caption{Reported slopes from Figures~\ref{fig:euption} and~\ref{fig:euption:eps}.}
\label{tbl:option:slopes}
\end{table}

\cite[Section~5.3]{CFL23} already provides a thorough discussion on SA, NSA and MLSA.
We retain here that MLSA achieves a partial gain on the gap of performance separating SA and NSA.

The novelties here are the adaptive schemes.
On Figure~\ref{fig:euption},
adNSA, $\sigma$-adNSA and u-adNSA achieve comparable performances, outperforming NSA by a margin that seems to widen for smaller accuracies.
adNSA seems to be slightly outperforming both $\sigma$-adNSA and u-adNSA.
The overperformance over $\sigma$-adNSA may be attributed to the overhead computation performed by the latter algorithm to recompute the confidence constant $C_\mathrm{ad}$.
The overperformance over u-adNSA can be explained by the desaturated refinement strategy used by the latter algorithm.
adMLSA, $\sigma$-adMLSA and u-adMLSA show to be outperforming MLSA:
for a target RMSE of order $10^{-2}$, adMLSA, $\sigma$-adMLSA and u-adMLSA run on average in $1\mathrm{s}$ while MLSA runs in $2\mathrm{s}$, achieving a $2$-fold speed-up over their non-adaptive counterpart.
$\sigma$-adMLSA seems to be slightly outperforming adMLSA, which can be explained by the precision brought in by the recomputation of the standard deviation $\sigma$, in spite of the entailed computational overhead.
The advantage of adMLSA over $\sigma$-adMLSA remains that, once calibrated on a few accuracies, the confidence constant $C_\mathrm{ad}$ for adMLSA needs not be readjusted, which eliminates all additional upstream fine-tuning.
u-adMLSA seems also to be slightly outperforming adMLSA, which can be attributed to the computational gain achieved by the simpler unsaturated refinement strategy, despite the implied refinement imprecision.

Eventually, an examination of the fitted slopes highlights that adNSA achieves the theoretical complexity exponent predicted by Theorem~\ref{thm:cost:adansa}.
It also demonstrates that adMLSA approaches the quadratic complexity anticipated by Theorem~\ref{thm:amlsa:complexity}(\ref{thm:amlsa:complexity:indicator:func-ia:N}) for $p_\star$ large.
%The performances of these schemes match that of SA that assumes the exact simulatability of the true loss $X^0$.
\\

Note that, in a comparable numerical setting, \cite{HST22} instead plots the quantity ``Average execution time $\times$ $\mathrm{RMSE}^2$'' versus the $\mathrm{RMSE}$ on a logarithmic scale. This choice is intended to illustrate the underlying complexity results, recalling that the average execution time serves as a proxy for the computational cost, while the RMSE represents the prescribed accuracy. Under a quasi-quadratic complexity regime ($\Cost\propto\varepsilon^{-2-\delta}$, with $\delta>0$ small), one expects a log-linear relationship between these quantities, with a mildly negative slope ($\ln(\Cost\varepsilon^2)\propto-\delta\ln\varepsilon$).
The main drawback of this graphical representation is that the quantity ``Average execution time $\times$ $\mathrm{RMSE}^2$'' lacks a clear interpretation, as it does not directly convey a tangible performance gain by directly comparing execution times for a given level of accuracy.

\subsection{Interest Rate Swap}

The following setting is recapitulated from \cite[Section~6]{CFL23}.

We consider a swap, of strike $\bar{K}$, maturity $T$ and nominal $\bar{N}$, with each leg being worth $1$ at inception, and issued at par on some underlying interest (or FX) rate $(S_t)_{0\leq t\leq T}$ that follows a Black-Scholes model with risk neutral drift $\bar\kappa$ and constant volatility $\bar\sigma$.
At the coupon dates $0<T_1<\dots<T_d=T$, the swap remunerates the cash flows $\Delta T_i(S_{T_{i-1}}-\bar{K})$, where $\Delta T_i=T_i-T_{i-1}$, $T_0=0$.
The risk-free rate is $\bar{r}$ and the risk neutral probability measure is $\mathbb{P}$.

For $t\in[0,T]$, let $i_t$ be the integer such that $t\in[T_{i_t-1},T_{i_t})$ if $t\in[0,T)$, and $+\infty$ otherwise.
The fair value of the swap at time $t\in[0,T]$ is
\begin{equation*}
P_t
=\bar{N}\mathbb{E}\bigg[\sum_{i=i_t}^d\e^{-\bar{r}(T_i-t)}\Delta T_i(S_{T_{i-1}}-\bar{K})\bigg|S_t\bigg].
\end{equation*}
The loss on a short position on the swap at a time horizon $\tau\in(0,T_1)$ is
\begin{equation*}
X^0=\e^{-\bar{r}\tau}P_\tau.
\end{equation*}
We are interested in computing the VaR $\xi^0_\star$ of this loss at some confidence level $\alpha\in(0,1)$.

\subsubsection{Analytical and Simulation Formulas}
On the one hand,
\begin{equation*}
X^0\stackrel{\mathcal{L}}{=}\bar{N} AS_0\bigg(\exp\bigg(-\frac{\bar\sigma^2}2\tau+\bar\sigma\sqrt\tau U\bigg)-1\bigg),
\quad\text{where}\quad
A:=\sum_{i=2}^d\e^{-\bar{r}T_i}\Delta T_i\e^{\bar\kappa T_{i-1}}
\end{equation*}
and $U\sim\mathcal{N}(0,1)$.
This allows simulating $X^0$ exactly, hence the availability of SA to approximate the VaR.

On the other hand, $X^0$ satisfies
\begin{equation}\label{eq:swap:nested:expect}
X^0\stackrel{\mathcal{L}}{=}\mathbb{E}[\varphi(Y,Z)|Y],
\end{equation}
where $Y\in\mathbb{R}$ is independent of $Z=(Z_1,\dots,Z_{d-1})\in\mathbb{R}^{d-1}$, with
\begin{equation*}
\begin{aligned}
Y
&:=\exp\bigg(-\frac{\bar\sigma^2}2\tau+\bar\sigma\sqrt\tau U_0\bigg),\\
Z_1
&:=\exp\bigg(-\frac{\bar\sigma^2}2(T_1-\tau)+\bar\sigma\sqrt{T_1-\tau}U_1\bigg),\\
Z_i
&:=\exp\bigg(-\frac{\bar\sigma^2}2\Delta T_i+\bar\sigma\sqrt{\Delta T_i}U_i\bigg),
&&2\leq i\leq d-1,\\
\varphi(y,z)
&:=\bar{N}S_0\sum_{i=2}^d\e^{-\bar{r}T_i}\Delta T_i\e^{\bar\kappa T_{i-1}}\bigg(y\prod_{j=1}^{i-1}z_j-1\bigg),
&&y\in\mathbb{R},
\;\; z=(z_1,\dots,z_{d-1})\in\mathbb{R}^{d-1},
\end{aligned}
\end{equation*}
and $(U_i)_{0\leq i\leq d-1}\stackrel{\text{\rm\tiny i.i.d.}}{\sim}\mathcal{N}(0,1)$.
The nested Monte Carlo averaging \eqref{eq:Xh} is thus available to approximate \eqref{eq:swap:nested:expect}, hence the applicability of NSA, adNSA, $\sigma$-adNSA, u-adNSA, MLSA, adMLSA, $\sigma$-adMLSA and u-adMLSA to approximate the VaR.

Finally, the VaR $\xi^0_\star$ at level $\alpha$ is available analytically:
\begin{equation}\label{trueVaR:2}
\xi^0_\star=\bar{N}AS_0\bigg(\exp\bigg(F^{-1}(\alpha)\bar\sigma\sqrt\tau-\frac{\bar\sigma^2}2\tau\bigg)-1\bigg),
\end{equation}
where $F$ is the standard Gaussian cdf.
The output of this formula will serve as a benchmark for the outcomes of the aforementioned algorithms.

\subsubsection{Numerical Results}
For the case study, we set $S_0=1\%$, $\bar{r}=2\%$, $\bar\kappa=12\%$, $\bar\sigma=20\%$, $T=1\,\mathrm{year}$, $\Delta T_i=3\,\mathrm{months}$, $\tau=7\,\mathrm{days}$ and $\alpha=85\%$. We use the $30/360$ day count fraction convention.
\eqref{trueVaR:2} yields $\xi^0_\star\approx219.64$.

We run all algorithms at their theoretical optimums, with $\beta=1$ and the corresponding optimal iteration amounts.
The learning rate $\gamma_n=\frac{100}n$ is employed for SA and $\gamma_n=\frac{50}n$ for NSA, adNSA, $\sigma$-adNSA and u-adNSA.
For adNSA, $\sigma$-adNSA, u-adNSA, MLSA, adMLSA, $\sigma$-adMLSA and u-adMLSA, we adopt the framework \eqref{assumption:finite:Lp:moment:bias} with the exponents $p_\star=8$ (fine-tuned on a grid to reduce RMSE and execution time) and $\delta=0.05$, and the geometric step $M=2$.
As suggested in Section~\ref{heuristics}, we set $\theta=\frac{p_\star-2}{p_\star+2}$ and $r=1+\frac1\theta$.
For every prescribed accuracy $\varepsilon\in\{\frac1{2},2,4,8,15\}$, we tune $h_0$ (governing the number of levels $L$ for MLSA, $\ell_\mathrm{ad}$ for adNSA, $\sigma$-adNSA and u-adNSA, and $L_\mathrm{ad}$ for adMLSA, $\sigma$-adMLSA and u-adMLSA) and the learning rate $(\gamma_n)_{n\geq1}$ for MLSA, adMLSA, $\sigma$-adMLSA and u-adMLSA on suitable grids.
Table~\ref{tbl:swap:mlsa} lists these parametrizations by prescribed accuracy.
The iteration amounts scaling factor $C$ and adaptive refinement confidence constant $C_\mathrm{ad}$ are tuned on grids.
We retain $C=2$ for adMLSA, $\sigma$-adMLSA and u-adMLSA, $C_\mathrm{ad}=30$ for adMLSA and u-adMLSA, $C=2$ for adNSA, $\sigma$-adNSA and u-adNSA, and $C_\mathrm{ad}=300$ for adNSA and u-adNSA.
We eventually set the critical value $C_p=3$ for $\sigma$-adNSA and $\sigma$-adMLSA.

\begin{table}[H]
\centerline{
\begin{tabular}{|c|c|c|c|c|c|c|}
\hline
$\varepsilon$ & $h_0$ & $\ell_\mathrm{ad}$ \\
\specialrule{0.1em}{0em}{0em}
$\frac{1}{32}$ & $\frac{1}{8}$ & $2$ \\
\hline
$\frac{1}{64}$ & $\frac{1}{16}$ & $2$ \\
\hline
$\frac{1}{128}$ & $\frac{1}{16}$ & $2$ \\
\hline
$\frac{1}{256}$ & $\frac{1}{16}$ & $3$ \\
\hline
$\frac{1}{512}$ & $\frac{1}{16}$ & $4$ \\
\hline
\end{tabular}
\hfill
\begin{tabular}{|c|c|c|c|c|c|c|}
\hline
$\varepsilon$ & $h_0$ & $L$ & $\gamma_n$ \\
\specialrule{0.1em}{0em}{0em}
$\frac{1}{32}$ & $\frac{1}{8}$ & $2$ & $\frac{6}{10+n}$ \\
\hline
$\frac{1}{64}$ & $\frac{1}{16}$ & $2$ & $\frac{20}{500+n}$ \\
\hline
$\frac{1}{128}$ & $\frac{1}{16}$ & $3$ & $\frac{21}{10^3+n}$ \\
\hline
$\frac{1}{256}$ & $\frac{1}{16}$ & $4$ & $\frac{20}{2\times10^3+n}$ \\
\hline
$\frac{1}{512}$ & $\frac{1}{16}$ & $5$ & $\frac{21}{3\times10^3+n}$ \\
\hline
\end{tabular}
\hfill
\begin{tabular}{|c|c|c|c|c|c|c|}
\hline
$\varepsilon$ & $h_0$ & $L_\mathrm{ad}$ & $\gamma_n$ \\
\specialrule{0.1em}{0em}{0em}
$\frac{1}{32}$ & $\frac{1}{8}$ & $2$ & $\frac{6}{10+n}$ \\
\hline
$\frac{1}{64}$ & $\frac{1}{16}$ & $2$ & $\frac{20}{500+n}$ \\
\hline
$\frac{1}{128}$ & $\frac{1}{16}$ & $2$ & $\frac{21}{10^3+n}$ \\
\hline
$\frac{1}{256}$ & $\frac{1}{16}$ & $3$ & $\frac{20}{2\times10^3+n}$ \\
\hline
$\frac{1}{512}$ & $\frac{1}{16}$ & $4$ & $\frac{21}{3\times10^3+n}$ \\
\hline
\end{tabular}
}
\caption{Parametrizations of adNSA and $\sigma$-adNSA (left), MLSA (center) and adMLSA and $\sigma$-adMLSA (right), by prescribed accuracy.}
\label{tbl:swap:mlsa}
\end{table}

The joint evolution of the RMSE and average execution time over $200$ runs for each SA scheme, for an accuracy $\varepsilon$ looping through $\{\frac1{2},2,4,8,15\}$, are plotted on a logarithmic scale on Figure~\ref{fig:swap}.
Since $\xi^0_\star\approx219.64$, these accuracies coincide with accuracy rates in the tune of $\frac{1/2}{219.64}\approx 6.83\%$ through $ \frac{15}{219.64}\approx 0.23\%$. The prescribed accuracies $\varepsilon$ have been adjusted by a factor of $C'=\frac1{450}$  into $C'\varepsilon$ before using them in the different SA schemes.
Figure~\ref{fig:swap:eps} showcases the average execution times against the prescribed accuracies.
Table~\ref{tbl:swap:slopes} reports the regressed slopes on these curves as depicted in dashed lines on Figures~\ref{fig:swap} and~\ref{fig:swap:eps}.

\begin{figure}[H]
\includegraphics[width=\textwidth]{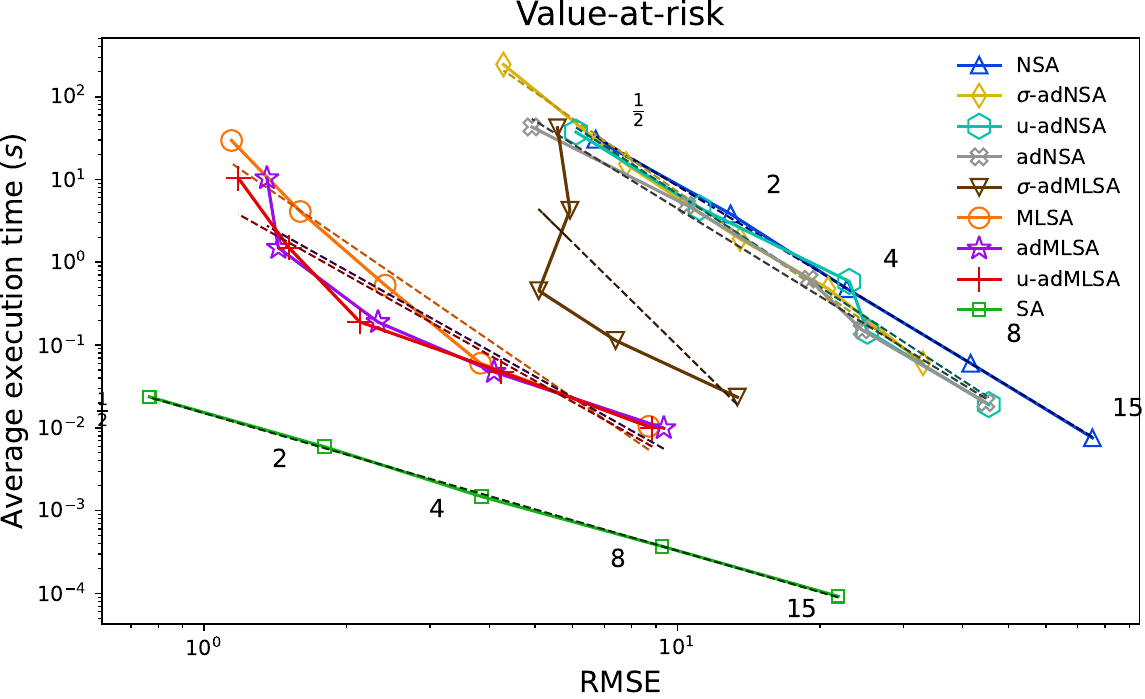}
\caption{Performance comparison of the different SA schemes.}
\label{fig:swap}
\end{figure}

\begin{figure}[H]
\includegraphics[width=\textwidth]{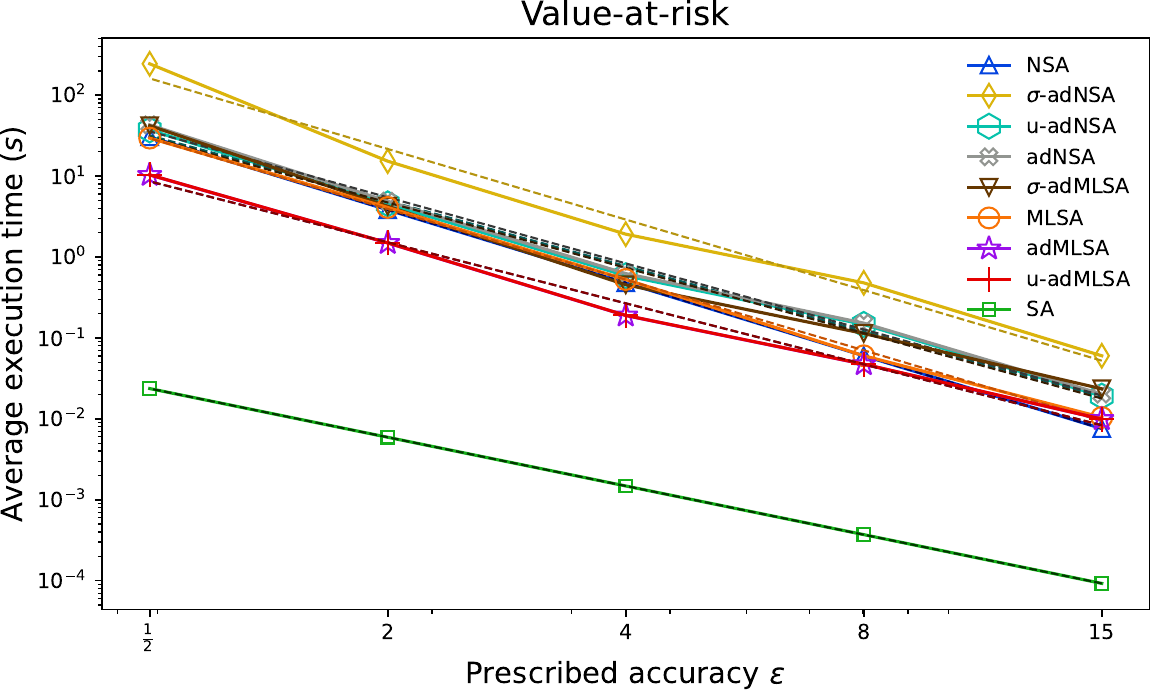}
\caption{Complexity comparison of the different SA schemes.}
\label{fig:swap:eps}
\end{figure}

\begin{table}[H]
\centerline{
\begin{tabular}{|c!{\vrule width 1pt}c|c|c|c|c|c|c|c|c|c|}
\hline
SA scheme
& NSA
& $\sigma$-adNSA
& u-adNSA
& adNSA
& $\sigma$-adMLSA
& MLSA
& adMLSA
& u-adMLSA
& SA \\
\specialrule{0.1em}{0em}{0em}
RMSE
& $-3.46$
& $-3.98$
& $-3.74$
& $-3.52$
& $-5.62$
& $-3.92$
& $-3.21$
& $-3.21$
& $-1.66$
\\\hline
$\varepsilon$
& $-3.00$
& $-2.90$
& $-2.68$
& $-2.71$
& $-2.68$
& $-2.90$
& $-2.50$
& $-2.50$
& $-2.00$
\\\hline
\end{tabular}
}
\caption{Reported slopes from Figures~\ref{fig:swap} and~\ref{fig:swap:eps}.}
\label{tbl:swap:slopes}
\end{table}

We refer to \cite[Section~6.3]{CFL23} for extended comments on SA, NSA and MLSA.

Some similar comments to the previous case study are applicable here.
Concerning the nested schemes, adNSA, $\sigma$-adNSA and u-adNSA score a slight speed-up over NSA.
As for the multilevel schemes, only adMLSA and u-adMLSA score a speed-up with respect to MLSA.
$\sigma$-adMLSA seems to score worse than MLSA.
This may be attributed to the computational overhead coming from the recomputation of the constant $C_\mathrm{ad}$ at each iteration.
The performance margins between adNSA, $\sigma$-adNSA and u-adNSA on the one hand and between adMLSA and u-adMLSA on the other hand remain slim.
%adNSA and adMLSA remain preferable as they do not require recomputing $C_\mathrm{ad}$ once it has been fine-tuned on a couple of prescribed accuracies.

Although further fine-tuning may be required, the fitted slopes already indicate that the adaptive multilevel schemes adMLSA and u-adMLSA are closer than MLSA to the desired quadratic complexity.

\section{Conclusion}

For a prescribed accuracy $\varepsilon>0$, the canonical $\varepsilon^{-2}$ multilevel complexity order is retrieved, up to a logarithmic factor, for an MLSA scheme with a Heaviside-type update function. This is made possible by adopting an adaptive refinement strategy on the samples driving the multilevel algorithm's inner nested schemes.
Our strategy allows us to close of the performance gap existing between the nested MLSA and the unbiased Robbins-Monro schemes for Heaviside-type update functions.
The performance gain achieved by adaptive MLSA compared to its regular counterpart is significant in practice, attaining a two-fold speed-up in certain cases and expected to increase exponentially with smaller prescribed accuracies.
Potential future lines of research could involve applying a Polyak-Ruppert averaging \cite{amlsa} to adMLSA to promote greater numerical stability, or extending our findings on adaptive refinement to triply nested X-value adjustment estimations \cite{GHS23}, or exploring the effect of sampling the risk factors $Y$ and $Z$ from an SDE trajectory \cite{RG12}.

\paragraph{Acknowledgement.}
The research of S.~Cr\'epey has benefited from the support of 
(by alphabetical order): \textit{Chair Capital Markets Tomorrow: Modeling and Computational Issues} under the aegis of the Institut Europlace de Finance, a joint initiative of Laboratoire de Probabilit\'es, Statistique et Modélisation (LPSM) / Université Paris Cit\'e and Cr\'edit Agricole CIB; \textit{Chair Futures of Quantitative Finance}, a partnership between LPSM at Université Paris Cit\'e, CERMICS at École nationale des ponts et chauss\'ees, and BNP Paribas Global Markets; \textit{Chair Stress Test, RISK Management and Financial Steering}, led by the French Ecole polytechnique and its Foundation and sponsored by BNP Paribas.

The research of N.~Frikha has benefited from the support of the Institut Europlace de Finance.

\appendix

\section{Auxiliary Results}
\label{aux}

Similarly to \cite[Sect.~2.1]{CG21} and \cite[Appendix~C]{CFL23}, we define, for $\mu\in\mathbb{R}_+$, $h\in\overline{\mathcal{H}}$ and $q\in\mathbb{N}^*$, the Lyapunov function
$\mathcal{L}^\mu_{h,q}:\mathbb{R}\to\mathbb{R}_+$ given by
\begin{equation}\label{eq:Lq}
\mathcal{L}^\mu_{h,q}(\xi)=\big(V_h(\xi)-V_h(\xi^h_\star)\big)^q\exp\Big(\mu\big(V_h(\xi)-V_h(\xi^h_\star)\big)\Big),
\quad\xi\in\mathbb{R}.
\end{equation}
The following lemma states some important properties of $\mathcal{L}^\mu_{h,q}$.

\begin{lemma}[{\cite[Lemma~C.1]{CFL23}}]\label{lmm:lyapunov}
Denote $k_\alpha=1\vee\frac{\alpha}{1-\alpha}$, $\mu_{h,q}= q\|V_h''\|_\infty$, and $\bar{\mathcal{L}}_{h,q}=\mathcal{L}^{\mu_{h,q}}_{h,q}$, $h\in\overline{\mathcal{H}}$, $q\in\mathbb{N}^*$.
Under Assumption \ref{asp:misc}, for all $\mu \in\mathbb{R}_+$, $h\in\overline{\mathcal{H}}$ and $q\in\mathbb{N}^*$,
\begin{enumerate}[(i)]
\item\label{lmm:lyapunov-i}
$\mathcal{L}^\mu_{h,q}$ is twice continuously differentiable on $\mathbb{R}$ and
\begin{equation}
\label{relation:first:deriv:lyapunov}
(\mathcal{L}^\mu_{h,q})'(\xi)=qV_h'(\xi) \mathcal{L}^\mu_{h,q-1}(\xi)+\mu V_{h}'(\xi) \mathcal{L}^\mu_{h,q}(\xi),
\quad\xi\in\mathbb{R}.
\end{equation}
\item\label{lmm:lyapunov-i-bis}
for all $\xi\in\mathbb{R}$,
\begin{equation*}
\bar{\mathcal{L}}_{h,q}(\xi)\leq k_\alpha^q|\xi-\xi^h_\star|^q\exp\bigg(\frac{qk_\alpha}{1-\alpha} \sup_{h'\in\overline{\mathcal{H}}}\|f_{X^{h'}}\|_\infty |\xi-\xi^h_\star|\bigg).
\end{equation*}
\item\label{lmm:lyapunov-ii}
for all $\xi\in\mathbb{R}$,
\begin{equation*}
V_h'(\xi)(\mathcal{L}^\mu_{h,q})'(\xi)\geq\lambda^\mu_{h,q}\mathcal{L}^\mu_{h,q}(\xi),
\quad\text{where}\quad
\lambda^\mu_{h,q}:=\frac38qV_h''(\xi^h_\star)\wedge \mu \frac{V_h''(\xi^h_\star)^4}{4[V_h'']_{\text{\rm Lip}}^2}.
\end{equation*}
Furthermore, for all $q'\geq q$,
$\inf _{h\in \overline{\mathcal{H}}}{\lambda^{\mu_{h,q'}}_{h,q}}>0$.
In particular, letting $\bar\lambda_{h,q}:=\lambda^{\mu_{h,q}}_{h,q}$, one has $\inf_{h\in\overline{\mathcal{H}}}\bar{\lambda}_{h,q}>0$.

\item\label{lmm:lyapunov-iii}
for all $\xi\in\mathbb{R}$,
\begin{equation*}
|(\mathcal{L}^\mu_{h,q})''(\xi)|
\leq\chi^\mu_{h,q}\big(\mathcal{L}^\mu_{h,q}(\xi)+\mathcal{L}^\mu_{h,q-1}(\xi)\big),
\end{equation*}
where
\begin{equation*}
\chi^\mu_{h,q}
:=(q\vee\mu )\|V_h''\|_\infty+k_\alpha^2\mu (\mu \vee2)+q\big(2\mu \vee(q-1)\big)
\bigg(\frac{3k_\alpha^2[V_h'']_{\text{\rm Lip}}^2}{V_h''(\xi^h_\star)^3}\vee\frac{3\|V_h''\|_\infty^2}{V_h''(\xi^h_\star)}\bigg).
\end{equation*}
Besides, for all $q'\geq q$,
$|\lambda^{\mu_{h,q'}}_{h,q}|^{2}\leq\chi^{\mu_{h,q'}}_{h,q}$, and
$\sup _{h\in \overline{\mathcal{H}}}{\chi^{\mu_{h,q'}}_{h,q}}<\infty$.
In particular, introducing $\bar{\chi}_{h,q}:=\chi^{\mu_{h,q}}_{h,q}$, then $|\bar{\lambda}_{h,q}|^2\leq\bar{\chi}_{h,q}$, and $\sup _{h\in \overline{\mathcal{H}}}{\bar{\chi}_{h,q}}<\infty$.

\item\label{lmm:lyapunov-iv}
for all $\mu'\in\mathbb{R}_+$ and all $\xi\in\mathbb{R}$,
\begin{equation*}
(\xi-\xi^h_\star)^{2q}
\leq\kappa_{h,q}\big(\mathcal{L}^\mu_{h,q}(\xi)+\mathcal{L}^{\mu'}_{h,2q}(\xi)\big),
\quad\text{where}\quad
\kappa_{h,q}
:=\frac{3^q}{V_h''(\xi^h_\star)^q}\vee\frac{3^{2q}[V_h'']_{\text{\rm Lip}}^{2q}}{V_h''(\xi^h_\star)^{4q}}.
\end{equation*}
Moreover, $\sup_{h\in\overline{\mathcal{H}}}{\kappa_{h,q}}<\infty$.
\end{enumerate}
\end{lemma}

For $\omega\in\mathbb{R}$, let
\begin{equation}\label{eq:phi}
\Phi_\omega(t):=
   \begin{cases}
      \omega^{-1}(t^\omega-1),
         &\text{if $\omega\neq0$,}\\
      \ln{t}+\gamma^\mathcal{E},
         &\text{if $\omega=0$,}
   \end{cases}
   \quad t\in\mathbb{R}^*_+,
\end{equation}
where $\gamma^\mathcal{E}$ is the Euler-Mascheroni constant.
The next result generalizes \cite[Lemma~C.2]{CFL23}.

\begin{lemma}\label{lmm:gamma}
Let $\gamma_n=\gamma_1n^{-\beta}$, $n\in\mathbb{N}^*$, with $\gamma_1>0$ and $\beta\in(0,1]$.
Using the convention $\prod_\varnothing=1$, define
\begin{equation*}
\Pi_{k+1:n}^\delta:=\prod_{j=k+1}^n(1-\lambda\gamma_j+\zeta\gamma_j^{1+\delta}),
\quad k\in[\![0,n]\!],
\quad n\in\mathbb{N}^*,
\end{equation*}
where $\delta,\lambda>0$ and $\zeta\geq0$ are constants independent of $k$ and $n$.
Assume that $1-\lambda\gamma_n+\zeta\gamma_n^{1+\delta}>0$, $n\in\mathbb{N}^*$.
Then, for all $n\in\mathbb{N}^*$,

\begin{enumerate}[(i)]
\item
for all $k\in[\![0,n-1]\!]$,
\begin{equation*}
\Pi_{k+1:n}^\delta\leq
\begin{cases}
\exp\{-\lambda\gamma_1(\Phi_{1-\beta}(n+1)-\Phi_{1-\beta}(k+1))\}\\
\;\times\exp\{2^{\beta(1+\delta)}\zeta\gamma_1^{1+\delta}(\Phi_{1-\beta(1+\delta)}(n+1)-\Phi_{1-\beta(1+\delta)}(k+1))\}\\
\hphantom{\exp(\zeta(1+\frac1\delta)\gamma_1^{1+\delta}+\frac{\lambda\gamma_1}2)
\frac{(k+1)^{\lambda\gamma_1}}{(n+1)^{\lambda\gamma_1}}
\qquad\qquad\qquad\qquad}
\text{if $\beta\in(0,1)$,}\\
\exp(\zeta(1+\frac1\delta)\gamma_1^{1+\delta}+\frac{\lambda\gamma_1}2)
\frac{(k+1)^{\lambda\gamma_1}}{(n+1)^{\lambda\gamma_1}}
\qquad\qquad\qquad\qquad
\text{if $\beta=1$.}
\end{cases}
\end{equation*}

\item
for all $p\geq1$,
\begin{equation*}
\begin{aligned}
&\sum_{k=1}^n\gamma_k^p\Pi_{k+1:n}^\delta\leq\\
&\begin{cases}
\zeta^{-1}(\gamma_1^{p-(1+\delta)}\vee(\frac\lambda{2\zeta})^\frac{p-(1+\delta)}\delta)\exp(2^{1+\beta(1+\delta)}\zeta\gamma_1^{1+\delta}\Phi_{1-\beta(1+\delta)}(n+1))\\
\hphantom{\zeta^{-1}(\gamma_1^{p-(1+\delta)}\vee(\frac\lambda{2\zeta})^\frac{p-(1+\delta)}\delta)}
\times\exp(-\frac{\lambda\gamma_1}2\Phi_{1-\beta}(n+1))\\
+2^\beta\gamma_1^p\exp(-2^{-(\beta+2)}\lambda\gamma_1n^{1-\beta})\Phi_{1-\beta}(n+1)+\frac{2^{1+(p-1)\beta}\gamma_1^{p-1}}{\lambda n^{(p-1)\beta}}
&\text{if $\beta\in(0,1)$,}\\
\gamma_1^p\exp\{\zeta(1+\frac1\delta)\gamma_1^{1+\delta}+\frac{\lambda\gamma_1}2+((\lambda\gamma_1)\vee p)\ln{(2)}\}
\frac{\Phi_{\lambda\gamma_1-(p-1)}(n+1)}{(n+1)^{\lambda\gamma_1}}
&\text{if $\beta=1$.}
\end{cases}
\end{aligned}
\end{equation*}

\end{enumerate}
\end{lemma}

\begin{proof}
We follow closely the strategy employed in the proof of~\cite[Lemma~C.2]{CFL23}.
\\

Let $n\in\mathbb{N}^*$, $k\in[\![0,n]\!]$ and $p\geq1$.
Using that $1+x\leq\e^x$, $x\in\mathbb{R}$,
\begin{equation}
\label{eq:exp<1+x}
\Pi_{k+1:n}^\delta\leq\exp\bigg(-\lambda\sum_{j=k+1}^n\gamma_j\bigg)\exp\bigg(\zeta\sum_{j=k+1}^n\gamma_j^{1+\delta}\bigg),
\end{equation}
with the convention $\sum_\varnothing=0$.
\\

\noindent
\emph{Step~1. Case $\beta=1$.}
\newline
\emph{Step~1.1. Inequality on $\Pi^\delta_{k+1:n}$.}
\newline
Denoting $\psi$ is the digamma function (c.f.~Abramowitz et al.~\cite[Section~6.3]{ASR88}), we have
\begin{equation*}
\sum_{j=k+1}^n\gamma_j=\gamma_1\big(\psi(n+1)-\psi(k+1)\big).
\end{equation*}
We recall that $\psi$ satisfies (see e.g.~Alzer \cite{Alz97})
\begin{equation*}
\ln{x}-\frac1x\leq\psi(x)\leq\ln{x}-\frac1{2x},
\quad x>0.
\end{equation*}
Thus
\begin{equation*}
\sum_{j=k+1}^n\gamma_j
\geq\gamma_1\ln\bigg(\frac{n+1}{k+1}\bigg)-\frac{\gamma_1}{n+1}+\frac{\gamma_1}{2(k+1)}
\geq\gamma_1\ln\bigg(\frac{n+1}{k+1}\bigg)-\frac{\gamma_1}2.
\end{equation*}
Besides, via a series-integral comparison,
\begin{equation*}
\sum_{j=1}^\infty\gamma_j^{1+\delta}
\leq\gamma_1^{1+\delta}\Big(1+\frac1\delta\Big).
\end{equation*}
Thus, recalling \eqref{eq:exp<1+x},
\begin{equation}
\label{eq:Pi:delta<}
\Pi_{k+1:n}^\delta\leq\exp\bigg(\zeta\Big(1+\frac1\delta\Big)\gamma_1^{1+\delta}+\frac{\lambda\gamma_1}2\bigg)\frac{(k+1)^{\lambda\gamma_1}}{(n+1)^{\lambda\gamma_1}}.
\end{equation}

\noindent
\emph{Step~1.2. Inequality on $\sum_{k=1}^n\gamma_k^p\Pi_{k+1:n}^\delta$.}
\newline
Via the upper bound in \eqref{eq:Pi:delta<},
\begin{equation*}
\sum_{k=1}^n\gamma_k^p\Pi_{k+1:n}^\delta
\leq\frac{\gamma_1^p}{(n+1)^{\lambda\gamma_1}}\exp\bigg(\zeta\Big(1+\frac1\delta\Big)\gamma_1^{1+\delta}+\frac{\lambda\gamma_1}2\bigg)\sum_{k=1}^n\frac{(k+1)^{\lambda\gamma_1}}{k^p}.
\end{equation*}
If $\lambda\gamma_1\leq p$, then
\begin{equation*}
\sum_{k=1}^n\frac{(k+1)^{\lambda\gamma_1}}{k^p}
\leq2^p\sum_{k=1}^n\frac1{(k+1)^{p-\lambda\gamma_1}}
\leq2^p\Phi_{\lambda\gamma_1-p+1}(n+1).
\end{equation*}
Otherwise, if $\lambda\gamma_1>p$, then
\begin{equation*}
\sum_{k=1}^n\frac{(k+1)^{\lambda\gamma_1}}{k^p}
\leq2^{\lambda\gamma_1}\sum_{k=1}^nk^{\lambda\gamma_1-p}
\leq2^{\lambda\gamma_1}\Phi_{\lambda\gamma_1-p+1}(n+1).
\end{equation*}
All in all,
\begin{equation*}
\sum_{k=1}^n\gamma_k^p\Pi_{k+1:n}^\delta
\leq\gamma_1^p\exp\bigg(\zeta\Big(1+\frac1\delta\Big)\gamma_1^{1+\delta}+\frac{\lambda\gamma_1}2+\big((\lambda\gamma_1)\vee p\big)\ln{(2)}\bigg)\frac{\Phi_{\lambda\gamma_1-p+1}(n+1)}{(n+1)^{\lambda\gamma_1}}.
\end{equation*}

\noindent
\emph{Step~2. Case $\beta\in(0,1)$.}
\newline
\emph{Step~2.1. Inequality on $\Pi_{k+1:n}^\delta$.}
\newline
Via a series-integral comparison,
\begin{equation}
\label{eq:sum:gamma>}
\sum_{j=k+1}^n\gamma_j
\geq\gamma_1\big(\Phi_{1-\beta}(n+1)-\Phi_{1-\beta}(k+1)\big),
\end{equation}
and also
\begin{equation}
\label{eq:sum:gamma<}
\sum_{j=k+1}^n\gamma_j^{1+\delta}
\leq2^{\beta(1+\delta)}\gamma_1^{1+\delta}\big(\Phi_{1-\beta(1+\delta)}(n+1)-\Phi_{1-\beta(1+\delta)}(k+1)\big).
\end{equation}
Therefore, recalling \eqref{eq:exp<1+x},
\begin{equation*}
\begin{aligned}
\Pi_{k+1:n}^\delta
&\leq\exp\Big(-\lambda\gamma_1\big(\Phi_{1-\beta}(n+1)-\Phi_{1-\beta}(k+1)\big)\Big)\\
&\hphantom{\leq}\times
\exp\Big(2^{\beta(1+\delta)}\zeta\gamma_1^{1+\delta}\big(\Phi_{1-\beta(1+\delta)}(n+1)-\Phi_{1-\beta(1+\delta)}(k+1)\big)\Big).
\end{aligned}
\end{equation*}

\noindent
\emph{Step~2.2. Inequality on $\sum_{k=1}^n\gamma_k^p\Pi_{k+1:n}^\delta$.}
\newline
Let
\begin{equation*}
n_0:=\inf{\Big\{n\in\mathbb{N}^*:\gamma_n\leq\Big(\frac\lambda{2\zeta}\Big)^\frac1\delta\Big\}}-1.
\end{equation*}
Note that
\begin{equation*}
1-\lambda\gamma_n+\zeta\gamma_n^{1+\delta}\leq1-\frac{\lambda}2\gamma_n,\quad n\geq n_0+1,
\end{equation*}
and that
\begin{equation}
\label{eq:n0<}
n_0
=\bigg\lceil\bigg(\gamma_1\Big(\frac{2\zeta}\lambda\Big)^\frac1\delta\bigg)^\frac1\beta\bigg\rceil-1
<\bigg(\gamma_1\Big(\frac{2\zeta}\lambda\Big)^\frac1\delta\bigg)^\frac1\beta.
\end{equation}
Now,
\begin{equation}
\label{eq:A+B}
\begin{aligned}
\sum_{k=1}^n\gamma_k^p\Pi_{k+1:n}^\delta
&=\sum_{k=1}^{n_0\wedge n}\gamma_k^p\Pi_{k+1:n_0\wedge n}^\delta\Pi_{n_0\wedge n+1:n}^\delta+\sum_{k=n_0\wedge n+1}^n\gamma_k^p\Pi_{k+1:n}^\delta\\
&\leq\bigg(\sum_{k=1}^{n_0\wedge n}\gamma_k^p\prod_{j=k+1}^{n_0\wedge n}(1+\zeta\gamma_j^{1+\delta})\bigg)\prod_{j=n_0\wedge n+1}^n\bigg(1-\frac{\lambda}2\gamma_j\bigg)\\
&\hphantom{\leq}
+\sum_{k=1}^n\gamma_k^p\prod_{j=k+1}^n\bigg(1-\frac{\lambda}2\gamma_j\bigg).
\end{aligned}
\end{equation}
Let us study each term in the upper bound separately.

On the one hand,
\begin{equation}
\label{eq:term:1:factor:1}
\begin{aligned}
\sum_{k=1}^{n_0\wedge n}\gamma_k^p&\prod_{j=k+1}^{n_0\wedge n}(1+\zeta\gamma_j^{1+\delta})\\
&\leq(\gamma_1^{p-(1+\delta)}\vee\gamma_{n_0}^{p-(1+\delta)})\sum_{k=1}^{n_0\wedge n}\gamma_k^{1+\delta}\prod_{j=k+1}^{n_0\wedge n}(1+\zeta\gamma_j^{1+\delta}).
\end{aligned}
\end{equation}
If $p\geq1+\delta$, then, given the decreasing monotonicity of $(\gamma_n)_{n\geq1}$,
\begin{equation}
\label{eq:case:1}
\gamma_1^{p-(1+\delta)}\vee\gamma_{n_0}^{p-(1+\delta)}\leq\gamma_1^{p-(1+\delta)}.
\end{equation}
Otherwise, if $p<1+\delta$, then, taking into account \eqref{eq:n0<},
\begin{equation}
\label{eq:case:2}
\gamma_1^{p-(1+\delta)}\vee\gamma_{n_0}^{p-(1+\delta)}
\leq\bigg(\frac{n_0^\beta}{\gamma_1}\bigg)^{(1+\delta)-p}
<\bigg(\frac{2\zeta}\lambda\bigg)^\frac{(1+\delta)-p}\delta.
\end{equation}
Besides, note that, by telescopic summation,
\begin{equation}
\label{eq:simple=complex}
\begin{aligned}
\sum_{k=1}^{n_0\wedge n}\gamma_k^{1+\delta}&\prod_{j=k+1}^{n_0\wedge n}(1+\zeta\gamma_j^{1+\delta})\\
&=\frac1\zeta\sum_{k=1}^{n_0\wedge n}\bigg(\prod_{j=k}^{n_0\wedge n}(1+\zeta\gamma_j^{1+\delta})-\prod_{j=k+1}^{n_0\wedge n}(1+\zeta\gamma_j^{1+\delta})\bigg),\\
&\leq\frac1\zeta\prod_{j=1}^{n_0\wedge n}(1+\zeta\gamma_j^{1+\delta}).
\end{aligned}
\end{equation}
Thus, coming back to \eqref{eq:term:1:factor:1}, using \eqref{eq:case:1}, \eqref{eq:case:2}, \eqref{eq:simple=complex}, telescopic summation and the inequality $1+x\leq\e^x$, $x\in\mathbb{R}$,
\begin{equation*}
\begin{aligned}
\sum_{k=1}^{n_0\wedge n}\gamma_k^p&\prod_{j=k+1}^{n_0\wedge n}(1+\zeta\gamma_j^{1+\delta})\\
&\leq\zeta^{-1}\bigg(\gamma_1^{p-(1+\delta)}\vee\Big(\frac\lambda{2\zeta}\Big)^\frac{p-(1+\delta)}\delta\bigg)
\prod_{j=1}^{n_0\wedge n}(1+\zeta\gamma_j^{1+\delta})\\
&\leq\zeta^{-1}\bigg(\gamma_1^{p-(1+\delta)}\vee\Big(\frac\lambda{2\zeta}\Big)^\frac{p-(1+\delta)}\delta\bigg)
\exp\bigg(\zeta\sum_{j=1}^{n_0\wedge n}\gamma_j^{1+\delta}\bigg).
\end{aligned}
\end{equation*}
Moreover, note that, for $j\leq n_0$, $\gamma_j\geq(\frac\lambda{2\zeta})^\frac1\delta$, so that $\frac\lambda2\gamma_j\leq\zeta\gamma_j^{1+\delta}$.
Hence, using the inequality $1+x\leq\e^x$, $x\in\mathbb{R}$,
\begin{equation*}
\begin{aligned}
\prod_{j=n_0\wedge n+1}^n\bigg(1-\frac{\lambda}2\gamma_j\bigg)
&\leq\exp\bigg(-\frac{\lambda}2\sum_{j=1}^n\gamma_j\bigg)\exp\bigg(\frac{\lambda}2\sum_{j=1}^{n_0\wedge n}\gamma_j\bigg)\\
&\leq\exp\bigg(-\frac{\lambda}2\sum_{j=1}^n\gamma_j\bigg)\exp\bigg(\zeta\sum_{j=1}^{n_0\wedge n}\gamma^{1+\delta}_j\bigg).
\end{aligned}
\end{equation*}
Combining the two preceding upper bounds and using \eqref{eq:sum:gamma>} and \eqref{eq:sum:gamma<} with $k=0$,
\begin{equation}
\label{eq:A}
\begin{aligned}
\bigg(\sum_{k=1}^{n_0\wedge n}\gamma_k^p&\prod_{j=k+1}^{n_0\wedge n}(1+\zeta\gamma_j^{1+\delta})\bigg)\prod_{j=n_0\wedge n+1}^n\bigg(1-\frac{\lambda}2\gamma_j\bigg)\\
&\leq\zeta^{-1}\bigg(\gamma_1^{p-(1+\delta)}\vee\Big(\frac\lambda{2\zeta}\Big)^\frac{p-(1+\delta)}\delta\bigg)
\exp\bigg(2\zeta\sum_{j=1}^{n_0\wedge n}\gamma_j^{1+\delta}\bigg)\exp\bigg(-\frac{\lambda}2\sum_{j=1}^n\gamma_j\bigg)\\
&\leq\zeta^{-1}\bigg(\gamma_1^{p-(1+\delta)}\vee\Big(\frac\lambda{2\zeta}\Big)^\frac{p-(1+\delta)}\delta\bigg)
\exp\big(2^{1+\beta(1+\delta)}\zeta\gamma_1^{1+\delta}\Phi_{1-\beta(1+\delta)}(n+1)\big)\\
&\hphantom{\leq\zeta^{-1}\bigg(\gamma_1^{p-(1+\delta)}\vee\Big(\frac\lambda{2\zeta}\Big)^\frac{p-(1+\delta)}\delta\bigg)}\times
\exp\bigg(-\frac{\lambda\gamma_1}2\Phi_{1-\beta}(n+1)\bigg).
\end{aligned}
\end{equation}

On the other hand, for $m\in[\![1,n]\!]$,
\begin{equation}
\label{eq:term:2}
\sum_{k=1}^n\gamma_k^p\prod_{j=k+1}^n\bigg(1-\frac{\lambda}2\gamma_j\bigg)
=\sum_{k=1}^m\gamma_k^p\prod_{j=k+1}^n\bigg(1-\frac{\lambda}2\gamma_j\bigg)
+\sum_{k=m+1}^n\gamma_k^p\prod_{j=k+1}^n\bigg(1-\frac{\lambda}2\gamma_j\bigg).
\end{equation}
By the inequality $1+x\leq\e^x$, $x\in\mathbb{R}$ and the decreasing monotony of $(\gamma_n)_{n\geq1}$,
\begin{equation}
\label{eq:prev:-2}
\begin{aligned}
\sum_{k=1}^m\gamma_k^p\prod_{j=k+1}^n\bigg(1-\frac{\lambda}2\gamma_j\bigg)
&\leq\prod_{j=m+1}^n\bigg(1-\frac{\lambda}2\gamma_j\bigg)\sum_{k=1}^n\gamma_k^p\\
&\leq\exp\bigg(-\frac{\lambda}2\sum_{j=m+1}^n\gamma_j\bigg)\gamma_1^{p-1}\sum_{k=1}^n\gamma_k.
\end{aligned}
\end{equation}
Using series-integral comparisons,
\begin{equation}
\label{eq:prev:-1}
2^{-\beta}\gamma_1\big(\Phi_{1-\beta}(n)-\Phi_{1-\beta}(m)\big)
\leq\sum_{j=m+1}^n\gamma_j
\leq\sum_{j=1}^n\gamma_j
\leq2^\beta\gamma_1\Phi_{1-\beta}(n+1).
\end{equation}
Hence, combining \eqref{eq:prev:-2} and \eqref{eq:prev:-1},
\begin{equation}
\label{eq:first:term}
\sum_{k=1}^m\gamma_k^p\prod_{j=k+1}^n\bigg(1-\frac{\lambda}2\gamma_j\bigg)
\leq2^\beta\gamma_1^p\exp{\Big(-2^{-(\beta+1)}\lambda\gamma_1\big(\Phi_{1-\beta}(n)-\Phi_{1-\beta}(m)\big)\Big)}\Phi_{1-\beta}(n+1).
\end{equation}
Besides, since $(\gamma_n)_{n\geq1}$ is decreasing,
\begin{equation}
\label{eq:prev:-2:bis}
\sum_{k=m+1}^n\gamma_k^p\prod_{j=k+1}^n\bigg(1-\frac{\lambda}2\gamma_j\bigg)
\leq\gamma_{m+1}^{p-1}\sum_{k=m+1}^n\gamma_k\prod_{j=k+1}^n\bigg(1-\frac{\lambda}2\gamma_j\bigg).
\end{equation}
By telescopic summation,
\begin{equation}
\label{eq:simple=complex:bis}
\sum_{k=m+1}^n\gamma_k\prod_{j=k+1}^n\bigg(1-\frac{\lambda}2\gamma_j\bigg)
=\frac2\lambda\sum_{k=m+1}^n\bigg(\prod_{j=k+1}^n\bigg(1-\frac{\lambda}2\gamma_j\bigg)-\prod_{j=k}^n\bigg(1-\frac{\lambda}2\gamma_j\bigg)\bigg)
\leq\frac2\lambda.
\end{equation}
Thus, gathering \eqref{eq:prev:-2:bis} and \eqref{eq:simple=complex:bis},
\begin{equation}
\label{eq:second:term}
\sum_{k=m+1}^n\gamma_k^p\prod_{j=k+1}^n\bigg(1-\frac{\lambda}2\gamma_j\bigg)
\leq\frac{2\gamma_{m+1}^{p-1}}\lambda
=\frac{2\gamma_1^{p-1}}{\lambda(m+1)^{(p-1)\beta}}.
\end{equation}
Collecting the upper bounds from \eqref{eq:term:2}, \eqref{eq:first:term} and \eqref{eq:second:term},
\begin{equation*}
\begin{aligned}
\sum_{k=1}^n&\gamma_k^p\prod_{j=k+1}^n\bigg(1-\frac{\lambda}2\gamma_j\bigg)\\
&\leq2^\beta\gamma_1^p\exp{\Big(-2^{-(\beta+1)}\lambda\gamma_1\big(\Phi_{1-\beta}(n)-\Phi_{1-\beta}(m)\big)\Big)}\Phi_{1-\beta}(n+1)
+\frac{2\gamma_1^{p-1}}{\lambda(m+1)^{(p-1)\beta}}.
\end{aligned}
\end{equation*}
Taking $m=\lfloor\frac{n}2\rfloor$, and using that $2^\beta\leq1+\beta$, $\beta\in(0,1]$, one has
\begin{equation*}
\Phi_{1-\beta}(n)
-\Phi_{1-\beta}\Big(\Big\lfloor\frac{n}2\Big\rfloor\Big)
\geq\frac{n^{1-\beta}}2.
\end{equation*}
Thus
\begin{equation}
\label{eq:B}
\begin{aligned}
\sum_{k=1}^n\gamma_k^p&\prod_{j=k+1}^n\bigg(1-\frac{\lambda}2\gamma_j\bigg)\\
&\leq2^\beta\gamma_1^p\exp(-2^{-(\beta+2)}\lambda\gamma_1n^{1-\beta})\Phi_{1-\beta}(n+1)
   +\frac{2^{1+(p-1)\beta}\gamma_1^{p-1}}{\lambda n^{(p-1)\beta}}.
\end{aligned}
\end{equation}

Finally, injecting the upper bounds of \eqref{eq:A} and \eqref{eq:B} into \eqref{eq:A+B} concludes the proof.
\end{proof}

The proof of the following result is included within the preliminary steps of the proof of \cite[Proposition~5.2]{CFL23}, and is therefore omitted.

\begin{lemma}\label{lmm:aux}
\begin{enumerate}[(i)]
\item\label{lmm:aux:i}
Assume that there exists $p\geq1$ such that
\begin{equation*}
\mathbb{E}\big[\big|\varphi(Y, Z)-\mathbb{E}[\varphi(Y,Z)|Y]\big|^p\big]<\infty.
\end{equation*}
Then, for all $h,h'\in\overline{\mathcal{H}}$,
\begin{equation*}
\mathbb{E}[|X^h-X^{h'}|^p]\leq C|h-h'|^\frac{p}2,
\end{equation*}
where
\begin{equation*}
C=\sqrt2B_p\mathbb{E}\big[\big|\varphi(Y, Z)-\mathbb{E}[\varphi(Y, Z)|Y]\big|^p\big]^\frac1p,
\end{equation*}
with $B_p$ being a positive constant that depends only on $p$.
\item\label{lmm:aux:ii}
Assume that there exists $\mathfrak{g}>0$ such that, for all $u\in\mathbb{R}$,
\begin{equation*}
\mathbb{E}\Big[\exp\Big(u\big(\varphi(Y,Z)-\mathbb{E}[\varphi(Y,Z)|Y]\big)\Big)\Big|Y\Big]\leq\e^{\mathfrak{g}u^2}
\quad\Pas.
\end{equation*}
Then, for all $h,h'\in\overline{\mathcal{H}}$ and all $u\in\mathbb{R}$,
\begin{equation*}
\mathbb{E}\big[\exp\big(u(X^h-X^{h'})\big)\big|Y\big]
\leq\exp(\mathfrak{g}u^2|h-h'|)
\quad\text{$\Pas$.}
\end{equation*}
\end{enumerate}
\end{lemma}

We refer to Remarks~\ref{rmk:framework}(\ref{rmk:framework:p}) and~(\ref{rmk:framework:sub:gaussian}) for related comments.

\section{Proofs of the Convergence Results}
\label{proofs}

\begin{proof}[Proof of Lemma~\ref{lmm:local:strong:error:indicator:func-ii:bis}]
\label{prf:local:strong:error:indicator:func-ii:bis}

The proof is similar to \cite[Proposition~4.2(ii)]{CFL23}.
Using the tower law, one has
\begin{equation*}
\begin{aligned}
\mathbb{E}[|\mathds1_{\{X^{h'}>\xi\}}&-\mathds1_{\{X^h>\xi\}}|]\\
&=\mathbb{P}(X^{h'}\leq\xi<X^h)+\mathbb{P}(X^h\leq\xi<X^{h'})\\
&=\mathbb{P}\big(X^h+{(h')}^\frac12G_h^{h'}\leq\xi<X^h\big)
+\mathbb{P}\big(X^h\leq\xi<X^h+{(h')}^\frac12G_h^{h'}\big)\\
&=\mathbb{P}\big(X^h-{(h')}^\frac12(G_h^{h'})^-\leq\xi<X^h\big)
+\mathbb{P}\big(X^h\leq\xi<X^h+{(h')}^\frac12(G_h^{h'})^+\big)\\
&=\mathbb{E}\big[F_{X^h|G_h^{h'}}\big(\xi+(h')^\frac12(G_h^{h'})^-\big)-F_{X^h|G_h^{h'}}(\xi)\big]\\
&\qquad+\mathbb{E}\big[F_{X^h|G_h^{h'}}(\xi)-F_{X^h|G_h^{h'}}\big(\xi-(h')^\frac12(G_h^{h'})^+\big)\big]\\
&=\mathbb{E}\big[F_{X^h|G_h^{h'}}\big(\xi+(h')^\frac12(G_h^{h'})^-\big)
-F_{X^h|G_h^{h'}}\big(\xi-(h')^\frac12(G_h^{h'})^+\big)\big]\\
&\leq\Big(\sup_{0\leq h_1<h_2\in\overline{\mathcal{H}}}{\mathbb{E}[K_{h_1}^{h_2}|G_{h_1}^{h_2}|]}\Big)(h')^\frac12.
\end{aligned}
\end{equation*}
\end{proof}

\begin{proof}[Proof of Lemma~\ref{lmm:bias}]
\label{prf:bias}

Let $\ell,n\in\mathbb{N}^*$ and $\xi\in\mathbb{R}$, $C_\mathrm{ad}>0$, $r>1$ and $\theta\in(0,1]$.

By \eqref{eq:eta} and the law of total probability,
\begin{equation}
\begin{aligned}
\mathbb{P}(X^{h_{\ell+\eta^\ell_n(\xi)}}\leq \xi)&-F_{X^{h_{\ell+\lceil\theta\ell\rceil}}}(\xi)\\
&=\sum_{k=0}^{\lceil\theta\ell\rceil}\mathbb{P}\big(X^{h_{\ell+k}}\leq\xi,\eta^\ell_n(\xi)=k\big)-\mathbb{P}\big(X^{h_{\ell+\lceil\theta\ell\rceil}}\leq\xi,\eta^\ell_n(\xi)=k\big)\\
&=\sum_{k=0}^{\lceil\theta\ell\rceil-1}\mathbb{E}[(\mathds1_{\{X^{h_{\ell+k}}\leq\xi\}}-\mathds1_{\{X^{h_{\ell+\lceil\theta\ell\rceil}}\leq\xi\}})\mathds1_{\{\eta^\ell_n(\xi)=k\}}].
\label{eq:*:0}
\end{aligned}
\end{equation}
Let $k\in[\![0,\lceil\theta\ell\rceil-1]\!]$. One the one hand,
\begin{equation}
\label{eq:*:1}
\mathds1_{\{X^{h_{\ell+k}}\leq\xi\}}-\mathds1_{\{X^{h_{\ell+\lceil\theta\ell\rceil}}\leq\xi\}}
% &=\begin{cases}
% 1&\text{on $\{X^{h_{\ell+k}}\leq\xi<X^{h_{\ell+\lceil\theta\ell\rceil}}\}$,}\\
% -1&\text{on $\{X^{h_{\ell+\lceil\theta\ell\rceil}}\leq\xi<X^{h_{\ell+k}}\}$,}\\
% 0&\text{elsewhere,}
% \end{cases}\\
=\mathds{1}_{\{X^{h_{\ell+k}}\leq\xi<X^{h_{\ell+\lceil\theta\ell\rceil}}\}}-\mathds{1}_{\{X^{h_{\ell+\lceil\theta\ell\rceil}}\leq\xi<X^{h_{\ell+k}}\}}.
\end{equation}
On the other hand, by \eqref{eq:eta},
\begin{equation}
\label{eq:*:2}
\{\eta^\ell_n(\xi)=k\}\subset\{|X^{h_{\ell+k}}-\xi|\geq C_\mathrm{ad}\psi^{\ell,k}_n\}.
\end{equation}
Therefore, combining \eqref{eq:*:1} and \eqref{eq:*:2},
\begin{equation}
\label{eq:*}
\begin{aligned}
|\mathbb{E}[(&\mathds1_{\{X^{h_{\ell+k}}\leq\xi\}}-\mathds1_{\{X^{h_{\ell+\lceil\theta\ell\rceil}}\leq\xi\}})\mathds1_{\{\eta^\ell_n(\xi)=k\}}]|\\
&\leq\mathbb{E}[(\mathds1_{\{X^{h_{\ell+k}}\leq\xi<X^{h_{\ell+\lceil\theta\ell\rceil}}\}}+\mathds1_{\{X^{h_{\ell+\lceil\theta\ell\rceil}}\leq\xi<X^{h_{\ell+k}}\}})
\mathds1_{\{|X^{h_{\ell+k}}-\xi|\geq C_\mathrm{ad}\psi^{\ell,k}_n\}}]\\
&\leq\mathbb{P}(|X^{h_{\ell+k}}-X^{h_{\ell+\lceil\theta\ell\rceil}}|\geq|X^{h_{\ell+k}}-\xi|,|X^{h_{\ell+k}}-\xi|\geq C_\mathrm{ad}\psi^{\ell,k}_n)\\
&\leq\mathbb{P}(|X^{h_{\ell+k}}-X^{h_{\ell+\lceil\theta\ell\rceil}}|\geq C_\mathrm{ad}\psi^{\ell,k}_n),
\end{aligned}
\end{equation}
where we used that
\begin{equation*}
\{X^{h_{\ell+k}}\leq\xi<X^{h_{\ell+\lceil\theta\ell\rceil}}\}\cup\{X^{h_{\ell+\lceil\theta\ell\rceil}}\leq\xi<X^{h_{\ell+k}}\}\subset\{|X^{h_{\ell+k}}-X^{h_{\ell+\lceil\theta\ell\rceil}}|\geq|X^{h_{\ell+k}}-\xi|\}.
\end{equation*}

\noindent (\ref{lmm:bias:i})\hyperref[lmm:bias:i:a]{\rm a}.\
By \eqref{eq:*}, Markov's inequality and Lemma~\ref{lmm:aux}(\ref{lmm:aux:i}),
\begin{equation}
\label{eq:prev}
\begin{aligned}
|\mathbb{E}[(&\mathds1_{\{X^{h_{\ell+k}}\leq\xi\}}-\mathds1_{\{X^{h_{\ell+\lceil\theta\ell\rceil}}\leq\xi\}})\mathds1_{\{\eta^\ell_n(\xi)=k\}}]|\\
&\leq\mathbb{P}(|X^{h_{\ell+k}}-X^{h_{\ell+\lceil\theta\ell\rceil}}|\geq C_\mathrm{ad}u_n^{-\frac1{p_\star}}h_{\theta\ell(r-1)+k}^\frac1r)\\
&\leq C_\mathrm{ad}^{-p_\star}u_nh_{\theta\ell(r-1)+k}^{-\frac{p_\star}r}\mathbb{E}[|X^{h_{\ell+k}}-X^{h_{\ell+\lceil\theta\ell\rceil}}|^{p_\star}]\\
&\leq C_{11}u_nh_{\theta\ell(r-1)+k}^{-\frac{p_\star}r}h_{\ell+k}^\frac{p_\star}2,
\end{aligned}
\end{equation}
where
\begin{equation*}
C_{11}=\sqrt2C_\mathrm{ad}^{-p_\star}B_{p_\star}\mathbb{E}\big[\big|\varphi(Y,Z)-\mathbb{E}[\varphi(Y,Z)|Y]\big|^{p_\star}\big]^\frac1{p_\star},
\end{equation*}
with $B_{p_\star}$ being a positive constant that depends only on $p_\star$.
Using \eqref{eq:prev}, \eqref{eq:*:0} and the condition \eqref{assumption:finite:Lp:moment:bias} on $r$ and on $\theta$, it follows
\begin{equation}
\begin{aligned}
|\mathbb{P}(X^{h_{\ell+\eta^\ell_n(\xi)}}\leq\xi)&-F_{X^{h_{\ell+\lceil\theta\ell\rceil}}}(\xi)|\\
&\leq C_{11}u_n\sum_{k=0}^{\lceil\theta\ell\rceil-1}h_{\theta\ell(r-1)+k}^{-\frac{p_\star}r}h_{\ell+k}^\frac{p_\star}2\\
&\leq\frac{C_{11}h_1^{-p_\star(\frac1r-\frac12)}}{M^{p_\star(\frac1r-\frac12)}-1}u_nM^{-\frac{p_\star\ell}2(1-\theta)}\\
&\leq C_{12}u_nh_\ell^{1+\theta},
\end{aligned}\label{eq:a}
\end{equation}
where
\begin{equation*}
C_{12}=\frac{C_{11}h_0^{-\frac{2p_\star}{p_\star+2}}h_1^{-p_\star(\frac1r-\frac12)}}{M^{p_\star(\frac1r-\frac12)}-1}.
\end{equation*}
\\

\noindent (\ref{lmm:bias:i})\hyperref[lmm:bias:i:b]{\rm b}.\
Let $\lambda>0$. By \eqref{eq:*}, the definition \eqref{eq:eta} for the case \eqref{assumption:conditional:gaussian:concentration}, Markov's exponential inequality,  the fact that $\e^{|x|}\leq\e^x+\e^{-x}$, $x\in\mathbb{R}$, and Lemma~\ref{lmm:aux}(\ref{lmm:aux:ii}),
\begin{equation*}
\begin{aligned}
|\mathbb{E}[(&\mathds1_{\{X^{h_{\ell+k}}\leq\xi\}}-\mathds1_{\{X^{h_{\ell+\lceil\theta\ell\rceil}}\leq\xi\}})\mathds1_{\{\eta^\ell_n(\xi)=k\}}]|\\
&\leq\mathbb{P}\Big(|X^{h_{\ell+k}}-X^{h_{\ell+\lceil\theta\ell\rceil}}|\geq C_\mathrm{ad}h^\frac1r_{\theta\ell(r-1)+k}\ln^\frac12{\big((\gamma_1^{-1}\gamma_n)^{-\frac12}h_{\ell+k}^{-\frac{1+\theta}2}\big)}\Big)\\
&\leq\exp\Big(-\lambda C_\mathrm{ad}h^\frac1r_{\theta\ell(r-1)+k}\ln^\frac12{\big((\gamma_1^{-1}\gamma_n)^{-\frac12}h_{\ell+k}^{-\frac{1+\theta}2}\big)}\Big)\mathbb{E}[\exp(\lambda|X^{h_{\ell+k}}-X^{h_{\ell+\lceil\theta\ell\rceil}}|)]\\
&\leq2\exp\Big(-\lambda C_\mathrm{ad}h^\frac1r_{\theta\ell(r-1)+k}\ln^\frac12{\big((\gamma_1^{-1}\gamma_n)^{-\frac12}h_{\ell+k}^{-\frac{1+\theta}2}\big)}+\mathfrak{g}\lambda^2h_{\ell+k}\Big).
\end{aligned}
\end{equation*}
Minimizing the above upper bound with respect to $\lambda$ yields
\begin{equation*}
\begin{aligned}
|\mathbb{E}[(&\mathds1_{\{X^{h_{\ell+k}}\leq\xi\}}-\mathds1_{\{X^{h_{\ell+\lceil\theta\ell\rceil}}\leq\xi\}})\mathds1_{\{\eta^\ell_n(\xi)=k\}}]|\\
&\leq2\exp\bigg(-\frac{C_\mathrm{ad}^2h^\frac2r_{\theta\ell(r-1)+k}\ln{((\gamma_1^{-1}\gamma_n)^{-\frac12}h_{\ell+k}^{-\frac{1+\theta}2})}}{4\mathfrak{g}h_{\ell+k}}\bigg)\\
&=2\bigg(\frac{\gamma_n}{\gamma_1}h_{\ell+k}^{1+\theta}\bigg)^\frac{C_\mathrm{ad}^2h_{\theta\ell(r-1)+k}^\frac2r}{8\mathfrak{g}h_{\ell+k}}.
\end{aligned}
\end{equation*}
Observe that $\frac{\gamma_n}{\gamma_1}h_{\ell+k}^{1+\theta}\leq1$.
Now, consider the condition \eqref{assump:unif:lipschitz:integrability:conditional:cdf:bias} on $h_0$, $r$ and $\theta$. On the one hand, one has $\theta\leq1\leq\frac{r}{2(r-1)}$, so that $-2\theta(1-\frac1r)+1\geq0$, i.e.~$M^{-2\theta\ell(1-\frac1r)+\ell}\geq1$.
On the other hand, one also has $\frac{C_\mathrm{ad}^2h_0^{\frac2r-1}}{8\mathfrak{g}}\geq1$.
Thus
\begin{equation*}
\begin{aligned}
|\mathbb{E}[(\mathds1_{\{X^{h_{\ell+k}}\leq\xi\}}&-\mathds1_{\{X^{h_{\ell+\lceil\theta\ell\rceil}}\leq\xi\}})\mathds1_{\{\eta^\ell_n(\xi)=k\}}]|\\
&\leq2\bigg(\frac{\gamma_n}{\gamma_1}h_{\ell+k}^{1+\theta}\bigg)^{\frac{C_\mathrm{ad}^2h_0^{\frac2r-1}}{8\mathfrak{g}}M^{-2\theta\ell(1-\frac1r)+\ell}}\\
&\leq C_{21}\gamma_nh_{\ell+k}^{1+\theta},
\end{aligned}
\end{equation*}
where
\begin{equation*}
C_{21}=2\gamma_1^{-1}.
\end{equation*}
Therefore, by \eqref{eq:*:0},
\begin{equation}
|\mathbb{P}(X^{h_{\ell+\eta^\ell_n(\xi)}}\leq \xi)-F_{X^{h_{\ell+\lceil\theta\ell\rceil}}}(\xi)|\leq C_{22}\gamma_nh_\ell^{1+\theta}.
\label{eq:b}
\end{equation}
where
\begin{equation*}
C_{22}=\frac{C_{21}}{1-M^{-1}}.
\end{equation*}

\noindent (\ref{lmm:bias:ii})\
Via \eqref{eq:*}, the definition \eqref{eq:eta} for the case \eqref{assump:unif:lipschitz:integrability:conditional:cdf:bis}, the definition \eqref{eq:G} and Markov's exponential inequality,
\begin{equation*}
\begin{aligned}
|\mathbb{E}[(&\mathds1_{\{X^{h_{\ell+k}}\leq\xi\}}-\mathds1_{\{X^{h_{\ell+\lceil\theta\ell\rceil}}\leq\xi\}})\mathds1_{\{\eta^\ell_n(\xi)=k\}}]|\\
&\leq\mathbb{P}\Big(|X^{h_{\ell+k}}-X^{h_{\ell+\lceil\theta\ell\rceil}}|\geq C_\mathrm{ad}h^\frac1r_{\theta\ell(r-1)+k}\ln^\frac12{\big((\gamma_1^{-1}\gamma_n)^{-\frac12}h_{\ell+k}^{-\frac{1+\theta}2}\big)}\Big)\\
&=\mathbb{P}\Big(|G_{h_{\ell+\lceil\theta\ell\rceil}}^{h_{\ell+k}}|\geq C_\mathrm{ad}h^\frac1r_{\theta\ell(r-1)+k}h_{\ell+k}^{-\frac12}\ln^\frac12{\big((\gamma_1^{-1}\gamma_n)^{-\frac12}h_{\ell+k}^{-\frac{1+\theta}2}}\big)\Big)\\
&\leq\exp{\Big(-\upsilon_0C_\mathrm{ad}^2h_{\theta\ell(r-1)+k}^\frac2rh_{\ell+k}^{-1}\ln{\big((\gamma_1^{-1}\gamma_n)^{-\frac12}h_{\ell+k}^{-\frac{1+\theta}2}\big)}\Big)}\mathbb{E}[\exp(\upsilon_0|G_{h_{\ell+\lceil\theta\ell\rceil}}^{h_{\ell+k}}|^2)].
\end{aligned}
\end{equation*}
Noting that $\frac{\gamma_n}{\gamma_1}h_{\ell+k}^{1+\theta}\leq1$, and using the condition \eqref{assump:unif:lipschitz:integrability:conditional:cdf:bias} on $h_0$, $r$ and $\theta$, which implies in particular that $\theta\leq1\leq\frac{r}{2(r-1)}$, so that~$-2\theta(1-\frac1r)+1\geq0$, hence $M^{-2\theta\ell(1-\frac1r)+\ell}\geq1$, and also that $\frac12\upsilon_0C_\mathrm{ad}^2h_0^{\frac2r-1}\geq1$, one has
\begin{equation*}
\begin{aligned}
|\mathbb{E}[(&\mathds1_{\{X^{h_{\ell+k}}\leq\xi\}}-\mathds1_{\{X^{h_{\ell+\lceil\theta\ell\rceil}}\leq\xi\}})\mathds1_{\{\eta^\ell_n(\xi)=k\}}]|\\
&\leq\Big(\sup_{0\leq h<h'\in\overline{\mathcal{H}}}\mathbb{E}[\exp(\upsilon_0|G_h^{h'}|^2)]\Big)\bigg(\frac{\gamma_n}{\gamma_1}h_{\ell+k}^{1+\theta}\bigg)^{\frac12\upsilon_0C_\mathrm{ad}^2h_0^{\frac2r-1}M^{-2\theta\ell(1-\frac1r)+\ell}}\\
&\leq C_{31}\gamma_nh_{\ell+k}^{1+\theta},
\end{aligned}
\end{equation*}
where
\begin{equation*}
C_{31}=\gamma_1^{-1}\Big(\sup_{0\leq h<h'\in\overline{\mathcal{H}}}\mathbb{E}[\exp(\upsilon_0|G_h^{h'}|^2)]\Big).
\end{equation*}
Thus, by \eqref{eq:*:0},
\begin{equation}
\label{eq:c}
|\mathbb{P}(X^{h_{\ell+\eta^\ell_n(\xi)}}\leq \xi)-F_{X^{h_{\ell+\lceil\theta\ell\rceil}}}(\xi)|\leq C_{32}\gamma_nh_\ell^{1+\theta}.
\end{equation}
where
\begin{equation*}
C_{32}=\frac{C_{31}}{1-M^{-1}}.
\end{equation*}
\end{proof}

\begin{proof}[Proof of Proposition~\ref{prp:error:statistical:bis}]
\label{prf:error:statistical:bis}
In this proof, we omit the superscript $\mu$ from the Lyapunov function $\mathcal{L}^\mu_{h,q}$ and the parameters $\lambda^\mu_{h,q}$ and $\chi^\mu_{h,q}$, as defined in Lemma~\ref{lmm:lyapunov}, and simply denote $\mathcal{L}_{h,q}$, $\lambda_{h,q}$ and $\chi_{h,q}$.
For $\ell\in\mathbb{N}$, let $(\widetilde{\mathcal{F}}^{h_\ell}_n)_{n\geq0}$ be the filtration defined by $\widetilde{\mathcal{F}}^{h_\ell}_0=\sigma(\widetilde\xi^{h_\ell}_0)$ and
$\widetilde{\mathcal{F}}^{h_\ell}_n=\sigma(\widetilde\xi^{h_\ell}_0,X^{h_{\ell+\lceil\theta\ell\rceil}}_1,\dots,X^{h_{\ell+\lceil\theta\ell\rceil}}_n)$, $n\in\mathbb{N}^*$.
\\

\noindent
\emph{Step~1. Linearization of the adNSA dynamics.}
\newline
Let $\ell,n\in\mathbb{N}^*$.
Define
\begin{equation}\label{eq:vee}
v^\ell_n(\xi):=\mathbb{E}[H(\xi,X^{h_{\ell+\eta^\ell_n(\xi)}})]=1-\frac1{1-\alpha}\mathbb{P}(X^{h_{\ell+\eta^\ell_n(\xi)}}\geq\xi),
\quad\xi\in\mathbb{R}.
\end{equation}
The dynamics \eqref{adNSA} can then be rewritten as
\begin{equation}
\label{xi:decomp}
\widetilde\xi^{h_\ell}_{n+1}
=\widetilde\xi^{h_\ell}_n
-\gamma_{n+1}V_{h_{\ell+\lceil\theta\ell\rceil}}'(\widetilde\xi^{h_\ell}_n)
-\gamma_{n+1}r^{h_\ell}_{n+1}
-\gamma_{n+1} e^{h_\ell}_{n+1},
\end{equation}
where
\begin{align}
r^{h_\ell}_{n+1}
&:= v^\ell_{n+1}(\widetilde\xi^{h_\ell}_n)-V_{h_{\ell+\lceil\theta\ell\rceil}}'(\widetilde\xi^{h_\ell}_n),
\label{eq:r:def}\\
e^{h_\ell}_{n+1}
&:= H(\widetilde\xi_n^{h_\ell},\widetilde{X}^{h_{\ell}}_{n+1})-v^\ell_{n+1}(\widetilde\xi^{h_\ell}_n).
\label{eq:e:def}
\end{align} 

By Lemma~\ref{lmm:bias},
\begin{equation}
    \label{bound:estimate:rnhell}
|r_{n+1}^{h_\ell}|
=\frac1{1-\alpha}|\mathbb{P}(X^{h_{\ell+\eta^\ell_{n+1}(\xi)}}\leq\xi)_{|\xi=\widetilde\xi^{h_\ell}_n}
-\mathbb{P}(X^{h_{\ell+\lceil\theta\ell\rceil}} \leq\xi)_{|\xi=\widetilde\xi^{h_\ell}_n}|
\leq CMh_{\ell+\lceil\theta\ell\rceil}\widetilde{u}_{n+1},
\end{equation}
where
\begin{equation}
\label{eq:u:tilde}
\widetilde{u}_{n+1}=\begin{cases}
u_{n+1}&\text{if \eqref{assumption:finite:Lp:moment:bias} holds,}\\
\gamma_{n+1}&\text{if \eqref{assumption:conditional:gaussian:concentration:bias} or \eqref{assump:unif:lipschitz:integrability:conditional:cdf:bias} holds.}
\end{cases}
\end{equation}
and
\begin{equation}
\label{eq:C=Ci,i=1,2,3}
C=\begin{cases}
C_1&\text{if \eqref{assumption:finite:Lp:moment:bias} holds,}\\
C_2&\text{if \eqref{assumption:conditional:gaussian:concentration:bias} holds,}\\
C_3&\text{if \eqref{assump:unif:lipschitz:integrability:conditional:cdf:bias} holds,}
\end{cases}
\end{equation}
with $C_1$, $C_2$ and $C_3$ being the positive constants defined in Lemma~\ref{lmm:bias}.
\\

\noindent
\emph{Step~2. Testing the Lyapunov function.}
\newline
Let $\mu\in\mathbb{R}_+$ and $q\in\mathbb{N}^*$. In the spirit of \cite[Theorem~3.7]{CFL23}, we test the Lyapunov $\mathcal{L}_{h_{\ell+\lceil\theta\ell\rceil},q}$ along the dynamics \eqref{xi:decomp}. Using a second order Taylor expansion,
\begin{equation}
\begin{aligned}
\mathcal{L}_{h_{\ell+\lceil\theta\ell\rceil},q}(\widetilde\xi^{h_\ell}_{n+1})
&=\mathcal{L}_{h_{\ell+\lceil\theta\ell\rceil},q}\big(\widetilde\xi^{h_\ell}_n-\gamma_{n+1}V_{h_{\ell+\lceil\theta\ell\rceil}}'(\widetilde\xi^{h_\ell}_{n})-\gamma_{n+1} r^{h_\ell}_{n+1}-\gamma_{n+1} e^{h_\ell}_{n+1}\big)\\
&=\mathcal{L}_{h_{\ell+\lceil\theta\ell\rceil},q}(\widetilde\xi^{h_\ell}_n)-\gamma_{n+1}\mathcal{L}_{h_{\ell+\lceil\theta\ell\rceil},q}'(\widetilde\xi^{h_\ell}_n)\big(V_{h_{\ell+\lceil\theta\ell\rceil}}'(\widetilde\xi^{h_\ell}_n)+r_{n+1}^{h_\ell} + e^{h_\ell}_{n+1}\big)\\
&\hphantom{=}
+\gamma_{n+1}^2H(\widetilde\xi^{h_\ell}_n,\widetilde{X}^{h_\ell}_{n+1})^2\int_0^1(1-t)\mathcal{L}_{h_{\ell+\lceil\theta\ell\rceil},q}''\big(t\widetilde\xi^{h_\ell}_{n+1}+(1-t)\widetilde\xi^{h_\ell}_n\big) \mathrm{d}t.
\end{aligned}
\end{equation}
It follows from Lemmas~\ref{lmm:lyapunov}(\ref{lmm:lyapunov-ii})-(\ref{lmm:lyapunov-iii}), \eqref{relation:first:deriv:lyapunov} and \eqref{bound:estimate:rnhell} that
\begin{equation}\label{eq:dynamics}
\begin{aligned}
&\mathcal{L}_{h_{\ell+\lceil\theta\ell\rceil},q}(\widetilde\xi^{h_\ell}_{n+1})
\leq\mathcal{L}_{h_{\ell+\lceil\theta\ell\rceil},q}(\widetilde\xi^{h_\ell}_n)\\
&\hphantom{\quad+}\times
(1-\lambda_{h_{\ell+\lceil\theta\ell\rceil},q}\gamma_{n+1}
+CM\mu\|V_{h_{\ell+\lceil\theta\ell\rceil}}'\|_\infty h_{\ell+\lceil\theta\ell\rceil}\gamma_{n+1}\widetilde{u}_{n+1})\\
&\quad
-\gamma_{n+1}\mathcal{L}_{h_{\ell+\lceil\theta\ell\rceil},q}'(\widetilde\xi^{h_\ell}_n)e^{h_\ell}_{n+1}
+CMq\|V_{h_{\ell+\lceil\theta\ell\rceil}}'\|_\infty h_{\ell+\lceil\theta\ell\rceil}\gamma_{n+1}\widetilde{u}_{n+1}\mathcal{L}_{h_{\ell+\lceil\theta\ell\rceil},q-1}(\widetilde\xi^{h_\ell}_n)\\
&\quad
+\chi_{h_{\ell+\lceil\theta\ell\rceil},q}\gamma_{n+1}^2H(\widetilde\xi^{h_\ell}_n,\widetilde{X}^{h_\ell}_{n+1})^2\\
&\hphantom{\quad+}\times
\int_0^1(1-t)\Big(\mathcal{L}_{h_{\ell+\lceil\theta\ell\rceil},q}\big(t\widetilde\xi^{h_\ell}_{n+1}+(1-t)\widetilde\xi^{h_\ell}_n\big)
+\mathcal{L}_{h_{\ell+\lceil\theta\ell\rceil},q-1}\big(t\widetilde\xi^{h_\ell}_{n+1}+(1-t)\widetilde\xi^{h_\ell}_n\big)\Big)\mathrm{d}t.
\end{aligned}
\end{equation}

The mean value theorem guarantees that, for all $t\in [0,1]$, there exists $a^{h_\ell}_n(t)\in\mathbb{R}$ such that
\begin{equation}\label{eq:mean-value}
V_{h_{\ell+\lceil\theta\ell\rceil}}\big(t\widetilde\xi^{h_\ell}_{n+1}+(1-t)\widetilde\xi^{h_\ell}_n\big)
=V_{h_{\ell+\lceil\theta\ell\rceil}}(\widetilde\xi^{h_\ell}_n)+tV'_{h_{\ell+\lceil\theta\ell\rceil}}\big(a^{h_\ell}_n(t)\big)(\widetilde\xi^{h_\ell}_{n+1}-\widetilde\xi^{h_\ell}_n).
\end{equation}
Besides, from \eqref{eq:H1},
\begin{equation}\label{eq:|thetap-p+1|<}
|\widetilde\xi^{h_\ell}_{n+1}-\widetilde\xi^{h_\ell}_n|
=\gamma_{n+1}|H(\widetilde\xi^{h_\ell}_n, \widetilde{X}^{h_\ell}_{n+1})|
\leq k_\alpha\gamma_{n+1},
\end{equation}
with $k_\alpha=\frac\alpha{1-\alpha}\vee1$.
From \eqref{eq:mean-value}, \eqref{eq:|thetap-p+1|<} and the fact that $\|V_{h_{\ell+\lceil\theta\ell\rceil}}'\|_\infty\leq k_\alpha$, it ensues that, for all $t\in[0,1]$,
\begin{equation}\label{eq:Vht-Vh<}
\begin{aligned}
V_{h_{\ell+\lceil\theta\ell\rceil}}\big(t \widetilde\xi^{h_\ell}_{n+1}&+(1-t)\widetilde\xi^{h_\ell}_n\big)-V_{h_{\ell+\lceil\theta\ell\rceil}}(\xi^{h_{\ell+\lceil\theta\ell\rceil}}_\star)\\
&\leq V_{h_{\ell+\lceil\theta\ell\rceil}}(\widetilde\xi^{h_\ell}_n)-V_{h_{\ell+\lceil\theta\ell\rceil}}(\xi^{h_{\ell+\lceil\theta\ell\rceil}}_\star)+k_\alpha^2\gamma_{n+1}.
\end{aligned}
\end{equation}
Using the inequality $\e^x\leq\e\mathds1_{x\leq1}+x^q\e^x\mathds1_{x>1}\leq\e(1+x^q\e^x)$, $x\geq0$, and the very definition \eqref{eq:Lq} of $\mathcal{L}_{h,q}$, it follows
\begin{equation}
\label{eq:L0<Lq}
\mathcal{L}^\mu_{h,0}(\xi)\leq\e\big(1+\mu^q\mathcal{L}^\mu_{h,q}(\xi)\big),
\quad\xi\in\mathbb{R},
\quad\mu\in\mathbb{R}_+,
\quad h\in\overline{\mathcal{H}},
\quad q\in\mathbb{N}^*.
\end{equation}
Thus, by \eqref{eq:Vht-Vh<} and \eqref{eq:L0<Lq}, for all $t\in[0,1]$,
\begin{equation}\label{eq:Lt<}
\begin{aligned}
&\mathcal{L}_{h_{\ell+\lceil\theta\ell\rceil},q}\big(t\widetilde\xi^{h_\ell}_{n+1}+(1-t)\widetilde\xi^{h_\ell}_n\big)\\
&\quad
\leq\big(V_{h_{\ell+\lceil\theta\ell\rceil}}(\widetilde\xi^{h_\ell}_n)-V_{h_{\ell+\lceil\theta\ell\rceil}}(\xi^{h_{\ell+\lceil\theta\ell\rceil}}_\star)+k_\alpha^2\gamma_{n+1}\big)^q\\
&\quad\hphantom{\leq}\times
\exp\Big(\mu\big(V_{h_{\ell+\lceil\theta\ell\rceil}}(\widetilde\xi^{h_\ell}_n)-V_{h_{\ell+\lceil\theta\ell\rceil}}(\xi^{h_{\ell+\lceil\theta\ell\rceil}}_\star)+k_\alpha^2\gamma_{n+1}\big)\Big)\\
&\quad
\leq2^{q-1}\exp(\mu k_\alpha^2\gamma_{n+1})
   \big(\mathcal{L}_{h_{\ell+\lceil\theta\ell\rceil},q}(\widetilde\xi^{h_\ell}_n)+k_\alpha^{2q}\gamma_{n+1}^q\mathcal{L}_{h_{\ell+\lceil\theta\ell\rceil},0}(\widetilde\xi^{h_\ell}_n)\big)\\
&\quad
\leq\sigma^\mu_q\big(\mathcal{L}_{h_{\ell+\lceil\theta\ell\rceil},q}(\widetilde\xi^{h_\ell}_n)+\gamma_{n+1}^q\big),
\end{aligned}
\end{equation}
with
\begin{equation*}
\sigma^\mu_q
:=2^{q-1}\exp(\mu k_\alpha^2\gamma_1)
\big((1+\e\mu^qk_\alpha^{2q}\gamma_1^q)
\vee\e k_\alpha^{2q}\big)
\geq 2^{q-1},
\quad\mu\in\mathbb{R}_+,
\quad h\in\overline{\mathcal{H}},
\quad q\in\mathbb{N}^*.
\end{equation*}
Plugging the upper bounds of \eqref{eq:|thetap-p+1|<} and \eqref{eq:Lt<} into \eqref{eq:dynamics} yields
\begin{equation}
\label{eq:Lhq(xih{n+1})<}
\begin{aligned}
&\mathcal{L}_{h_{\ell+\lceil\theta\ell\rceil},q}(\widetilde\xi^{h_\ell}_{n+1})
\leq\mathcal{L}_{h_{\ell+\lceil\theta\ell\rceil},q}(\widetilde\xi^{h_\ell}_n)\\
&\hphantom{\quad+}\times
\big(1-\lambda_{h_{\ell+\lceil\theta\ell\rceil},q}\gamma_{n+1}
+(CM\mu\|V'_{h_{\ell+\lceil\theta\ell\rceil}}\|_\infty h_{\ell+\lceil\theta\ell\rceil}\gamma_{n+1}\widetilde{u}_{n+1}
+\nu^\mu_{h_{\ell+\lceil\theta\ell\rceil},q}\gamma^2_{n+1})\big)\\
&\quad
+(CMq\|V'_{h_{\ell+\lceil\theta\ell\rceil}}\|_\infty\gamma_{n+1}\widetilde{u}_{n+1}h_{\ell+\lceil\theta\ell\rceil}
+\nu^\mu_{h_{\ell+\lceil\theta\ell\rceil},q}\gamma_{n+1}^2)\mathcal{L}_{h_{\ell+\lceil\theta\ell\rceil},q-1}(\widetilde\xi^{h_\ell}_n)\\
&\quad
-\gamma_{n+1}\mathcal{L}_{h_{\ell+\lceil\theta\ell\rceil},q}'(\widetilde\xi^{h_\ell}_n)e^{h_\ell}_{n+1}
+\nu^\mu_{h_{\ell+\lceil\theta\ell\rceil},q}\gamma_{n+1}^{q+1},
\end{aligned}
\end{equation}
where
\begin{equation}
\label{eq:xihq}
\nu^\mu_{h,q}
:=\frac12\chi^\mu_{h,q}
k_\alpha^2
\big((\gamma_1\sigma^\mu_{h,q}+\sigma_{h,q-1})
\vee\sigma^\mu_{h,q}\big),
\quad\mu\in\mathbb{R}_+,
\quad h\in\overline{\mathcal{H}},
\quad q\in\mathbb{N}^*.
\end{equation}
Hereafter, to simplify notation, we will simply denote $\nu_{h,q}$ instead of $\nu^\mu_{h,q}$.
\\

\noindent
\emph{Step~3. Inequality on $\mathbb{E}[\mathcal{L}_{h_{\ell+\lceil\theta\ell\rceil},q}(\widetilde\xi^{h_\ell}_n)]$.}
\newline
Via the tower law,
\begin{equation}\label{eq:E[E]}
\mathbb{E}[\mathcal{L}_{h_{\ell+\lceil\theta\ell\rceil},q}'(\widetilde\xi^{h_\ell}_n) e^{h_\ell}_{n+1}]
=\mathbb{E}\big[\mathcal{L}_{h_{\ell+\lceil\theta\ell\rceil},q}'(\widetilde\xi^{h_\ell}_n)\mathbb{E}[e^{h_\ell}_{n+1}|\widetilde{\mathcal{F}}^{h_\ell}_n]\big].
\end{equation}
For $n\in\mathbb{N}$, by the independence of the $(X^{h_{\ell+k}}_{n+1})_{1\leq k\leq\lceil\theta\ell\rceil}$ and $\widetilde{\mathcal{F}}^{h_\ell}_n$,
\begin{equation*}
\begin{aligned}
\mathbb{E}[&e^{h_\ell}_{n+1}|\widetilde{\mathcal{F}}^{h_\ell}_n]\\
&=\mathbb{E}[H(\widetilde\xi^{h_\ell}_n,X^{h_{\ell+\eta^\ell_{n+1}(\widetilde\xi^{h_\ell}_n)}}_{n+1})|\widetilde{\mathcal{F}}^{h_\ell}_n]-\mathbb{E}[H(\xi,X^{h_{\ell+\eta^\ell_{n+1}(\xi)}}_{n+1})]_{|\xi=\widetilde\xi^{h_\ell}_n}\\
&=\mathbb{E}\bigg[\sum_{k=1}^{\lceil\theta\ell\rceil}H(\widetilde\xi^{h_\ell}_n,X^{h_\ell+k}_{n+1})\mathds1_{\{\eta^\ell_{n+1}(\widetilde\xi^{h_\ell}_n)=k\}}\bigg|\widetilde{\mathcal{F}}^{h_\ell}_n\bigg]-\sum_{k=1}^{\lceil\theta\ell\rceil}\mathbb{E}[H(\xi,X^{h_\ell+k}_{n+1})\mathds1_{\{\eta^\ell_{n+1}(\xi)=k\}}]_{|\xi=\widetilde\xi^{h_\ell}_n}\\
&=0.
\end{aligned}
\end{equation*}
Hence $(e^{h_\ell}_n)_{n\geq1}$ is a $((\widetilde{\mathcal{F}}^{h_\ell}_n)_{n\geq0},\mathbb{P})$-martingale increment sequence, and, by \eqref{eq:E[E]},
\begin{equation*}
\mathbb{E}[\mathcal{L}_{h_{\ell+\lceil\theta\ell\rceil},q}'(\widetilde\xi^{h_\ell}_n) e^{h_\ell}_{n+1}]=0.
\end{equation*}

Thus, by taking the expectation in both sides of the inequality \eqref{eq:Lhq(xih{n+1})<},
\begin{equation}
\label{eq:ELq<}
\begin{aligned}
\mathbb{E}[&\mathcal{L}_{h_{\ell+\lceil\theta\ell\rceil},q}(\widetilde\xi^{h_\ell}_{n+1})]
\leq\mathbb{E}[\mathcal{L}_{h_{\ell+\lceil\theta\ell\rceil},q}(\widetilde\xi^{h_\ell}_n)]\\
&\hphantom{+}\times
\big(1-\lambda_{h_{\ell+\lceil\theta\ell\rceil},q}\gamma_{n+1}
+(CM\mu\|V'_{h_{\ell+\lceil\theta\ell\rceil}}\|_\infty h_{\ell+\lceil\theta\ell\rceil}\gamma_{n+1}\widetilde{u}_{n+1}
+\nu_{h_{\ell+\lceil\theta\ell\rceil},q}\gamma^2_{n+1})\big)\\
&+(CMq\|V'_{h_{\ell+\lceil\theta\ell\rceil}}\|_\infty h_{\ell+\lceil\theta\ell\rceil}\gamma_{n+1}\widetilde{u}_{n+1}
+\nu_{h_{\ell+\lceil\theta\ell\rceil},q}\gamma_{n+1}^2)\mathbb{E}[\mathcal{L}_{h_{\ell+\lceil\theta\ell\rceil},q-1}(\widetilde\xi^{h_\ell}_n)]\\
&+\nu_{h_{\ell+\lceil\theta\ell\rceil},q}\gamma_{n+1}^{q+1}.
\end{aligned}
\end{equation}

\noindent(\ref{prp:error:statistical:bis:i})\
We hereby aim to use Lemma~\ref{lmm:lyapunov}(\ref{lmm:lyapunov-iv}) with $h=h_{\ell+\lceil\theta\ell\rceil}$ and $q=1$.
To that end, we provide below sharper controls for $\mathbb{E}[\bar{\mathcal{L}}_{h_{\ell+\lceil\theta\ell\rceil},1}(\widetilde\xi^{h_\ell}_n)]$ and $\mathbb{E}[\bar{\mathcal{L}}_{h_{\ell+\lceil\theta\ell\rceil},2}(\widetilde\xi^{h_\ell}_n)]$.
According to Lemma~\ref{lmm:lyapunov}(\ref{lmm:lyapunov-i-bis}), by assumption, for $\mu\in\{\mu_{h_{\ell+\lceil\theta\ell\rceil},q'},q'\in\{1,2\}\}$,
\begin{equation*}
\begin{aligned}
\mathbb{E}[\mathcal{L}_{h_{\ell+\lceil\theta\ell\rceil},1}(\widetilde\xi^{h_\ell}_0)]
&\vee\mathbb{E}[\mathcal{L}_{h_{\ell+\lceil\theta\ell\rceil},2}(\widetilde\xi^{h_\ell}_0)]\\
&\leq\mathbb{E}\bigg[\big(1+2k_\alpha^2(|\widetilde\xi^{h_\ell}_0|^2+|\xi^{h_{\ell+\lceil\theta\ell\rceil}}_\star|^2)\big)\\
&\hphantom{\leq\mathbb{E}\bigg[}\times
\exp\bigg(\frac{2k_\alpha}{1-\alpha}
\|f_{X^{h_{\ell+\lceil\theta\ell\rceil}}}\|_\infty
(|\widetilde\xi_0^{h_\ell}|+|\xi^{h_{\ell+\lceil\theta\ell\rceil}}_\star|)\bigg)\bigg]
<\infty.
\end{aligned}
\end{equation*}

\noindent
\emph{Step~1. Inequality on $\mathbb{E}[\mathcal{L}_{h_{\ell+\lceil\theta\ell\rceil},1}(\widetilde\xi^{h_\ell}_n)]$.}
\newline
Taking $q=1$ in \eqref{eq:ELq<} and \eqref{eq:L0<Lq},
\begin{equation*}
\begin{aligned}
&\mathbb{E}[\mathcal{L}_{h_{\ell+\lceil\theta\ell\rceil},1}(\widetilde\xi^{h_\ell}_{n+1})]
\leq\mathbb{E}[\mathcal{L}_{h_{\ell+\lceil\theta\ell\rceil},1}(\widetilde\xi^{h_\ell}_n)]\\
&\times
\big(1-\lambda_{h_{\ell+\lceil\theta\ell\rceil},1}\gamma_{n+1}
+(1+\e\mu)
(CM\mu\|V'_{h_{\ell+\lceil\theta\ell\rceil}}\|_\infty h_{\ell+\lceil\theta\ell\rceil}\gamma_{n+1}\widetilde{u}_{n+1}
+\nu_{h_{\ell+\lceil\theta\ell\rceil},1}\gamma^2_{n+1})\big)\\
&
+CM\e\|V'_{h_{\ell+\lceil\theta\ell\rceil}}\|_\infty h_{\ell+\lceil\theta\ell\rceil}\gamma_{n+1}\widetilde{u}_{n+1}+(1+\e\mu)\nu_{h_{\ell+\lceil\theta\ell\rceil},1}\gamma_{n+1}^2.
\end{aligned}
\end{equation*}
Let
\begin{equation}
\label{eq:delta}
\delta'=
\begin{cases}
\delta
&\text{if \eqref{assumption:finite:Lp:moment:bias} holds,}\\
\beta
&\text{if \eqref{assumption:conditional:gaussian:concentration:bias} or \eqref{assump:unif:lipschitz:integrability:conditional:cdf:bias} holds.}
\end{cases}
\end{equation}
Note that, by \eqref{eq:u:tilde},
\begin{equation}
\label{eq:u=f(g)}
\widetilde{u}_n
=\frac{\gamma_1}{n^{\delta'}}
=\gamma_1^{1-\frac{\delta'}\beta}\gamma_n^\frac{\delta'}\beta,
\quad n\in\mathbb{N}^*,
\end{equation}
so that
\begin{equation}
\label{eq:gamma:u}
\gamma_n\widetilde{u}_n
=\frac{\gamma_1^2}{n^{\beta+\delta'}}
=\gamma_1^{1-\frac{\delta'}\beta}\gamma_n^{1+\frac{\delta'}\beta},
\quad n\in\mathbb{N}^*,
\end{equation}
hence
\begin{equation}
\label{eq:max:gamma<gamma}
\gamma_n\widetilde{u}_n\vee\gamma_n^2
=\frac{\gamma_1^2}{n^{\beta+(\delta'\wedge\beta)}}
=\gamma_1^{1-\frac{\delta'\wedge\beta}\beta}\gamma_n^{1+\frac{\delta'\wedge\beta}\beta},
\quad n\in\mathbb{N}^*.
\end{equation}
Therefore, using in particular that $\gamma_n\widetilde{u}_n\leq\gamma_1^{1-\frac{\delta'\wedge\beta}\beta}\gamma_n^{1+\frac{\delta'\wedge\beta}\beta}$, $n\in\mathbb{N}^*$,
\begin{equation}
\label{ineq:bar:L1:before:iteration}
\begin{aligned}
&\mathbb{E}[\mathcal{L}_{h_{\ell+\lceil\theta\ell\rceil},1}(\widetilde\xi^{h_\ell}_{n+1})]
\leq\mathbb{E}[\mathcal{L}_{h_{\ell+\lceil\theta\ell\rceil},1}(\widetilde\xi^{h_\ell}_n)]\\
&\times\Big(1-\lambda_{h_{\ell+\lceil\theta\ell\rceil},1}\gamma_{n+1}\\
&\hphantom{\times\Big(}
+(1+\e\mu)(CM\mu\|V'_{h_{\ell+\lceil\theta\ell\rceil}}\|_\infty h_{\ell+\lceil\theta\ell\rceil}
+\nu_{h_{\ell+\lceil\theta\ell\rceil},1})\gamma_1^{1-\frac{\delta'\wedge\beta}\beta}\gamma_{n+1}^{1+\frac{\delta'\wedge\beta}\beta}\Big)\\
&
+CM\gamma_1^{1-\frac{\delta'}\beta}\e\|V'_{h_{\ell+\lceil\theta\ell\rceil}}\|_\infty h_{\ell+\lceil\theta\ell\rceil}\gamma_{n+1}^{1+\frac{\delta'}\beta}
+(1+\e\mu)\nu_{h_{\ell+\lceil\theta\ell\rceil},1}\gamma_{n+1}^2.
\end{aligned}
\end{equation}

Recalling the definition of $\chi^\mu_{h,q}$ from Lemma~\ref{lmm:lyapunov}(\ref{lmm:lyapunov-iii}) and the definition of $\nu^\mu_{h,q}$ given in \eqref{eq:xihq}, let
\begin{multline}
\label{eq:c:hq}
c^\mu_{h,q}
:=\frac{\nu^\mu_{h,q}}{\chi^\mu_{h,q}}
=\frac12
\big(\sigma^\mu_q\vee(\gamma_1\sigma^\mu_q+\sigma^\mu_{q-1})\big)
k_\alpha^2
>\frac12,
\quad
\bar{c}_{h,q}:=c^{\mu_{h,q}}_{h,q},\\
\mu\in\mathbb{R}_+,
\quad h\in\overline{\mathcal{H}},
\quad q\in\mathbb{N}^*.
\end{multline}
Let $q'\geq q$ and set $\mu=\mu_{h,q'}$.
Recalling that $|\lambda^{\mu_{h,q'}}_{h,q}|^2
\leq\chi^{\mu_{h,q'}}_{h,q}$ according to Lemma~\ref{lmm:lyapunov}(\ref{lmm:lyapunov-iii}), and that $\gamma_1^{1-\frac{\delta'\wedge\beta}\beta}\gamma_n^{1+\frac{\delta'\wedge\beta}\beta}\geq\gamma_n^2$, $n\in\mathbb{N}^*$, according to \eqref{eq:max:gamma<gamma}, it ensues that, for all $n\in\mathbb{N}^*$,
\begin{equation}
\label{eq:1-lambda>}
\begin{aligned}
1
-\lambda^{\mu_{h,q'}}_{h,q}\gamma_n
&+(1+\e\mu_{h,q'}^q)
\big(CM\mu_{h,q'}\|V'_h\|_\infty h
+\nu^{\mu_{h,q'}}_{h,q}\big)\gamma_1^{1-\frac{\delta'\wedge\beta}\beta}\gamma_n^{1+\frac{\delta'\wedge\beta}\beta}\\
&\geq1-\lambda^{\mu_{h,q'}}_{h,q}\gamma_n
+\nu^{\mu_{h,q'}}_{h,q}\gamma^2_n\\
&\geq1
-(\chi^{\mu_{h,q'}}_{h,q})^\frac12\gamma_{n+1}
+\chi^{\mu_{h,q'}}_{h,q}c^{\mu_{h,q'}}_{h,q}\gamma_n^2\\
&\geq1-\frac1{4c^{\mu_{h,q'}}_{h,q}}\\
&\geq\frac12>0.
\end{aligned}
\end{equation}

Therefore, for $q'\geq1$, setting $\mu=\mu_{h_{\ell+\lceil\theta\ell\rceil},q'}$ and iterating $n$ times the inequality \eqref{ineq:bar:L1:before:iteration},
\begin{equation}
\label{eq:E[L1]<}
\begin{aligned}
\mathbb{E}[\mathcal{L}_{h_{\ell+\lceil\theta\ell\rceil},1}(\widetilde\xi^{h_\ell}_n)]
&\leq\mathbb{E}[\mathcal{L}_{h_{\ell+\lceil\theta\ell\rceil},1}(\widetilde\xi^{h_\ell}_0)]
\Pi_{1:n}^{h_{\ell+\lceil\theta\ell\rceil},1,q'}\\
&\hphantom{\leq}
+CM\gamma_1^{1-\frac{\delta'}\beta}\e\|V'_{h_{\ell+\lceil\theta\ell\rceil}}\|_\infty
\sum_{k=1}^nh_{\ell+\lceil\theta\ell\rceil}\gamma_k^{1+\frac{\delta'}\beta}
\Pi_{k+1:n}^{h_{\ell+\lceil\theta\ell\rceil},1,q'}\\
&\hphantom{\leq}
+(1+\e\mu_{h_{\ell+\lceil\theta\ell\rceil},q'})\nu_{h_{\ell+\lceil\theta\ell\rceil},1}
\sum_{k=1}^n\gamma_k^2
\Pi_{k+1:n}^{h_{\ell+\lceil\theta\ell\rceil},1,q'},
\end{aligned}
\end{equation}
where
\begin{multline*}
\Pi_{k:n}^{h,q,q'}
:=\prod_{j=k}^n
\big(
1
-\lambda^{\mu_{h,q'}}_{h,q}\gamma_j
+(1+\e\mu_{h,q'}^q)
(
    CM\mu_{h,q'}\|V'_h\|_\infty h
    +\nu^{\mu_{h,q'}}_{h,q}
)
\gamma_1^{1-\frac{\delta'\wedge\beta}\beta}
\gamma_j^{1+\frac{\delta'\wedge\beta}\beta}
\big),\\
\quad h\in\overline{\mathcal{H}},
\quad q\in\mathbb{N}^*,
\quad q'\geq q,
\quad k,n\in\mathbb{N}^*,
\end{multline*}
with the convention $\prod_\varnothing=1$.

If $\gamma_n=\gamma_1n^{-1}$, $n\in\mathbb{N}^*$, then, applying Lemma~\ref{lmm:gamma},
\begin{equation*}
\begin{aligned}
\mathbb{E}[\mathcal{L}_{h_{\ell+\lceil\theta\ell\rceil},1}(\widetilde\xi^{h_\ell}_n)]
&\leq
\widehat{C}^{\ell,1}_{1,q'}\frac{\mathbb{E}[\mathcal{L}_{h_{\ell+\lceil\theta\ell\rceil},1}(\widetilde\xi^{h_\ell}_0)]}{(n+1)^{\lambda_{h_{\ell+\lceil\theta\ell\rceil},1}\gamma_1}}
+\bar{C}^{\ell,1}_{1,q'}\frac{\Phi_{\lambda_{h_{\ell+\lceil\theta\ell\rceil},1}\gamma_1-\delta'}(n+1)}{(n+1)^{\lambda_{h_{\ell+\lceil\theta\ell\rceil},1}\gamma_1}}h_{\ell+\lceil\theta\ell\rceil}\\
&\hphantom{\leq}
+\widetilde{C}^{\ell,1}_{1,q'}\frac{\Phi_{\lambda_{h_{\ell+\lceil\theta\ell\rceil},1}\gamma_1-1}(n+1)}{(n+1)^{\lambda_{h_{\ell+\lceil\theta\ell\rceil},1}\gamma_1}},
\end{aligned}
\end{equation*}
where
\begin{equation*}
\begin{aligned}
\widehat{C}^{\ell,1}_{1,q'}
&=\exp\bigg(\zeta^{1,\delta'}_{h_{\ell+\lceil\theta\ell\rceil},1,q'}
\Big(1+\frac1{\delta'}\Big)\gamma_1^{1+\delta'}
+\frac{\lambda_{h_{\ell+\lceil\theta\ell\rceil},1}\gamma_1}2\bigg),\\
\bar{C}^{\ell,1}_{1,q'}
&=2^{(\lambda_{h_{\ell+\lceil\theta\ell\rceil},1}\gamma_1)\vee(1+\delta')}
C\widehat{C}^{\ell,1}_{1,q'}M\e\|V'_{h_{\ell+\lceil\theta\ell\rceil}}\|_\infty
\gamma_1^2,\\
\widetilde{C}^{\ell,1}_{1,q'}
&=2^{(\lambda_{h_{\ell+\lceil\theta\ell\rceil},1}\gamma_1)\vee 2}(1+\e\mu_{h_{\ell+\lceil\theta\ell\rceil},q'})\nu_{h_{\ell+\lceil\theta\ell\rceil},1}\widehat{C}^{\ell,1}_{1,q'}
\gamma_1^2,
\end{aligned}
\end{equation*}
with
\begin{multline*}
\zeta^{\beta,\delta}_{h,q,q'}
:=(1+\e\mu_{h,q'}^q)
\big(CM\mu_{h,q'}\|V'_h\|_\infty h
+\nu^{\mu_{h,q'}}_{h,q}\big)
\gamma_1^{1-\frac{\delta\wedge\beta}\beta},\\
\quad h\in\overline{\mathcal{H}},
\quad q\in\mathbb{N}^*,
\quad q'\geq q,
\quad\beta,\delta\in(0,1].
\end{multline*}
If $\bar\lambda_{h_{\ell+\lceil\theta\ell\rceil},1}\gamma_1
>1$, so that, for $q'\geq1$, $\lambda^{\mu_{h_{\ell+\lceil\theta\ell\rceil},q'}}_{h_{\ell+\lceil\theta\ell\rceil},1}\gamma_1
\geq\bar\lambda_{h_{\ell+\lceil\theta\ell\rceil},1}\gamma_1
>1
\geq\delta'$, then, via the definition \eqref{eq:phi},
\begin{equation}
\label{estimate:L1}
\mathbb{E}[\mathcal{L}_{h_{\ell+\lceil\theta\ell\rceil},1}(\widetilde\xi^{h_\ell}_n)] \leq C^{\ell,1}_{1,q'}\big((h_\ell^{1+\theta}\widetilde{u}_n)\vee\gamma_n\big),
\end{equation}
where
\begin{equation*}
C^{\ell,1}_{1,q'}
=\gamma_1^{-1}(\widehat{C}^{\ell,1}_{1,q'}\mathbb{E}[\mathcal{L}_{h_{\ell+\lceil\theta\ell\rceil},1}(\widetilde\xi^{h_\ell}_0)]
+h_0^{-1}\bar{C}^{\ell,1}_{1,q'}
+\widetilde{C}^{\ell,1}_{1,q'}).
\end{equation*}

Likewise, if $\gamma_n=\gamma_1n^{-\beta}$, $n\in\mathbb{N}^*$, with $\beta\in(0,1)$, then, by Lemma~\ref{lmm:gamma},
\begin{equation*}
\begin{aligned}
&\mathbb{E}[\mathcal{L}_{h_{\ell+\lceil\theta\ell\rceil},1}(\widetilde\xi^{h_\ell}_n)]
\leq(\mathbb{E}[\mathcal{L}_{h_{\ell+\lceil\theta\ell\rceil},1}(\widetilde\xi^{h_\ell}_0)]+\bar{C}^{\ell,\beta}_{1,q'}+\widetilde{C}^{\ell,\beta}_{1,q'})\\
&\hphantom{+}\times
\exp\bigg(2^{1+\beta+\delta'\wedge\beta}\zeta^{\beta,\delta'}_{h_{\ell+\lceil\theta\ell\rceil},1,q'}\gamma_1^{1+\frac{\delta'\wedge\beta}\beta}\Phi_{1-\beta-\delta'\wedge\beta}(n+1)-\frac{\lambda_{h_{\ell+\lceil\theta\ell\rceil},1}\gamma_1}2\Phi_{1-\beta}(n+1)\bigg)\\
&
+2^\beta(\bar{C}^{\ell,\beta}_{1,q'}\gamma_1^{1+\frac{\delta'}\beta}h_{\ell+\lceil\theta\ell\rceil}
+\widetilde{C}^{\ell,\beta}_{1,q'}\gamma_1^2)
\exp(-2^{-(\beta+2)}\lambda_{h_{\ell+\lceil\theta\ell\rceil},1}\gamma_1n^{1-\beta})\Phi_{1-\beta}(n+1)\\
&
+\bar{C}^{\ell,\beta}_{1,q'}
\frac{2^{1+\delta'}\gamma_1^{\frac{\delta'}\beta}}{\lambda_{h_{\ell+\lceil\theta\ell\rceil},1}n^{\delta'}}
h_{\ell+\lceil\theta\ell\rceil}
+\widetilde{C}^{\ell,\beta}_{1,q'}
\frac{2^{1+\beta}\gamma_1}{\lambda_{h_{\ell+\lceil\theta\ell\rceil},1}n^\beta},
\end{aligned}
\end{equation*}
where
\begin{equation*}
\begin{aligned}
\bar{C}^{\ell,\beta}_{1,q'}
&=(\zeta^{\beta,\delta'}_{h_{\ell+\lceil\theta\ell\rceil},1,q'})^{-1}
\bigg(\gamma_1^\frac{\delta'-(\delta'\wedge\beta)}\beta
\vee\bigg(\frac{\lambda_{h_{\ell+\lceil\theta\ell\rceil},1}}{2\zeta^{\beta,\delta'}_{h_{\ell+\lceil\theta\ell\rceil},1,q'}}\bigg)^\frac{\delta'-(\delta'\wedge\beta)}{\delta'\wedge\beta}\bigg)
CM\gamma_1^{1-\frac{\delta'}\beta}\e\|V'_{h_{\ell+\lceil\theta\ell\rceil}}\|_\infty,\\
\widetilde{C}^{\ell,\beta}_{1,q'}
&=(\zeta^{\beta,\delta'}_{h_{\ell+\lceil\theta\ell\rceil},1,q'})^{-1}
\bigg(\gamma_1^\frac{\beta-(\delta'\wedge\beta)}\beta
\vee\bigg(\frac{\lambda_{h_{\ell+\lceil\theta\ell\rceil},1}}{2\zeta^{\beta,\delta'}_{h_{\ell+\lceil\theta\ell\rceil},1,q'}}\bigg)^\frac{\beta-\delta'\wedge\beta}{\delta'\wedge\beta}\bigg)
(1+\e\mu_{h_{\ell+\lceil\theta\ell\rceil},q'})\nu_{h_{\ell+\lceil\theta\ell\rceil},1}.
\end{aligned}
\end{equation*}
Hence
\begin{equation}
\label{estimate:L1:beta}
\mathbb{E}[\mathcal{L}_{h_{\ell+\lceil\theta\ell\rceil},1}(\widetilde\xi^{h_\ell}_n)]
\leq C^{\ell,\beta}_{1,q'}\big((h_\ell^{1+\theta}\widetilde{u}_n)\vee\gamma_n\big),
\end{equation}
where
\begin{equation*}
\begin{aligned}
C^{\ell,\beta}_{1,q'}
&=\gamma_1^{-1}(\mathbb{E}[\mathcal{L}_{h_{\ell+\lceil\theta\ell\rceil},1}(\widetilde\xi^{h_\ell}_0)]+\bar{C}^{\ell,\beta}_{1,q'}+\widetilde{C}^{\ell,\beta}_{1,q'})\\
&\hphantom{=+}\times
\sup_{m\geq1}\bigg\{m^\beta
\exp\bigg(2^{1+\beta+\delta'\wedge\beta}\zeta^{\beta,\delta'}_{h_{\ell+\lceil\theta\ell\rceil},1,q'}\gamma_1^{1+\frac{\delta'\wedge\beta}\beta}\Phi_{1-\beta-\delta'\wedge\beta}(m+1)\bigg)\\
&\hphantom{=+\times\sup_{m\geq1}\bigg\{}\times
\exp\bigg(-\frac{\lambda_{h_{\ell+\lceil\theta\ell\rceil},1}\gamma_1}2\Phi_{1-\beta}(m+1)\bigg)
\bigg\}\\
&\hphantom{=}
+2^\beta(\bar{C}^{\ell,\beta}_{1,q'}\gamma_1^{\frac{\delta'}\beta}h_{\ell+\lceil\theta\ell\rceil}
+\widetilde{C}^{\ell,\beta}_{1,q'}\gamma_1)\\
&\hphantom{=+}\times
\sup_{m\geq1}\{m^\beta
\exp(-2^{-(\beta+2)}\lambda_{h_{\ell+\lceil\theta\ell\rceil},1}\gamma_1m^{1-\beta})\Phi_{1-\beta}(m+1)
\}\\
&\hphantom{=}
+\bar{C}^{\ell,\beta}_{1,q'}
\frac{2^{1+\delta'}\gamma_1^{-1+\frac{\delta'}\beta}}{h_0\lambda_{h_{\ell+\lceil\theta\ell\rceil},1}}
+\widetilde{C}^{\ell,\beta}_{1,q'}
\frac{2^{1+\beta}}{\lambda_{h_{\ell+\lceil\theta\ell\rceil},1}},
\end{aligned}
\end{equation*}

Note that $\sup_{\ell\geq1}C^{\ell,\beta}_{1,q'}<\infty$, $q'\in\{1,2\}$, $\beta\in(0,1]$.
\\

\noindent
\emph{Step~2. Inequality on $\mathbb{E}[\mathcal{L}_{h_{\ell+\lceil\theta\ell\rceil},2}(\widetilde\xi^{h_\ell}_n)]$.}
\newline
Via \eqref{estimate:L1} and \eqref{estimate:L1:beta}, assuming that $\bar\lambda_{h_{\ell+\lceil\theta\ell\rceil},1}\gamma_1>1$ if $\beta=1$, for $q'\geq1$,
\begin{equation}
\label{eq:L1:expanded}
\mathbb{E}[\mathcal{L}_{h_{\ell+\lceil\theta\ell\rceil},1}(\widetilde\xi^{h_\ell}_n)]
\leq C^{\ell,\beta}_{1,q'}(2^{\delta'}Mh_{\ell+\lceil\theta\ell\rceil}\widetilde{u}_{n+1}+2^\beta\gamma_{n+1}).
\end{equation}
Thus, taking $q=2$, $h=h_{\ell+\lceil\theta\ell\rceil}$ and $\mu=\mu_{h_{\ell+\lceil\theta\ell\rceil},q'}$, with $q'\geq2$, in \eqref{eq:ELq<}, and using \eqref{eq:max:gamma<gamma} and \eqref{eq:L1:expanded},
\begin{equation}
\label{eq:E[L2]<}
\begin{aligned}
&\mathbb{E}[\mathcal{L}_{h_{\ell+\lceil\theta\ell\rceil},2}(\widetilde\xi^{h_\ell}_{n+1})]
\leq\mathbb{E}[\mathcal{L}_{h_{\ell+\lceil\theta\ell\rceil},2}(\widetilde\xi^{h_\ell}_n)]\\
&\times
\big(1-\lambda_{h_{\ell+\lceil\theta\ell\rceil},2}\gamma_{n+1}
+(CM\mu_{h_{\ell+\lceil\theta\ell\rceil},q'}\|V'_{h_{\ell+\lceil\theta\ell\rceil}}\|_\infty h_{\ell+\lceil\theta\ell\rceil}+\nu_{h_{\ell+\lceil\theta\ell\rceil},2})\gamma_1^{1-\frac{\delta'\wedge\beta}\beta}\gamma_{n+1}^{1+\frac{\delta'\wedge\beta}\beta}\big)\\
&\hphantom{\leq}
+2^{1+\delta'}CM^2\|V_{h_{\ell+\lceil\theta\ell\rceil}}'\|_\infty C^{\ell,\beta}_{1,q'}h_{\ell+\lceil\theta\ell\rceil}^2\gamma_{n+1}\widetilde{u}_{n+1}^2\\
&\hphantom{\leq}
+C^{\ell,\beta}_{1,q'}M(2^{1+\beta}C\|V_{h_{\ell+\lceil\theta\ell\rceil}}'\|_\infty+2^{\delta'}\nu_{h_{\ell+\lceil\theta\ell\rceil},2})h_{\ell+\lceil\theta\ell\rceil}\gamma_{n+1}^2\widetilde{u}_{n+1}\\
&\hphantom{\leq}
+(2^\beta C^{\ell,\beta}_{1,q'}+1)\nu_{h_{\ell+\lceil\theta\ell\rceil},2}\gamma_{n+1}^3\\
&
\leq\mathbb{E}[\mathcal{L}_{h_{\ell+\lceil\theta\ell\rceil},2}(\widetilde\xi^{h_\ell}_n)]
(1-\lambda_{h_{\ell+\lceil\theta\ell\rceil},2}\gamma_{n+1}
+\bar\zeta^{\beta,\delta'}_{h_{\ell+\lceil\theta\ell\rceil},2,q'}\gamma_{n+1}^{1+\frac{\delta'\wedge\beta}\beta})\\
&\hphantom{\leq}
+d^{\beta,\delta'}_{h_{\ell+\lceil\theta\ell\rceil},2,q'}\gamma_{n+1}(h_{\ell+\lceil\theta\ell\rceil}\widetilde{u}_{n+1}+\gamma_{n+1})^2\\
&
\leq\mathbb{E}[\mathcal{L}_{h_{\ell+\lceil\theta\ell\rceil},2}(\widetilde\xi^{h_\ell}_n)]
(1-\lambda_{h_{\ell+\lceil\theta\ell\rceil},2}\gamma_{n+1}
+\bar\zeta^{\beta,\delta'}_{h_{\ell+\lceil\theta\ell\rceil},2,q'}\gamma_{n+1}^{1+\frac{\delta'\wedge\beta}\beta})\\
&\hphantom{\leq}
+2d^{\beta,\delta'}_{h_{\ell+\lceil\theta\ell\rceil},2,q'}\gamma_1^{2-2\frac{\delta'}\beta}
h_{\ell+\lceil\theta\ell\rceil}^2\gamma_{n+1}^{1+2\frac{\delta'}\beta}
+2d^{\beta,\delta'}_{h_{\ell+\lceil\theta\ell\rceil},2,q'}\gamma_{n+1}^3,
\end{aligned}
\end{equation}
where we used
\begin{equation*}
\begin{aligned}
\bar\zeta^{\beta,\delta}_{h,q,q'}
&=(CM\mu_{h,q'}\|V'_h\|_\infty h+\nu^{\mu_{h,q'}}_{h,q})\gamma_1^{1-\frac{\delta\wedge\beta}\beta},\\
d^{\beta,\delta}_{h,q,q'}
&=(2^{1+\delta'}CC^{\ell,\beta}_{1,q'}M^2\|V_h'\|_\infty)
\vee\bigg(\frac{C^{\ell,\beta}_{1,q'}}2(2^{\delta'}M+2^{1+\beta}CM\|V_h'\|_\infty)\bigg)\\
&\hphantom{=(2^{1+\delta'}CC^{\ell,\beta}_{1,q'}M^2\|V_h'\|_\infty)}\;
\vee\big((2^\beta C^{\ell,\beta}_{1,q'}+1)\nu^{\mu_{h,q'}}_{h,q}\big)\\
&\qquad\qquad\qquad
\qquad\qquad\qquad
\qquad\qquad
\quad h\in\overline{\mathcal{H}},
\quad q\in\mathbb{N}^*,
\quad q'\geq q,
\quad \beta,\delta\in(0,1].
\end{aligned}
\end{equation*}
and the fact that
\begin{equation}
\label{eq:gamma:u^2}
\gamma_n\widetilde{u}_n^2
=\frac{\gamma_1^3}{n^{\beta+2\delta'}}
=\gamma_1^{2-2\frac{\delta'}\beta}\gamma_n^{1+2\frac{\delta'}\beta},
\quad n\in\mathbb{N}^*.
\end{equation}

Recall that $\gamma_1^{1-\frac{\delta'\wedge\beta}\beta}\gamma_n^{1+\frac{\delta'\wedge\beta}\beta}\geq\gamma_n^2$, $n\in\mathbb{N}^*$ by \eqref{eq:max:gamma<gamma}.
Therefore, via \eqref{eq:1-lambda>}, for all $n\in\mathbb{N}^*$, $h\in\overline{\mathcal{H}}$, $q\in\mathbb{N}^*$ and $q'\geq q$,
\begin{equation}
\label{eq:1-lambda*gamma>0}
\begin{aligned}
1-\lambda^{\mu_{h,q'}}_{h,q}\gamma_n
    &+(
        CM\mu_{h,q'}\|V'_h\|_\infty h
        +\nu^{\mu_{h,q'}}_{h,q}
    )
    \gamma_1^{1-\frac{\delta'\wedge\beta}\beta}
    \gamma_n^{1+\frac{\delta'\wedge\beta}\beta}\\
&\geq1
-\lambda^{\mu_{h,q'}}_{h,q}\gamma_n
+\nu^{\mu_{h,q'}}_{h,q}\gamma^2_n\\
&\geq\frac12>0.
\end{aligned}
\end{equation}
Hence, iterating $n$ times the inequality \eqref{eq:E[L2]<} yields, for $q'\geq2$,
\begin{equation}
\label{eq:E[L2(theta)]<}
\begin{aligned}
\mathbb{E}[\mathcal{L}_{h_{\ell+\lceil\theta\ell\rceil},2}(\widetilde\xi^{h_\ell}_n)]
&\leq\mathbb{E}[\mathcal{L}_{h_{\ell+\lceil\theta\ell\rceil},2}(\widetilde\xi^{h_\ell}_0)]
\widetilde\Pi^{h_{\ell+\lceil\theta\ell\rceil},2,q'}_{1:n}\\
&\hphantom{\leq}
+2d^{\beta,\delta'}_{h_{\ell+\lceil\theta\ell\rceil},2,q'}\gamma_1^{2-2\frac{\delta'}\beta}
\sum_{k=1}^nh_{\ell+\lceil\theta\ell\rceil}^2\gamma_k^{1+2\frac{\delta'}\beta}
\widetilde\Pi^{h_{\ell+\lceil\theta\ell\rceil},2,q'}_{k+1:n}\\
&\hphantom{\leq}
+2d^{\beta,\delta'}_{h_{\ell+\lceil\theta\ell\rceil},2,q'}
\sum_{k=1}^n\gamma_k^3
\widetilde\Pi^{h_{\ell+\lceil\theta\ell\rceil},2,q'}_{k+1:n},
\end{aligned}
\end{equation}
where
\begin{multline}
\label{eq:Pi:~}
\widetilde\Pi^{h,q,q'}_{k:n}:=\prod_{j=k}^n
\big(
    1-\lambda^{\mu_{h,q'}}_{h,q}\gamma_j
    +(
        CM\mu_{h,q'}\|V'_h\|_\infty h
        +\nu^{\mu_{h,q'}}_{h,q}
    )
    \gamma_1^{1-\frac{\delta'\wedge\beta}\beta}
    \gamma_j^{1+\frac{\delta'\wedge\beta}\beta}
\big),\\
\quad h\in\overline{\mathcal{H}},
\quad q\in\mathbb{N}^*,
\quad q'\geq q,
\quad k,n\in\mathbb{N}^*,
\end{multline}
with the convention $\prod_\varnothing=1$.

If $\gamma_n=\gamma_1n^{-1}$, $n\in\mathbb{N}^*$, then, by Lemma~\ref{lmm:gamma},
\begin{equation}
\label{estimate:L2:1}
\begin{aligned}
\mathbb{E}[\mathcal{L}_{h_{\ell+\lceil\theta\ell\rceil},2}(\widetilde{\xi}^{h_\ell}_n)]
&\leq\widehat{C}^{\ell,1}_{2,q'}\frac{\mathbb{E}[\mathcal{L}_{h_{\ell+\lceil\theta\ell\rceil},2}(\widetilde{\xi}^{h_\ell}_0)]}{(n+1)^{\lambda_{h_{\ell+\lceil\theta\ell\rceil},2}\gamma_1}}
+\bar{C}^{\ell,1}_{2,q'}\frac{\Phi_{\lambda_{h_{\ell+\lceil\theta\ell\rceil},2}\gamma_1-2\delta'}(n+1)}{(n+1)^{\lambda_{h_{\ell+\lceil\theta\ell\rceil},2}\gamma_1}}h_{\ell+\lceil\theta\ell\rceil}^2\\
&\hphantom{\leq}
+\widetilde{C}^{\ell,1}_{2,q'}\frac{\Phi_{\lambda_{h_{\ell+\lceil\theta\ell\rceil},2}\gamma_1-2}(n+1)}{(n+1)^{\lambda_{h_{\ell+\lceil\theta\ell\rceil},2}\gamma_1}},
\end{aligned}
\end{equation}
where
\begin{equation*}
\begin{aligned}
\widehat{C}^{\ell,1}_{2,q'}
&=\exp\bigg(\bar\zeta^{1,\delta'}_{h_{\ell+\lceil\theta\ell\rceil},2,q'}\Big(1+\frac1{\delta'}\Big)\gamma_1^{1+\delta'}+\frac{\lambda_{h_{\ell+\lceil\theta\ell\rceil},2}\gamma_1}2\bigg),\\
\bar{C}^{\ell,1}_{2,q'}
&=2^{1+(\lambda_{h_{\ell+\lceil\theta\ell\rceil},2}\gamma_1)\vee(1+\delta')}
\widehat{C}^{\ell,1}_{2,q'}
d^{1,\delta'}_{h_{\ell+\lceil\theta\ell\rceil},2,q'}
\gamma_1^3,\\
\widetilde{C}^{\ell,1}_{2,q'}
&=2^{1+(\lambda_{h_{\ell+\lceil\theta\ell\rceil},2}\gamma_1)\vee2}
\widehat{C}^{\ell,1}_{2,q'}
d^{1,\delta'}_{h_{\ell+\lceil\theta\ell\rceil},2,q'}
\gamma_1^3.
\end{aligned}
\end{equation*}
If $\bar\lambda_{h_{\ell+\lceil\theta\ell\rceil},1}\gamma_1
>1$, so that $\lambda_{h_{\ell+\lceil\theta\ell\rceil},2}\gamma_1
\geq\bar\lambda_{h_{\ell+\lceil\theta\ell\rceil},2}\gamma_1
\geq\bar\lambda_{h_{\ell+\lceil\theta\ell\rceil},1}\gamma_1
>1
\geq\delta'$, then
\begin{equation}
\label{estimate:lyapunov:2:1}
\mathbb{E}[\mathcal{L}_{h_{\ell+\lceil\theta\ell\rceil},2}(\widetilde\xi^{h_\ell}_n)]
\leq C^{\ell,1}_{2,q'}\big((h_\ell^{1+\theta}\widetilde{u}_n)\vee\gamma_n\big),
\end{equation}
where
\begin{equation*}
\begin{aligned}
C^{\ell,1}_{2,q'}
&=\gamma_1^{-1}
\bigg(\widehat{C}^{\ell,1}_{2,q'}\mathbb{E}[\mathcal{L}_{h_{\ell+\lceil\theta\ell\rceil},2}(\widetilde\xi^{h_\ell}_0)]\\
&\hphantom{=\gamma_1^{-1}
\bigg(}
+\frac{\bar{C}^{\ell,1}_{2,q'}}{M^{\ell+\lceil\theta\ell\rceil}}
\sup_{m\geq1}\bigg\{m^{\delta'}\frac{\Phi_{\lambda_{h_{\ell+\lceil\theta\ell\rceil},2}\gamma_1-2\delta'}(m+1)}{(m+1)^{\lambda_{h_{\ell+\lceil\theta\ell\rceil},2}\gamma_1}}\bigg\}\\
&\hphantom{=\gamma_1^{-1}
\bigg(}
+\widetilde{C}^{\ell,1}_{2,q'}
\sup_{m\geq1}\bigg\{m\frac{\Phi_{\lambda_{h_{\ell+\lceil\theta\ell\rceil},2}\gamma_1-2}(m+1)}{(m+1)^{\lambda_{h_{\ell+\lceil\theta\ell\rceil},2}\gamma_1}}\bigg\}\bigg).
\end{aligned}
\end{equation*}

If $\gamma_n=\gamma_1n^{-\beta}$, $n\in\mathbb{N}^*$, with $\beta\in(0,1)$, then, as a consequence of Lemma~\ref{lmm:gamma},
\begin{equation}
\label{estimate:L2:beta}
\begin{aligned}
&\mathbb{E}[\mathcal{L}_{h_{\ell+\lceil\theta\ell\rceil},2}(\widetilde\xi^{h_\ell}_n)]
\leq
(\mathbb{E}[\mathcal{L}_{h_{\ell+\lceil\theta\ell\rceil},2}(\widetilde\xi^{h_\ell}_0)]+\bar{C}^{\ell,\beta}_{2,q'}+\widetilde{C}^{\ell,\beta}_{2,q'})\\
&\quad\times
\exp\bigg(2^{1+\beta+\delta'\wedge\beta}\bar\zeta^{\beta,\delta'}_{h_{\ell+\lceil\theta\ell\rceil},2,q'}\gamma_1^{1+\frac{\delta'\wedge\beta}\beta}\Phi_{1-\beta-\delta'\wedge\beta}(n+1)
-\frac{\lambda_{h_{\ell+\lceil\theta\ell\rceil},2}\gamma_1}2\Phi_{1-\beta}(n+1)\bigg)\\
&+2^\beta(\bar{C}^{\ell,\beta}_{2,q'}\gamma_1^{1+2\frac{\delta'}\beta}h_{\ell+\lceil\theta\ell\rceil}^2
+\widetilde{C}^{\ell,\beta}_{2,q'}\gamma_1^3)
\exp(-2^{-(\beta+2)}\lambda_{h_{\ell+\lceil\theta\ell\rceil},2}\gamma_1n^{1-\beta})\Phi_{1-\beta}(n+1)\\
&
+\bar{C}^{\ell,\beta}_{2,q'}\frac{2^{1+2\delta'}\gamma_1^{2\frac{\delta'}\beta}}{\lambda_{h_{\ell+\lceil\theta\ell\rceil},2}n^{2\delta'}}
h_{\ell+\lceil\theta\ell\rceil}^2
+\widetilde{C}^{\ell,\beta}_{2,q'}
\frac{2^{1+2\beta}\gamma_1^2}{\lambda_{h_{\ell+\lceil\theta\ell\rceil},2}n^{2\beta}},
\end{aligned}
\end{equation}
where
\begin{equation*}
\begin{aligned}
\bar{C}^{\ell,\beta}_{2,q'}
&=
(\bar\zeta^{\beta,\delta'}_{h_{\ell+\lceil\theta\ell\rceil},2,q'})^{-1}
\bigg(\gamma_1^\frac{2\delta'-\delta'\wedge\beta}\beta\vee\bigg(\frac{\lambda_{h_{\ell+\lceil\theta\ell\rceil},2}}{2\bar\zeta^{\beta,\delta'}_{h_{\ell+\lceil\theta\ell\rceil},2,q'}}\bigg)^\frac{2\delta'-\delta'\wedge\beta}{\delta'\wedge\beta}\bigg),\\
\widetilde{C}^{\ell,\beta}_{2,q'}
&=
(\bar\zeta^{\beta,\delta'}_{h_{\ell+\lceil\theta\ell\rceil},2,q'})^{-1}
\bigg(\gamma_1^\frac{2\beta-\delta'\wedge\beta}\beta\vee\bigg(\frac{\lambda_{h_{\ell+\lceil\theta\ell\rceil},2}}{2\bar\zeta^{\beta,\delta'}_{h_{\ell+\lceil\theta\ell\rceil},2,q'}}\bigg)^\frac{2\beta-\delta'\wedge\beta}{\delta'\wedge\beta}\bigg).
\end{aligned}
\end{equation*}
Thus
\begin{equation}
\label{estimate:lyapunov:2:beta}
    \mathbb{E}[\mathcal{L}_{h_{\ell+\lceil\theta\ell\rceil},2}(\widetilde\xi^{h_\ell}_n)]
    \leq C^{\ell,\beta}_{2,q'}\big((h_\ell^{1+\theta}\widetilde{u}_n)\vee\gamma_n\big),
\end{equation}
where
\begin{equation*}
\begin{aligned}
C^{\ell,\beta}_{2,q'}
&=\gamma_1^{-1}(\mathbb{E}[\mathcal{L}_{h_{\ell+\lceil\theta\ell\rceil},2}(\widetilde\xi^{h_\ell}_0)]+\bar{C}^{\ell,\beta}_{2,q'}+\widetilde{C}^{\ell,\beta}_{2,q'})\\
&\hphantom{=+}\times
\sup_{m\geq1}\bigg\{m^\beta\exp\bigg(2^{1+\beta+\delta'\wedge\beta}\bar\zeta^{\beta,\delta'}_{h_{\ell+\lceil\theta\ell\rceil},2,q'}\gamma_1^{1+\frac{\delta'\wedge\beta}\beta}\Phi_{1-\beta-\delta'\wedge\beta}(m+1)\bigg)\\
&\hphantom{=+\times\sup_{m\geq1}\bigg\{}\times
\exp\bigg(-\frac{\lambda_{h_{\ell+\lceil\theta\ell\rceil},2}\gamma_1}2\Phi_{1-\beta}(m+1)\bigg)\bigg\}\\
&\hphantom{=}
+2^\beta(\bar{C}^{\ell,\beta}_{2,q'}\gamma_1^{2\frac{\delta'}\beta}h_{\ell+\lceil\theta\ell\rceil}^2
+\widetilde{C}^{\ell,\beta}_{2,q'}\gamma_1^2)\\
&\hphantom{=+}\times
\sup_{m\geq1}\{m^\beta\exp(-2^{-(\beta+2)}\lambda_{h_{\ell+\lceil\theta\ell\rceil},2}\gamma_1m^{1-\beta})\Phi_{1-\beta}(m+1)\}\\
&\hphantom{=}
+\bar{C}^{\ell,\beta}_{2,q'}\frac{2^{1+2\delta'}\gamma_1^{-1+2\frac{\delta'}\beta}}{h_0M^{\ell+\lceil\theta\ell\rceil}\lambda_{h_{\ell+\lceil\theta\ell\rceil},2}}
+\widetilde{C}^{\ell,\beta}_{2,q'}
\frac{2^{1+2\beta}\gamma_1}{\lambda_{h_{\ell+\lceil\theta\ell\rceil},2}}.
\end{aligned}
\end{equation*}

Note that $\sup_{\ell\geq1}C^{\ell,\beta}_{2,2}<\infty$, $\beta\in(0,1]$.
\\

\noindent
\emph{Step~3. Conclusion.}
\newline
Applying Lemma \ref{lmm:lyapunov}(\ref{lmm:lyapunov-iv}) with $h=h_{\ell+\lceil\theta\ell\rceil}$, $q=1$ yields
\begin{equation*}
\mathbb{E}[(\widetilde\xi^{h_\ell}_n-\xi^{h_{\ell+\lceil\theta\ell\rceil}}_\star)^2]
\leq
\kappa_{h_{\ell+\lceil\theta\ell\rceil},1}
\big(\mathbb{E}[\bar{\mathcal{L}}_{h_{\ell+\lceil\theta\ell\rceil},1}(\widetilde\xi^{h_\ell}_n)]
+\mathbb{E}[\bar{\mathcal{L}}_{h_{\ell+\lceil\theta\ell\rceil},2}(\widetilde\xi^{h_\ell}_n)]\big).
\end{equation*}
Thus, via \eqref{estimate:L1}, \eqref{estimate:L1:beta}, \eqref{estimate:lyapunov:2:1} and \eqref{estimate:lyapunov:2:beta}, if $\bar\lambda_{h_{\ell+\lceil\theta\ell\rceil},1}\gamma_1
>1$ when $\beta=1$, so that $\bar\lambda_{h_{\ell+\lceil\theta\ell\rceil},2}\gamma_1
\geq\bar\lambda_{h_{\ell+\lceil\theta\ell\rceil},1}\gamma_1
>1$ when $\beta=1$, then
\begin{equation*}
\mathbb{E}[(\widetilde\xi^{h_\ell}_n-\xi^{h_{\ell+\lceil\theta\ell\rceil}}_\star)^2]
\leq C^{\ell,\beta}\big((h_\ell^{1+\theta}\widetilde{u}_n)\vee\gamma_n\big),
\end{equation*}
where
\begin{equation}
\label{eq:C^l,b}
C^{\ell,\beta}=\kappa_{h_{\ell+\lceil\theta\ell\rceil},1}(C^{\ell,\beta}_{1,1}+C^{\ell,\beta}_{2,2}).
\end{equation}
Observe that $\sup_{\ell\geq1}C^{\ell,\beta}<\infty$, $\beta\in(0,1]$.
\\

\noindent(\ref{prp:error:statistical:bis:ii})\
Our goal is to apply Lemma~\ref{lmm:lyapunov}(\ref{lmm:lyapunov-iv}) with $q=2$ and $h=h_{\ell+\lceil\theta\ell\rceil}$.
Via Lemma~\ref{lmm:lyapunov}(\ref{lmm:lyapunov-i-bis}), by assumption, for $\mu\in\{\mu_{h_{\ell+\lceil\theta\ell\rceil},q'},q'\in\{1,2,3,4\}\}$,
\begin{equation*}
\begin{aligned}
\mathbb{E}[\mathcal{L}_{h_{\ell+\lceil\theta\ell\rceil},1}(\widetilde\xi^{h_\ell}_0)]
&\vee\mathbb{E}[\mathcal{L}_{h_{\ell+\lceil\theta\ell\rceil},2}(\widetilde\xi^{h_\ell}_0)]
\vee\mathbb{E}[\mathcal{L}_{h_{\ell+\lceil\theta\ell\rceil},3}(\widetilde\xi^{h_\ell}_0)]
\vee\mathbb{E}[\mathcal{L}_{h_{\ell+\lceil\theta\ell\rceil},4}(\widetilde\xi^{h_\ell}_0)]\\
&\leq\mathbb{E}\big[\big(1+8k_\alpha^4(|\widetilde\xi^{h_\ell}_0|^4+|\xi^{h_{\ell+\lceil\theta\ell\rceil}}_\star|^4)\big)
\exp\big(k_\alpha(|\widetilde\xi_0^{h_\ell}|+|\xi^{h_{\ell+\lceil\theta\ell\rceil}}_\star|)\big)\big]
<\infty.
\end{aligned}
\end{equation*}
\\

\noindent
\emph{Step~1. Alternative inequality on $\mathbb{E}[\mathcal{L}_{h_{\ell+\lceil\theta\ell\rceil},2}(\widetilde\xi^{h_\ell}_n)]$.}
\newline
If $\gamma_n=\gamma_1n^{-1}$, $n\in\mathbb{N}^*$, and $\bar\lambda_{h_{\ell+\lceil\theta\ell\rceil},2}\gamma_1>2$, which implies in particular that $\bar\lambda_{h_{\ell+\lceil\theta\ell\rceil},1}\gamma_1
=\frac{\bar\lambda_{h_{\ell+\lceil\theta\ell\rceil},2}\gamma_1}2>1$
and that, for $q'\geq2$, $\lambda_{h_{\ell+\lceil\theta\ell\rceil},2}^{\mu_{h_{\ell+\lceil\theta\ell\rceil},q'}}\gamma_1
\geq\bar\lambda_{h_{\ell+\lceil\theta\ell\rceil},2}\gamma_1
>2
\geq2\delta'$,
then, recalling the control on $\mathbb{E}[\mathcal{L}_{h_{\ell+\lceil\theta\ell\rceil},2}(\widetilde\xi^{h_\ell}_n)]$ derived in \eqref{estimate:L2:1}, for $q'\geq2$, 
\begin{equation}
\label{eq:EL2<:1}
\mathbb{E}[\mathcal{L}_{h_{\ell+\lceil\theta\ell\rceil},2}(\widetilde\xi^{h_\ell}_n)]
\leq\check{C}^{\ell,1}_{2,q'}\big((h_\ell^{2(1+\theta)}\widetilde{u}_n^2)\vee\gamma_n^2\big),
\end{equation}
where
\begin{equation*}
\check{C}^{\ell,1}_{2,q'}
=\gamma_1^{-2}
(\widehat{C}^{\ell,1}_{2,q'}\mathbb{E}[\mathcal{L}_{h_{\ell+\lceil\theta\ell\rceil},2}(\widetilde\xi^{h_\ell}_0)]
+h_0^{-2}\bar{C}^{\ell,1}_{2,q'}
+\widetilde{C}^{\ell,1}_{2,q'}).
\end{equation*}

Reconsidering the inequality on $\mathbb{E}[\mathcal{L}_{h_{\ell+\lceil\theta\ell\rceil},2}(\widetilde\xi^{h_\ell}_n)]$ obtained in \eqref{estimate:L2:beta},
if $\gamma_n=\gamma_1n^{-\beta}$, $n\in\mathbb{N}^*$, with $\beta\in(0,1)$, then
\begin{equation}
\label{eq:EL2<:beta}
    \mathbb{E}[\mathcal{L}_{h_{\ell+\lceil\theta\ell\rceil},2}(\widetilde\xi^{h_\ell}_n)]\leq \check{C}^{\ell,\beta}_{2,q'}\big((h_\ell^{2(1+\theta)}\widetilde{u}_n^2)\vee\gamma_n^2\big),
\end{equation}
where
\begin{equation*}
\begin{aligned}
\check{C}^{\ell,\beta}_{2,q'}
&=\gamma_1^{-2}(\mathbb{E}[\mathcal{L}_{h_{\ell+\lceil\theta\ell\rceil},2}(\widetilde\xi^{h_\ell}_0)]+\bar{C}^{\ell,\beta}_{2,q'}+\widetilde{C}^{\ell,\beta}_{2,q'})\\
&\hphantom{=+}\times
\sup_{m\geq1}\bigg\{m^{2\beta}\exp\bigg(2^{1+\beta+\delta'\wedge\beta}\bar\zeta^{\beta,\delta'}_{h_{\ell+\lceil\theta\ell\rceil},2,q'}\gamma_1^{1+\frac{\delta'\wedge\beta}\beta}\Phi_{1-\beta-\delta'\wedge\beta}(m+1)\bigg)\\
&\hphantom{=+\times\sup_{m\geq1}\bigg\{}\times
\exp\bigg(-\frac{\lambda_{h_{\ell+\lceil\theta\ell\rceil},2}\gamma_1}2\Phi_{1-\beta}(m+1)\bigg)\bigg\}\\
&\hphantom{=}
+2^\beta(\bar{C}^{\ell,\beta}_{2,q'}\gamma_1^{-1+2\frac{\delta'}\beta}h_{\ell+\lceil\theta\ell\rceil}^2
+\widetilde{C}^{\ell,\beta}_{2,q'}\gamma_1)\\
&\hphantom{=+}\times
\sup_{m\geq1}\{m^{2\beta}\exp(-2^{-(\beta+2)}\lambda_{h_{\ell+\lceil\theta\ell\rceil},2}\gamma_1m^{1-\beta})\Phi_{1-\beta}(m+1)\}\\
&\hphantom{=}
+\bar{C}^{\ell,\beta}_{2,q'}\frac{2^{1+2\delta'}\gamma_1^{-2+2\frac{\delta'}\beta}}{h_0^2\lambda_{h_{\ell+\lceil\theta\ell\rceil},2}}
+\widetilde{C}^{\ell,\beta}_{2,q'}
\frac{2^{1+2\beta}}{\lambda_{h_{\ell+\lceil\theta\ell\rceil},2}}.
\end{aligned}
\end{equation*}

Note that $\sup_{\ell\geq1}\check{C}^{\ell,\beta}_{2,q'}<\infty$, $q'\in\{2,3,4\}$, $\beta\in(0,1]$.
\\

\noindent
\emph{Step~2. Inequality on $\mathbb{E}[\mathcal{L}_{h_{\ell+\lceil\theta\ell\rceil},3}(\widetilde\xi^{h_\ell}_n)]$.}
\newline
To be able to apply Lemma~\ref{lmm:lyapunov}(\ref{lmm:lyapunov-iv}) with $q=2$ and $h=h_{\ell+\lceil\theta\ell\rceil}$, one needs a control on $\mathbb{E}[\bar{\mathcal{L}}_{h_{\ell+\lceil\theta\ell\rceil},4}(\widetilde\xi^{h_\ell}_n)]$, which, ostensibly from the recursive relation \eqref{eq:ELq<}, requires in turn a control on $\mathbb{E}[\mathcal{L}_{h_{\ell+\lceil\theta\ell\rceil},3}(\widetilde\xi^{h_\ell}_n)]$.
For $q'\geq3$, assuming that $\bar\lambda_{h_{\ell+\lceil\theta\ell\rceil},2}\gamma_1
>2$ if $\beta=1$, then, by \eqref{eq:EL2<:1} and \eqref{eq:EL2<:beta},
\begin{equation*}
\mathbb{E}[\mathcal{L}_{h_{\ell+\lceil\theta\ell\rceil},2}(\widetilde\xi^{h_\ell}_n)]
\leq\check{C}^{\ell,\beta}_{2,q'}(2^{2\delta'}M^2h_{\ell+\lceil\theta\ell\rceil}^2\widetilde{u}_{n+1}^2
+2^{2\beta}\gamma_{n+1}^2).
\end{equation*}
Thus, by taking $q=3$ and $h=h_{\ell+\lceil\theta\ell\rceil}$ in \eqref{eq:ELq<}, for $q'\geq3$,
\begin{equation}
\label{eq:EL3<}
\begin{aligned}
\mathbb{E}&[\mathcal{L}_{h_{\ell+\lceil\theta\ell\rceil},3}(\widetilde\xi^{h_\ell}_{n+1})]
\leq\mathbb{E}[\mathcal{L}_{h_{\ell+\lceil\theta\ell\rceil},3}(\widetilde\xi^{h_\ell}_n)]
(1
-\lambda_{h_{\ell+\lceil\theta\ell\rceil},3}\gamma_{n+1}
+\bar\zeta^{\beta,\delta'}_{h_{\ell+\lceil\theta\ell\rceil},3,q'}\gamma_{n+1}^{1+\frac{\delta'\wedge\beta}\beta})\\
&\quad
+2^{2\delta'}3C\check{C}^{\ell,\beta}_{2,q'}M^3\|V_{h_{\ell+\lceil\theta\ell\rceil}}'\|_\infty
h_{\ell+\lceil\theta\ell\rceil}^3\widetilde{u}_{n+1}^3\gamma_{n+1}\\
&\quad
+2^{2\beta}3C\check{C}^{\ell,\beta}_{2,q'}M\|V_{h_{\ell+\lceil\theta\ell\rceil}}'\|_\infty
h_{\ell+\lceil\theta\ell\rceil}\widetilde{u}_{n+1}\gamma_{n+1}^3\\
&\quad
+2^{2\delta'}\check{C}^{\ell,\beta}_{2,q'}M^2\nu_{h_{\ell+\lceil\theta\ell\rceil},3}
h_{\ell+\lceil\theta\ell\rceil}^2\widetilde{u}_{n+1}^2\gamma_{n+1}^2\\
&\quad
+(2^{2\beta}\check{C}^{\ell,\beta}_{2,q'}+1)\nu_{h_{\ell+\lceil\theta\ell\rceil},3}
\gamma_{n+1}^4\\
&\leq\mathbb{E}[\mathcal{L}_{h_{\ell+\lceil\theta\ell\rceil},3}(\widetilde\xi^{h_\ell}_n)]
(1
-\lambda_{h_{\ell+\lceil\theta\ell\rceil},3}\gamma_{n+1}
+\bar\zeta^{\beta,\delta'}_{h_{\ell+\lceil\theta\ell\rceil},3,q'}\gamma_{n+1}^{1+\frac{\delta'\wedge\beta}\beta})\\
&\quad
+\bar{d}^{\beta,\delta'}_{h_{\ell+\lceil\theta\ell\rceil},3,q'}
\gamma_{n+1}(h_{\ell+\lceil\theta\ell\rceil}\widetilde{u}_{n+1}+\gamma_{n+1})^3\\
&\leq\mathbb{E}[\mathcal{L}_{h_{\ell+\lceil\theta\ell\rceil},3}(\widetilde\xi^{h_\ell}_n)]
(1
-\lambda_{h_{\ell+\lceil\theta\ell\rceil},3}\gamma_{n+1}
+\bar\zeta^{\beta,\delta'}_{h_{\ell+\lceil\theta\ell\rceil},3,q'}\gamma_{n+1}^{1+\frac{\delta'\wedge\beta}\beta})\\
&\quad
+4\bar{d}^{\beta,\delta'}_{h_{\ell+\lceil\theta\ell\rceil},3,q'}\gamma_1^{3-3\frac{\delta'}\beta}
h_{\ell+\lceil\theta\ell\rceil}^3\gamma_{n+1}^{1+3\frac{\delta'}\beta}
+4\bar{d}^{\beta,\delta'}_{h_{\ell+\lceil\theta\ell\rceil},3,q'}
\gamma_{n+1}^4,
\end{aligned}
\end{equation}
where we introduced
\begin{equation*}
\begin{aligned}
\bar{d}_{h,q,q'}^{\beta,\delta}
&
=(2^{2\delta'}3C\check{C}^{\ell,\beta}_{2,q'}M^3\|V_h'\|_\infty)
\vee
(2^{2\beta}C\check{C}^{\ell,\beta}_{2,q'}M\|V_h'\|_\infty)\\
&\hphantom{=(2^{2\delta'}3C\check{C}^{\ell,\beta}_{2,q'}M^3\|V_h'\|_\infty)}\;
\vee
\bigg(\frac{2^{2\delta'}}3\check{C}^{\ell,\beta}_{2,q'}M^2\nu^{\mu_{h,q'}}_{h,q}\bigg)
\vee
\big((2^{2\beta}\check{C}^{\ell,\beta}_{2,q'}+1)\nu^{\mu_{h,q'}}_{h,q}\big),\\
&\qquad\qquad\qquad
\qquad\qquad\qquad
\qquad\qquad
      h\in\overline{\mathcal{H}},
\quad q\in\mathbb{N}^*,
\quad q'\geq q,
\quad \beta,\delta\in(0,1],
\end{aligned}
\end{equation*}
and used that
\begin{equation}
\label{eq:gamma:u^3}
\gamma_n\widetilde{u}_n^3
=\frac{\gamma_1^4}{n^{\beta+3\delta'}}
=\gamma_1^{3-3\frac{\delta'}\beta}\gamma_n^{1+3\frac{\delta'}\beta},
\quad n\in\mathbb{N}^*.
\end{equation}
Considering \eqref{eq:1-lambda*gamma>0} and recalling \eqref{eq:Pi:~}, for $q'\geq3$, by iterating $n$ times the inequality \eqref{eq:EL3<},
\begin{equation*}
\begin{aligned}
\mathbb{E}[\mathcal{L}_{h_{\ell+\lceil\theta\ell\rceil},3}(\widetilde\xi^{h_\ell}_n)]
&\leq\mathbb{E}[\mathcal{L}_{h_{\ell+\lceil\theta\ell\rceil},3}(\widetilde\xi^{h_\ell}_0)]\widetilde\Pi^{h_{\ell+\lceil\theta\ell\rceil},3,q'}_{1:n}\\
&\hphantom{\leq}
+4\bar{d}^{\beta,\delta'}_{h_{\ell+\lceil\theta\ell\rceil},3,q'}\gamma_1^{3-3\frac{\delta'}\beta}
\sum_{k=1}^nh_{\ell+\lceil\theta\ell\rceil}^3
\gamma_k^{1+3\frac{\delta'}\beta}
\widetilde\Pi^{h_{\ell+\lceil\theta\ell\rceil},3,q'}_{k+1:n}\\
&\hphantom{\leq}
+4\bar{d}^{\beta,\delta'}_{h_{\ell+\lceil\theta\ell\rceil},3,q'}
\sum_{k=1}^n\gamma_k^4
\widetilde\Pi^{h_{\ell+\lceil\theta\ell\rceil},3,q'}_{k+1:n}.
\end{aligned}
\end{equation*}

If $\gamma_n=\gamma_1n^{-1}$, $n\in\mathbb{N}^*$, then, by Lemma~\ref{lmm:gamma},
\begin{equation*}
\begin{aligned}
\mathbb{E}[\mathcal{L}_{h_{\ell+\lceil\theta\ell\rceil},3}(\widetilde{\xi}^{h_\ell}_n)]
&\leq\widehat{C}^{\ell,1}_{3,q'}
\frac{\mathbb{E}[\mathcal{L}_{h_{\ell+\lceil\theta\ell\rceil},3}(\widetilde{\xi}^{h_\ell}_0)]}{(n+1)^{\lambda_{h_{\ell+\lceil\theta\ell\rceil},3}\gamma_1}}
+\bar{C}^{\ell,1}_{3,q'}
\frac{\Phi_{\lambda_{h_{\ell+\lceil\theta\ell\rceil},3}\gamma_1-3\delta'}(n+1)}{(n+1)^{\lambda_{h_{\ell+\lceil\theta\ell\rceil},3}\gamma_1}}
h_{\ell+\lceil\theta\ell\rceil}^3\\
&\hphantom{\leq}
+\widetilde{C}^{\ell,1}_{3,q'}
\frac{\Phi_{\lambda_{h_{\ell+\lceil\theta\ell\rceil},3}\gamma_1-3}(n+1)}{(n+1)^{\lambda_{h_{\ell+\lceil\theta\ell\rceil},3}\gamma_1}},
\end{aligned}
\end{equation*}
where
\begin{equation*}
\begin{aligned}
\widehat{C}^{\ell,1}_{3,q'}
&=\exp\bigg(\bar\zeta^{1,\delta'}_{h_{\ell+\lceil\theta\ell\rceil},3,q'}\Big(1+\frac1{\delta'}\Big)\gamma_1^{1+\delta'}+\frac{\lambda_{h_{\ell+\lceil\theta\ell\rceil},3}\gamma_1}2\bigg),\\
\bar{C}^{\ell,1}_{3,q'}
&=2^{2+(\lambda_{h_{\ell+\lceil\theta\ell\rceil},3}\gamma_1)\vee(1+3\delta')}
\widehat{C}^{\ell,1}_{3,q'}
\bar{d}^{1,\delta'}_{h_{\ell+\lceil\theta\ell\rceil},3,q'}
\gamma_1^4,\\
\widetilde{C}^{\ell,1}_{3,q'}
&=2^{2+(\lambda_{h_{\ell+\lceil\theta\ell\rceil},3}\gamma_1)\vee4}
\widehat{C}^{\ell,1}_{3,q'}
\bar{d}^{1,\delta'}_{h_{\ell+\lceil\theta\ell\rceil},3,q'}
\gamma_1^4.
\end{aligned}
\end{equation*}
In particular, if $\bar\lambda_{h_{\ell+\lceil\theta\ell\rceil},2}>2$, so that $\lambda_{h_{\ell+\lceil\theta\ell\rceil},3}\gamma_1
\geq\lambda_{h_{\ell+\lceil\theta\ell\rceil},2}\gamma_1
\geq\bar\lambda_{h_{\ell+\lceil\theta\ell\rceil},2}\gamma_1
>2
\geq2\delta'$, then
\begin{equation}
\label{eq:EL3<:1}
\mathbb{E}[\mathcal{L}_{h_{\ell+\lceil\theta\ell\rceil},3}(\widetilde\xi^{h_\ell}_n)]
\leq C^{\ell,1}_{3,q'}\big((h_\ell^{2(1+\theta)}\widetilde{u}_n^2)\vee\gamma_n^2\big),
\end{equation}
where
\begin{equation*}
\begin{aligned}
C^{\ell,1}_{3,q'}
&=\gamma_1^{-2}
\bigg(\widehat{C}^{\ell,1}_{3,q'}\mathbb{E}[\mathcal{L}_{h_{\ell+\lceil\theta\ell\rceil},3}(\widetilde\xi^{h_\ell}_0)]\\
&\hphantom{=\gamma_1^{-1}
\bigg(}
+\frac{\bar{C}^{\ell,1}_{3,q'}}{h_0M^{\ell+\lceil\theta\ell\rceil}}
\sup_{m\geq1}\bigg\{m^{2\delta'}\frac{\Phi_{\lambda_{h_{\ell+\lceil\theta\ell\rceil},3}\gamma_1-3\delta'}(m+1)}{(m+1)^{\lambda_{h_{\ell+\lceil\theta\ell\rceil},3}\gamma_1}}\bigg\}\\
&\hphantom{=\gamma_1^{-2}
\bigg(}
+\widetilde{C}^{\ell,1}_{3,q'}
\sup_{m\geq1}\bigg\{m^2\frac{\Phi_{\lambda_{h_{\ell+\lceil\theta\ell\rceil},3}\gamma_1-3}(m+1)}{(m+1)^{\lambda_{h_{\ell+\lceil\theta\ell\rceil},3}\gamma_1}}\bigg\}\bigg).
\end{aligned}
\end{equation*}

If $\gamma_n=\gamma_1n^{-\beta}$, $n\in\mathbb{N}^*$, with $\beta\in(0,1)$, then, by Lemma~\ref{lmm:gamma},
\begin{equation*}
\begin{aligned}
&\mathbb{E}[\mathcal{L}_{h_{\ell+\lceil\theta\ell\rceil},3}(\widetilde\xi^{h_\ell}_n)]
\leq
(\mathbb{E}[\mathcal{L}_{h_{\ell+\lceil\theta\ell\rceil},3}(\widetilde\xi^{h_\ell}_0)]+\bar{C}^{\ell,\beta}_{3,q'}+\widetilde{C}^{\ell,\beta}_{3,q'})\\
&\quad\times
\exp\bigg(2^{1+\beta+\delta'\wedge\beta}\bar\zeta^{\beta,\delta'}_{h_{\ell+\lceil\theta\ell\rceil},3,q'}\gamma_1^{1+\frac{\delta'\wedge\beta}\beta}\Phi_{1-\beta-\delta'\wedge\beta}(n+1)
-\frac{\lambda_{h_{\ell+\lceil\theta\ell\rceil},3}\gamma_1}2\Phi_{1-\beta}(n+1)\bigg)\\
&+2^\beta(\bar{C}^{\ell,\beta}_{3,q'}\gamma_1^{1+3\frac{\delta'}\beta}h_{\ell+\lceil\theta\ell\rceil}^3
+\widetilde{C}^{\ell,\beta}_{3,q'}\gamma_1^4)
\exp(-2^{-(\beta+2)}\lambda_{h_{\ell+\lceil\theta\ell\rceil},3}\gamma_1n^{1-\beta})\Phi_{1-\beta}(n+1)\\
&
+\bar{C}^{\ell,\beta}_{3,q'}\frac{2^{1+3\delta'}\gamma_1^{3\frac{\delta'}\beta}}{\lambda_{h_{\ell+\lceil\theta\ell\rceil},3}n^{3\delta'}}
h_{\ell+\lceil\theta\ell\rceil}^3
+\widetilde{C}^{\ell,\beta}_{3,q'}
\frac{2^{1+3\beta}\gamma_1^3}{\lambda_{h_{\ell+\lceil\theta\ell\rceil},3}n^{3\beta}},
\end{aligned}
\end{equation*}
where
\begin{equation*}
\begin{aligned}
\bar{C}^{\ell,\beta}_{3,q'}
&=
(\bar\zeta^{\beta,\delta'}_{h_{\ell+\lceil\theta\ell\rceil},3,q'})^{-1}
\bigg(\gamma_1^\frac{3\delta'-\delta'\wedge\beta}\beta\vee\bigg(\frac{\lambda_{h_{\ell+\lceil\theta\ell\rceil},3}}{2\bar\zeta^{\beta,\delta'}_{h_{\ell+\lceil\theta\ell\rceil},3,q'}}\bigg)^\frac{3\delta'-\delta'\wedge\beta}{\delta'\wedge\beta}\bigg),\\
\widetilde{C}^{\ell,\beta}_{3,q'}
&=
(\bar\zeta^{\beta,\delta'}_{h_{\ell+\lceil\theta\ell\rceil},3,q'})^{-1}
\bigg(\gamma_1^\frac{3\beta-\delta'\wedge\beta}\beta\vee\bigg(\frac{\lambda_{h_{\ell+\lceil\theta\ell\rceil},3}}{2\bar\zeta^{\beta,\delta'}_{h_{\ell+\lceil\theta\ell\rceil},3,q'}}\bigg)^\frac{3\beta-\delta'\wedge\beta}{\delta'\wedge\beta}\bigg).
\end{aligned}
\end{equation*}
Thus
\begin{equation}
\label{eq:EL3<:beta}
    \mathbb{E}[\mathcal{L}_{h_{\ell+\lceil\theta\ell\rceil},3}(\widetilde\xi^{h_\ell}_n)]\leq C^{\ell,\beta}_{3,q'}\big((h_\ell^{2(1+\theta)}\widetilde{u}_n^2)\vee\gamma_n^2\big),
\end{equation}
where
\begin{equation*}
\begin{aligned}
C^{\ell,\beta}_{3,q'}
&=\gamma_1^{-2}(\mathbb{E}[\mathcal{L}_{h_{\ell+\lceil\theta\ell\rceil},3}(\widetilde\xi^{h_\ell}_0)]+\bar{C}^{\ell,\beta}_{3,q'}+\widetilde{C}^{\ell,\beta}_{3,q'})\\
&\hphantom{=+}\times
\sup_{m\geq1}\bigg\{m^{2\beta}\exp\bigg(2^{1+\beta+\delta'\wedge\beta}\bar\zeta^{\beta,\delta'}_{h_{\ell+\lceil\theta\ell\rceil},3,q'}\gamma_1^{1+\frac{\delta'\wedge\beta}\beta}\Phi_{1-\beta-\delta'\wedge\beta}(m+1)\bigg)\\
&\hphantom{=+\times\sup_{m\geq1}\bigg\{}\times
\exp\bigg(-\frac{\lambda_{h_{\ell+\lceil\theta\ell\rceil},3}\gamma_1}2\Phi_{1-\beta}(m+1)\bigg)\bigg\}\\
&\hphantom{=}
+2^\beta(\bar{C}^{\ell,\beta}_{3,q'}\gamma_1^{-1+3\frac{\delta'}\beta}h_{\ell+\lceil\theta\ell\rceil}^3
+\widetilde{C}^{\ell,\beta}_{3,q'}\gamma_1^2)\\
&\hphantom{=+}\times
\sup_{m\geq1}\{m^{2\beta}\exp(-2^{-(\beta+2)}\lambda_{h_{\ell+\lceil\theta\ell\rceil},3}\gamma_1m^{1-\beta})\Phi_{1-\beta}(m+1)\}\\
&\hphantom{=}
+\bar{C}^{\ell,\beta}_{3,q'}\frac{2^{1+3\delta'}\gamma_1^{-2+3\frac{\delta'}\beta}}{h_0M^{\ell+\lceil\theta\ell\rceil}\lambda_{h_{\ell+\lceil\theta\ell\rceil},3}}
+\widetilde{C}^{\ell,\beta}_{3,q'}
\frac{2^{1+3\beta}\gamma_1}{\lambda_{h_{\ell+\lceil\theta\ell\rceil},3}}.
\end{aligned}
\end{equation*}

Note that $\sup_{\ell\geq1}C^{\ell,\beta}_{3,q'}<\infty$, $q'\in\{3,4\}$, $\beta\in(0,1]$.
\\

\noindent
\emph{Step~3. Inequality on $\mathbb{E}[\mathcal{L}_{h_{\ell+\lceil\theta\ell\rceil},4}(\widetilde\xi^{h_\ell}_n)]$.}
\newline
For $q'\geq4$, assuming that $\bar\lambda_{h_{\ell+\lceil\theta\ell\rceil},2}\gamma_1>2$ if $\beta=1$, following from \eqref{eq:EL3<:1} and \eqref{eq:EL3<:beta},
\begin{equation*}
\mathbb{E}[\mathcal{L}_{h_{\ell+\lceil\theta\ell\rceil},3}(\widetilde\xi^{h_\ell}_n)]
\leq C^{\ell,\beta}_{3,q'}(2^{2\delta'}M^2h_{\ell+\lceil\theta\ell\rceil}^2\widetilde{u}_{n+1}^2
+2^{2\beta}\gamma_{n+1}^2).
\end{equation*}
Thus, the relation \eqref{eq:ELq<}, with $q=4$ and $h=h_{\ell+\lceil\theta\ell\rceil}$, together with the decreasing monotonicity of $(\gamma_n)_{n\geq1}$ and the relation \eqref{eq:gamma:u^3}, give
\begin{equation}
\label{eq:EL4<}
\begin{aligned}
\mathbb{E}&[\mathcal{L}_{h_{\ell+\lceil\theta\ell\rceil},4}(\widetilde\xi^{h_\ell}_{n+1})]
\leq\mathbb{E}[\mathcal{L}_{h_{\ell+\lceil\theta\ell\rceil},4}(\widetilde\xi^{h_\ell}_n)]
(1
-\lambda_{h_{\ell+\lceil\theta\ell\rceil},4}\gamma_{n+1}
+\bar\zeta^{\beta,\delta'}_{h_{\ell+\lceil\theta\ell\rceil},4,q'}\gamma_{n+1}^{1+\frac{\delta'\wedge\beta}\beta})\\
&\quad
+2^{2+2\delta'}CC^{\ell,\beta}_{3,q'}M^3\|V_{h_{\ell+\lceil\theta\ell\rceil}}'\|_\infty
h_{\ell+\lceil\theta\ell\rceil}^3\widetilde{u}_{n+1}^3\gamma_{n+1}\\
&\quad
+2^{2+2\beta}CC^{\ell,\beta}_{3,q'}M\|V_{h_{\ell+\lceil\theta\ell\rceil}}'\|_\infty
h_{\ell+\lceil\theta\ell\rceil}\widetilde{u}_{n+1}\gamma_{n+1}^3\\
&\quad
+2^{2\delta'}C^{\ell,\beta}_{3,q'}M^2\nu_{h_{\ell+\lceil\theta\ell\rceil},4}
h_{\ell+\lceil\theta\ell\rceil}^2\widetilde{u}_{n+1}^2\gamma_{n+1}^2\\
&\quad
+(2^{2\beta}C^{\ell,\beta}_{3,q'}+\gamma_1)\nu_{h_{\ell+\lceil\theta\ell\rceil},4}
\gamma_{n+1}^4\\
&\leq\mathbb{E}[\mathcal{L}_{h_{\ell+\lceil\theta\ell\rceil},4}(\widetilde\xi^{h_\ell}_n)]
(1
-\lambda_{h_{\ell+\lceil\theta\ell\rceil},4}\gamma_{n+1}
+\bar\zeta^{\beta,\delta'}_{h_{\ell+\lceil\theta\ell\rceil},4,q'}\gamma_{n+1}^{1+\frac{\delta'\wedge\beta}\beta})\\
&\quad
+\hat{d}^{\beta,\delta'}_{h_{\ell+\lceil\theta\ell\rceil},4,q'}
\gamma_{n+1}(h_{\ell+\lceil\theta\ell\rceil}\widetilde{u}_{n+1}+\gamma_{n+1})^3\\
&\leq\mathbb{E}[\mathcal{L}_{h_{\ell+\lceil\theta\ell\rceil},4}(\widetilde\xi^{h_\ell}_n)]
(1
-\lambda_{h_{\ell+\lceil\theta\ell\rceil},4}\gamma_{n+1}
+\bar\zeta^{\beta,\delta'}_{h_{\ell+\lceil\theta\ell\rceil},4,q'}\gamma_{n+1}^{1+\frac{\delta'\wedge\beta}\beta})\\
&\quad
+4\hat{d}^{\beta,\delta'}_{h_{\ell+\lceil\theta\ell\rceil},4,q'}\gamma_1^{3-3\frac{\delta'}\beta}
h_{\ell+\lceil\theta\ell\rceil}^3\gamma_{n+1}^{1+3\frac{\delta'}\beta}
+4\hat{d}^{\beta,\delta'}_{h_{\ell+\lceil\theta\ell\rceil},4,q'}
\gamma_{n+1}^4,
\end{aligned}
\end{equation}
where we introduced
\begin{equation*}
\begin{aligned}
\hat{d}_{h,q,q'}^{\beta,\delta}
&
=(2^{2+2\delta'}CC^{\ell,\beta}_{3,q'}M^3\|V_h'\|_\infty)
\vee
\bigg(\frac{2^{2+2\beta}}3CC^{\ell,\beta}_{3,q'}M\|V_h'\|_\infty\bigg)\\
&\hphantom{=(2^{2+2\delta'}CC^{\ell,\beta}_{3,q'}M^3\|V_h'\|_\infty)}\;
\vee
\bigg(\frac{2^{2\delta'}}3C^{\ell,\beta}_{3,q'}M^2\nu^{\mu_{h,q'}}_{h,q}\bigg)
\vee
\big((2^{2\beta}C^{\ell,\beta}_{3,q'}+\gamma_1)\nu^{\mu_{h,q'}}_{h,q}\big),\\
&\qquad\qquad\qquad
\qquad\qquad\qquad
\qquad\qquad
      h\in\overline{\mathcal{H}},
\quad q\in\mathbb{N}^*,
\quad q'\geq q,
\quad \beta,\delta\in(0,1].
\end{aligned}
\end{equation*}
Hence, iterating $n$ times the inequality \eqref{eq:EL4<},
\begin{equation*}
\begin{aligned}
\mathbb{E}[\mathcal{L}_{h_{\ell+\lceil\theta\ell\rceil},4}(\widetilde\xi^{h_\ell}_n)]
&\leq\mathbb{E}[\mathcal{L}_{h_{\ell+\lceil\theta\ell\rceil},4}(\widetilde\xi^{h_\ell}_0)]\widetilde\Pi^{h_{\ell+\lceil\theta\ell\rceil},4,q'}_{1:n}\\
&\hphantom{\leq}
+4\hat{d}^{\beta,\delta'}_{h_{\ell+\lceil\theta\ell\rceil},4,q'}\gamma_1^{3-3\frac{\delta'}\beta}
\sum_{k=1}^nh_{\ell+\lceil\theta\ell\rceil}^3
\gamma_k^{1+3\frac{\delta'}\beta}
\widetilde\Pi^{h_{\ell+\lceil\theta\ell\rceil},4,q'}_{k+1:n}\\
&\hphantom{\leq}
+4\hat{d}^{\beta,\delta'}_{h_{\ell+\lceil\theta\ell\rceil},4,q'}
\sum_{k=1}^n\gamma_k^4
\widetilde\Pi^{h_{\ell+\lceil\theta\ell\rceil},4,q'}_{k+1:n}.
\end{aligned}
\end{equation*}

If $\gamma_n=\gamma_1n^{-1}$, $n\in\mathbb{N}^*$, then, by Lemma~\ref{lmm:gamma},
\begin{equation*}
\begin{aligned}
\mathbb{E}[\mathcal{L}_{h_{\ell+\lceil\theta\ell\rceil},4}(\widetilde{\xi}^{h_\ell}_n)]
&\leq\widehat{C}^{\ell,1}_{4,q'}
\frac{\mathbb{E}[\mathcal{L}_{h_{\ell+\lceil\theta\ell\rceil},4}(\widetilde{\xi}^{h_\ell}_0)]}{(n+1)^{\lambda_{h_{\ell+\lceil\theta\ell\rceil},4}\gamma_1}}
+\bar{C}^{\ell,1}_{4,q'}
\frac{\Phi_{\lambda_{h_{\ell+\lceil\theta\ell\rceil},4}\gamma_1-3\delta'}(n+1)}{(n+1)^{\lambda_{h_{\ell+\lceil\theta\ell\rceil},4}\gamma_1}}
h_{\ell+\lceil\theta\ell\rceil}^3\\
&\hphantom{\leq}
+\widetilde{C}^{\ell,1}_{4,q'}
\frac{\Phi_{\lambda_{h_{\ell+\lceil\theta\ell\rceil},4}\gamma_1-3}(n+1)}{(n+1)^{\lambda_{h_{\ell+\lceil\theta\ell\rceil},4}\gamma_1}},
\end{aligned}
\end{equation*}
where
\begin{equation*}
\begin{aligned}
\widehat{C}^{\ell,1}_{4,q'}
&=\exp\bigg(\bar\zeta^{1,\delta'}_{h_{\ell+\lceil\theta\ell\rceil},4,q'}\Big(1+\frac1{\delta'}\Big)\gamma_1^{1+\delta'}+\frac{\lambda_{h_{\ell+\lceil\theta\ell\rceil},4}\gamma_1}2\bigg),\\
\bar{C}^{\ell,1}_{4,q'}
&=2^{2+(\lambda_{h_{\ell+\lceil\theta\ell\rceil},4}\gamma_1)\vee(1+3\delta')}
\widehat{C}^{\ell,1}_{4,q'}
\hat{d}^{1,\delta'}_{h_{\ell+\lceil\theta\ell\rceil},4,q'}
\gamma_1^4,\\
\widetilde{C}^{\ell,1}_{4,q'}
&=2^{2+(\lambda_{h_{\ell+\lceil\theta\ell\rceil},4}\gamma_1)\vee4}
\widehat{C}^{\ell,1}_{4,q'}
\hat{d}^{1,\delta'}_{h_{\ell+\lceil\theta\ell\rceil},4,q'}
\gamma_1^4.
\end{aligned}
\end{equation*}
In particular, if $\bar\lambda_{h_{\ell+\lceil\theta\ell\rceil},2}>2$, so that $\lambda_{h_{\ell+\lceil\theta\ell\rceil},4}\gamma_1
\geq\lambda_{h_{\ell+\lceil\theta\ell\rceil},2}\gamma_1
\geq\bar\lambda_{h_{\ell+\lceil\theta\ell\rceil},2}\gamma_1
>2
\geq2\delta'$, then
\begin{equation}
\label{eq:EL4<:1}
\mathbb{E}[\mathcal{L}_{h_{\ell+\lceil\theta\ell\rceil},4}(\widetilde\xi^{h_\ell}_n)]
\leq C^{\ell,1}_{4,q'}\big((h_\ell^{2(1+\theta)}\widetilde{u}_n^2)\vee\gamma_n^2\big),
\end{equation}
where
\begin{equation*}
\begin{aligned}
C^{\ell,1}_{4,q'}
&=\gamma_1^{-2}
\bigg(\widehat{C}^{\ell,1}_{4,q'}\mathbb{E}[\mathcal{L}_{h_{\ell+\lceil\theta\ell\rceil},4}(\widetilde\xi^{h_\ell}_0)]\\
&\hphantom{=\gamma_1^{-1}\bigg(}
+\frac{\bar{C}^{\ell,1}_{4,q'}}{h_0M^{\ell+\lceil\theta\ell\rceil}}
\sup_{m\geq1}\bigg\{m^{2\delta'}\frac{\Phi_{\lambda_{h_{\ell+\lceil\theta\ell\rceil},4}\gamma_1-3\delta'}(m+1)}{(m+1)^{\lambda_{h_{\ell+\lceil\theta\ell\rceil},4}\gamma_1}}\bigg\}\\
&\hphantom{=\gamma_1^{-2}
\bigg(}
+\widetilde{C}^{\ell,1}_{4,q'}
\sup_{m\geq1}\bigg\{m^2\frac{\Phi_{\lambda_{h_{\ell+\lceil\theta\ell\rceil},4}\gamma_1-3}(m+1)}{(m+1)^{\lambda_{h_{\ell+\lceil\theta\ell\rceil},4}\gamma_1}}\bigg\}\bigg).
\end{aligned}
\end{equation*}

If $\gamma_n=\gamma_1n^{-\beta}$, $n\in\mathbb{N}^*$, with $\beta\in(0,1)$, then, by Lemma~\ref{lmm:gamma},
\begin{equation*}
\begin{aligned}
&\mathbb{E}[\mathcal{L}_{h_{\ell+\lceil\theta\ell\rceil},4}(\widetilde\xi^{h_\ell}_n)]
\leq
(\mathbb{E}[\mathcal{L}_{h_{\ell+\lceil\theta\ell\rceil},4}(\widetilde\xi^{h_\ell}_0)]+\bar{C}^{\ell,\beta}_{4,q'}+\widetilde{C}^{\ell,\beta}_{4,q'})\\
&\quad\times
\exp\bigg(2^{1+\beta+\delta'\wedge\beta}\bar\zeta^{\beta,\delta'}_{h_{\ell+\lceil\theta\ell\rceil},4,q'}\gamma_1^{1+\frac{\delta'\wedge\beta}\beta}\Phi_{1-\beta-\delta'\wedge\beta}(n+1)
-\frac{\lambda_{h_{\ell+\lceil\theta\ell\rceil},4}\gamma_1}2\Phi_{1-\beta}(n+1)\bigg)\\
&+2^\beta(\bar{C}^{\ell,\beta}_{4,q'}\gamma_1^{1+3\frac{\delta'}\beta}h_{\ell+\lceil\theta\ell\rceil}^3
+\widetilde{C}^{\ell,\beta}_{4,q'}\gamma_1^4)
\exp(-2^{-(\beta+2)}\lambda_{h_{\ell+\lceil\theta\ell\rceil},4}\gamma_1n^{1-\beta})\Phi_{1-\beta}(n+1)\\
&
+\bar{C}^{\ell,\beta}_{4,q'}\frac{2^{1+3\delta'}\gamma_1^{3\frac{\delta'}\beta}}{\lambda_{h_{\ell+\lceil\theta\ell\rceil},4}n^{3\delta'}}
h_{\ell+\lceil\theta\ell\rceil}^3
+\widetilde{C}^{\ell,\beta}_{4,q'}
\frac{2^{1+3\beta}\gamma_1^3}{\lambda_{h_{\ell+\lceil\theta\ell\rceil},4}n^{3\beta}},
\end{aligned}
\end{equation*}
where
\begin{equation*}
\begin{aligned}
\bar{C}^{\ell,\beta}_{4,q'}
&=
(\bar\zeta^{\beta,\delta'}_{h_{\ell+\lceil\theta\ell\rceil},4,q'})^{-1}
\bigg(\gamma_1^\frac{3\delta'-\delta'\wedge\beta}\beta\vee\bigg(\frac{\lambda_{h_{\ell+\lceil\theta\ell\rceil},4}}{2\bar\zeta^{\beta,\delta'}_{h_{\ell+\lceil\theta\ell\rceil},4,q'}}\bigg)^\frac{3\delta'-\delta'\wedge\beta}{\delta'\wedge\beta}\bigg),\\
\widetilde{C}^{\ell,\beta}_{4,q'}
&=
(\bar\zeta^{\beta,\delta'}_{h_{\ell+\lceil\theta\ell\rceil},4,q'})^{-1}
\bigg(\gamma_1^\frac{3\beta-\delta'\wedge\beta}\beta\vee\bigg(\frac{\lambda_{h_{\ell+\lceil\theta\ell\rceil},4}}{2\bar\zeta^{\beta,\delta'}_{h_{\ell+\lceil\theta\ell\rceil},4,q'}}\bigg)^\frac{3\beta-\delta'\wedge\beta}{\delta'\wedge\beta}\bigg).
\end{aligned}
\end{equation*}
Thus
\begin{equation}
\label{eq:EL4<:beta}
    \mathbb{E}[\mathcal{L}_{h_{\ell+\lceil\theta\ell\rceil},4}(\widetilde\xi^{h_\ell}_n)]\leq C^{\ell,\beta}_{4,q'}\big((h_\ell^{2(1+\theta)}\widetilde{u}_n^2)\vee\gamma_n^2\big),
\end{equation}
where
\begin{equation*}
\begin{aligned}
C^{\ell,\beta}_{4,q'}
&=\gamma_1^{-2}(\mathbb{E}[\mathcal{L}_{h_{\ell+\lceil\theta\ell\rceil},4}(\widetilde\xi^{h_\ell}_0)]+\bar{C}^{\ell,\beta}_{4,q'}+\widetilde{C}^{\ell,\beta}_{4,q'})\\
&\hphantom{=+}\times
\sup_{m\geq1}\bigg\{m^{2\beta}\exp\bigg(2^{1+\beta+\delta'\wedge\beta}\bar\zeta^{\beta,\delta'}_{h_{\ell+\lceil\theta\ell\rceil},4,q'}\gamma_1^{1+\frac{\delta'\wedge\beta}\beta}\Phi_{1-\beta-\delta'\wedge\beta}(m+1)\bigg)\\
&\hphantom{=+\times\sup_{m\geq1}\bigg\{}\times
\exp\bigg(-\frac{\lambda_{h_{\ell+\lceil\theta\ell\rceil},4}\gamma_1}2\Phi_{1-\beta}(m+1)\bigg)\bigg\}\\
&\hphantom{=}
+2^\beta(\bar{C}^{\ell,\beta}_{4,q'}\gamma_1^{-1+3\frac{\delta'}\beta}h_{\ell+\lceil\theta\ell\rceil}^3
+\widetilde{C}^{\ell,\beta}_{4,q'}\gamma_1^2)\\
&\hphantom{=+}\times
\sup_{m\geq1}\{m^{2\beta}\exp(-2^{-(\beta+2)}\lambda_{h_{\ell+\lceil\theta\ell\rceil},4}\gamma_1m^{1-\beta})\Phi_{1-\beta}(m+1)\}\\
&\hphantom{=}
+\bar{C}^{\ell,\beta}_{4,q'}\frac{2^{1+3\delta'}\gamma_1^{-2+3\frac{\delta'}\beta}}{h_0M^{\ell+\lceil\theta\ell\rceil}\lambda_{h_{\ell+\lceil\theta\ell\rceil},4}}
+\widetilde{C}^{\ell,\beta}_{4,q'}
\frac{2^{1+3\beta}\gamma_1}{\lambda_{h_{\ell+\lceil\theta\ell\rceil},4}}.
\end{aligned}
\end{equation*}

Note that $\sup_{\ell\geq1}C^{\ell,\beta}_{4,4}<\infty$, $\beta\in(0,1]$.
\\

\noindent
\emph{Step~4. Conclusion.}
\newline
All in all, if $\bar\lambda_{h_{\ell+\lceil\theta\ell\rceil},2}\gamma_1>2$ when $\beta=1$, then, considering Lemma~\ref{lmm:lyapunov}(\ref{lmm:lyapunov-iv}) with $q=2$, $h=h_{\ell+\lceil\theta\ell\rceil}$, $\mu=\mu_{h_{\ell+\lceil\theta\ell\rceil},2}$ and $\mu'=\mu_{h_{\ell+\lceil\theta\ell\rceil},4}$, the inequalities \eqref{eq:EL2<:1}, \eqref{eq:EL2<:beta}, \eqref{eq:EL4<:1} and \eqref{eq:EL4<:beta} yield
\begin{equation*}
\mathbb{E}[(\widetilde\xi^{h_\ell}_n-\xi^{h_{\ell+\lceil\theta\ell\rceil}}_\star)^4]
\leq \bar{C}^{\ell,\beta}\big((h_\ell^{2(1+\theta)}\widetilde{u}_n^2)\vee\gamma_n^2\big),
\end{equation*}
where the
\begin{equation}
\label{eq:C:bar:l,b}
\bar{C}^{\ell,\beta}=\kappa_{h_{\ell+\lceil\theta\ell\rceil},2}(\check{C}^{\ell,\beta}_{2,2}+C^{\ell,\beta}_{4,4})
\end{equation}
satisfy $\sup_{\ell\geq1}\bar{C}^{\ell,\beta}<\infty$, $\beta\in(0,1]$.
\end{proof}

\begin{proof}[Proof of Lemma~\ref{lmm:local:strong:error:indicator:func:variance}]
\label{prf:local:strong:error:indicator:func:variance}

Let $\ell,n\in\mathbb{N}^*$ and $\xi\in\mathbb{R}$, $C_\mathrm{ad}>0$, $r>1$ and $\theta\in(0,1]$.

Via \eqref{eq:eta} and the law of total probability,
\begin{equation*}
\mathbb{E}[|\mathds1_{\{X^{h_{\ell+\eta^\ell_n(\xi)}}>\xi\}}-\mathds1_{\{X^0>\xi\}}|]
=\sum_{k=0}^{\lceil\theta\ell\rceil}\mathbb{E}[|\mathds1_{\{X^{h_{\ell+k}}>\xi\}}-\mathds1_{\{X^0>\xi\}}|\mathds1_{\{\eta^\ell_n(\xi)=k\}}]
\leq A^\ell+B^\ell,
\end{equation*}
where
\begin{align}
A^\ell&:=\mathbb{E}[|\mathds1_{\{X^{h_{\ell+\lceil\theta\ell\rceil}}>\xi\}}-\mathds1_{\{X^0>\xi\}}|],\notag\\
B^\ell&:=\sum_{k=0}^{\lceil\theta\ell\rceil-1}\mathbb{E}[|\mathds1_{\{X^{h_{\ell+k}}\leq\xi\}}-\mathds1_{\{X^0\leq\xi\}}|\mathds1_{\{\eta^\ell_n(\xi)=k\}}].\label{eq:carre}
\end{align}

Let $k\in[\![0,\lceil\theta\ell\rceil-1]\!]$.
Observe that
\begin{equation*}
|\mathds1_{\{X^{h_{\ell+k}}\leq\xi\}}-\mathds1_{\{X^{h_{\ell+\lceil\theta\ell\rceil}}\leq\xi\}}|
=\mathds1_{\{X^{h_{\ell+k}}\leq\xi<X^{h_{\ell+\lceil\theta\ell\rceil}}\}}+\mathds1_{\{X^{h_{\ell+\lceil\theta\ell\rceil}}\leq\xi<X^{h_{\ell+k}}\}},
\end{equation*}
that
\begin{equation*}
\{X^{h_{\ell+k}}\leq\xi<X^{h_{\ell+\lceil\theta\ell\rceil}}\}\cup\{X^{h_{\ell+\lceil\theta\ell\rceil}}\leq\xi<X^{h_{\ell+k}}\}\subset\{|X^{h_{\ell+k}}-X^{h_{\ell+\lceil\theta\ell\rceil}}|\geq|X^{h_{\ell+k}}-\xi|\},
\end{equation*}
and that, by \eqref{eq:eta},
\begin{equation*}
\{\eta^\ell_n(\xi)=k\}\subset\{|X^{h_{\ell+k}}-\xi|\geq C_\mathrm{ad}\psi^{\ell,k}_n\}.
\end{equation*}
Hence
\begin{equation}\label{eq:E[|1|]<}
\begin{aligned}
\mathbb{E}[|&\mathds1_{\{X^{h_{\ell+k}}\leq\xi\}}-\mathds1_{\{X^{h_{\ell+\lceil\theta\ell\rceil}}\leq\xi\}}|\mathds1_{\{\eta^\ell_n(\xi)=k\}}]\\
&\leq\mathbb{E}[(\mathds1_{\{X^{h_{\ell+k}}\leq\xi<X^{h_{\ell+\lceil\theta\ell\rceil}}\}}+\mathds1_{\{X^{h_{\ell+\lceil\theta\ell\rceil}}\leq\xi<X^{h_{\ell+k}}\}})
\mathds1_{\{|X^{h_{\ell+k}}-\xi|\geq C_\mathrm{ad}\psi^{\ell,k}_n\}}]\\
&\leq\mathbb{P}(|X^{h_{\ell+k}}-X^{h_{\ell+\lceil\theta\ell\rceil}}|\geq|X^{h_{\ell+k}}-\xi|,|X^{h_{\ell+k}}-\xi|\geq C_\mathrm{ad}\psi^{\ell,k}_n)\\
&\leq\mathbb{P}(|X^{h_{\ell+k}}-X^{h_{\ell+\lceil\theta\ell\rceil}}|\geq C_\mathrm{ad}\psi^{\ell,k}_n).
\end{aligned}
\end{equation}

In the following, we separately study $A^\ell$ and $B^\ell$ within each framework.
\\

\noindent (\ref{lmm:local:strong:error:indicator:func-i:variance})\hyperref[lmm:local:strong:error:indicator:func-ia:variance]{\rm a}.\
\emph{Step~1. Study of $A^\ell$.}
\newline
Using Lemma~\ref{lmm:local:strong:error:indicator:func}(\ref{lmm:local:strong:error:indicator:func-i})\hyperref[lmm:local:strong:error:indicator:func-ia]{\rm a},
\begin{equation*}
A^\ell
\leq C_{1A}h_\ell^{(1+\theta)\frac{p_\star}{2(p_\star+1)}},
\end{equation*}
where
\begin{equation*}
C_{1A}=h_0^{-\frac{p_\star}{2(p_\star+1)}}B_{p_\star}\mathbb{E}\big[\big|\varphi(Y, Z)-\mathbb{E}[\varphi(Y,Z)|Y]\big|^{p_\star}\big]^\frac1{p_\star+1}\Big(\sup_{h\in\overline{\mathcal{H}}}\|f_{X^h}\|_\infty\Big)^\frac{p_\star}{p_\star+1},
\end{equation*}
with $B_{p_\star}$ being a positive constant that depends only on $p_\star$.
\\

\noindent\emph{Step~2. Study of $B^\ell$.}
\newline
Via \eqref{eq:E[|1|]<}, the definition \eqref{eq:eta}, Markov's inequality and Lemma~\ref{lmm:aux}(\ref{lmm:aux:i}),
\begin{equation*}
\begin{aligned}
\mathbb{E}[|&\mathds1_{\{X^{h_{\ell+k}}\leq\xi\}}-\mathds1_{\{X^0\leq\xi\}}|\mathds1_{\{\eta^\ell_n(\xi)=k\}}]\\
&\leq\mathbb{P}(|X^{h_{\ell+k}}-X^0|\geq C_\mathrm{ad}u_n^{-\frac1{p_\star}}h_{\theta\ell(r-1)+k}^\frac1r)\\
&\leq C_\mathrm{ad}^{-p_\star}u_nh_{\theta\ell(r-1)+k}^{-\frac{p_\star}r}\mathbb{E}[|X^{h_{\ell+k}}-X^0|^{p_\star}]\\
&\leq C_{1b}h_{\theta\ell(r-1)+k}^{-\frac{p_\star}r}h_{\ell+k}^\frac{p_\star}2,
\end{aligned}
\end{equation*}
where
\begin{equation*}
C_{1b}=\gamma_1\sqrt2C_\mathrm{ad}^{-p_\star}B_{p_\star}\mathbb{E}\big[\big|\varphi(Y, Z)-\mathbb{E}[\varphi(Y, Z)|Y]\big|^{p_\star}\big]^\frac1{p_\star}.
\end{equation*}
Thus, by \eqref{eq:carre} and using the condition \eqref{assumption:finite:Lp:moment:variance} on $r$ and $\theta$,
\begin{equation*}
\begin{aligned}
B^\ell
&\leq C_{1b}\sum_{k=0}^{\lceil\theta\ell\rceil-1}h_{\theta\ell(r-1)+k}^{-\frac{p_\star}r}h_{\ell+k}^\frac{p_\star}2\\
&\leq\frac{C_{1b}h_1^{-p_\star(\frac1r-\frac12)}}{M^{p_\star(\frac1r-\frac12)}-1}M^{-\frac{p_\star\ell}2(1-\theta)}\\
&\leq C_{1B}h_\ell^{(1+\theta)\frac{p_\star}{2(p_\star+1)}},
\end{aligned}
\end{equation*}
where
\begin{equation*}
C_{1B}=\frac{C_{1b}h_0^{-\frac{2(p_\star+1)}{p_\star+2}}h_1^{-p_\star(\frac1r-\frac12)}}{M^{p_\star(\frac1r-\frac12)}-1}.
\end{equation*}

\noindent (\ref{lmm:local:strong:error:indicator:func-i:variance})\hyperref[lmm:local:strong:error:indicator:func-ib:variance]{\rm b}.\
\emph{Step~1. Study of $A^\ell$.}
\newline
Lemma~\ref{lmm:local:strong:error:indicator:func}(\ref{lmm:local:strong:error:indicator:func-i})\hyperref[lmm:local:strong:error:indicator:func-ib]{\rm b} entails
\begin{equation*}
A^\ell\leq2\sqrt{\mathfrak{g}h_{\ell+\lceil\theta\ell\rceil}}\bigg(1+\Big(\sup_{h\in\overline{\mathcal{H}}}{\|f_{X^h}\|_\infty}\Big)\sqrt{2|\ln{(\mathfrak{g}h_{\ell+\lceil\theta\ell\rceil})}|}\bigg).
\end{equation*}
On the one hand,
\begin{equation*}
2\sqrt{\mathfrak{g}h_{\ell+\lceil\theta\ell\rceil}}
\leq2\sqrt{\mathfrak{g}h_0^{-1}}h_\ell^\frac{1+\theta}2.
\end{equation*}
On the other hand, using that $\theta\leq1$ under the condition~\eqref{assumption:conditional:gaussian:concentration:variance} and the inequality $\sqrt{x+y}\leq\sqrt{x}+\sqrt{y}$, $x,y\geq0$,
\begin{equation*}
\sqrt{2|\ln{(\mathfrak{g}h_{\ell+\theta\ell\rceil})}|}
\leq\sqrt{2|\ln{(\mathfrak{g}h_0^{-1})}+(1+\theta)\ln{h_\ell}|}
\leq\sqrt{2|\ln{(\mathfrak{g}h_0^{-1})}|}+2\sqrt{|\ln{h_\ell}|}.
\end{equation*}
Thus
\begin{equation*}
A^\ell
\leq C_{2A}h_\ell^\frac{1+\theta}2|\ln{h_\ell}|^\frac12,
\end{equation*}
where
\begin{equation*}
C_{2A}=4\mathfrak{g}^\frac12h_0^{-\frac12}\bigg(1+\Big(\sup_{h\in\overline{\mathcal{H}}}{\|f_{X^h}\|_\infty}\Big)\Big(2\big(2\vee|\ln{(\mathfrak{g}h_0^{-1})}|\big)\Big)^\frac12\bigg).
\end{equation*}

\noindent\emph{Step~2. Study of $B^\ell$.}
\newline
Let $\lambda>0$. By the inequality \eqref{eq:E[|1|]<}, the definition \eqref{eq:eta}, Markov's exponential inequality and Lemma~\ref{lmm:aux}(\ref{lmm:aux:ii}),
\begin{equation*}
\begin{aligned}
\mathbb{E}[|&\mathds1_{\{X^{h_{\ell+k}}\leq\xi\}}-\mathds1_{\{X^0\leq\xi\}}|\mathds1_{\{\eta^\ell_n(\xi)=k\}}]\\
&\leq\mathbb{P}\Big(|X^{h_{\ell+k}}-X^0|\geq C_\mathrm{ad}h^\frac1r_{\theta\ell(r-1)+k}\ln^\frac12{\big((\gamma_1^{-1}\gamma_n)^{-\frac12}h_{\ell+k}^{-\frac{1+\theta}2}\big)}\Big)\\
&\leq\exp\Big(-\lambda C_\mathrm{ad}h^\frac1r_{\theta\ell(r-1)+k}\ln^\frac12{\big((\gamma_1^{-1}\gamma_n)^{-\frac12}h_{\ell+k}^{-\frac{1+\theta}2}\big)}\Big)\mathbb{E}[\exp(\lambda|X^{h_{\ell+k}}-X^0|)]\\
&\leq2\exp\Big(-\lambda C_\mathrm{ad}h^\frac1r_{\theta\ell(r-1)+k}\ln^\frac12{\big((\gamma_1^{-1}\gamma_n)^{-\frac12}h_{\ell+k}^{-\frac{1+\theta}2}\big)}+\mathfrak{g}\lambda^2h_{\ell+k}\Big),
\end{aligned}
\end{equation*}
where we used the inequality $\e^{|x|}\leq\e^x+\e^{-x}$, $x\in\mathbb{R}$.
Minimizing the above upper bound with respect to $\lambda$ yields
\begin{equation*}
\begin{aligned}
\mathbb{E}[|\mathds1_{\{X^{h_{\ell+k}}\leq\xi\}}&-\mathds1_{\{X^0\leq\xi\}}|\mathds1_{\{\eta^\ell_n(\xi)=k\}}]\\
&\leq2\exp\bigg(-\frac{C_\mathrm{ad}^2h^\frac2r_{\theta\ell(r-1)+k}\ln{\big((\gamma_1^{-1}\gamma_n)^{-\frac12}h_{\ell+k}^{-\frac{1+\theta}2}\big)}}{4\mathfrak{g}h_{\ell+k}}\bigg)\\
&=2\bigg(\frac{\gamma_n}{\gamma_1}h_{\ell+k}^{1+\theta}\bigg)^\frac{C_\mathrm{ad}^2h_{\theta\ell(r-1)+k}^\frac2r}{8\mathfrak{g}h_{\ell+k}}.
\end{aligned}
\end{equation*}
Note that $\gamma_n\leq\gamma_1$ and $h_{\ell+k}^{1+\theta}\leq1$.
The condition \eqref{assumption:conditional:gaussian:concentration:variance} on $h_0$, $r$ and $\theta$ implies in particular that $\theta\leq1\leq\frac{r}{2(r-1)}$, so that $-2\theta(1-\frac1r)+1\geq0$, hence $M^{-2\theta\ell(1-\frac1r)+\ell}\geq1$.
It also entails that $\frac{C_\mathrm{ad}^2h_0^{\frac2r-1}}{8\mathfrak{g}}\geq\frac12$.
Hence
\begin{equation*}
\begin{aligned}
\mathbb{E}[|\mathds1_{\{X^{h_{\ell+k}}\leq\xi\}}&-\mathds1_{\{X^0\leq\xi\}}|\mathds1_{\{\eta^\ell_n(\xi)=k\}}]\\
&\leq2h_{\ell+k}^{(1+\theta)\frac{C_\mathrm{ad}^2h_0^{\frac2r-1}}{8\mathfrak{g}}M^{-2\theta\ell(1-\frac1r)+\ell}}\\
&\leq2h_{\ell+k}^{\frac{1+\theta}2}.
\end{aligned}
\end{equation*}
Eventually, \eqref{eq:carre} yields
\begin{equation*}
B^\ell
\leq C_{2B}h_\ell^\frac{1+\theta}2
\leq C_{2B}h_\ell^\frac{1+\theta}2|\ln{h_\ell}|^\frac12,
\end{equation*}
where
\begin{equation*}
C_{2B}=\frac2{1-M^{-\frac12}}.
\end{equation*}

\noindent (\ref{lmm:local:strong:error:indicator:func-ii:variance})\
\emph{Step~1. Study of $A^\ell$.}
\newline
Consequently from Lemma~\ref{lmm:local:strong:error:indicator:func-ii:bis},
\begin{equation*}
A^\ell\leq C_{3A}h_\ell^\frac{1+\theta}2,
\end{equation*}
where
\begin{equation*}
C_{3A}=h_0^{-\frac12}\sup_{0\leq h_1<h_2\in\overline{\mathcal{H}}}{\mathbb{E}[K_{h_1}^{h_2}|G_{h_1}^{h_2}|]}.
\end{equation*}

\noindent
\emph{Step~2. Study of $B^\ell$.}
\newline
Via \eqref{eq:E[|1|]<}, the definition \eqref{eq:eta} and Markov's exponential inequality, and recalling the condition \eqref{assump:unif:lipschitz:integrability:conditional:cdf:variance} on $h_0$, $r$ and $\theta$,
\begin{equation*}
\begin{aligned}
|\mathbb{E}[(&\mathds1_{\{X^{h_{\ell+k}}\leq\xi\}}-\mathds1_{\{X^0\leq\xi\}})\mathds1_{\{\eta^\ell_n(\xi)=k\}}]|\\
&\leq\mathbb{P}\Big(|X^{h_{\ell+k}}-X^0|\geq C_\mathrm{ad}h^\frac1r_{\theta\ell(r-1)+k}\ln^\frac12{\big((\gamma_1^{-1}\gamma_n)^{-\frac12}h_{\ell+k}^{-\frac{1+\theta}2}\big)}\Big)\\
&=\mathbb{P}\Big(|G_0^{h_{\ell+k}}|\geq C_\mathrm{ad}h^\frac1r_{\theta\ell(r-1)+k}h_{\ell+k}^{-\frac12}\ln^\frac12{\big((\gamma_1^{-1}\gamma_n)^{-\frac12}h_{\ell+k}^{-\frac{1+\theta}2}}\big)\Big)\\
&\leq\exp{\Big(-\upsilon_0C_\mathrm{ad}^2h_{\theta\ell(r-1)+k}^\frac2rh_{\ell+k}^{-1}\ln{\big((\gamma_1^{-1}\gamma_n)^{-\frac12}h_{\ell+k}^{-\frac{1+\theta}2}\big)}\Big)}\mathbb{E}[\exp(\upsilon_0|G_0^{h_{\ell+k}}|^2)]\\
&\leq\Big(\sup_{h\in\mathcal{H}}\mathbb{E}[\exp(\upsilon_0|G_0^h|^2)]\Big)\bigg(\frac{\gamma_n}{\gamma_1}h_{\ell+k}^{1+\theta}\bigg)^{\frac12\upsilon_0C_\mathrm{ad}^2h_0^{\frac2r-1}M^{-2\theta\ell(1-\frac1r)+\ell}}\\
&\leq C_{3b}h_{\ell+k}^\frac{1+\theta}2,
\end{aligned}
\end{equation*}
with
\begin{equation*}
C_{3b}=\sup_{h\in\mathcal{H}}\mathbb{E}[\exp(\upsilon_0|G_0^h|^2)],
\end{equation*}
where we used that $\gamma_n\leq\gamma_1$, that $\theta\leq1\leq\frac{r}{2(r-1)}$, so that $M^{-2\theta\ell(1-\frac1r)+\ell}\geq1$, and that $\upsilon_0C_\mathrm{ad}^2h_0^{\frac2r-1}\geq1$.
All in all, via \eqref{eq:carre},
\begin{equation*}
B^\ell\leq C_{3B}h_\ell^\frac{1+\theta}2,
\end{equation*}
where
\begin{equation*}
C_{3B}=\frac{C_{3b}}{1-M^{-\frac12}}.
\end{equation*}
\end{proof}

\section{Proof of Theorem~\ref{thm:amlsa:L2}}
\label{apx:prf:main}

For all $\ell\in\mathbb{N}$, let $(\widetilde{\mathcal{F}}^{h_\ell}_n)_{n\geq0}$ be the filtration defined by $\widetilde{\mathcal{F}}^{h_\ell}_0=\sigma(\widetilde\xi^{h_\ell}_0)$ and \newline $\widetilde{\mathcal{F}}^{h_\ell}_n=\sigma(\widetilde\xi^{h_\ell}_0,X^{h_{\ell+\lceil\theta\ell\rceil}}_{1},\dots,X^{h_{\ell+\lceil\theta\ell\rceil}}_{n})$, $n\geq1$.

Following \eqref{xi:decomp}, the dynamics \eqref{adNSA} can be decomposed  into
\begin{equation}
\label{decomposition:var:sa}
\widetilde\xi^{h_\ell}_n-\xi^{h_{\ell+\lceil\theta\ell\rceil}}_\star
=\big(1-\gamma_nV_0''(\xi^0_\star)\big)(\widetilde\xi^{h_\ell}_{n-1}-\xi^{h_{\ell+\lceil\theta\ell\rceil}}_\star)-\gamma_ng^{h_\ell}_n-\gamma_n\rho^{h_\ell}_n-\gamma_nr^{h_\ell}_n-\gamma_ne^{h_\ell}_n,
\end{equation}
where $(e^{h_\ell}_n)_{n\geq1}$ and $(r^{h_\ell}_n)_{n\geq1}$ are defined in \eqref{eq:r:def} and \eqref{eq:e:def}, and
\begin{align}
g^{h_\ell}_n
&:=\big(V_{h_{\ell+\lceil\theta\ell\rceil}}''(\xi^{h_{\ell+\lceil\theta\ell\rceil}}_\star)-V_0''(\xi^0_\star)\big)(\widetilde\xi^{h_\ell}_{n-1}-\xi^{h_{\ell+\lceil\theta\ell\rceil}}_\star),
\label{eq:ghn}\\
\rho^{h_\ell}_n
&:= V_{h_{\ell+\lceil\theta\ell\rceil}}'(\widetilde\xi^{h_\ell}_{n-1})-V_{h_{\ell+\lceil\theta\ell\rceil}}''(\xi^{h_{\ell+\lceil\theta\ell\rceil}}_\star)(\widetilde\xi^{h_\ell}_{n-1}-\xi^{h_{\ell+\lceil\theta\ell\rceil}}_\star).
\label{eq:rhohn}
\end{align}
By iterating \eqref{decomposition:var:sa},
\begin{equation}
\label{eq:xih-xi*}
\begin{aligned}
\widetilde\xi^{h_\ell}_n-\xi^{h_{\ell+\lceil\theta\ell\rceil}}_\star=
(\widetilde\xi^{h_\ell}_0-\xi^{h_{\ell+\lceil\theta\ell\rceil}}_\star)\Pi_{1:n}
&-\sum_{k=1}^n\gamma_k\Pi_{k+1:n}g^{h_\ell}_k
-
\sum_{k=1}^n\gamma_k\Pi_{k+1:n}\rho^{h_\ell}_k\\
-\sum_{k=1}^n\gamma_k\Pi_{k+1:n}r^{h_\ell}_k
&-\sum_{k=1}^n\gamma_k\Pi_{k+1:n}e^{h_\ell}_k,
\end{aligned}
\end{equation}
where
\begin{equation}
\Pi_{k:n}=\prod_{j=k}^n\big(1-\gamma_jV_0''(\xi^0_\star)\big) 
\label{eq:Pi}
\end{equation}
with the convention $\prod_\varnothing=1$.

Given that $\gamma_k\downarrow0$ as $k\uparrow\infty$, there exists $k_0\geq0$ such that for $j\geq k_0$, $(1-\gamma_jV_0''(\xi^0_\star))>0$.
Hence, using the inequality $1+x\leq\e^x$, $x\in\mathbb{R}$, for all $n\in\mathbb{N}^*$,
\begin{equation}
\label{upper:estimate:pi:i:n}
\begin{aligned}
|\Pi_{k:n}|
&=|\Pi_{k:k_0-1}|\prod_{j=k_0 \vee k}^n\big(1-\gamma_jV_0''(\xi^0_\star)\big)\\
&\leq|\Pi_{k:k_0-1}|\exp\bigg(-V_0''(\xi^0_\star)\sum_{j=k_0 \vee k}^n\gamma_j\bigg)\\
&\leq\bar{C}\exp\bigg(-V_0''(\xi^0_\star)\sum_{j=k}^n\gamma_j\bigg),
\end{aligned}
\end{equation}
where
\begin{equation*}
\bar{C}=1\vee\max_{1\leq k\leq k_0}
\bigg\{|\Pi_{k:k_0-1}|
\exp\bigg(V_0''(\xi^0_\star)\sum_{j=k_0\wedge k}^{k-1}\gamma_j\bigg)\bigg\},
\end{equation*}
with the convention $\sum_\varnothing=0$.

According to \eqref{alg:amlsa:var} and the decomposition \eqref{eq:xih-xi*},
\begin{equation}
\label{eq:xiML-xi*(hL)}
\begin{aligned}
\widetilde\xi^\text{\tiny\rm ML}_\mathbf{N}-\xi^{h_{L+\lceil\theta L\rceil}}_\star
&=\widetilde\xi^{h_0}_{N_0}-\xi^{h_0}_\star
+\sum_{\ell=1}^L\big(\widetilde\xi^{h_\ell}_{N_\ell}-\xi^{h_{\ell+\lceil\theta\ell\rceil}}_\star-(\widetilde\xi^{h_{\ell-1}}_{N_\ell}-\xi^{h_{\ell-1+\lceil\theta(\ell-1)\rceil}}_\star)\big)\\
&=\widetilde\xi^{h_0}_{N_0}-\xi^{h_0}_\star
+\sum_{\ell=1}^L\big(\widetilde\xi^{h_\ell}_0-\xi^{h_{\ell+\lceil\theta\ell\rceil}}_\star-(\widetilde\xi^{h_{\ell-1}}_0-\xi^{h_{\ell-1+\lceil\theta(\ell-1)\rceil}}_\star)\big)\Pi_{1:N_\ell}\\
&\quad-\sum_{\ell=1}^L\sum_{k=1}^{N_\ell}\gamma_k\Pi_{k+1:N_\ell}(g^{h_\ell}_k-g^{h_{\ell-1}}_k)
-\sum_{\ell=1}^L\sum_{k=1}^{N_\ell}\gamma_k\Pi_{k+1:N_\ell}(\rho^{h_\ell}_k-\rho^{h_{\ell-1}}_k)\\
&\quad-\sum_{\ell=1}^L\sum_{k=1}^{N_\ell}\gamma_k\Pi_{k+1:N_\ell}(r^{h_\ell}_k-r^{h_{\ell-1}}_k)
-\sum_{\ell=1}^L\sum_{k=1}^{N_\ell}\gamma_k\Pi_{k+1:N_\ell}(e^{h_\ell}_k-e^{h_{\ell-1}}_k).
\end{aligned}
\end{equation}
We study each term in the above decomposition separately.
\\

\noindent
\emph{Step~1. Study of $\widetilde\xi^{h_0}_{N_0}-\xi^{h_0}_\star$.}
\newline
Following Remark~\ref{rmk:refine}(\ref{rmk:level:0}), no refinement is applied at level $0$, so that by Lemma~\ref{lmm:error}(\ref{lmm:error:statistical}),
\begin{equation*}
\mathbb{E}[(\widetilde\xi^{h_0}_{N_0}-\xi^{h_0}_\star)^2]
\leq C_{-1}\gamma_{N_0},
\end{equation*}
where $C_{-1}>0$ is the constant described in \eqref{eq:C:recall}.
\\

\noindent
\emph{Step~2. Study of $\sum_{\ell=1}^L(\widetilde\xi^{h_\ell}_0-\xi^{h_{\ell+\lceil\theta\ell\rceil}}_\star-(\widetilde\xi^{h_{\ell-1}}_0-\xi^{h_{\ell-1+\lceil\theta(\ell-1)\rceil}}_\star))\Pi_{1:N_\ell}$.}
\newline
By assumption, $\sup_{\ell\geq0}\mathbb{E}[|\widetilde\xi^{h_\ell}_0|^2]<\infty$, and according to Lemma~\ref{lmm:error}(\ref{lmm:error:weak}), $(\xi^h_\star)_{h\in\overline{\mathcal{H}}}$ is bounded.
Hence $\sup_{\ell\geq0}{\mathbb{E}[(\widetilde\xi^{h_\ell}_0-\xi^{h_{\ell+\lceil\theta\ell\rceil}}_\star)^2]^\frac12}\leq\sup_{\ell\ge0}{\mathbb{E}[|\widetilde\xi^{h_\ell}_0|^2]^\frac12}+\sup_{\ell\geq0}{|\xi^{h_\ell}_\star|}<\infty$.
Besides, by \eqref{upper:estimate:pi:i:n} and \cite[Lemma~A.1(ii)]{amlsa}, $\limsup_{n\uparrow\infty}\gamma_n^{-1}\exp(-V_0''(\xi^0_\star)\sum_{k=1}^n\gamma_k)=0$ under the condition $\bar\lambda_2\gamma_1>2$, since the latter entails that $\gamma_1V_0''(\xi^0_\star)>\bar\lambda_2\gamma_1>2>0$.
Hence,
\begin{equation*}
\begin{aligned}
\mathbb{E}\bigg[\bigg(\sum_{\ell=1}^L\big(\widetilde\xi^{h_\ell}_0-\xi^{h_{\ell+\lceil\theta\ell\rceil}}_\star&-(\widetilde\xi^{h_{\ell-1}}_0-\xi^{h_{\ell-1+\lceil\theta(\ell-1)\rceil}}_\star)\big)\Pi_{1:N_\ell}\bigg)^2\bigg]^\frac12\\
&\leq2\sup_{\ell'\geq0}{\mathbb{E}[(\widetilde\xi^{h_{\ell'}}_0-\xi^{h_{\ell'+\lceil\theta\ell'\rceil}}_\star)^2]^\frac12}\sum_{\ell=1}^L|\Pi_{1:N_\ell}|\\
&\leq C_0\sum_{\ell=1}^L\gamma_{N_\ell},
\end{aligned}
\end{equation*}
where
\begin{equation*}
C_0=2\bar{C}\Big(\sup_{\ell'\ge0}{\mathbb{E}[|\widetilde\xi^{h_{\ell'}}_0|^2]^\frac12}
+\sup_{\ell'\geq0}{|\xi^{h_{\ell'}}_\star|}\Big)
\sup_{m\geq1}\bigg\{\gamma_m^{-1}\exp\bigg(-V_0''(\xi^0_\star)\sum_{k=1}^m\gamma_k\bigg)\bigg\}.
\end{equation*}

\noindent
\emph{Step~3. Study of $\sum_{\ell=1}^L\sum_{k=1}^{N_\ell}\gamma_k\Pi_{k+1:N_\ell}(g^{h_\ell}_k-g^{h_{\ell-1}}_k)$.}
\newline
Recalling that, by Lemma~\ref{lmm:error}(\ref{lmm:error:weak}), $(\xi^h_\star)_{h\in\overline{\mathcal{H}}}$ is bounded, there exists a compact set $\mathcal{K}\subset\mathbb{R}$ such that $\xi^{h_{\ell+\lceil\theta\ell\rceil}}_\star\in\mathcal{K}$, $\ell\in\mathbb{N}$. From Lemma~\ref{lmm:error}(\ref{lmm:error:weak}) and Assumptions~\ref{asp:misc}(\ref{asp:misc:iv}) and~\ref{asp:fh-f0},
\begin{equation}
\label{eq:V-V}
\begin{aligned}
|V_0''(\xi^0_\star)&-V_{h_{\ell+\lceil\theta\ell\rceil}}''(\xi_\star^{h_{\ell+\lceil\theta\ell\rceil}})|\\
&\leq|V_0''(\xi^0_\star)-V_0''(\xi_\star^{h_{\ell+\lceil\theta\ell\rceil}})|
+|V_0''(\xi_\star^{h_{\ell+\lceil\theta\ell\rceil}})-V_{h_{\ell+\lceil\theta\ell\rceil}}''(\xi_\star^{h_{\ell+\lceil\theta\ell\rceil}})|\\
&\leq\frac1{1-\alpha}\Big([f_{X^0}]_{\text{\rm Lip}}|\xi_\star^{h_{\ell+\lceil\theta\ell\rceil}}-\xi^0_\star|
+\sup_{\xi\in\mathcal{K}}{|f_{X^0}(\xi)-f_{X^{h_{\ell+\lceil\theta\ell\rceil}}}(\xi)|}\Big)\\
&\leq\frac{[f_{X_0}]_{\text{\rm Lip}}c'+c}{1-\alpha}(h_{\ell+\lceil\theta\ell\rceil}\vee h_{\ell+\lceil\theta\ell\rceil}^{\frac14+\delta_0})\\
&\leq C_g^1h_\ell^{((\frac14+\delta_0)\wedge1)(1+\theta)},
\end{aligned}
\end{equation}
where $c'\geq0$ is the constant described in \eqref{eq:oO}, and
\begin{equation*}
C_g^1=M^{(\frac14+\delta_0)\wedge1}\frac{[f_{X_0}]_{\text{\rm Lip}}c'+c}{1-\alpha}.
\end{equation*}
Recalling the definitions \eqref{def:zeta}, \eqref{eq:u:tilde} and \eqref{eq:delta}, using Proposition~\ref{prp:error:statistical:bis}(\ref{prp:error:statistical:bis:ii}), \eqref{eq:u=f(g)} and the fact that
\begin{equation}
\label{eq:u^2}
\widetilde{u}_n^2
=\gamma_1^{2\frac{\beta-\delta'}\beta}\gamma_n^{2\frac{\delta'}\beta},
\quad n\in\mathbb{N}^*,
\end{equation}
one gets
\begin{equation}
\label{eq:total:recall}
\begin{aligned}
\mathbb{E}[(\widetilde\xi^{h_\ell}_{n-1}&-\xi^{h_{\ell+\lceil\theta\ell\rceil}}_\star)^4]\\
&\leq
\begin{cases}
8(\sup_{\ell'\geq1}\mathbb{E}[|\widetilde{\xi}^{h_{\ell'}}_0|^4]+\sup_{\ell'\geq1}|\xi^{h_{\ell'}}_\star|^4)
&\text{if $n=1$,}\\
C^\beta(\widetilde{\gamma}_{n-1}^\ell)^2
&\text{if $n\geq2$,}
\end{cases}\\
&\leq\bigg\{\bigg(8\gamma_1^{-2}\Big(\sup_{\ell'\geq1}\mathbb{E}[|\widetilde{\xi}^{h_{\ell'}}_0|^4]+\sup_{\ell'\geq1}|\xi^{h_{\ell'}}_\star|^4\Big)\bigg)\vee(2^{2(\beta\vee\delta')}C^\beta)\bigg\}
(h_\ell^{2(1+\theta)}\widetilde{u}_n^2+\gamma_n^2)\\
&\leq C_g^2(h_\ell^{2(1+\theta)}\gamma_n^{2\frac{\delta'}\beta}+\gamma_n^2),
\quad n\in\mathbb{N}^*,
\end{aligned}
\end{equation}
where
\begin{equation*}
C_g^2=\bigg\{\bigg(8\gamma_1^{-2}\Big(\sup_{\ell'\geq1}\mathbb{E}[|\widetilde{\xi}^{h_{\ell'}}_0|^4]+\sup_{\ell'\geq1}|\xi^{h_{\ell'}}_\star|^4\Big)\bigg)\vee(2^{2(\beta\vee\delta')}C^\beta)\bigg\}(\gamma_1^{2\frac{\beta-\delta'}\beta}\vee1),
\end{equation*}
with $C^\beta>0$ being the constant described in \eqref{eq:L4:control}.
Consequently, using \eqref{eq:ghn}, \eqref{eq:V-V} and the inequality $(x+y)^\frac14\leq x^\frac14+y^\frac14$, $x,y\geq0$, one gets
\begin{equation}
\label{eq:E[|g|]<}
\begin{aligned}
\mathbb{E}[|g^{h_\ell}_n|^2]^\frac12
&\leq|V_0''(\xi^0_\star)-V_{h_{\ell+\lceil\theta\ell\rceil}}''(\xi^{h_{\ell+\lceil\theta\ell\rceil}}_\star)|\mathbb{E}[(\widetilde\xi^{h_\ell}_{n-1}-\xi^{h_{\ell+\lceil\theta\ell\rceil}}_\star)^4]^\frac14\\
&\leq C_g^3h_\ell^{((\frac14+\delta_0)\wedge1)(1+\theta)}(h_\ell^\frac{1+\theta}2\gamma_n^\frac{\delta'}{2\beta}+\gamma_n^\frac12),
\end{aligned}
\end{equation}
where
\begin{equation*}
C_g^3=C_g^1(C_g^2)^\frac14.
\end{equation*}
Recalling \eqref{upper:estimate:pi:i:n}, under the condition $\bar\lambda_2\gamma_1>2$ if $\beta=1$, which entails that $\gamma_1V_0''(\xi^0_\star)>\bar\lambda_2\gamma_1>2\geq\frac12\vee\frac{\delta'}2$ if $\beta=1$, by \cite[Lemma~A.1(i)]{amlsa},
\begin{equation*}
\begin{aligned}
\mathbb{E}\bigg[\bigg(&\sum_{\ell=1}^L\sum_{k=1}^{N_\ell}\gamma_k\Pi_{k+1:N_\ell}(g^{h_\ell}_k-g^{h_{\ell-1}}_k)\bigg)^2\bigg]^\frac12\\
&\leq\sum_{\ell=1}^L\sum_{k=1}^{N_\ell}\gamma_k|\Pi_{k+1:N_\ell}|(\mathbb{E}[|g^{h_\ell}_k|^2]^\frac12+\mathbb{E}[|g^{h_{\ell-1}}_k|^2]^\frac12)\\
&\leq C_g^3(M+1)\sum_{\ell=1}^Lh_\ell^{((\frac14+\delta_0)\wedge1)(1+\theta)}
\bigg(\sum_{k=1}^{N_\ell}h_\ell^\frac{1+\theta}2\gamma_k^{1+\frac{\delta'}{2\beta}}|\Pi_{k+1:N_\ell}|
+\sum_{k=1}^{N_\ell}\gamma_k^\frac32|\Pi_{k+1:N_\ell}|\bigg)\\
&\leq C_g^4\sum_{\ell=1}^Lh_\ell^{((\frac14+\delta_0)\wedge1)(1+\theta)}(h_\ell^\frac{1+\theta}2\gamma_{N_\ell}^\frac{\delta'}{2\beta}+\gamma_{N_\ell}^\frac12)\\
&\leq C_g^5\sum_{\ell=1}^Lh_\ell^{(\frac14+(\delta_0\wedge\frac34))(1+\theta)}
(\widetilde{\gamma}^\ell_{N_\ell})^\frac12,
\end{aligned}
\end{equation*}
where
\begin{equation*}
\begin{aligned}
C_g^4
&=C_g^3\bar{C}(M+1)\bigg(\sup_{m\geq1}\bigg\{\gamma_m^{-\frac{\delta'}{2\beta}}\sum_{k=1}^m\gamma_k^{1+\frac{\delta'}{2\beta}}\exp\bigg(-V_0''(\xi^0_\star)\sum_{j=k+1}^m\gamma_j\bigg)\bigg\}\\
&\hphantom{=C_g^3\bar{C}(M+1)\bigg\{}\;
\vee\sup_{m\geq1}\bigg\{\gamma_m^{-\frac12}\sum_{k=1}^m\gamma_k^\frac32\exp\bigg(-V_0''(\xi^0_\star)\sum_{j=k+1}^m\gamma_j\bigg)\bigg\}\bigg),\\
C_g^5
&=\sqrt2C_g^4(\gamma_1^{-1+\frac{\delta'}{2\beta}}\vee1).
\end{aligned}
\end{equation*}
Hence, using the Cauchy-Schwarz inequality and the fact that, recalling \eqref{eq:eps(hl)}, $h^\frac12\leq\epsilon(h)$, $h\in\mathcal{H}\setminus\{1\}$, one gets
\begin{equation}
\begin{aligned}
\mathbb{E}\bigg[\bigg(\sum_{\ell=1}^L\sum_{k=1}^{N_\ell}&\gamma_k\Pi_{k+1:N_\ell}(g^{h_\ell}_k-g^{h_{\ell-1}}_k)\bigg)^2\bigg]^\frac12\\
&\leq C_g^5\bigg(\sum_{\ell=1}^L\widetilde\gamma_{N_\ell}^\ell h_\ell^{\frac12(1+\theta)}\bigg)^\frac12\bigg(\sum_{\ell=1}^Lh_\ell^{2(1+\theta)(\delta_0\wedge\frac34)}\bigg)^\frac12\\
&\leq C_g\bigg(\sum_{\ell=1}^L\widetilde\gamma_{N_\ell}^\ell\epsilon(h_\ell)^{1+\theta}\bigg)^\frac12,
\end{aligned}
\end{equation}
where
\begin{equation*}
C_g
=C_g^5\bigg(\sum_{\ell=1}^\infty h_\ell^{2(1+\theta)(\delta_0\wedge\frac34)}\bigg)^\frac12
=\frac{h_0^{(1+\theta)(\delta_0\wedge\frac34)}C_g^5}{(1-M^{-2(1+\theta)(\delta_0\wedge\frac34)})^\frac12}.
\end{equation*}

\noindent
\emph{Step~4. Study of $\sum_{\ell=1}^L\sum_{k=1}^{N_\ell}\gamma_k\Pi_{k+1:N_\ell}(\rho^{h_\ell}_k-\rho^{h_{\ell-1}}_k)$.}
\newline
Taking into account the fact that $V_{h_{\ell+\lceil\theta\ell\rceil}}'(\xi^{h_{\ell+\lceil\theta\ell\rceil}}_\star)=0$ and Assumption~\ref{asp:misc}(\ref{asp:misc:iv}), a first order Taylor expansion yields
\begin{equation*}
\begin{aligned}
&|V_{h_{\ell+\lceil\theta\ell\rceil}}'(\widetilde\xi^{h_\ell}_{n-1})-V_{h_{\ell+\lceil\theta\ell\rceil}}''(\xi^{h_{\ell+\lceil\theta\ell\rceil}}_\star)(\widetilde\xi^{h_\ell}_{n-1}-\xi^{h_{\ell+\lceil\theta\ell\rceil}}_\star)|\\
&\qquad=\bigg|(\widetilde\xi^{h_\ell}_{n-1}-\xi^{h_{\ell+\lceil\theta\ell\rceil}}_\star)\int_0^1\Big(V_{h_{\ell+\lceil\theta\ell\rceil}}''\big(t\widetilde\xi^{h_\ell}_{n-1}+(1-t)\xi^{h_{\ell+\lceil\theta\ell\rceil}}_\star\big)-V_{h_{\ell+\lceil\theta\ell\rceil}}''(\xi^{h_{\ell+\lceil\theta\ell\rceil}}_\star)\Big)\mathrm{d}t\bigg|\\
&\qquad\leq\frac{[f_{X^{h_\ell}}]_\text{Lip}}{2(1-\alpha)}(\widetilde\xi^{h_\ell}_{n-1}-\xi^{h_{\ell+\lceil\theta\ell\rceil}}_\star)^2.
\end{aligned}
\end{equation*}
Hence, by \eqref{eq:rhohn}, \eqref{eq:total:recall} and the inequality $(x+y)^\frac12\leq x^\frac12+y^\frac12$, $x,y\geq0$,
\begin{equation*}
\mathbb{E}[|\rho^{h_\ell}_n|^2]^\frac12
\leq\frac{[f_{X^{h_\ell}}]_\text{Lip}}{2(1-\alpha)}\mathbb{E}[(\widetilde\xi^{h_\ell}_{n-1}-\xi^{h_{\ell+\lceil\theta\ell\rceil}}_\star)^4]^\frac12
\leq C_\rho^1(h_\ell^{1+\theta}\gamma_n^{\frac{\delta'}\beta}+\gamma_n),
\end{equation*}
where
\begin{equation*}
C_\rho^1=\frac{(C_g^2)^\frac12\sup_{\ell'\geq0}{[f_{X^{h_{\ell'}}}]_\text{Lip}}}{2(1-\alpha)}.
\end{equation*}
Recalling \eqref{upper:estimate:pi:i:n} and that, by assumption, $\bar\lambda_2\gamma_1>2$ if $\beta=1$, i.e.~$\gamma_1V_0''(\xi^0_\star)>\bar\lambda_2\gamma_1>2>1\vee\delta'$ if $\beta=1$, via \cite[Lemma~A.1(i)]{amlsa}
\begin{equation*}
\begin{aligned}
\mathbb{E}\bigg[\bigg(\sum_{\ell=1}^L\sum_{k=1}^{N_\ell}&\gamma_k\Pi_{k+1:N_\ell}(\rho^{h_\ell}_k-\rho^{h_{\ell-1}}_k)\bigg)^2\bigg]^\frac12\\
&\leq\sum_{\ell=1}^L\sum_{k=1}^{N_\ell}\gamma_k|\Pi_{k+1:N_\ell}|(\mathbb{E}[|\rho^{h_\ell}_k|^2]^\frac12+\mathbb{E}[|\rho^{h_{\ell-1}}_k|^2]^\frac12)\\
&
\leq C_\rho^1(M+1)\sum_{\ell=1}^L\bigg(
\sum_{k=1}^{N_\ell}h_\ell^{1+\theta}\gamma_k^{1+\frac{\delta'}\beta}|\Pi_{k+1:N_\ell}|
+\sum_{k=1}^{N_\ell}\gamma_k^2|\Pi_{k+1:N_\ell}|\bigg)\\
&\leq C_\rho^2\sum_{\ell=1}^L(h_\ell^{1+\theta}\gamma_{N_\ell}^\frac{\delta'}\beta+\gamma_{N_\ell})\\
&\leq C_\rho\sum_{\ell=1}^L\widetilde{\gamma}^\ell_{N_\ell},
\end{aligned}
\end{equation*}
where
\begin{equation*}
\begin{aligned}
C_\rho^2&=C_\rho^1\bar{C}(M+1)\bigg(\sup_{m\geq1}\bigg\{\gamma_m^{-\frac{\delta'}\beta}\sum_{k=1}^m\gamma_k^{1+\frac{\delta'}\beta}\exp\bigg(-V_0''(\xi^0_\star)\sum_{j=k+1}^m\gamma_j\bigg)\bigg\}\\
&\hphantom{=C_{33}\bar{C}(M+1)\bigg\{}\;
\vee\sup_{m\geq1}\bigg\{\gamma_m^{-1}\sum_{k=1}^m\gamma_k^2\exp\bigg(-V_0''(\xi^0_\star)\sum_{j=k+1}^m\gamma_j\bigg)\bigg\}\bigg),\\
C_\rho&=2C_\rho^2(\gamma_1^{-1+\frac{\delta'}\beta}\vee1).
\end{aligned}
\end{equation*}

\noindent
\emph{Step~5. Study of $\sum_{\ell=1}^L\sum_{k=1}^{N_\ell}\gamma_k\Pi_{k+1:N_\ell}(r^{h_\ell}_k-r^{h_{\ell-1}}_k)$.}
\newline
Recalling the definition \eqref{eq:C=Ci,i=1,2,3}, it follows from the inequality \eqref{bound:estimate:rnhell} that
\begin{equation*}
|r^{h_\ell}_n|\leq C_r^1h_\ell^{1+\theta}\gamma_n^\frac{\delta'}\beta.
\end{equation*}
where
\begin{equation*}
C_r^1=Ch_1^{-1}\gamma_1^{1-\frac{\delta'}\beta}.
\end{equation*}
Thus, using \eqref{upper:estimate:pi:i:n} and recalling that $\bar\lambda_2\gamma_1>2$ if $\beta=1$, hence $\gamma_1V_0''(\xi^0_\star)>\bar\lambda_2\gamma_1>2>\delta'$ if $\beta=1$, by \cite[Lemma~A.1(i)]{amlsa},
\begin{equation*}
\begin{aligned}
\mathbb{E}\bigg[\bigg(\sum_{\ell=1}^L\sum_{k=1}^{N_\ell}&\gamma_k\Pi_{k+1:N_\ell}(r^{h_\ell}_k-r^{h_{\ell-1}}_k)\bigg)^2\bigg]^\frac12\\
&\leq\sum_{\ell=1}^L\sum_{k=1}^{N_\ell}\gamma_k|\Pi_{k+1:N_\ell}|(\mathbb{E}[|r^{h_\ell}_k|^2]^\frac12+\mathbb{E}[|r^{h_{\ell-1}}_k|^2]^\frac12)\\
&\leq C_r^1(M+1)\sum_{\ell=1}^L\sum_{k=1}^{N_\ell}h_\ell^{1+\theta}\gamma_k^{1+\frac{\delta'}\beta}|\Pi_{k+1:N_\ell}|\\
&\leq C_r\sum_{\ell=1}^Lh_\ell^{1+\theta}\gamma_{N_\ell}^\frac{\delta'}\beta\\
&\leq C_r\gamma_1^{-1+\frac{\delta'}\beta}\sum_{\ell=1}^L\widetilde\gamma_{N_\ell}^\ell,
\end{aligned}
\end{equation*}
where
\begin{equation*}
C_r=C_r^1(M+1)\sup_{m\geq1}\bigg\{\gamma_m^{-\frac{\delta'}\beta}\sum_{k=1}^m\gamma_k^{1+\frac{\delta'}\beta}\exp\bigg(-V_0''(\xi^0_\star)\sum_{j=k+1}^m\gamma_j\bigg)\bigg\}.
\end{equation*}

\noindent
\emph{Step~6. Study of $\sum_{\ell=1}^L\sum_{k=1}^{N_\ell}\gamma_k\Pi_{k+1:N_\ell}(e^{h_\ell}_k-e^{h_{\ell-1}}_k)$.}
\newline
Note that the random variables $(\sum_{k=1}^{N_\ell}\gamma_k\Pi_{k+1:N_\ell}(e^{h_\ell}_k-e^{h_{\ell-1}}_k))_{\ell\geq1}$ are independent with zero mean and that, at each level $\ell\in\mathbb{N}^*$, $(e^{h_\ell}_k-e^{h_{\ell-1}}_k)_{k\geq1}$ are $(\widetilde{\mathcal{F}}^{h_\ell}_k)_{k\geq1}$-martingale increments.
Therefore
\begin{equation*}
\mathbb{E}\bigg[\bigg(\sum_{\ell=1}^L\sum_{k=1}^{N_\ell}\gamma_k\Pi_{k+1:N_\ell}(e^{h_\ell}_k-e^{h_{\ell-1}}_k)\bigg)^2\bigg]
=\sum_{\ell=1}^L\sum_{k=1}^{N_\ell}\gamma_k^2|\Pi_{k+1:n}|^2\mathbb{E}[(e^{h_\ell}_k-e^{h_{\ell-1}}_k)^2].
\end{equation*}

From \eqref{eq:e:def}, \eqref{eq:vee} and \eqref{eq:H1},
\begin{equation*}
\begin{aligned}
\mathbb{E}[(e^{h_\ell}_n-e^{h_{\ell-1}}_n)^2]
&=\mathbb{E}\big[\Var\big(H(\widetilde\xi^{h_\ell}_{n-1},\widetilde{X}^{h_\ell}_n)-H(\widetilde\xi^{h_{\ell-1}}_{n-1},\widetilde{X}^{h_{\ell-1}}_n)\big|\widetilde{\mathcal{F}}^{h_\ell}_{n-1}\big)\big]\\
&\leq\mathbb{E}\big[\big(H(\widetilde\xi^{h_\ell}_{n-1},\widetilde{X}^{h_\ell}_n)-H(\widetilde\xi^{h_{\ell-1}}_{n-1},\widetilde{X}^{h_{\ell-1}}_n)\big)^2\big]\\
&\leq\frac1{(1-\alpha)^2}(\mathbb{E}[|\mathds1_{\{\widetilde{X}^{h_\ell}_n>\widetilde\xi^{h_\ell}_{n-1}\}}-\mathds1_{\{X^0>\widetilde\xi^{h_\ell}_{n-1}\}}|]\\
&\hphantom{\leq\frac1{(1-\alpha)^2}(}
+\mathbb{E}[|\mathds1_{\{X^0>\widetilde\xi^{h_\ell}_{n-1}\}}-\mathds1_{\{X^0>\widetilde\xi^{h_{\ell-1}}_{n-1}\}}|]\\
&\hphantom{\leq\frac1{(1-\alpha)^2}(}
+\mathbb{E}[|\mathds1_{\{X^0>\widetilde\xi^{h_{\ell-1}}_{n-1}\}}-\mathds1_{\{\widetilde{X}^{h_{\ell-1}}_n>\widetilde\xi^{h_{\ell-1}}_{n-1}\}}|]).
\end{aligned}
\end{equation*}
On the one hand, for $j\in\{\ell-1,\ell\}$, recalling the definition \eqref{eq:eps(hl)}, by the tower law and Lemma~\ref{lmm:local:strong:error:indicator:func:variance},
\begin{equation*}
\begin{aligned}
\mathbb{E}[|\mathds1_{\{\widetilde{X}^{h_j}_n>\widetilde\xi^{h_j}_{n-1}\}}-\mathds1_{\{X^0>\widetilde\xi^{h_j}_{n-1}\}}|]
&=\mathbb{E}\Big[\mathbb{E}\big[|\mathds1_{\{\widetilde{X}^{h_j}_n>\widetilde\xi^{h_j}_{n-1}\}}-\mathds1_{\{X^0>\widetilde\xi^{h_j}_{n-1}\}}|\big|\widetilde{\mathcal{F}}^{h_j}_{n-1}\big]\Big]\\
&=\mathbb{E}\big[\mathbb{E}[|\mathds1_{\{\widetilde{X}^{h_j}_n>\xi\}}-\mathds1_{\{X^0>\xi\}}|]_{|\xi=\widetilde\xi^{h_j}_{n-1}}\big]\\
&\leq C'\epsilon(h_j)^{1+\theta}\\
&\leq C_e^1\epsilon(h_\ell)^{1+\theta},
\end{aligned}
\end{equation*}
where
\begin{equation*}
C_e^1=C'M^\frac12(\ln{M})^\frac12,
\end{equation*}
and
\begin{equation*}
C'=\begin{cases}
C_1&\text{if \eqref{assumption:finite:Lp:moment:variance} holds,}\\
C_2&\text{if \eqref{assumption:conditional:gaussian:concentration:variance} holds,}\\
C_3&\text{if \eqref{assump:unif:lipschitz:integrability:conditional:cdf:variance} holds,}
\end{cases}
\end{equation*}
with $C_1$, $C_2$ and $C_3$ being the positive constants defined in Lemma~\ref{lmm:local:strong:error:indicator:func:variance}.

On the other hand, using Assumption~\ref{asp:misc}(\ref{asp:misc:iv}),
\begin{equation*}
\begin{aligned}
\mathbb{E}[|&\mathds1_{\{X^0>\widetilde\xi^{h_\ell}_{n-1}\}}-\mathds1_{\{X^0>\widetilde\xi^{h_{\ell-1}}_{n-1}\}}|]\\
&=\mathbb{E}[\mathds1_{\{\widetilde\xi^{h_\ell}_{n-1}<X^0\leq\widetilde\xi^{h_{\ell-1}}_{n-1}\}}+\mathds1_{\{\widetilde\xi^{h_{\ell-1}}_{n-1}<X^0\leq\widetilde\xi^{h_\ell}_{n-1}\}}]\\
&=\mathbb{E}[|F_{X^0}(\widetilde\xi^{h_\ell}_{n-1})-F_{X^0}(\widetilde\xi^{h_{\ell-1}}_{n-1})|]\\
&\leq\|f_{X^0}\|_\infty\mathbb{E}[(\widetilde\xi^{h_\ell}_{n-1}-\widetilde\xi^{h_{\ell-1}}_{n-1})^4]^\frac14\\
&\leq\|f_{X^0}\|_\infty\big(\mathbb{E}[(\widetilde\xi^{h_\ell}_{n-1}-\xi^{h_{\ell+\lceil\theta\ell\rceil}}_\star)^4]^\frac14
+|\xi^{h_{\ell+\lceil\theta\ell\rceil}}_\star-\xi^0_\star|\\
&\hphantom{\leq\|f_{X^0}\|_\infty\big(}
+|\xi^0_\star-\xi^{h_{\ell-1+\lceil\theta(\ell-1)\rceil}}_\star|
+\mathbb{E}[(\xi^{h_{\ell-1+\lceil\theta(\ell-1)\rceil}}_\star-\widetilde\xi^{h_{\ell-1}}_{n-1})^4]^\frac14\big).
\end{aligned}
\end{equation*}
Recalling \eqref{eq:total:recall}, by using the inequality $(x+y)^\frac14\leq x^\frac14+y^\frac14$, $x,y\geq0$,
\begin{equation*}
\mathbb{E}[(\widetilde\xi^{h_\ell}_{n-1}-\xi^{h_{\ell+\lceil\theta\ell\rceil}}_\star)^4]^\frac14
\leq(C_g^2)^\frac14(h_\ell^\frac{1+\theta}2\gamma_n^\frac{\delta'}{2\beta}+\gamma_n^\frac12).
\end{equation*}
Besides, by Lemma~\ref{lmm:error}(\ref{lmm:error:weak})
\begin{equation*}
|\xi^{h_{\ell+\lceil\theta\ell\rceil}}_\star-\xi^0_\star|
\leq c'h_0^{-1}h_\ell^{1+\theta},
\end{equation*}
recalling that $c'\geq0$ is the constant described in \eqref{eq:oO}.
Thus
\begin{equation*}
\begin{aligned}
\mathbb{E}[|&\mathds1_{\{X^0>\widetilde\xi^{h_\ell}_{n-1}\}}-\mathds1_{\{X^0>\widetilde\xi^{h_{\ell-1}}_{n-1}\}}|]\\
&\leq\|f_{X^0}\|_\infty\big((C_g^2)^\frac14(M+1)(h_\ell^\frac{1+\theta}2\gamma_n^\frac{\delta'}{2\beta}+\gamma_n^\frac12)+c'h_0^{-1}(M^2+1)h_\ell^{1+\theta}\big)\\
&\leq C_e^2\big(h_\ell^{1+\theta}+h_\ell^\frac{1+\theta}2\gamma_n^\frac{\delta'}{2\beta}+\gamma_n^\frac12\big),
\end{aligned}
\end{equation*}
where
\begin{equation*}
C_e^2=2\|f_{X^0}\|_\infty\Big(\big((C_g^2)^\frac14(M+1)\big)\vee\big(c'h_0^{-1}(M^2+1)\big)\Big).
\end{equation*}
Combining the previous results and invoking \cite[Lemma~A.1(i)]{amlsa}, recalling the decreasing monotonicity of $(\gamma_n)_{n\geq1}$ and the fact that $\bar\lambda_2\gamma_1>2$ if $\beta=1$, i.e.~$\gamma_1V_0''(\xi^0_\star)>\bar\lambda_2\gamma_1>2>(1+\frac{\delta'}2)\vee\frac32$ if $\beta=1$,
\begin{equation*}
\begin{aligned}
\mathbb{E}\bigg[\bigg(&\sum_{\ell=1}^L\sum_{k=1}^{N_\ell}\gamma_k\Pi_{k+1:n}(e^{h_\ell}_k-e^{h_{\ell-1}}_k)\bigg)^2\bigg]\\
&
\leq\frac1{(1-\alpha)^2}\sum_{\ell=1}^L\bigg(
C_e^1\sum_{k=1}^{N_\ell}\epsilon(h_\ell)^{1+\theta}\gamma_k^2|\Pi_{k+1:N_\ell}|^2
+C_e^2\sum_{k=1}^{N_\ell}h_\ell^{1+\theta}\gamma_k^2|\Pi_{k+1:N_\ell}|^2\\
&\hphantom{\leq\frac1{(1-\alpha)^2}\sum_{\ell=1}^L\bigg(}
+C_e^2\sum_{k=1}^{N_\ell}h_\ell^\frac{1+\theta}2\gamma_k^{2+\frac{\delta'}{2\beta}}|\Pi_{k+1:N_\ell}|^2
+C_e^2\sum_{k=1}^{N_\ell}\gamma_k^\frac52|\Pi_{k+1:N_\ell}|^2\bigg)\\
&\leq C_e\sum_{\ell=1}^L\big(\gamma_{N_\ell}\epsilon(h_\ell)^{1+\theta}+\gamma_{N_\ell}(\widetilde{\gamma}_{N_\ell}^\ell)^\frac12\big),
\end{aligned}
\end{equation*}
where
\begin{equation*}
\begin{aligned}
C_e&=\frac1{(1-\alpha)^2}\bar{C}^2\bigg(
(C_e^1+C_e^2)\sup_{m\geq1}\bigg\{\gamma_m^{-1}\sum_{k=1}^m\gamma_k^2\exp\bigg(-2V_0''(\xi^0_\star)\sum_{j=k+1}^m\gamma_j\bigg)\bigg\}\\
&\hphantom{=\frac{C_e^1+C_e^2}{(1-\alpha)^2}\bar{C}^2\bigg(}\;
\vee C_e^2\bigg(\sup_{m\geq1}\bigg\{\gamma_m^{-1-\frac{\delta'}{2\beta}}\sum_{k=1}^m\gamma_k^{2+\frac{\delta'}{2\beta}}\exp\bigg(-2V_0''(\xi^0_\star)\sum_{j=k+1}^m\gamma_j\bigg)\bigg\}\\
&\hphantom{=\frac{C_e^1+C_e^2}{(1-\alpha)^2}\bar{C}^2\bigg(\vee C_e^2\bigg(}
+\sup_{m\geq1}\bigg\{\gamma_m^{-\frac32}\sum_{k=1}^m\gamma_k^\frac52\exp\bigg(-2V_0''(\xi^0_\star)\sum_{j=k+1}^m\gamma_j\bigg)\bigg\}
\bigg)\bigg).
\end{aligned}
\end{equation*}

\noindent
\emph{Step~7. Conclusion.}
\newline
Combining the previous upper bounds, we obtain
\begin{equation*}
\mathbb{E}[(\widetilde\xi^\text{\tiny\rm ML}_\mathbf{N}-\xi^{h_{L+\lceil\theta L\rceil}}_\star)^2]
\leq \widehat{C}\bigg(\gamma_{N_0}+\bigg(\sum_{\ell=1}^L\widetilde\gamma_{N_\ell}^\ell\bigg)^2+\sum_{\ell=1}^L\gamma_{N_\ell}(\widetilde{\gamma}_{N_\ell}^\ell)^\frac12+\sum_{\ell=1}^L\gamma_{N_\ell}\epsilon(h_\ell)^{1+\theta}\bigg),
\end{equation*}
where
\begin{equation}
\label{eq:master:cst}
\widehat{C}=6\Big(C_{-1}\vee(C_0)^2\vee(C_g)^2\vee(C_\rho)^2\vee\big((C_r)^2\gamma_1^{-2+{2\frac{\delta'}\beta}}\big)\vee C_e\Big).
\end{equation}

\section{Proofs of the Complexity Results}
\label{complexity}

\begin{proof}[Proof of Proposition~\ref{prp:selection:number:level:adNSA}]
By Lemma~\ref{lmm:error}(\ref{lmm:error:weak}), the bias error of adNSA satisfies
\begin{equation}
\label{eq:set:l:1}
|\xi^{h_{\ell+\lceil\theta\ell\rceil}}_\star-\xi^0_\star|
\leq Ch_{\ell+\lceil\theta\ell\rceil}
\leq\frac{Ch_0}{M^{\ell(1+\theta)}},
\end{equation}
where $C>0$ is the constant described in \eqref{eq:oO}.
Hence, assuming that $h_0>C^{-1}\varepsilon$, in order to achieve a bias error of order $\varepsilon>0$, it suffices to take $\ell$ minimal such that
\begin{equation}
\label{eq:set:l:2}
\frac{Ch_0}{M^{\ell(1+\theta)}}\leq\varepsilon,
\quad\text{i.e.}\quad
\ell=\bigg\lceil\frac{\ln{(Ch_0\varepsilon^{-1})}}{(1+\theta)\ln{M}}\bigg\rceil\geq1.
\end{equation}
\end{proof}

\begin{proof}[Proof of Lemma~\ref{lmm:E:h:-1}]
Let $\ell,n\in\mathbb{N}^*$.
By the law of total probability,
\begin{equation*}
\mathbb{E}[h_{\ell+\widetilde\eta^\ell_n}^{-1}]
=\sum_{k=0}^{\lceil\theta\ell\rceil}h_{\ell+k}^{-1}\mathbb{P}(\widetilde\eta^\ell_n=k)
\leq h_\ell^{-1}+\sum_{k=1}^{\lceil\theta\ell\rceil}h_{\ell+k}^{-1}\mathbb{P}(\widetilde\eta^\ell_n=k).
\end{equation*}
For $k\in\intl1,\lceil\theta\ell\rceil\intr$,
by the definition \eqref{eq:eta}, the tower law and Assumption~\ref{asp:misc}(\ref{asp:misc:iv}),
\begin{equation*}
\begin{aligned}
\mathbb{P}(\widetilde\eta^\ell_n=k)
&\leq\mathbb{P}(|X^{h_{\ell+k-1}}_n-\widetilde\xi^{h_\ell}_{n-1}|<C_\mathrm{ad}\psi^{\ell,k-1}_n)\\
&=\mathbb{E}[\mathbb{P}(|X^{h_{\ell+k-1}}_n-\xi|<C_\mathrm{ad}\psi^{\ell,k-1}_n)_{|\xi=\widetilde\xi^{h_\ell}_{n-1}}]\\
&\leq2C_\mathrm{ad}\Big(\sup_{\ell'\geq0}{\|f_{X^{h_{\ell'}}}\|_\infty}\Big)\psi^{\ell,k-1}_n.
\end{aligned}
\end{equation*}
Thus
\begin{equation*}
\mathbb{E}[h_{\ell+\widetilde\eta^\ell_n}^{-1}]
\leq C_0\bigg(h_\ell^{-1}+\sum_{k=1}^{\lceil\theta\ell\rceil}h_{\ell+k}^{-1}\psi^{\ell,k-1}_n\bigg),
\end{equation*}
where
\begin{equation*}
C_0=\Big(2C_\mathrm{ad}\sup_{\ell'\geq0}{\|f_{X^{h_{\ell'}}}\|_\infty}\Big)\vee1.
\end{equation*}

Recall now the definition \eqref{eq:psi} and the fact that $r>1$. If \eqref{assumption:finite:Lp:moment} holds, then
\begin{equation*}
\begin{aligned}
\mathbb{E}[h_{\ell+\widetilde\eta^\ell_n}^{-1}]
&\leq C_0\bigg(h_\ell^{-1}+u_n^{-\frac1{p_\star}}\sum_{k=1}^{\lceil\theta\ell\rceil}h_{\ell+k}^{-1}h_{\theta\ell(r-1)+k}^\frac1r\bigg)\\
&=C_0\bigg(h_\ell^{-1}+\frac{h_0^{-1+\frac1r}n^\frac\delta{p_\star}}{\gamma_1^\frac1{p_\star}M^{-\ell+\theta\ell(1-\frac1r)}}\sum_{k=1}^{\lceil\theta\ell\rceil}M^{(1-\frac1r)k}\bigg)\\
&\leq C_0n^\frac\delta{p_\star}\bigg(h_\ell^{-1}+\frac{h_0^{-1+\frac1r}}{\gamma_1^\frac1{p_\star}M^{-\ell+\theta\ell(1-\frac1r)}}M^{1-\frac1r}\frac{M^{(2+\theta\ell)(1-\frac1r)}}{M^{1-\frac1r}-1}\bigg)\\
&=C_1n^\frac\delta{p_\star}h_\ell^{-1},
\end{aligned}
\end{equation*}
where
\begin{equation*}
C_1=C_0\bigg(1+\frac{h_0^\frac1rM^{3(1-\frac1r)}}{\gamma_1^\frac1{p_\star}(M^{1-\frac1r}-1)}\bigg).
\end{equation*}

Next, recalling that $h_0=\frac1K$, note that, for $k\in[\![1,\lceil\theta\ell\rceil]\!]$,
\begin{equation*}
\begin{aligned}
\ln\bigg(\bigg(\frac{\gamma_n}{\gamma_1}\bigg)^{-\frac12}h_{\ell+k}^{-\frac{1+\theta}2}\bigg)
&=\frac\beta2\ln{n}+\frac{1+\theta}2(\ell+k)\ln{M}+\frac{1+\theta}2\ln{K}\\
&\leq\frac\beta2\vee\bigg(\frac12(1+\theta)(2+\theta)\ln{M}\bigg)\vee\bigg(\frac{1+\theta}2\ln{K}\bigg)(1+\ell+\ln{n})\\
&\leq C_2'(\ell+\ln{n}),
\end{aligned}
\end{equation*}
where
\begin{equation*}
C_2'=\beta\vee\big((1+\theta)(2+\theta)\ln{M}\big)\vee\big((1+\theta)\ln{K}\big).
\end{equation*}
Hence, if \eqref{assumption:conditional:gaussian:concentration} or \eqref{assump:unif:lipschitz:integrability:conditional:cdf:bis} holds, then
\begin{equation*}
\begin{aligned}
\mathbb{E}[h_{\ell+\widetilde\eta^\ell_n}^{-1}]
&\leq C_0\bigg(h_\ell^{-1}+\sum_{k=1}^{\lceil\theta\ell\rceil}\ln^\frac12\bigg(\bigg(\frac{\gamma_n}{\gamma_1}\bigg)^{-\frac12}h^{-\frac{1+\theta}2}_{\ell+k}\bigg)h_{\ell+k}^{-1}h_{\theta\ell(r-1)+k}^\frac1r\bigg)\\
&\leq C_0(\ell+\ln{n})^\frac12\bigg(h_\ell^{-1}+C_2'\frac{h_0^{-1+\frac1r}}{M^{-\ell+\theta\ell(1-\frac1r)}}\sum_{k=1}^{\lceil\theta\ell\rceil}M^{(1-\frac1r)k}\bigg)\\
&\leq C_0(\ell+\ln{n})^\frac12\bigg(h_\ell^{-1}+C_2'\frac{h_0^{-1+\frac1r}}{M^{-\ell+\theta\ell(1-\frac1r)}}M^{1-\frac1r}\frac{M^{(2+\theta\ell)(1-\frac1r)}}{M^{1-\frac1r}-1}\bigg)\\
&=C_2(\ell+\ln{n})^\frac12h_\ell^{-1},
\end{aligned}
\end{equation*}
where
\begin{equation*}
C_2=C_0\bigg(1+C_2'\frac{h_0^\frac1rM^{3(1-\frac1r)}}{M^{1-\frac1r}-1}\bigg).
\end{equation*}

All in all,
\begin{equation*}
\mathbb{E}[h_{\ell+\widetilde\eta^\ell_n}^{-1}]
\leq C\begin{cases}
    n^\frac\delta{p_\star}h_\ell^{-1}
    &\text{if \eqref{assumption:finite:Lp:moment} holds,}\\
    (\ell+\ln{n})^\frac12h_\ell^{-1}
    &\text{if \eqref{assumption:conditional:gaussian:concentration} or \eqref{assump:unif:lipschitz:integrability:conditional:cdf:bis} holds,}
\end{cases}
\end{equation*}
where
\begin{equation}
\label{eq:C=C1:max:C2}
C=\begin{cases}
    C_1
    &\text{if \eqref{assumption:finite:Lp:moment} holds,}\\
    C_2
    &\text{if \eqref{assumption:conditional:gaussian:concentration} or \eqref{assump:unif:lipschitz:integrability:conditional:cdf:bis} holds.}
\end{cases}
\end{equation}

\end{proof}

\begin{proof}[Proof of Proposition~\ref{prp:ansa:complexity}]
Let $\ell,n\in\mathbb{N}^*$.
Based on \eqref{adNSA}, adNSA performs $n$ iterations, where, at each iteration $k\in[\![1,n]\!]$, it simulates one copy of $Y$ of cost $c_0$, and simulates $h_{\ell+\widetilde{\eta}^\ell_k}^{-1}$ i.i.d.~copies of $Z$, each copy being of cost $c_0$. Overall, the average computational cost of adNSA satisfies
\begin{equation*}
\Cost_{\text{\rm adNSA}}
\leq c_0\sum_{k=1}^n(1+\mathbb{E}[h_{\ell+\widetilde\eta^\ell_k}^{-1}])
\leq2c_0\sum_{k=1}^n\mathbb{E}[h_{\ell+\widetilde\eta^\ell_k}^{-1}].
\end{equation*}

By Lemma~\ref{lmm:E:h:-1}, if \eqref{assumption:finite:Lp:moment} holds,
\begin{equation*}
\Cost_{\text{\rm adNSA}}
\leq2c_0C'\sum_{k=1}^nh_\ell^{-1}k^\frac\delta{p_\star},
\end{equation*}
where $C'>0$ is the constant described in \eqref{eq:C=C1:max:C2}.
Via a series-integral comparison,
\begin{equation*}
\Cost_{\text{\rm adNSA}}
\leq C_1n^{1+\frac\delta{p_\star}}h_\ell^{-1},
\quad\text{where}\quad
C_1=2c_0C'\bigg(1\vee\frac{2^\frac\delta{p_\star}p_\star}{p_\star+\delta}\bigg).
\end{equation*}

If \eqref{assumption:conditional:gaussian:concentration} or \eqref{assump:unif:lipschitz:integrability:conditional:cdf:bis} holds, using Lemma~\ref{lmm:E:h:-1} and the inequality $(x+y)^\frac12\leq x^\frac12+y^\frac12$, $x,y\geq0$,
\begin{equation*}
\Cost_{\text{\rm adNSA}}
\leq2c_0C'\sum_{k=1}^nh_\ell^{-1}(\ell+\ln{k})^\frac12
\leq2c_0C'h_\ell^{-1}\bigg(n\ell^\frac12+\sum_{k=1}^n\ln^\frac12{k}\bigg).
\end{equation*}
If $n\geq2$, then, via a series-integral comparison, the inequality $(x+y)^\frac12\leq x^\frac12+y^\frac12$, $x,y\geq0$, and the increasing monotony of $\mathbb{R}_+^*\ni x\mapsto\ln^\frac12{x}$,
\begin{equation*}
\sum_{k=1}^n\ln^\frac12{k}
\leq n\ln^\frac12{2}+\sum_{k=2}^n\ln^\frac12(k-1)
\leq n\ln^\frac12{2}+\int_1^n\ln^\frac12{s}\mathrm{d}s
\leq2n\ln^\frac12{n}.
\end{equation*}
The above upper bound can also be obtained by integration by parts.
Thus, for all $n\in\mathbb{N}^*$,
\begin{equation*}
\Cost_{\text{\rm adNSA}}
\leq C_2n(\ell^\frac12+\ln^\frac12{n})h_\ell^{-1},
\quad\text{where}\quad
C_2=4c_0C'.
\end{equation*}

Overall,
\begin{equation*}
\Cost_{\text{\rm adNSA}}
\leq C
\begin{cases}
n^{1+\frac\delta{p_\star}}h_\ell^{-1}
&\text{if \eqref{assumption:finite:Lp:moment} holds,}\\
n(\ell^\frac12+\ln^\frac12{n})h_\ell^{-1}
&\text{if \eqref{assumption:conditional:gaussian:concentration} or \eqref{assump:unif:lipschitz:integrability:conditional:cdf:bis} holds,}
\end{cases}
\end{equation*}
where
\begin{equation}
\label{eq:C:meta}
C=\begin{cases}
    C_1
    &\text{if \eqref{assumption:finite:Lp:moment} holds,}\\
    C_2
    &\text{if \eqref{assumption:conditional:gaussian:concentration} or \eqref{assump:unif:lipschitz:integrability:conditional:cdf:bis} holds.}
\end{cases}
\end{equation}
\end{proof}

\begin{proof}[Proof of Theorem~\ref{thm:cost:adansa}]

Let $n\in\mathbb{N}^*$ and $\varepsilon>0$, and set $\ell\in\mathbb{N}^*$ as in \eqref{selection:number:level:adNSA}.
According to \eqref{eq:set:l:2}, one has $h_\ell^{1+\theta}\leq\frac{\varepsilon}{C'}$, where $C'>0$ is the constant described in \eqref{eq:oO}.
Thus, consequently from Proposition~\ref{prp:error:statistical:bis}(\ref{prp:error:statistical:bis:i}), the statistical error of adNSA satisfies
\begin{equation*}
\begin{aligned}
\mathbb{E}[(\widetilde\xi^{h_\ell}_n-\xi^{h_{\ell+\lceil\theta\ell\rceil}}_\star)^2]^\frac12
&\leq(C^\beta)^\frac12(\widetilde{\gamma}^\ell_n)^\frac12\\
&\leq C
\begin{cases}
(n^{-\frac\delta2}\varepsilon^\frac12)\vee n^{-\frac\beta2}
&\text{if \eqref{assumption:finite:Lp:moment:bias} holds,}\\
n^{-\frac\beta2}
&\text{if \eqref{assumption:conditional:gaussian:concentration:bias} or \eqref{assump:unif:lipschitz:integrability:conditional:cdf:bias} holds,}
\end{cases}
\end{aligned}
\end{equation*}
where
\begin{equation}
\label{eq:cst}
C=\Big(C^\beta\big((C')^{-1}\vee1\big)\gamma_1\Big)^\frac12,
\end{equation}
with $C^\beta>0$ denoting the constant described in \eqref{eq:L2:stat:error}.

Assume \eqref{assumption:finite:Lp:moment:bias} holds, with $(n^{-\frac\delta2}\varepsilon^\frac12)\vee n^{-\frac\beta2}=n^{-\frac\delta2}\varepsilon^\frac12$.
To achieve a statistical error of $\varepsilon$, it suffices to take $n$ minimal such that
\begin{equation}
\label{eq:n:1}
Cn^{-\frac\delta2}\varepsilon^\frac12\leq\varepsilon,
\quad\text{i.e.}\quad
n=\lceil C^\frac2\delta\varepsilon^{-\frac1\delta}\rceil.
\end{equation}
If \eqref{assumption:finite:Lp:moment:bias} holds, with $(n^{-\frac\delta2}\varepsilon^\frac12)\vee n^{-\frac\beta2}=n^{-\frac\beta2}$, or if \eqref{assumption:conditional:gaussian:concentration:bias} or \eqref{assump:unif:lipschitz:integrability:conditional:cdf:bias} holds, then, to achieve a statistical error of order $\varepsilon$, we must take $n$ minimal such that
\begin{equation}
\label{eq:n:2}
Cn^{-\frac\beta2}\leq\varepsilon,
\quad\text{i.e.}\quad
n=\lceil C^\frac2\beta\varepsilon^{-\frac2\beta}\rceil.
\end{equation}

We now go back to the framework \eqref{assumption:finite:Lp:moment:bias}. If $n$ is set as in \eqref{eq:n:1}, then
\begin{equation}
\label{eq:intermediate}
n^{-\frac\delta2}\varepsilon^\frac12\leq C^{-1}\varepsilon,
\quad\text{and}\quad
n^{-\frac\beta2}>\frac{\varepsilon^\frac\beta{2\delta}}{(C^\frac2\delta+\varepsilon^\frac1\delta)^\frac\beta2}.
\end{equation}
Thus, if $n^{-\frac\beta2}\leq n^{-\frac\delta2}\varepsilon^\frac12$ for all $\varepsilon>0$, then, necessarily, $\delta\leq\frac\beta2$.
Conversely, if $\delta\leq\frac\beta2$, then, according to \eqref{eq:intermediate}, by taking $\varepsilon$ small enough, it holds that $n^{-\frac\beta2}\leq n^{-\frac\delta2}\varepsilon^\frac12$.
Applying a similar reasoning to \eqref{eq:n:2} shows that $n^{-\frac\delta2}\varepsilon^\frac12\leq n^{-\frac\beta2}$ if and only if $\delta\geq\frac\beta2$, provided that $\varepsilon$ is taken small enough.

Let $\widehat{C}>0$ denote the constant defined in \eqref{eq:C:meta}, set $r>1$, and $\ell$ as in \eqref{selection:number:level:adNSA}, and take $\varepsilon>0$ small enough. Following Proposition~\ref{prp:ansa:complexity}, by setting $n$ as in \eqref{eq:n:1} if \eqref{assumption:finite:Lp:moment:bias} holds with $\delta\leq\frac\beta2$, and as in \eqref{eq:n:1} otherwise, it holds
\begin{equation*}
\Cost_{\text{\rm adNSA}}
\leq\bar{C}\begin{cases}
\varepsilon^{-\frac{p_\star+\delta}{p_\star\delta}-\frac1{1+\theta}}
&\text{if \eqref{assumption:finite:Lp:moment:bias} holds with $\delta\leq\frac\beta2$,}\\
\varepsilon^{-\frac{2(\delta+p_\star)}{\beta p_\star}-\frac1{1+\theta}}
&\text{if \eqref{assumption:finite:Lp:moment:bias} holds with $\delta\geq\frac\beta2$,}\\
\varepsilon^{-\frac2\beta-\frac1{1+\theta}}|\ln{\varepsilon}|^\frac12
&\text{if \eqref{assumption:conditional:gaussian:concentration:bias} or \eqref{assump:unif:lipschitz:integrability:conditional:cdf:bias} holds,}
\end{cases}
\end{equation*}
where
\begin{equation}
\label{eq:cst:bar}
\bar{C}=\widehat{C}(C')^\frac1{1+\theta}h_0^\frac1{1+\theta}h_1^{-1}
\begin{cases}
(C^\frac2\delta+1)^\frac{p_\star+\delta}{p_\star}
&\text{if \eqref{assumption:finite:Lp:moment:bias} holds with $\delta\leq\frac\beta2$,}\\
(C^\frac2\beta+1)^\frac{p_\star+\delta}{p_\star}
&\text{if \eqref{assumption:finite:Lp:moment:bias} holds with $\delta\geq\frac\beta2$,}\\
(1+C^\frac2\beta)\\
\times
(\ln{(1+C^\frac2\beta)}+\frac2\beta+1
+\frac{1+|\ln{(Ch_0^{-1})}|}{(1+\theta)\ln{M}})
&\text{if \eqref{assumption:conditional:gaussian:concentration:bias} or \eqref{assump:unif:lipschitz:integrability:conditional:cdf:bias} holds.}
\end{cases}
\end{equation}
The best complexity rates are then easily obtained by optimizing the upper bounds with respect to $\theta$, $\delta$ and $\beta$.

\end{proof}

\begin{proof}[Proof of Proposition~\ref{prp:selection:number:level:adMLSA}]
Lemma~\ref{lmm:error}(\ref{lmm:error:weak}) guarantees that the bias error of adMLSA satisfies
\begin{equation*}
|\xi^{h_{L+\lceil\theta L\rceil}}_\star-\xi^0_\star|
\leq Ch_{L+\lceil\theta L\rceil}
\leq\frac{Ch_0}{M^{L(1+\theta)}},
\end{equation*}
where $C>0$ is the constant described in \eqref{eq:oO}.
Thus, provided that $h_0>C^{-1}\varepsilon$, it suffices to take $L$ minimal so that
\begin{equation*}
\frac{Ch_0}{M^{L(1+\theta)}}\leq\varepsilon,
\quad\text{i.e.}\quad
L=\bigg\lceil\frac{\ln{(Ch_0\varepsilon^{-1})}}{(1+\theta)\ln{M}}\bigg\rceil\geq1.
\end{equation*}
\end{proof}

\begin{proof}[Proof of Proposition~\ref{prp:amlsa:complexity}]
Recalling that no adaptation applies at level $0$ (c.f.~Remark~\ref{rmk:refine}(\ref{rmk:level:0})), the average complexity of the adaptive multilevel SA algorithm satisfies
\begin{equation*}
\begin{aligned}
\Cost_{\text{adMLSA}}
&\leq c_0\bigg(\sum_{n=1}^{N_0}(1+h_0^{-1})
+\sum_{\ell=1}^L\sum_{n=1}^{N_\ell}
(1+\mathbb{E}[h_{\ell+\widetilde\eta^\ell_n}^{-1}\vee h_{\ell-1+\widetilde\eta^{\ell-1}_n}^{-1}])
\bigg)\\
&\leq3c_0\bigg(N_0h_0^{-1}
+\sum_{\ell=1}^L\sum_{n=1}^{N_\ell}(\mathbb{E}[h_{\ell+\widetilde\eta^\ell_n}^{-1}]+\mathbb{E}[h_{\ell-1+\widetilde\eta^{\ell-1}_n}^{-1}])
\bigg)
\end{aligned}
\end{equation*}
Using a similar reasoning to the proof of Proposition~\ref{prp:ansa:complexity}, for $\ell\in\mathbb{N}^*$, $N_\ell\in\mathbb{N}^*$, and $j\in\{\ell-1,\ell\}$,
\begin{equation*}
\sum_{n=1}^{N_\ell}\mathbb{E}[h_{j+\widetilde{\eta}^j_n}^{-1}]
\leq (C'\vee1)
\begin{cases}
N_\ell^{1+\frac\delta{p_\star}}h_\ell^{-1}
&\text{if \eqref{assumption:finite:Lp:moment} holds,}\\
N_\ell(\ell^\frac12+\ln^\frac12{N_\ell})h_\ell^{-1}
&\text{if \eqref{assumption:conditional:gaussian:concentration} or \eqref{assump:unif:lipschitz:integrability:conditional:cdf:bis} holds,}
\end{cases}
\end{equation*}
where $C'>0$ is the constant described in \eqref{eq:C:meta}.
Note that the previous inequality takes into account the case where $\ell=1$ and $j=0$, for which no adaptation occurs.
Therefore
\begin{equation*}
\Cost_{\text{adMLSA}}
\leq C
\begin{cases}
        \sum_{\ell=0}^LN_\ell^{1+\frac\delta{p_\star}}h_\ell^{-1}
    &\text{if \eqref{assumption:finite:Lp:moment} holds,}\\
        \sum_{\ell=0}^LN_\ell(L^\frac12+\ln^\frac12{N_\ell})h_\ell^{-1}
    &\text{if \eqref{assumption:conditional:gaussian:concentration} or \eqref{assump:unif:lipschitz:integrability:conditional:cdf:bis} holds.}
\end{cases}
\end{equation*}
where
\begin{equation}
\label{eq:new:C}
C=6c_0(C'\vee1).
\end{equation}

\end{proof}

\begin{proof}[Proof of Theorem~\ref{thm:amlsa:complexity}]
Define
\begin{equation*}
\phi_L(N_\ell):=
\begin{cases}
N_\ell^{1+\frac\delta{p_\star}}
&\text{if \eqref{assumption:finite:Lp:moment:bias} holds,}\\
N_\ell(L^\frac12+\ln^\frac12{N_\ell})
&\text{if \eqref{assumption:conditional:gaussian:concentration:bias} or \eqref{assump:unif:lipschitz:integrability:conditional:cdf:bias} holds.}
\end{cases}
\end{equation*}
Following \cite{GH19,CFL23}, a heuristic proxy for the upper estimate in \eqref{L2:norm:adaML:VaR} is $\bar{C}\sum_{\ell=1}^L\gamma_{N_\ell}\epsilon(h_\ell)^{1+\theta}$,
where $\bar{C}$ is the constant given in \eqref{eq:master:cst}.
To determine the optimal number of iterations $N_0,\dots,N_L\in\mathbb{N}^*$, we relax the problem to seeking the optimal values of $N_0,\dots,N_L\in\mathbb{R}$ that minimize the complexity while constraining the proxy $L^2(\mathbb{P})$-error control to $\varepsilon^2$:
\begin{equation*}
\begin{cases}
\min_{N_0,\dots,N_L\in\mathbb{R}}
&C\sum_{\ell=0}^L\phi_L(N_\ell)h_\ell^{-1},\\
\text{s.t.}
&\bar{C}\sum_{\ell=0}^L\gamma_{N_\ell}\epsilon(h_\ell)^{1+\theta}=\varepsilon^2.
\end{cases}
\end{equation*}
where $C>0$ is the constant given in \eqref{eq:new:C}.
We can check easily that the obtained solutions are positive and that, through them, the remaining terms in the upper estimate \eqref{L2:norm:adaML:VaR} are of order $\varepsilon^2$.
\\

\noindent
(\ref{thm:amlsa:complexity:indicator:func-ia:N})\
\emph{Step~1. Iteration amounts.}
\newline
The problem to solve is
\begin{equation*}
\begin{cases}
\min_{N_0,\dots,N_L\in\mathbb{R}}
&C\sum_{\ell=0}^LN_\ell^\frac{p_\star+\delta}{p_\star}h_\ell^{-1},\\
\text{s.t.}
&\bar{C}\gamma_1\sum_{\ell=0}^LN_\ell^{-\beta}h_\ell^\frac{(1+\theta)p_\star}{2(p+1)}=\varepsilon^2.
\end{cases}
\end{equation*}
The associated Lagrangian can be written as
\begin{equation*}
\mathfrak{L}
=C\sum_{\ell=0}^L\frac{N_\ell^\frac{p_\star+\delta}{p_\star}}{h_\ell}
+\lambda\bigg(\bar{C}\gamma_1\sum_{\ell=0}^L\frac{h_\ell^\frac{(1+\theta)p_\star}{2(p+1)}}{N_\ell^\beta}-\varepsilon^2\bigg),
\end{equation*}
where $\lambda\in\mathbb{R}$ is the Lagrangian multiplier.
Optimizing with respect to $N_\ell$, $\ell\in[\![0,L]\!]$, leads to
\begin{equation}
\label{eq:N:l:1}
N_\ell
=\lambda^\frac{p_\star}{(1+\beta)p_\star+\delta}
\bigg(\frac{\bar{C}\gamma_1\beta p_\star}{C(p_\star+\delta)}\bigg)^\frac{p_\star}{(1+\beta)p_\star+\delta}
h_\ell^{\frac{(3+\theta)p_\star+2}{2(p_\star+1)((1+\beta)p_\star+\delta)}p_\star}.
\end{equation}
Injecting this value back into the Lagrangian gives
\begin{equation*}
\begin{aligned}
\mathfrak{L}
&=C^\frac{\beta p_\star}{(1+\beta)p_\star+\delta}(\bar{C}\gamma_1)^\frac{p_\star+\delta}{(1+\beta)p_\star+\delta}\\
&\hphantom{=}\times
\bigg(\bigg(\frac{\beta p_\star}{(1+\beta)p_\star+\delta}\bigg)^\frac{p_\star+\delta}{(1+\beta)p_\star+\delta}+\bigg(\frac{\beta p_\star}{(1+\beta)p_\star+\delta}\bigg)^{-\frac{\beta p_\star}{(1+\beta)p_\star+\delta}}\bigg)\\
&\hphantom{=}\times
\lambda^\frac{p_\star+\delta}{(1+\beta)p_\star+\delta}
\sum_{\ell=0}^Lh_\ell^{-\frac{(2\beta-(1+\theta))p_\star+(2\beta-(1+\theta)\delta)}{2(p_\star+1)((1+\beta)p_\star+\delta)}p_\star}
-\lambda\varepsilon^2.
\end{aligned}
\end{equation*}
The optimal Lagrangian multiplier $\lambda$ satisfies
\begin{equation*}
\begin{aligned}
\lambda^\frac{p\star}{(1+\beta)p_\star+\delta}
&=\bigg\{C^\frac{\beta p_\star}{(1+\beta)p_\star+\delta}(\bar{C}\gamma_1)^\frac{p_\star+\delta}{(1+\beta)p_\star+\delta}\bigg(\frac{p_\star+\delta}{(1+\beta)p_\star+\delta}\bigg)\\
&\hphantom{=\bigg\{}\times
\bigg(\bigg(\frac{\beta p_\star}{(1+\beta)p_\star+\delta}\bigg)^\frac{p_\star+\delta}{(1+\beta)p_\star+\delta}+\bigg(\frac{\beta p_\star}{(1+\beta)p_\star+\delta}\bigg)^{-\frac{\beta p_\star}{(1+\beta)p_\star+\delta}}\bigg)\bigg\}^\frac1\beta\\
&\hphantom{=}\times
\varepsilon^{-\frac2\beta}
\bigg(\sum_{\ell=0}^Lh_\ell^{-\frac{(2\beta-(1+\theta))p_\star+(2\beta-(1+\theta)\delta)}{2(p_\star+1)((1+\beta)p_\star+\delta)}p_\star}\bigg)^\frac1\beta.
\end{aligned}
\end{equation*}
The optimal iteration amounts $N_\ell$, $\ell\in[\![0,L]\!]$ are then obtained by injecting the previous value of the Lagrangian multiplier into \eqref{eq:N:l:1}:
\begin{equation}
\label{eq:N:l:continuous}
\boxed{
N_\ell
=\widehat{C}_1^{\beta,\delta}\varepsilon^{-\frac2\beta}\bigg(\sum_{\ell'=0}^Lh_{\ell'}^{-\frac{(2\beta-(1+\theta))p_\star+(2\beta-(1+\theta)\delta)}{2(p_\star+1)((1+\beta)p_\star+\delta)}p_\star}\bigg)^\frac1\beta h_\ell^{\frac{(3+\theta)p_\star+2}{2(p_\star+1)((1+\beta)p_\star+\delta)}p_\star},
}
\end{equation}
where
\begin{equation*}
\begin{aligned}
\widehat{C}_1^{\beta,\delta}
&=C^\frac{p_\star}{(1+\beta)p_\star+\delta}(\bar{C}\gamma_1)^\frac{(p_\star+\delta)p_\star}{(1+\beta)p_\star+\delta}\\
&\hphantom{=}\times
\bigg(\bigg(\frac{\beta p_\star}{p_\star+\delta}\bigg)^\frac{p_\star+\delta}{(1+\beta)p_\star+\delta}+\bigg(\frac{\beta p_\star}{p_\star+\delta}\bigg)^{-\frac{\beta p_\star}{(1+\beta)p_\star+\delta}}\bigg)^\frac1\beta
\bigg(\frac{p_\star+\delta}{(1+\beta)p_\star+\delta}\bigg)^\frac1\beta.
\end{aligned}
\end{equation*}
It is easy to see that $N_\ell>0$, $\ell\in[\![0,L]\!]$. Since we expect such values to be in $\mathbb{N}^*$, it suffices to take the ceiling thereof.
\\

\noindent
\emph{Step~2. Complexity.}
\newline
Denote the rational fractions
\begin{align}
A_{p_\star}
&:=\frac{(3+\theta)p_\star+2}
{2(p_\star+1)((1+\beta)p_\star+\delta)}
p_\star,
\label{Ap}\\
B_{p_\star}
&:=\frac{(2\beta-(1+\theta))p_\star+(2\beta-(1+\theta)\delta)}
{2(p_\star+1)((1+\beta)p_\star+\delta)}
p_\star.
\label{Bp}
\end{align}
Observe that
\begin{equation*}
\bigg(1+\frac\delta{p_\star}\bigg)A_{p_\star}-1=-B_{p_\star}.
\end{equation*}
Thus, recalling that the $N_\ell$, $\ell\in[\![0,L]\!]$, are set as the ceilings of the value indicated in \eqref{eq:N:l:continuous}, by considering Proposition~\ref{prp:amlsa:complexity},
\begin{equation}
\label{eq:Cost:1}
\Cost_\text{adMLSA}
\leq\bar{C}^{\beta,\delta}_1
\bigg(\sum_{\ell=0}^Lh_\ell^{-1}
+\varepsilon^{-\frac{2(p_\star+\delta)}{\beta p_\star}}
\bigg(\sum_{\ell=0}^Lh_\ell^{-B_{p_\star}}\bigg)^\frac{(1+\beta)p_\star+\delta}{\beta p_\star}\bigg).
\end{equation}
where
\begin{equation*}
\bar{C}^{\beta,\delta}_1
=2^\frac\delta{p_\star}\bar{C}
\big(1\vee(\widehat{C}^{\beta,\delta}_1)^\frac{p_\star+\delta}{p_\star}\big).
\end{equation*}
In view of \eqref{selection:number:level:adMLSA},
\begin{equation}
\label{eq:sum:hl:-1}
\sum_{\ell=0}^Lh_\ell^{-1}
\leq\frac{h_L^{-1}}{1-M^{-1}}
\leq\frac{(Ch_0)^\frac1{1+\theta}h_1^{-1}}{1-M^{-1}}
\varepsilon^{-\frac1{1+\theta}}.
\end{equation}
Moreover,
\begin{equation*}
\sum_{\ell=0}^Lh_\ell^{-B_{p_\star}}
\leq\begin{cases}
\frac{h_0^{-B_{p_\star}}}{1-M^{B_{p_\star}}}
&\text{if $B_{p_\star}<0$,}\\
L+1
&\text{if $B_{p_\star}=0$,}\\
\frac{h_L^{-B_{p_\star}}}{1-M^{-B_{p_\star}}}
&\text{if $B_{p_\star}>0$.}
\end{cases}
\end{equation*}
Take $\varepsilon$ small enough so that $|\ln{\varepsilon}|\geq1$.
Note that, following \eqref{selection:number:level:adMLSA},
\begin{equation}
\label{eq:L+1<}
L+1\leq\bigg(2+\frac{|\ln{(Ch_0)}|+1}{(1+\theta)\ln{M}}\bigg)|\ln{\varepsilon}|.
\end{equation}
Besides, if $B_{p_\star}>0$, then, it ensues from \eqref{selection:number:level:adMLSA} that
\begin{equation*}
h_L^{-B_{p_\star}}\leq(Ch_0)^\frac{B_{p_\star}}{1+\theta}h_1^{-B_{p_\star}}\varepsilon^{-\frac{B_{p_\star}}{1+\theta}}.
\end{equation*}
Thus
\begin{equation}
\label{eq:sum:hl:-Bp<eps}
\sum_{\ell=0}^Lh_\ell^{-B_{p_\star}}
\leq\widetilde{C}^{\beta,\delta,\theta}_1
\times\begin{cases}
1
&\text{if $B_{p_\star}<0$,}\\
|\ln{\varepsilon}|
&\text{if $B_{p_\star}=0$,}\\
\varepsilon^{-\frac{B_{p_\star}}{1+\theta}}
&\text{if $B_{p_\star}>0$,}
\end{cases}
\end{equation}
where
\begin{equation*}
\widetilde{C}^{\beta,\delta,\theta}_1
=\begin{cases}
\frac{h_0^{-B_{p_\star}}}{1-M^{B_{p_\star}}}
&\text{if $B_{p_\star}<0$,}\\
2+\frac{|\ln{(C'h_0)}|+1}{(1+\theta)\ln{M}}
&\text{if $B_{p_\star}=0$,}\\
\frac{(Ch_0)^\frac{B_{p_\star}}{1+\theta}h_1^{-B_{p_\star}}}{1-M^{-B_{p_\star}}}
&\text{if $B_{p_\star}>0$.}
\end{cases}
\end{equation*}
Therefore, via \eqref{eq:Cost:1}, \eqref{eq:sum:hl:-1} and \eqref{eq:sum:hl:-Bp<eps},
\begin{equation*}
\Cost_\text{adMLSA}
\leq\check{C}_1
\varepsilon^{-\frac1{1+\theta}}
+\check{C}_1
\begin{cases}
\varepsilon^{-\frac{2(p_\star+\delta)}{\beta p_\star}}
&\text{if $B_{p_\star}<0$,}\\
\varepsilon^{-\frac{2(p_\star+\delta)}{\beta p_\star}}
|\ln{\varepsilon}|^\frac{(1+\beta)p_\star+\delta}{\beta p_\star}
&\text{if $B_{p_\star}=0$,}\\
\varepsilon^{-\frac{2(p_\star+\delta)}{\beta p_\star}-\frac{(2\beta-(1+\theta))p_\star+2\beta-(1+\theta)\delta}{(1+\theta)\beta(p_\star+1)}}
&\text{if $B_{p_\star}>0$,}
\end{cases}
\end{equation*}
where
\begin{equation}
\label{eq:1:C}
\check{C}_1
=\sup_{\beta,\delta,\theta}\check{C}^{\beta,\delta,\theta}_1,
\quad\text{with}\quad
\check{C}^{\beta,\delta,\theta}_1
=\bar{C}^{\beta,\delta}_1
\bigg(\frac{(Ch_0)^\frac1{1+\theta}h_1^{-1}}{1-M^{-1}}
\vee(\widetilde{C}^{\beta,\delta,\theta}_1)^\frac{p_\star+\delta}{p_\star}\bigg).
\end{equation}
Hence
\begin{equation}
\label{eq:cost=O(eps)}
\boxed{
\Cost_\text{adMLSA}
\leq2\check{C}_1
\begin{cases}
\varepsilon^{-\frac{2(p_\star+\delta)}{\beta p_\star}}
&\text{if $B_{p_\star}<0$,}\\
\varepsilon^{-\frac{2(p_\star+\delta)}{\beta p_\star}}
|\ln{\varepsilon}|^\frac{(1+\beta)p_\star+\delta}{\beta p_\star}
&\text{if $B_{p_\star}=0$,}\\
\varepsilon^{-\frac{2(p_\star+\delta)}{\beta p_\star}-\frac{(2\beta-(1+\theta))p_\star+2\beta-(1+\theta)\delta}{(1+\theta)\beta(p_\star+1)}}
&\text{if $B_{p_\star}>0$.}
\end{cases}
}
\end{equation}

\noindent
\emph{Step~3. Best complexity.}
\newline
\emph{Step~3.1. Case where $B_{p_\star}<0$.}
\newline
If $B_{p_\star}<0$, then, by the definition \eqref{Bp},
\begin{equation*}
\big(2\beta-(1+\theta)\big)p_\star
+\big(2\beta-(1+\theta)\delta\big)
<0,
\quad\text{i.e.}\quad
\theta
>\frac{(2\beta-1)p_\star+2\beta-\delta}{p_\star+\delta}.
\end{equation*}
Given that $\theta\in(0,\frac{p_\star-2}{p_\star+2}]$, then, necessarily,
\begin{equation}
\label{eq:beta=f(delta):0}
\frac{(2\beta-1)p_\star+2\beta-\delta}{p_\star+\delta}
<\frac{p_\star-2}{p_\star+2},
\quad\text{i.e.}\quad
\beta<\frac{p_\star(p_\star+\delta)}{(p_\star+1)(p_\star+2)}<1.
\end{equation}
Reconsidering \eqref{eq:cost=O(eps)}, by using \eqref{eq:beta=f(delta):0}, and letting $\theta=\frac{p_\star-2}{p_\star+2}$ and $\beta\uparrow\frac{p_\star(p_\star+\delta)}{(p_\star+1)(p_\star+2)}$,
\begin{equation*}
\inf_{\beta,\theta}\Cost_{\text{adMLSA}}
\leq2\check{C}_1\varepsilon^{-2-\frac6{p_\star}-\frac4{p_\star^2}-\iota},
\end{equation*}
for any $\iota>0$.
This bound is independent of the choice of $\delta\in(0,1]$.
\\

\noindent
\emph{Step~3.2. Case where $B_{p_\star}=0$.}
\newline
If $B_{p_\star}=0$, then, by the definition \eqref{Bp},
\begin{equation*}
\big(2\beta-(1+\theta)\big)p_\star
+\big(2\beta-(1+\theta)\delta\big)
=0,
\quad\text{i.e.}\quad
\theta
=\frac{(2\beta-1)p_\star+2\beta-\delta}{p_\star+\delta}.
\end{equation*}
In particular,
\begin{equation}
\label{eq:theta=f(beta,delta)}
\frac1{1+\theta}=\frac{p_\star+\delta}{2\beta(p_\star+1)}.
\end{equation}
Since $\theta\in(0,\frac{p_\star-2}{p_\star+2}]$, then, necessarily,
\begin{equation}
\label{eq:beta=f(delta)}
\frac{(2\beta-1)p_\star+2\beta-\delta}{p_\star+\delta}
\leq\frac{p_\star-2}{p_\star+2},
\quad\text{i.e.}\quad
\beta\leq\frac{p_\star(p_\star+\delta)}{(p_\star+1)(p_\star+2)}<1.
\end{equation}
Coming back to \eqref{eq:cost=O(eps)}, via \eqref{eq:theta=f(beta,delta)} and \eqref{eq:beta=f(delta)}, by letting $\beta=\frac{p_\star(p_\star+\delta)}{(p_\star+1)(p_\star+2)}$,
\begin{equation*}
\inf_\beta\Cost_{\text{adMLSA}}
\leq2\check{C}_1\varepsilon^{-2-\frac6{p_\star}-\frac4{p_\star^2}}|\ln{\varepsilon}|^\frac{2p_\star^3+(3+2\delta)p_\star^2+(2+3\delta)p_\star+2\delta}{\delta p_\star^2+\delta^2p_\star}.
\end{equation*}
Note in particular that, for such choice of $\beta$,
\begin{equation*}
\theta=\frac{p_\star-2}{p_\star+2}.
\end{equation*}
This upper bound is minimal for $\delta=1$, in which case,
\begin{equation*}
\inf_{\beta,\delta}\Cost_{\text{adMLSA}}
\leq2\check{C}_1
\varepsilon^{-2-\frac6{p_\star}-\frac4{p_\star^2}}
|\ln{\varepsilon}|^\frac{2p_\star^2+3p_\star+2}{p_\star}.
\end{equation*}

\noindent
\emph{Step~3.3. Case where $B_{p_\star}>0$.}
\newline
If $B_{p_\star}>0$, then, by the definition \eqref{Bp},
\begin{equation*}
\big(2\beta-(1+\theta)\big)p_\star
+\big(2\beta-(1+\theta)\delta\big)
>0,
\quad\text{i.e.}\quad
\theta
<\frac{(2\beta-1)p_\star+2\beta-\delta}{p_\star+\delta}.
\end{equation*}
Note that, following \eqref{Bp} and \eqref{eq:cost=O(eps)},
\begin{equation}
\label{eq:cost=O(eps):3}
\Cost_{\text{adMLSA}}
\leq2\check{C}_1
\varepsilon^{-\frac{2(p_\star+\delta)}{\beta p_\star}-\frac{(2p_\star+\delta)((2\beta-(1+\theta))p_\star+(2\beta-(1+\theta)\delta))}{2(1+\theta)(p_\star+1)((1+\beta)p_\star+\delta)}}.
\end{equation}

First, assume
\begin{equation*}
\frac{p_\star-2}{p_\star+2}
<\frac{(2\beta-1)p_\star+2\beta-\delta}{p_\star+\delta},
\quad\text{i.e.}\quad
\beta>\frac{p_\star(p_\star+\delta)}{(p_\star+1)(p_\star+2)}.
\end{equation*}
Per \eqref{eq:cost=O(eps):3}, taking $\theta=\frac{p_\star-2}{p_\star+2}$,
\begin{equation*}
\inf_\theta\Cost_{\text{adMLSA}}
\leq2\check{C}_1
\varepsilon^{-\frac{2(p_\star+\delta)}{\beta p_\star}-\frac{(2p_\star+\delta)((p_\star+1)(p_\star+2)\beta-p_\star(p_\star+\delta))}{2p_\star(p_\star+1)((1+\beta)p_\star+\delta)}}.
\end{equation*}
Thus, taking $\beta=1$,
\begin{equation*}
\inf_{\beta,\theta}
\Cost_{\text{adMLSA}}
\leq2\check{C}_1
\varepsilon^{-\frac{2(p_\star+\delta)}{p_\star}
-\frac{(3-\delta)p_\star+2}{2p_\star(p_\star+1)}}.
\end{equation*}
Hence, letting $\delta\downarrow0$,
\begin{equation*}
\inf_{\beta,\delta,\theta}
\Cost_{\text{adMLSA}}
\leq2\check{C}_1
\varepsilon^{-2
-\frac{3p_\star+2}{2p_\star(p_\star+1)}-\iota},
\end{equation*}
for any $\iota>0$.

Assume now that
\begin{equation*}
\frac{(2\beta-1)p_\star+2\beta-\delta}{p_\star+\delta}
\leq\frac{p_\star-2}{p_\star+2},
\quad\text{i.e.}\quad
\beta\leq\frac{p_\star(p_\star+\delta)}{(p_\star+1)(p_\star+2)}
<1.
\end{equation*}
Via \eqref{eq:cost=O(eps):3}, letting $\theta\uparrow\frac{(2\beta-1)p_\star+2\beta-\delta}{p_\star+\delta}$,
\begin{equation*}
\inf_\theta\Cost_{\text{adMLSA}}
\leq2\check{C}_1
\varepsilon^{-\frac{2(p_\star+\delta)}{\beta p_\star}-\iota},
\end{equation*}
for any $\iota>0$.
Hence, setting $\beta=\frac{p_\star(p_\star+\delta)}{(p_\star+1)(p_\star+2)}$,
\begin{equation*}
\inf_{\beta,\theta}\Cost_{\text{adMLSA}}
\leq2\check{C}_1\varepsilon^{-2-\frac6{p_\star}-\frac4{p_\star^2}-\iota},
\end{equation*}
for any $\iota>0$,
independently from the choice of $\delta\in(0,1]$.
\\

\noindent
\emph{Step~3.4. Conclusion.}
\newline
All in all, the best complexity is attained when $\beta=1$, $\delta\downarrow0$ and $\theta=\frac{p_\star-2}{p_\star+2}$.
The ensuing complexity satisfies
\begin{equation*}
\boxed{
\inf_{\beta,\delta,\theta}\Cost_{\text{adMLSA}}
\leq2\check{C}_1
\varepsilon^{-2
-\frac{3p_\star+2}{2p_\star(p_\star+1)}-\iota}.
}
\end{equation*}

\noindent
(\ref{thm:amlsa:complexity:indicator:func-ib:N})\
\emph{Step~0. Preliminaries.}
\newline
Take $\varepsilon$ sufficiently small to ensure that $|\ln{\varepsilon}|\geq1$.
Note that, consequently from \eqref{selection:number:level:adMLSA},
\begin{equation}
\label{eq:L<ln:eps}
L\leq\widetilde{C}_1|\ln{\varepsilon}|,
\quad\text{where}\quad
\widetilde{C}_1=1+\frac{1+|\ln{(Ch_0)}|}{(1+\theta)\ln{M}}.
\end{equation}
We postulate that, for some $\widetilde{C}_2>0$,
\begin{equation}
\label{eq:adhoc}
\ln{N_\ell}\leq\widetilde{C}_2|\ln{\varepsilon}|,
\quad \ell\in[\![0,L]\!],
\end{equation}
a claim that we check a posteriori by  explicitating $\widetilde{C}_2$.
Then, one has
\begin{equation}
\label{eq:cost:proxy}
C\sum_{\ell=0}^L\frac{N_\ell}{h_\ell}(L^\frac12+\ln^\frac12{N_\ell})
\leq C\widetilde{C}|\ln{\varepsilon}|^\frac12\sum_{\ell=0}^L\frac{N_\ell}{h_\ell},
\quad\text{where}\quad
\widetilde{C}=(\widetilde{C}_1)^\frac12+(\widetilde{C}_2)^\frac12.
\end{equation}

\noindent
\emph{Step~1. Iteration amounts.}
\newline
We now solve the proxy problem
\begin{equation*}
\begin{cases}
\min_{N_0,\dots,N_L\in\mathbb{R}}
&C\widetilde{C}|\ln{\varepsilon}|^\frac12\sum_{\ell=0}^LN_\ell h_\ell^{-1},\\
\text{s.t.}
&\bar{C}\gamma_1\sum_{\ell=0}^LN_\ell^{-\beta}h_\ell^\frac{1+\theta}2|\ln{h_\ell}|^\frac{1+\theta}2=\varepsilon^2.
\end{cases}
\end{equation*}
The associated Lagrangian is
\begin{equation*}
\mathfrak{L}
=C\widetilde{C}|\ln{\varepsilon}|^\frac12\sum_{\ell=0}^L\frac{N_\ell}{h_\ell}
+\bigg(\bar{C}\gamma_1\sum_{\ell=0}^L\frac{h_\ell^\frac{1+\theta}2|\ln{h_\ell}|^\frac{1+\theta}2}{N_\ell^\beta}
-\varepsilon^2\bigg),
\end{equation*}
where $\lambda\in\mathbb{R}$ is the Lagrangian multiplier.
The optimal $N_\ell$, $\ell\in[\![0,L]\!]$ satisfy
\begin{equation}
\label{eq:N:l:2}
N_\ell
=\bigg(\frac{\bar{C}\gamma_1\beta}{C\widetilde{C}}\bigg)^\frac1{1+\beta}
\lambda^\frac1{1+\beta}
|\ln{\varepsilon}|^{-\frac1{2(1+\beta)}}
h_\ell^\frac{3+\theta}{2(1+\beta)}
|\ln{h_\ell}|^\frac{1+\theta}{2(1+\beta)}.
\end{equation}
Coming back to the Lagrangian, using the previous value for the $N_\ell$, $\ell\in[\![0,L]\!]$, one has
\begin{equation*}
\begin{aligned}
\mathfrak{L}
&=(C\widetilde{C})^\frac\beta{1+\beta}
(\bar{C}\gamma_1)^\frac1{1+\beta}
(\beta^\frac1{\beta+1}+\beta^{-\frac\beta{1+\beta}})\\
&\hphantom{=}\times
\lambda^\frac1{1+\beta}
|\ln{\varepsilon}|^\frac\beta{2(1+\beta)}
\sum_{\ell=0}^Lh_\ell^{-\frac{2\beta-(1+\theta)}{2(1+\beta)}}
|\ln{h_\ell}|^\frac{1+\theta}{2(1+\beta)}
-\lambda\varepsilon^2.
\end{aligned}
\end{equation*}
The optimal Lagrangian multiplier verifies
\begin{equation*}
\begin{aligned}
\lambda^\frac1{1+\beta}
&=\bigg(\frac{(C\widetilde{C})^\frac\beta{1+\beta}
(\bar{C}\gamma_1)^\frac1{1+\beta}
(\beta^\frac1{\beta+1}+\beta^{-\frac\beta{1+\beta}})}{1+\beta}\bigg)^\frac1\beta\\
&\hphantom{=}\times
\varepsilon^{-\frac2\beta}|\ln{\varepsilon}|^\frac1{2(1+\beta)}
\bigg(\sum_{\ell=0}^Lh_\ell^{-\frac{2\beta-(1+\theta)}{2(1+\beta)}}|\ln{h_\ell}|^\frac{1+\theta}{2(1+\beta)}\bigg)^\frac1\beta.
\end{aligned}
\end{equation*}
Injecting this value back into \eqref{eq:N:l:2} yields
\begin{equation}
\label{eq:N:l:2:bis}
\boxed{N_\ell
=(\bar{C}\gamma_1)^\frac1\beta
\varepsilon^{-\frac2\beta}
\bigg(\sum_{\ell'=0}^Lh_{\ell'}^{-\frac{2\beta-(1+\theta)}{2(1+\beta)}}|\ln{h_{\ell'}}|^\frac{1+\theta}{2(1+\beta)}\bigg)^\frac1\beta
h_\ell^\frac{3+\theta}{2(1+\beta)}
|\ln{h_\ell}|^\frac{1+\theta}{2(1+\beta)}.}
\end{equation}

Recalling that $h_0=\frac1K$, note that
\begin{equation}
\label{eq:ln:h<L}
|\ln{h_\ell}|\leq(\ln{K}+\ln{M})L.
\end{equation}
Besides,
\begin{equation}
\label{eq:sum:hl:-}
\sum_{\ell=0}^L
h_\ell^{-\frac{2\beta-(1+\theta)}{2(1+\beta)}}
\leq\begin{cases}
\frac{h_0^{-\frac{2\beta-(1+\theta)}{2(1+\beta)}}}{1-M^\frac{2\beta-(1+\theta)}{2(1+\beta)}}
&\text{if $2\beta<1+\theta$,}\\
L+1
&\text{if $2\beta=1+\theta$,}\\
\frac{h_L^{-\frac{2\beta-(1+\theta)}{2(1+\beta)}}}{1-M^{-\frac{2\beta-(1+\theta)}{2(1+\beta)}}}
&\text{if $2\beta>1+\theta$.}
\end{cases}
\end{equation}
Hence
\begin{equation}
\label{eq:ln:sum<}
\begin{aligned}
\ln\bigg(&\sum_{\ell=0}^Lh_\ell^{-\frac{2\beta-(1+\theta)}{2(1+\beta)}}\bigg)\\
&\leq\begin{cases}
-\ln{(1-M^\frac{2\beta-(1+\theta)}{2(1+\beta)})}
&\text{if $2\beta<1+\theta$,}\\
\ln{(L+1)}
&\text{if $2\beta=1+\theta$,}\\
\frac{2\beta-(1+\theta)}{2(1+\beta)}(\ln{K}+L\ln{M})-\ln{(1-M^{-\frac{2\beta-(1+\theta)}{2(1+\beta)}})}
&\text{if $2\beta>1+\theta$,}
\end{cases}\\
&\leq C'_2L,
\end{aligned}
\end{equation}
where
\begin{equation}
\label{eq:C'2}
C'_2=
\bigg(\bigg(\frac{2\beta-(1+\theta)}{2(1+\beta)}\bigg)^+\ln{(MK)}-\ln{(1-M^{-|\frac{2\beta-(1+\theta)}{2(1+\beta)}|})}
\mathds{1}_{\{2\beta\neq1+\theta\}}\bigg)
\vee1.
\end{equation}
Hence, considering \eqref{eq:L<ln:eps}, the inequality \eqref{eq:adhoc} holds with
\begin{equation*}
\begin{aligned}
\widetilde{C}_2
=\frac{2+|\ln{(\bar{C}\gamma_1)}|}\beta
+\widetilde{C}_1
\bigg(\frac{3+\theta}{2(1+\beta)}\ln{(MK)}
+\frac{1+\theta}{2\beta}\Big(1+\big|\ln{\big(\ln{(MK)}\big)}\big|\Big)
+\frac{C'_2}\beta\bigg).
\end{aligned}
\end{equation*}

\noindent
\emph{Step~2. Complexity.}
\newline
We set each $N_\ell$, $\ell\in[\![0,L]\!]$, as the ceiling of the value in \eqref{eq:N:l:2:bis} in order to retrieve iteration amounts in $\mathbb{N}^*$.
Via Proposition~\ref{prp:amlsa:complexity}, \eqref{eq:cost:proxy} and \eqref{eq:ln:h<L},
\begin{equation}
\label{eq:Cost:2}
\begin{aligned}
\Cost_\text{adMLSA}
&\leq C\widetilde{C}|\ln{\varepsilon}|^\frac12
\sum_{\ell=0}^L\frac{N_\ell}{h_\ell}\\
&\leq \widehat{C}^{\beta,\theta}_2
\bigg(|\ln{\varepsilon}|^\frac12
\sum_{\ell=0}^Lh_\ell^{-1}
+\varepsilon^{-\frac2\beta}
|\ln{\varepsilon}|^\frac{1+\theta+\beta}{2\beta}
\bigg(\sum_{\ell=0}^Lh_\ell^{-\frac{2\beta-(1+\theta)}{2(1+\beta)}}\bigg)^\frac{1+\beta}\beta\bigg),
\end{aligned}
\end{equation}
where
\begin{equation*}
\widehat{C}^{\beta,\theta}_2
=C\widetilde{C}\Big(1\vee(\bar{C}\gamma_1)^\frac1\beta
\big(\widetilde{C}_1\ln{(KM)}\big)^\frac{1+\theta}{2\beta}\Big).
\end{equation*}
If $2\beta>1+\theta$, it follows from \eqref{selection:number:level:adMLSA} that
\begin{equation*}
h_L^{-\frac{2\beta-(1+\theta)}{2(1+\beta)}}
\leq(Ch_0)^\frac{2\beta-(1+\theta)}{2(1+\theta)(1+\beta)}h_1^{-\frac{2\beta-(1+\theta)}{2(1+\theta)(1+\beta)}}
\varepsilon^{-\frac{2\beta-(1+\theta)}{2(1+\theta)(1+\beta)}}.
\end{equation*}
Hence, via \eqref{eq:sum:hl:-} and \eqref{eq:L+1<},
\begin{equation}
\label{eq:sum:hl:-2b-(1+o)}
\sum_{\ell=0}^L
h_\ell^{-\frac{2\beta-(1+\theta)}{2(1+\beta)}}
\leq\widetilde{C}^{\beta,\theta}_2
\times\begin{cases}
1
&\text{if $2\beta<1+\theta$,}\\
|\ln{\varepsilon}|
&\text{if $2\beta=1+\theta$,}\\
\varepsilon^{-\frac{2\beta-(1+\theta)}{2(1+\theta)(1+\beta)}}
&\text{if $2\beta>1+\theta$,}
\end{cases}
\end{equation}
where
\begin{equation}
\label{eq:C~2}
\widetilde{C}^{\beta,\theta}_2
=\begin{cases}
(1-M^\frac{2\beta-(1+\theta)}{2(1+\beta)})^{-1}
h_0^{-\frac{2\beta-(1+\theta)}{2(1+\beta)}}
&\text{if $2\beta<1+\theta$,}\\
2+\frac{|\ln{(Ch_0)}|+1}{(1+\theta)\ln{M}}
&\text{if $2\beta=1+\theta$,}\\
(1-M^{-\frac{2\beta-(1+\theta)}{2(1+\beta)}})^{-1}
(Ch_0)^\frac{2\beta-(1+\theta)}{2(1+\theta)(1+\beta)}h_1^{-\frac{2\beta-(1+\theta)}{2(1+\theta)(1+\beta)}}
&\text{if $2\beta>1+\theta$.}
\end{cases}
\end{equation}
Thus, by gathering the results of \eqref{eq:sum:hl:-1} and \eqref{eq:sum:hl:-2b-(1+o)} and coming back to \eqref{eq:Cost:2},
\begin{equation*}
\Cost_\text{adMLSA}
\leq\check{C}_2
\varepsilon^{-\frac1{1+\theta}}
|\ln{\varepsilon}|^\frac12
+\check{C}_2
\begin{cases}
\varepsilon^{-\frac2\beta}
|\ln{\varepsilon}|^\frac{1+\theta+\beta}{2\beta}
&\text{if $2\beta<1+\theta$,}\\
\varepsilon^{-\frac2\beta}
|\ln{\varepsilon}|^\frac{3(1+\beta)+\theta}{2\beta}
&\text{if $2\beta=1+\theta$,}\\
\varepsilon^{-\frac{3(1+\theta)+2\beta}{2(1+\theta)\beta}}
|\ln{\varepsilon}|^\frac{1+\theta+\beta}{2\beta}
&\text{if $2\beta>1+\theta$,}
\end{cases}
\end{equation*}
where
\begin{equation}
\label{eq:2:C}
\check{C}_2
=\sup_{\beta,\theta}\check{C}^{\beta,\theta}_2,
\quad\text{with}\quad
\check{C}^{\beta,\theta}_2
=\widehat{C}^{\beta,\theta}_2\bigg(\frac{(Ch_0)^\frac1{1+\theta}h_1^{-1}}{1-M^{-1}}
\vee(\widetilde{C}^{\beta,\theta}_2)^\frac{1+\beta}\beta\bigg).
\end{equation}
Thus
\begin{equation}
\label{eq:cost=O(eps):bis}
\boxed{
\Cost_\text{adMLSA}
\leq2\check{C}_2
\begin{cases}
\varepsilon^{-\frac2\beta}
|\ln{\varepsilon}|^\frac{1+\theta+\beta}{2\beta}
&\text{if $2\beta<1+\theta$,}\\
\varepsilon^{-\frac2\beta}
|\ln{\varepsilon}|^\frac{3(1+\beta)+\theta}{2\beta}
&\text{if $2\beta=1+\theta$,}\\
\varepsilon^{-\frac{3(1+\theta)+2\beta}{2(1+\theta)\beta}}
|\ln{\varepsilon}|^\frac{1+\theta+\beta}{2\beta}
&\text{if $2\beta>1+\theta$.}
\end{cases}
}
\end{equation}

\noindent
\emph{Step~3. Best complexity.}
\newline
If $2\beta<1+\theta$, then, the complexity bound is minimal if $\theta=1$ and $\beta\uparrow1$, in which case
\begin{equation*}
\inf_{\beta,\theta}\Cost_{\text{adMLSA}}
\leq\check{C}_2
\varepsilon^{-2-\iota}
|\ln{\varepsilon}|^{\frac32+\iota},
\end{equation*}
for any $\iota>0$.

If $2\beta=1+\theta$, then, the complexity bound is minimal if $\beta=\theta=1$, in which case
\begin{equation*}
\inf_{\beta,\theta}\Cost_{\text{adMLSA}}
\leq2\check{C}_2
\varepsilon^{-2}
|\ln{\varepsilon}|^\frac72.
\end{equation*}

If $2\beta>1+\theta$, then, the complexity bound is minimal if $\beta=1$ and $\theta\uparrow1$, in which case
\begin{equation*}
\inf_{\beta,\theta}\Cost_{\text{adMLSA}}
\leq\check{C}_2
\varepsilon^{-2-\iota}
|\ln{\varepsilon}|^{\frac32+\iota},
\end{equation*}
for any $\iota>0$.

Overall, the best complexity is attained for $\beta=\theta=1$, and it satisfies
\begin{equation*}
\boxed{
\inf_{\beta,\theta}\Cost_{\text{adMLSA}}
\leq2\check{C}_2
\varepsilon^{-2}
|\ln{\varepsilon}|^\frac72.
}
\end{equation*}

\noindent
(\ref{thm:amlsa:complexity:indicator:func-ii:N})\
\emph{Step~0. Preliminaries.}
\newline
We proceed similarly to the previous framework by assuming that $\varepsilon$ is small enough to ensure that $|\ln{\varepsilon}|\geq1$, and that \eqref{eq:adhoc} holds.
\\

\noindent
\emph{Step~1. Iteration amounts.}
\newline
Thus, taking into account \eqref{eq:L<ln:eps} and \eqref{eq:cost:proxy}, we look into the proxy problem
\begin{equation*}
\begin{cases}
\min_{N_0,\dots,N_L\in\mathbb{R}}
&C\widetilde{C}|\ln{\varepsilon}|^\frac12\sum_{\ell=0}^LN_\ell h_\ell^{-1},\\
\text{s.t.}
&\bar{C}\gamma_1\sum_{\ell=0}^LN_\ell^{-\beta}h_\ell^\frac{1+\theta}2=\varepsilon^2.
\end{cases}
\end{equation*}

The Lagrangian of the previous problem is
\begin{equation*}
\mathfrak{L}
=C\widetilde{C}|\ln{\varepsilon}|^\frac12\sum_{\ell=0}^L\frac{N_\ell}{h_\ell}
+\lambda\bigg(\bar{C}\gamma_1\sum_{\ell=0}^L\frac{h_\ell^\frac{1+\theta}2}{N_\ell^\beta}
-\varepsilon^2\bigg),
\end{equation*}
with $\lambda\in\mathbb{R}$ being the Lagrangian multiplier.
The optimal $N_\ell$, $\ell\in[\![0,L]\!]$, are given by
\begin{equation*}
N_\ell
=\bigg(\frac{\bar{C}\gamma_1\beta}{C\widetilde{C}}\bigg)^\frac1{1+\beta}
\lambda^\frac1{1+\beta}
|\ln{\varepsilon}|^{-\frac1{2(1+\beta)}}
h_\ell^\frac{3+\theta}{2(1+\beta)}.
\end{equation*}
Injecting these values back into the Lagrangian yields
\begin{equation*}
\mathfrak{L}
=(C\widetilde{C})^\frac\beta{1+\beta}
(\bar{C}\gamma_1)^\frac1{1+\beta}
(\beta^\frac1{\beta+1}+\beta^{-\frac\beta{1+\beta}})
\lambda^\frac1{1+\beta}
|\ln{\varepsilon}|^\frac\beta{2(1+\beta)}
\sum_{\ell=0}^Lh_\ell^{-\frac{2\beta-(1+\theta)}{2(1+\beta)}}
-\lambda\varepsilon^2.
\end{equation*}
The optimal Lagrangian multiplier $\lambda$ satisfies
\begin{equation*}
\lambda^\frac1{1+\beta}
=\bigg(\frac{(C\widetilde{C})^\frac\beta{1+\beta}
(\bar{C}\gamma_1)^\frac1{1+\beta}
(\beta^\frac1{\beta+1}+\beta^{-\frac\beta{1+\beta}})}{1+\beta}\bigg)^\frac1\beta
\varepsilon^{-\frac2\beta}|\ln{\varepsilon}|^\frac1{2(1+\beta)}
\bigg(\sum_{\ell=0}^Lh_\ell^{-\frac{2\beta-(1+\theta)}{2(1+\beta)}}\bigg)^\frac1\beta.
\end{equation*}
Thus
\begin{equation*}
\boxed{
N_\ell
=(\bar{C}\gamma_1)^\frac1\beta
\varepsilon^{-\frac2\beta}
\bigg(\sum_{\ell'=0}^Lh_{\ell'}^{-\frac{2\beta-(1+\theta)}{2(1+\beta)}}
\bigg)^\frac1\beta
h_\ell^\frac{3+\theta}{2(1+\beta)}.
}
\end{equation*}

Using \eqref{eq:ln:sum<} and \eqref{eq:L<ln:eps}, we obtain that \eqref{eq:adhoc} holds, with
\begin{equation*}
\widetilde{C}_2
=\frac{2+|\ln{(\bar{C}\gamma_1)}|}\beta
+\widetilde{C}_1
\bigg(\frac{3+\theta}{2(1+\beta)}\ln{(MK)}
+\frac{C'_2}\beta\bigg),
\end{equation*}
where $C_2'>0$ is the constant defined in \eqref{eq:C'2}.
\\

\noindent
\emph{Step~2. Complexity.}
\newline
We set each $N_\ell$, $\ell\in[\![0,L]\!]$, as the ceiling of the value in \eqref{eq:N:l:2:bis} in order to retrieve iteration amounts in $\mathbb{N}^*$.
Via Proposition~\ref{prp:amlsa:complexity}, \eqref{eq:cost:proxy} and \eqref{eq:ln:h<L},
\begin{equation*}
\begin{aligned}
\Cost_\text{adMLSA}
&\leq C\widetilde{C}|\ln{\varepsilon}|^\frac12
\sum_{\ell=0}^L\frac{N_\ell}{h_\ell}\\
&\leq \widehat{C}^{\beta,\theta}_3
\bigg(|\ln{\varepsilon}|^\frac12
\sum_{\ell=0}^Lh_\ell^{-1}
+\varepsilon^{-\frac2\beta}
|\ln{\varepsilon}|^\frac12
\bigg(\sum_{\ell=0}^Lh_\ell^{-\frac{2\beta-(1+\theta)}{2(1+\beta)}}\bigg)^\frac{1+\beta}\beta\bigg),
\end{aligned}
\end{equation*}
where
\begin{equation*}
\widehat{C}^{\beta,\theta}_3
=C\widetilde{C}\big(1\vee(\bar{C}\gamma_1)^\frac1\beta\big).
\end{equation*}
Using \eqref{eq:sum:hl:-1} and \eqref{eq:sum:hl:-2b-(1+o)},
\begin{equation*}
\Cost_\text{adMLSA}
\leq\check{C}_3
\varepsilon^{-\frac1{1+\theta}}
|\ln{\varepsilon}|^\frac12
+\check{C}_3
\begin{cases}
\varepsilon^{-\frac2\beta}
|\ln{\varepsilon}|^\frac12
&\text{if $2\beta<1+\theta$,}\\
\varepsilon^{-\frac2\beta}
|\ln{\varepsilon}|^\frac{1+2\beta}\beta
&\text{if $2\beta=1+\theta$,}\\
\varepsilon^{-\frac{3(1+\theta)+2\beta}{2(1+\theta)\beta}}
|\ln{\varepsilon}|^\frac12
&\text{if $2\beta>1+\theta$,}
\end{cases}
\end{equation*}
where
\begin{equation}
\label{eq:3:C}
\check{C}_3
=\sup_{\beta,\theta}\check{C}^{\beta,\theta}_3,
\quad\text{with}\quad
\check{C}^{\beta,\theta}_3
=\widehat{C}^{\beta,\theta}_3\bigg(\frac{(Ch_0)^\frac1{1+\theta}h_1^{-1}}{1-M^{-1}}
\vee(\widetilde{C}^{\beta,\theta}_2)^\frac{1+\beta}\beta\bigg),
\end{equation}
with $\widetilde{C}^{\beta,\theta}_2>0$ being defined in \eqref{eq:C~2}.
Hence
\begin{equation}
\label{eq:cost=O(eps):ter}
\boxed{
\Cost_\text{adMLSA}
\leq2\check{C}_3
\begin{cases}
\varepsilon^{-\frac2\beta}
|\ln{\varepsilon}|^\frac12
&\text{if $2\beta<1+\theta$,}\\
\varepsilon^{-\frac2\beta}
|\ln{\varepsilon}|^\frac{1+2\beta}\beta
&\text{if $2\beta=1+\theta$,}\\
\varepsilon^{-\frac{3(1+\theta)+2\beta}{2(1+\theta)\beta}}
|\ln{\varepsilon}|^\frac12
&\text{if $2\beta>1+\theta$.}
\end{cases}
}
\end{equation}

\noindent
\emph{Step~3. Best complexity.}
\newline
Similarly to the previous framework, the best complexity is attained for $\beta=\theta=1$, and it satisfies
\begin{equation*}
\boxed{
\inf_{\beta,\theta}\Cost_{\text{adMLSA}}
\leq2\check{C}_3
\varepsilon^{-2}
|\ln{\varepsilon}|^\frac52.
}
\end{equation*}

\end{proof}

\end{document}